\PassOptionsToPackage{dvipsnames}{xcolor}

\documentclass[acmsmall,screen,authorversion]{acmart}

\newif\ifextended
\extendedtrue 

\makeatletter
\@ifundefined{lhs2tex.lhs2tex.sty.read}%
  {\@namedef{lhs2tex.lhs2tex.sty.read}{}%
   \newcommand\SkipToFmtEnd{}%
   \newcommand\EndFmtInput{}%
   \long\def\SkipToFmtEnd#1\EndFmtInput{}%
  }\SkipToFmtEnd

\newcommand\ReadOnlyOnce[1]{\@ifundefined{#1}{\@namedef{#1}{}}\SkipToFmtEnd}
\usepackage{amstext}
\usepackage{amssymb}
\usepackage{stmaryrd}
\DeclareFontFamily{OT1}{cmtex}{}
\DeclareFontShape{OT1}{cmtex}{m}{n}
  {<5><6><7><8>cmtex8
   <9>cmtex9
   <10><10.95><12><14.4><17.28><20.74><24.88>cmtex10}{}
\DeclareFontShape{OT1}{cmtex}{m}{it}
  {<-> ssub * cmtt/m/it}{}

\DeclareFontShape{OT1}{cmtt}{bx}{n}
  {<5><6><7><8>cmtt8
   <9>cmbtt9
   <10><10.95><12><14.4><17.28><20.74><24.88>cmbtt10}{}
\DeclareFontShape{OT1}{cmtex}{bx}{n}
  {<-> ssub * cmtt/bx/n}{}

\newcommand{\Conid}[1]{\mathit{#1}}
\newcommand{\Varid}[1]{\mathit{#1}}
\newcommand{\anonymous}{\kern0.06em \vbox{\hrule\@width.5em}}

\renewcommand{\leq}{\leqslant}
\renewcommand{\geq}{\geqslant}
\usepackage{polytable}

\@ifundefined{mathindent}%
  {\newdimen\mathindent\mathindent\leftmargini}%
  {}%

\def\resethooks{%
  \global\let\SaveRestoreHook\empty
  \global\let\ColumnHook\empty}
\newcommand*{\savecolumns}[1][default]%
  {\g@addto@macro\SaveRestoreHook{\savecolumns[#1]}}
\newcommand*{\restorecolumns}[1][default]%
  {\g@addto@macro\SaveRestoreHook{\restorecolumns[#1]}}
\newcommand*{\aligncolumn}[2]%
  {\g@addto@macro\ColumnHook{\column{#1}{#2}}}

\resethooks

\newcommand{\onelinecommentchars}{\quad-{}- }
\newcommand{\commentbeginchars}{\enskip\{-}
\newcommand{\commentendchars}{-\}\enskip}

\newcommand{\visiblecomments}{%
  \let\onelinecomment=\onelinecommentchars
  \let\commentbegin=\commentbeginchars
  \let\commentend=\commentendchars}

\newcommand{\invisiblecomments}{%
  \let\onelinecomment=\empty
  \let\commentbegin=\empty
  \let\commentend=\empty}

\visiblecomments

\newlength{\blanklineskip}
\newcommand{\hsindent}[1]{\quad}
\let\hspre\empty
\let\hspost\empty

\EndFmtInput
\makeatother

\newcommand{\stepT}[2]{\xrightarrow{#1}_{#2}}

\newcommand{\Moon}[1]{\llparenthesis #1 \rrparenthesis}

\newcommand{\protectMS}{\ensuremath{\mathbf{protect }\xspace}}

\newcommand{\flows}{\sqsubseteq}

\newcommand{\fresh}[1]{\ensuremath{\textrm{fresh}(#1)}}
\newcommand{\update}[3]{\ensuremath{#1 [ #2 \mapsto #3 ]}}

\newcommand{\srule}[1]{{[}\textsc{#1}{]}}

\newcommand{\while}{\textsc{While}\xspace}

\newcommand{\tool}{\textsc{Blade}\xspace}
\newcommand{\Hacl}{\textsc{HACL*}\xspace}
\newcommand{\mypara}[1]{\smallskip\noindent\emph{\textbf{{#1.}}}}
\newcommand{\eg}{e.g.,\xspace}
\newcommand{\ie}{i.e.,\xspace}
\newcommand{\exNum}[1]{\textsc{Example}}

\newcommand{\theSet}[1]{\{#1\}}

\newcommand{\Transient}{\sourceColorOf{\textbf{T}}}
\newcommand{\Concrete}{\sinkColorOf{\textbf{S}}}

\newcommand{\graphsize}{\footnotesize}
\newcommand{\examplesize}{\small}
\newcommand{\rulesize}{\small}

\newcommand{\sourceColor}{Magenta}
\newcommand{\sinkColor}{BlueGreen}

\newcommand{\sourceColorOf}[1]{{\color{\sourceColor} #1}}
\newcommand{\sinkColorOf}[1]{{\color{\sinkColor} #1}}

\newcommand{\protectedSet}{\mathsf{Prot}}

\definecolor{constColor}{RGB}{227, 248, 250}

\newcommand{\shadedCons}[1]{\colorbox{constColor}{\ensuremath{#1}}} 
\newcommand{\constrColorOf}[1]{\ensuremath{#1}} 
\newcommand{\cnsSep}{\Rightarrow}

\newcommand{\ghstPatch}[1]{\textcolor{ForestGreen}{\textbf{#1}}}
\newcommand{\orangePatch}[1]{\textcolor{BurntOrange}{\textbf{#1}}}

\usepackage{microtype}

\usepackage{bm}
\usepackage[normalem]{ulem}
\usepackage{url}
\usepackage{enumerate}
\usepackage{enumitem}

\newenvironment{CompactItemize}%
  {\begin{list}{$\; \; \; \; \; \; \ \  \triangleright$}%
   {\leftmargin=0pt \itemsep=2pt \topsep=2pt
     \parsep=0pt \partopsep=0pt}}%
  {\end{list}}

  {\begin{list}{$\; \; \; \; \; \; \ \  \blacktriangleright$}%
   {\leftmargin=0pt \itemsep=2pt \topsep=2pt
     \parsep=0pt \partopsep=0pt}}%
  {\end{list}}
\usepackage{graphicx}
\usepackage{listings}
\usepackage{xspace}
\usepackage{verbatim}
\usepackage{hyperref}
\usepackage{multirow}
\usepackage{subcaption}
\usepackage[dvipsnames,prologue]{xcolor}

\usepackage{tikz,pgfplots}
\pgfplotsset{compat=1.16}
\usetikzlibrary{shapes,arrows,positioning}
\usepackage{tikz-cd}
\usetikzlibrary{calc,decorations.pathmorphing}
\usetikzlibrary{shapes.misc,fit,backgrounds,positioning,arrows,shapes,patterns}

\DeclareCaptionFormat{caption-with-line}{#1#2#3\hrulefill}

\usepackage{mathrsfs}
\usepackage{amsmath}
\usepackage{amsthm}
\usepackage{mathpartir}
\usepackage{mathtools}

\usetikzlibrary{calc,decorations.pathmorphing,shapes}
\usetikzlibrary{decorations.pathmorphing}
\newcommand\xrsquigarrow[1]{%
    \mathrel{%
        \begin{tikzpicture}[%
            baseline={(current bounding box.south)}
            ]
        \node[%
            ,inner sep=.44ex
            ,align=center
            ] (tmp) {$\scriptstyle #1$};
        \path[%
            ,draw,<-
            ,decorate,decoration={%
                ,zigzag
                ,amplitude=0.7pt
                ,segment length=1.2mm,pre length=3.5pt
                }
            ]
        (tmp.south east) -- (tmp.south west);
        \end{tikzpicture}
        }
    }

\newtheorem{thm}{Theorem}
\newtheorem{lemma}{\noindent Lemma}
\newenvironment{proofsketch}{\noindent\emph{Proof (Sketch).}}{\hfill$\square$}
\newtheorem{cor}{Corollary}
\newtheorem{prop}{Proposition}

\newtheorem{definition}{Definition}
\usepackage{cleveref}
\crefname{prop}{Proposition}{Propositions}
\crefname{thm}{Theorem}{Theorems}

\newcommand{\Blue}[1]{{\color{blue} #1}}
\newcommand{\Red}[1]{{\color{red} #1}}

\newcommand{\New}[1]{#1}

\title{Automatically Eliminating Speculative Leaks from Cryptographic Code with Blade}
\ifextended
\subtitle{Extended Version}
\fi

\author[M. Vassena]{Marco Vassena}
\affiliation{%
  \institution{CISPA Helmholtz Center for Information Security}
  \country{Germany}
}
\email{marco.vassena@cispa.saarland}

\author[C. Disselkoen]{Craig Disselkoen}
\affiliation{%
  \institution{UC San Diego}
  \country{USA}
}
\email{cdisselk@cs.ucsd.edu}

\author[K. Gleissenthall]{Klaus von Gleissenthall}
\affiliation{
  \institution{Vrije Universiteit Amsterdam}
  \country{Netherlands}
}
\email{k.freiherrvongleissenthal@vu.nl}

\author[S. Cauligi]{Sunjay Cauligi}
\affiliation{%
  \institution{UC San Diego}
  \country{USA}
}
\email{scauligi@eng.ucsd.edu}

\author[R. Kıcı]{Rami Gökhan Kıcı}
\affiliation{%
  \institution{UC San Diego}
  \country{USA}
}
\email{rkici@eng.ucsd.edu}

\author[R. Jhala]{Ranjit Jhala}
\affiliation{%
  \institution{UC San Diego}
  \country{USA}
}
\email{jhala@cs.ucsd.edu}

\author[D. Tullsen]{Dean Tullsen}
\affiliation{%
  \institution{UC San Diego}
  \country{USA}
}
\email{tullsen@cs.ucsd.edu}

\author[D. Stefan]{Deian Stefan}
\affiliation{%
  \institution{UC San Diego}
  \country{USA}
}
\email{deian@cs.ucsd.edu}

\begin{abstract}
We introduce \tool, a new approach to automatically and efficiently eliminate
  speculative leaks from cryptographic code.
\tool is built on the insight that to stop leaks via speculative execution, it
  suffices to \emph{cut} the dataflow from expressions that speculatively
  introduce secrets (\emph{sources}) to those that leak them through the cache
  (\emph{sinks}), rather than prohibit speculation altogether.
We formalize this insight in a \emph{static type system} that
(1) types each expression as either \emph{transient}, \ie possibly containing
  speculative secrets or as being \emph{stable}, and
(2) prohibits speculative leaks by requiring that all \emph{sink} expressions
  are stable.
\tool relies on a new abstract primitive, \protectMS{}, to halt speculation
  at fine granularity.
We formalize and implement \protectMS{} using existing architectural
  mechanisms, and show how \tool's type system can automatically synthesize a
  \emph{minimal} number of \protectMS{}s to provably eliminate speculative leaks.
We implement \tool in the Cranelift WebAssembly compiler and evaluate our
  approach by repairing several verified, yet vulnerable WebAssembly
  implementations of cryptographic primitives. We find that \tool can fix
  existing programs that leak via speculation \emph{automatically}, without
  user intervention, and \emph{efficiently}
  even when using fences to implement \protectMS{}.\looseness=-1
\end{abstract}



\setcopyright{rightsretained}
\acmPrice{}
\acmDOI{10.1145/3434330}
\acmYear{2021}
\copyrightyear{2021}
\acmSubmissionID{popl21main-p403-p}
\acmJournal{PACMPL}
\acmVolume{5}
\acmNumber{POPL}
\acmArticle{49}
\acmMonth{1}

\begin{document}

\maketitle

\section{Introduction}
Implementing secure cryptographic algorithms is hard.
The code must not only be functionally correct, memory safe, and efficient, it must also avoid divulging secrets indirectly through side channels
like control-flow, memory-access patterns, or execution time.
Consequently, much recent work focuses on how to ensure
implementations do not leak secrets \eg via type
systems~\cite{cauligi:2019:fact,Watt:2019}, verification~\cite{libsignalSP, ct-verif}, and
program transformations~\cite{ct-c-compiler}.

Unfortunately, these efforts can be foiled by speculative execution.
Even if secrets are closely controlled via guards and access checks,
the processor can simply ignore these checks when executing
speculatively. This, in turn, can be exploited by an attacker to leak protected secrets.

In principle, memory fences block speculation, and hence, offer a way to
recover the original security guarantees. In practice, however, fences
pose a dilemma.
Programmers can restore security by conservatively inserting fences after every
load (e.g., using Microsoft's Visual Studio compiler pass~\cite{MSCVC-2020}),
but at huge performance costs.
Alternatively, they can rely on heuristic approaches for inserting
fences~\cite{oo7}, but forgo guarantees about the absence of
side-channels.
Since missing even one fence can allow an attacker to leak any secret from the
address space, secure runtime systems---in particular, browsers like Chrome and
Firefox---take yet another approach and isolate secrets from untrusted code
in different processes to avoid the risk altogether~\cite{site-isolation,
mozilla-sandbox}.
Unfortunately, the engineering effort of such a multi-process redesign is
huge---\eg Chrome's redesign took five years and roughly 450K lines of code
changes~\cite{site-isolation-paper}.

In this paper, we introduce \tool, a new, fully automatic approach to
provably and efficiently eliminate speculation-based leakage from constant-time cryptographic code.
\tool is based on the key insight that to prevent leaking data via
speculative execution, it is not necessary to stop \emph{all}
speculation.
Instead, it suffices to \emph{cut} the data flow from
expressions (sources) that could speculatively introduce
secrets to those that leak them through the cache (sinks).
We develop this insight into an automatic enforcement
algorithm via four contributions.


\mypara{A JIT-Step Semantics for Speculation}
A key aim of BLADE is to enable \emph{source}-level reasoning about
the absence of speculation-based information leaks. This is crucial
to let the source-level type system use control- and data-flow
information to optimally prevent leaks.
High-level reasoning requires a source-level semantic model
of speculation which has, so far, proven to be challenging
as the effects of speculation manifest at the very low
\emph{machine}-level: \eg as branch-prediction affects
the \emph{streams} of low-level instructions that are
fetched and (speculatively) executed.
We address this challenge with our first contribution:
a JIT-step semantics for a high-level \while language
that lets us reconcile the tension between high-level
reasoning and low-level speculation.
These semantics translate high-level \emph{commands}
to low-level machine \emph{instructions} in
a \emph{just-in-time} (JIT) fashion, whilst tracking
control-flow predictions at branch points and modeling
the speculative execution path of the program as well as the ``paths not taken''
as stacks of high-level commands.

Our low-level instructions are inspired by a previous
formal model of a speculative and out-of-order
processor~\cite{pitchfork}, and let us model
the essence of speculation-based attacks---in
particular Spectre-PHT~\cite{Kocher2018spectre,Kiriansky18,Canella:2019}---by
modeling precisely how speculation can occur and what an attacker can
observe via speculation (\S~\ref{sec:semantics}).
%
%
To prevent leakage, we propose and formalize the semantics of an
abstract primitive called \protectMS{} that embodies several hardware mechanisms proposed in recent work~\cite{tullsen-asplos,STT-Yu-19}.
Crucially, and in contrast to a regular fence which stops \emph{all}
speculation, \protectMS{} only stops speculation for
a given \emph{variable}.
For example~\ensuremath{\Varid{x}\mathbin{:=}\mathbf{protect}(\Varid{e})} ensures that the value of $e$ is
assigned to $x$ only \emph{after} $e$ has been assigned
its \emph{stable}, non-speculative value.
Though we encourage hardware vendors to implement \protectMS{} in future
processors, for backwards compatibility, we implement and evaluate two
versions of \protectMS{} on existing hardware---one using fences, another
using \emph{speculative load hardening} (SLH)~\cite{SLH}.

\mypara{A Type System for Speculation}
Our second contribution is an approach to conservatively approximating
the dynamic semantics of speculation via a \emph{static type sytem}
that types each expression as either transient ($\Transient$),
\ie expressions that \emph{may} contain speculative secrets, or
stable ($\Concrete$), \ie those that cannot (\S~\ref{sec:types}).
Our system prohibits speculative leaks by requiring that all
\emph{sink}  expressions that can influence intrinsic attacker
visible behavior (\eg cache addresses) are typed as stable.
The type system does \emph{not} rely on user-provided security annotations to
identify sensitive sources and public sinks.
%
%
Instead, we \emph{conservatively} reject programs that exhibit any flow of
information from transient sources to stable sinks.
%
%
This, in turn, allows us to automatically identify speculative vulnerabilities
in existing cryptographic code (where secrets are not explicitly annotated).
We connect the static and dynamic semantics by proving that well-typed
constant-time programs are secure, i.e., they are also \emph{speculative
  constant-time}~\cite{pitchfork} (§~\ref{sec:correctness}).
This result extends the pre-existing guarantees about sequential constant-time execution of verified cryptographic code to our speculative execution model.
%

\mypara{Automatic Protection}
Existing programs that are free of \protectMS{} statements
are likely insecure under speculation (see Section~\ref{sec:eval}
and~\cite{pitchfork}) and will be rejected by our type system.
Thus, our third contribution is an algorithm that finds potential speculative
leaks and automatically synthesizes a \emph{minimal} number of \protectMS{}
statements to ensure that the program is speculatively constant-time
(\S~\ref{sec:fence-inference}).
To this end, we extend the type checker to construct
a \emph{def-use graph} that captures the data-flow between program
expressions.
The presence of a \emph{path} from transient sources to stable sinks
in the graph indicates a potential speculative leak in the program.
To repair the program, we only need to identify a \emph{cut-set}, a
set of variables whose removal eliminates all the leaky paths in the
graph.
We show that inserting a \protectMS{} statement for each variable in
a cut-set suffices to yield a program that is well-typed, and hence,
secure with respect to speculation.
%
%
Finding such cuts is an instance of the classic
Max-Flow/Min-Cut problem, so existing polynomial time algorithms let
us efficiently synthesize \protectMS{} statements that resolve the
dilemma of enforcing security with \emph{minimal} number of protections.

\mypara{\tool Tool}
Our final contribution is an automatic push-button tool, \tool, which
eliminates potential speculative leaks using this min-cut algorithm.
\tool extends the Cranelift compiler~\cite{cranelift}, which compiles
WebAssembly (Wasm) to x86 machine code; thus, \tool operates on programs
expressed in Wasm.
However, operating on Wasm is not fundamental to our approach---we believe
that \tool's techniques can be applied to other programming languages and
bytecodes.\looseness=-1

We evaluate \tool by repairing verified yet vulnerable (to transient execution
attacks) constant-time cryptographic primitives from Constant-time Wasm
(CT-Wasm)~\cite{Watt:2019} and HACL*~\cite{HACL} (\S~\ref{sec:eval}).
Compared to existing fully automatic speculative mitigation approaches (as
notably implemented in Clang), \tool inserts an order of magnitude fewer
protections (fences or SLH masks).
%
%
\tool's fence-based implementation imposes modest performance overheads:
(geometric mean) 5.0\% performance overhead on our benchmarks to defend from Spectre
v1, or 15.3\% overhead to also defend from Spectre v1.1.  Both results are
significant gains over current solutions.
Our fence-free implementation, which automatically inserts SLH masks, is
faster in the Spectre v1 case---geometric mean 1.7\% overhead---but slower when
including Spectre v1.1 protections, imposing 26.6\% overhead.




\section{Overview}
\label{sec:overview}
This section gives a brief overview of the kinds of speculative leaks
that \tool identifies and how it repairs such programs by careful placement of
\protectMS{} statements.
We then describe how \tool (1) automatically repairs existing programs using
our minimal \protectMS{} inference algorithm and (2) proves that the repairs are
correct using our transient-flow type system.

\subsection{Two Kinds of Speculative Leaks}
\label{subsec:kinds_of_leaks}
\definecolor{burgundy}{rgb}{0.5, 0.0, 0.13}
\begin{figure}[t]
  \centering
  \begin{lstlisting}[language=C,
    keywordstyle=\color{blue}\ttfamily\bfseries,%
    keywords={uint8_t,void,if, int, return},
    basicstyle=\ttfamily\small,
    numbers=left,
    emph={load128_be, store128_be},
    emphstyle=\color{cyan}\ttfamily,
    escapechar=&,
    commentstyle=\color{burgundy}\ttfamily,
    columns=fullflexible,
    numbersep=.5pt,
    xleftmargin=2.5ex]
   void SHA2_update_last(int *input_len, ...)
   {
     if (! valid(input_len)) { return; } &\label{line:validate}& // Input validation
     int len = &\ghstPatch{protect}&(*input_len);  &\label{line:protect}& // Can speculatively read secret data
     ...
     int *dst3 = len + base;  &\label{line:double_index:dest}& // Secret-tainted address
     ...
     *dst3 = pad; &\label{line:double_index:store}&  // Secret-dependent memory access
     ...
     for ( i = 0; i < len + ...)       &\label{line:control:leak}&  // Secret-dependent branch
       dst2[i] = 0;                                   &\label{line:control:write}&
     ...
   }
  \end{lstlisting}
\caption{Code fragment adapted from the HACL* SHA2 implementation, containing
  two potential speculative execution vulnerabilities: one through the data
  cache by writing memory at a secret-tainted address, and one through the
  instruction cache via a secret-tainted control-flow dependency.
  The patch computed by \tool is shown in \ghstPatch{green}.
  }
\label{fig:sha2}
\end{figure}
%
%
%
Figure \ref{fig:sha2} shows a code fragment of the
\text{\ttfamily SHA2\char95{}update\char95{}last} function, a core piece of the SHA2 cryptographic hash
function implementation, adapted (to simplify exposition) from the \Hacl
library.
%
%
%
This function takes as input a pointer \text{\ttfamily input\char95{}len}, validates the
input (line~\ref{line:validate}), loads from memory the public length
of the hash (line~\ref{line:protect}, ignore \ghstPatch{\text{\ttfamily protect}} for now),
calculates a target address \text{\ttfamily dst3}~(line~\ref{line:double_index:dest}),
and pads the buffer pointed to by
\text{\ttfamily dst3}~(line~\ref{line:double_index:store}).
Later, it uses \text{\ttfamily len} to determine the number of initialization rounds in
the condition of the for-loop on line~\ref{line:control:leak}.

\mypara{Leaking Through a Memory Write}
During normal, sequential execution this code is not a problem: the
function validates the input to prevent classic buffer-overflow
vulnerabilities.
However, during speculation, an attacker can exploit this function to leak
sensitive data.
To do this, the attacker first has to
\emph{mistrain} the branch predictor to predict the next input to be valid.
Since \text{\ttfamily input\char95{}len} is a function parameter, the attacker can do this
by, \eg calling the function repeatedly with legitimate addresses.
%
After mistraining the branch predictor this way, the attacker manipulates
\text{\ttfamily input\char95{}len} to point to an address containing secret data and calls the
function again, this time with an invalid pointer.
As a result of the mistraining, the branch predictor causes the processor to
skip validation and load the secret into \text{\ttfamily len}, which in turn is used to
calculate pointer \text{\ttfamily dst3}.
The location pointed to by \text{\ttfamily dst3} is then written on
line~\ref{line:double_index:store}, leaking the secret data.
Even though pointer \text{\ttfamily dst3} is invalid and the subsequent write will not
be visible at the program level (the processor disregards it), the
side-effects of the memory access persist in the cache and therefore become
visible to the attacker.
%
In particular, the attacker can extract the target address---and thereby the
secret---using cache timing measurements~\cite{cache-survey}.

\mypara{Leaking Through a Control-Flow Dependency}
The code in \Cref{fig:sha2} contains a second potential speculative
vulnerability, which leaks secrets through a control-flow dependency instead
of a memory access.
%
%
To exploit this vulnerability, the attacker can similarly manipulate the pointer
\text{\ttfamily input\char95{}len} to point to a secret after mistraining the branch predictor to
skip validation.
But instead of leaking the secret directly through the data cache, the
attacker can leak the value \emph{indirectly} through a control-flow
dependency: in this case, the secret determines how often the initialization
loop (\cref{line:control:leak}) is executed during speculation.
The attacker can then infer the value of the secret from a timing attack on
the instruction cache or (much more easily) on iteration-dependent
lines of the data cache.\looseness=-1

\subsection{Eliminating Speculative Leaks}


\mypara{Preventing the Leak using Memory Fences}
Since these leaks exploit the fact that input validation is speculatively
skipped, we can prevent them by making sure that dangerous operations such as
the write on \cref{line:double_index:store} or the loop condition check on
\cref{line:control:leak} are not executed until the input has been validated.
\citet{Intel-sw-mitigations},
\citet{AMD-sw-mitigations}, and others~\cite{MSCVC, MSCVC-2020} recommend
doing this by inserting a speculation barrier after critical validation check-points.
This would place a \emph{memory fence} after line 3, but anywhere between lines 3 and 8 would work.
%
%
This fence would stop speculation over the fence: statements after the fence
will not be executed until all statements up to the fence (including input
validation) executed.
%
%
While fences can prevent leaks, using fences as such is more restrictive than
necessary---they stop speculative execution of all following instructions, not
only of the instructions that leak---and thus incur a high performance
cost~\cite{Tkachenko18,tullsen-asplos}.\looseness=-1

\mypara{Preventing the Leak Efficiently}
We do not need to stop all speculation to prevent leaks.
Instead, we only need to ensure that potentially secret data, when read speculatively, cannot be leaked.
To this end, we propose an alternative way to stop speculation from reaching the
operations on \cref{line:double_index:store} and \cref{line:control:leak},
through a new primitive called \protectMS{}.
Rather than eliminate \emph{all} speculation, \protectMS{} only
stops speculation along a \emph{particular data-path}.
We use \protectMS{} to patch the program on
\cref{line:protect}.
Instead of assigning the value \text{\ttfamily len} directly from the result of the load, the
memory load is guarded by a \protectMS{} statement.
This ensures that the value assigned to \text{\ttfamily len} is
always guaranteed to use the \text{\ttfamily input\char95{}len} pointer's final, nonspeculative value.
This single \ghstPatch{\protectMS{}} statement on \cref{line:protect} is
sufficient to fix both of the speculative leaks described in
\Cref{subsec:kinds_of_leaks}---it prevents any speculative, secret data from
reaching lines~\ref{line:double_index:store} or~\ref{line:control:leak} where
it could be leaked to the attacker.
%

%

%
%

\mypara{Implementation of protect}
Our \protectMS{} primitive provides an \emph{abstract interface} for
fine grained control of speculation.
This allows us to eliminate speculation-based
leaks precisely and only when needed.
%
However, whether \protectMS{} can eliminate leaks with tolerable
runtime overhead depends on its concrete implementation.
We consider and formalize two implementations: an ideal implementation and one
we can implement on today's hardware.
To have fine grain control of speculation, \protectMS{} must be implemented in
\emph{hardware} and exposed as part of the ISA.
Though existing processors provide only coarse grained control
over speculation through memory fence instructions,
%
%
this might change in the future.
For example, recently proposed microprocessor
designs~\cite{tullsen-asplos,STT-Yu-19} provide new hardware
mechanisms to control speculation, in particular to restrict targeted types of
speculation while allowing other speculation to proceed: this suggests that \protectMS{}
could be implemented efficiently in hardware in the future.\looseness=-1
Even if future processors implement \protectMS{}, we still need to
address Spectre attacks on existing hardware.
Hence, we formalize and implement \protectMS{} in \emph{software},
building on recent Spectre attack mitigations~\cite{ConTExT20}.
%
%
Specifically, we propose a self-contained approach
inspired by Clang's Speculative Load Hardening (SLH)~\cite{SLH}.
%
%
At a high level, Clang's SLH stalls speculative load instructions in a
conditional block by inserting artificial data-dependencies between
loaded addresses and the value of the condition.
This ensures that the load
is not executed before the branch condition is resolved.
%
%
Unfortunately, this approach unnecessarily stalls \emph{all} non-constant conditional load instructions, regardless of whether they can influence a subsequent load and thus can actually cause speculative data leaks.
Furthermore, this approach is unsound---it can also miss some speculative leaks, e.g., if a load instruction is not in the same conditional block that validates its address (like the code in \Cref{fig:sha2}).
%
In contrast to Clang, our approach applies SLH \emph{selectively}, only to individual load instructions whose
result flows to an instruction which might leak, and \emph{precisely}, by
using accurate bounds-check conditions to ensure that only data from valid
addresses can be loaded.\looseness=-1
%

%
%

\subsection{Automatically Repairing Speculative Leaks via \protectMS{}-Inference}
\label{sec:computing-minimal}
\begin{figure}[t]

  \begin{subfigure}[b]{.5\columnwidth}
  \centering
  \begin{tabular}{l}
  \multicolumn{1}{c}{\exNum{1}} \\
  \hline
  \ensuremath{\Varid{x}\mathbin{:=}\Varid{a}{[}\Varid{i}_{1}{]}} \\
  \ensuremath{\Varid{y}\mathbin{:=}\Varid{a}{[}\Varid{i}_{2}{]}} \\
  \ensuremath{\Varid{z}\mathbin{:=}\Varid{x}\mathbin{+}\Varid{y}} \\
  \ensuremath{\Varid{w}\mathbin{:=}\Varid{b}{[}\Varid{z}{]}}
  \end{tabular}
  \end{subfigure}%
  \begin{subfigure}[b]{.5\columnwidth}
  \centering
  \begin{tabular}{l}
  \multicolumn{1}{c}{\exNum{1} \textsc{Patched}} \\
  \hline
  \ensuremath{\Varid{x}\mathbin{:=}\textcolor{BurntOrange}{\mathbf{protect}}(\Varid{a}{[}\Varid{i}_{1}{]})} \\
  \ensuremath{\Varid{y}\mathbin{:=}\textcolor{BurntOrange}{\mathbf{protect}}(\Varid{a}{[}\Varid{i}_{2}{]})} \\
  \ensuremath{\Varid{z}\mathbin{:=}\textcolor{ForestGreen}{\mathbf{protect}}(\Varid{x}\mathbin{+}\Varid{y})} \\
  \ensuremath{\Varid{w}\mathbin{:=}\Varid{b}{[}\Varid{z}{]}}
  \end{tabular}
  \end{subfigure}

   \caption{Running example. Program \exNum{1} is shown on the left and the patched program is shown on the right.
The \orangePatch{orange} patch is sub-optimal because it requires more \ensuremath{\mathbf{protect}} statements than the optimal
\ghstPatch{green} patch.
%
\label{fig:ex1-code}}
\end{figure}
\tool automatically finds potential speculative leaks and synthesizes a
\emph{minimal} number of \protectMS{} statements to eliminate the leaks.
%
We illustrate this process using the simple program
\exNum{1} in Figure~\ref{fig:ex1-code} as a running example.
The program reads two values from an array (\ensuremath{\Varid{x}\mathbin{:=}\Varid{a}{[}\Varid{i}_{1}{]}} and \ensuremath{\Varid{y}\mathbin{:=}\Varid{a}{[}\Varid{i}_{2}{]}}), adds them (\ensuremath{\Varid{z}\mathbin{:=}\Varid{x}\mathbin{+}\Varid{y}}), and indexes another array with
the result (\ensuremath{\Varid{w}\mathbin{:=}\Varid{b}{[}\Varid{z}{]}}).
This program is written using our formal calculus in which all array operations are implicitly bounds-checked and thus no explicit validation code is needed.

Like the SHA2 example from \Cref{fig:sha2}, \exNum{1} contains a speculative execution
vulnerability: the speculative array reads could bypass their bounds checks and
so \ensuremath{\Varid{x}} and \ensuremath{\Varid{y}} can contain transient secrets (\ie secrets introduced by
misprediction).
This secret data then flows to \ensuremath{\Varid{z}}, and finally leaks through the data
cache when reading \ensuremath{\Varid{b}{[}\Varid{z}{]}}.

\mypara{Def-Use Graph}
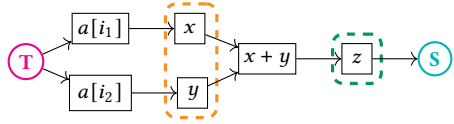
\begin{figure}[t]
  \centering
  \begin{tikzpicture}[node distance=0.4cm and 0.6cm,minimum size=0.4cm,auto]

    \node[draw, circle, color=\sourceColor, thick, inner sep=2pt] (src)
    {\graphsize $\Transient$};

    \node[draw, inner sep=3pt] (ai2src) [below right=0cm and 0.4cm of src]
    {\graphsize $a[i_2]$};


     \node[draw, inner sep=2pt] (ai1src)
     [above = of ai2src] {\graphsize $a[i_1]$};




    \node[draw,inner sep=1pt] (x) [right=of ai1src] {\graphsize
      $x$};


    \node[draw,inner sep=1pt] (y) [right=of ai2src] {\graphsize $y$};


      \node[draw,inner sep=2pt] (xplusy) [above right = 1pt and 0.4cm of y]
    {\graphsize $x+y$};

    \node[draw,inner sep=1pt] (z) [ right = of xplusy]
    {\graphsize $z$};




    \node[draw,circle, color=\sinkColor, thick, inner sep=2pt] (snk) [right = of z]
    {\graphsize $\Concrete$};

    \node (box) [color=BurntOrange,draw,rounded corners, dashed, very thick, fit = (x) (y)] {};
    \node (box) [color=ForestGreen,draw,rounded corners, dashed, very thick, fit = (z)] {};






   \path[->] (ai2src) edge node {} (y);
   \path[->] (ai1src) edge node {} (x);
   \path[->] (z) edge node {} (snk);


    \path[->]
        (src) edge node {} (ai1src)
        (src) edge node {} (ai2src)
        (x) edge node {} (xplusy)
        (y) edge node {} (xplusy)
        (xplusy) edge node {} (z)
    ; 
    \end{tikzpicture}
    \caption{\emph{Subset} of the def-use graph of \exNum{1}.
      %
      The dashed lines identify two valid choices of cut-sets. The \textcolor{BurntOrange}{\textbf{left cut}} requires removing two nodes and thus
inserting two \protectMS{} statements. The \textcolor{ForestGreen}{\textbf{right cut}} shows a
minimal solution, which only requires removing a single node.}
\label{fig:ex1:data-flow}
\end{figure}
To secure the program, we need to \emph{cut the dataflow} between the
array reads which could introduce \emph{transient} secret values into
the program, and the index in the array read where they are leaked
through the cache.
For this, we first build a \emph{def-use graph} whose nodes and
directed edges capture the data dependencies between the expressions
and variables of a program.
For example, consider (a subset of) the def-use graph of program
\exNum{1} in \Cref{fig:ex1:data-flow}.
%
In the graph, the edge \ensuremath{\Varid{x}\to \Varid{x}\mathbin{+}\Varid{y}} indicates that \ensuremath{\Varid{x}} is used to
compute \ensuremath{\Varid{x}\mathbin{+}\Varid{y}}.
%
To track how transient values propagate in the def-use graph, we
extend the graph with the special node~$\Transient$, which
represents the \emph{source} of \sourceColorOf{\emph{transient}}
values of the program.
Since reading memory creates transient values, we connect the
$\Transient$ node to all nodes containing expressions that explicitly
read memory, \eg \ensuremath{\ensuremath{\Transient}\to \Varid{a}{[}\Varid{i}_{1}{]}}.
Following the data dependencies along the edges of the def-use graph,
we can see that node \ensuremath{\ensuremath{\Transient}} is transitively connected to node \ensuremath{\Varid{z}}, which
indicates that \ensuremath{\Varid{z}} can contain transient data at runtime.
To detect insecure uses of transient values, we then extend the graph
with the special node \ensuremath{\ensuremath{\Concrete}}, which represents the \emph{sink} of
\sinkColorOf{\emph{stable}} (\ie non-transient) values of a program.
Intuitively, this node draws all the values of a program that
\emph{must} be stable to avoid transient execution attacks.
Therefore, we connect all expression used as array indices in the
program to the \ensuremath{\ensuremath{\Concrete}} node, \eg \ensuremath{\Varid{z}\to \ensuremath{\Concrete}}.
The fact that the graph in Figure~\ref{fig:ex1:data-flow} contains a
\emph{path} from \ensuremath{\ensuremath{\Transient}} to \ensuremath{\ensuremath{\Concrete}} indicates that transient data flows
through data dependencies into (what should be) a stable index
expression and thus the program may be leaky.
\mypara{Cutting the Dataflow}
In order to make the program safe, we need to \emph{cut} the data-flow
between $\Transient$ and $\Concrete$ by introducing
\protectMS{} statements.
This problem can be equivalently restated as follows: find a
\emph{cut-set}, \ie a set of variables, such that removing the
variables from the graph eliminates all paths from ~$\Transient$ to
$\Concrete$.
Each choice of cut-set defines a way to repair the program: simply add
a \protectMS{} statement for each variable in the set.
\Cref{fig:ex1:data-flow} contains two choices of cut-sets, shown as
dashed lines.
The cut-set on the left requires two protect statements, for variables
$x$ and $y$ respectively, corresponding to the \orangePatch{orange}
patch in \Cref{fig:ex1-code}.
The cut-set on the right is \emph{minimal}, it requires only a single
protect, for variable~$z$, and corresponds to the \ghstPatch{green}
patch in \Cref{fig:ex1-code}.
Intuitively, minimal cut-sets result in patches that introduce as few
\protectMS{}s as needed and therefore allow more speculation.
Luckily, the problem of finding a minimal cut-set is an instance of the classic
Min-Cut/Max-Flow problem, which can be solved using efficient, polynomial-time
algorithms~\cite{MinCut}.
For simplicity, \tool adopts a uniform cost model and therefore synthesizes
patches that contain a minimal \emph{number} of \protectMS{} statements,
regardless of their position in the code and how many times they can be
executed.
Though our evaluation shows that even this simple model imposes modest
overheads (\S\ref{sec:eval}), our implementation can be easily optimized by
considering additional criteria when searching for a minimal cut set, with
further performance gain likely.
For example, we could assign weights proportional to execution frequency, or
introduce penalties for placing \protectMS{} inside loops.

\subsection{Proving Correctness via Transient-Flow Types}
\label{sec:proving}
%
%
To ensure that we add \protectMS{} statements in all the right places (without formally
verifying our repair algorithm),
we use a type system to \emph{prove} that patched
programs are secure, \ie they satisfy a \emph{semantic} security condition.
The type system simplifies the security analysis---we can reason about program
execution directly rather than through generic flows of information in the
def-use graph.
Moreover, restricting the security analysis to the type system
makes the security proofs independent of the specific algorithm used
to compute the program repairs (\eg the Max-Flow/Min-Cut algorithm).
As long as the repaired program type checks, \tool's
formal guarantees hold.
To show that the patches obtained from cutting the def-use graph of a
given program are \emph{correct} (\ie they make the program
well-typed), our transient-flow type system constructs its def-use
graph from the type-constraints generated during type inference.
%
%
%

\mypara{Typing Judgement}
Our type system statically assigns a transient-flow type to each
variable: a variable is typed as \sourceColorOf{\emph{transient}}
(written as $\Transient$), if it can contain transient data (\ie
potential secrets) at runtime, and as \sinkColorOf{\emph{stable}}
(written as $\Concrete$), otherwise.
For instance, in program \exNum{1} (Fig.~\ref{fig:ex1:data-flow}) variables \ensuremath{\Varid{x}} and \ensuremath{\Varid{y}} (and hence \ensuremath{\Varid{z}}) are typed as
\sourceColorOf{\emph{transient}} because they may temporarily contain secret data originated from speculatively reading the array out of bounds.
Given a typing environment~$\Gamma$ which assigns a transient-flow
type to each variable,
and a command~$c$, the type system defines a
judgement~$\Gamma \vdash c$ saying that $c$ is
free of speculative execution bugs (\S~\ref{sec:types}).
%
%
%
%
The type system ensures that transient expressions are not used in
positions that may leak their value by affecting memory reads and
writes, \eg they may not be used as array indices and in loop
conditions.
Additionally, it ensures that transient expressions are not
written to stable variables, except via \protectMS{}.
For example, our type system rejects program \exNum{1} because it uses  \sourceColorOf{\emph{transient}} variable \ensuremath{\Varid{z}} as an array index, but it accepts program \exNum{1} \textsc{Patched} in which
\ensuremath{\Varid{z}} can be typed \sinkColorOf{\emph{stable}} thanks to the \protectMS{} statements.
\New{We say that variables whose assignment is guarded by a \ensuremath{\mathbf{protect}} statement (like \ensuremath{\Varid{z}} in \exNum{1} \textsc{Patched}) are ``protected variables''.}

To study the security guarantees of our type system, we define a JIT-step semantics for speculative execution of a high-level \while
language (\S~\ref{sec:semantics}), which
resolves the tension between source-level reasoning and machine-level speculation.
We then show that our type system indeed prevents speculative execution
attacks, i.e., we prove that well-typed programs
remain constant-time under the speculative semantics (\S~\ref{sec:correctness}).
%

%

\mypara{Type Inference}
Given an input program, we construct the corresponding def-use graph
by collecting the type constraints generated during type inference.
Type inference is formalized by a typing-inference judgment \ensuremath{\Gamma,\protectedSet\vdash\Varid{c}\cnsSep\Varid{k}} (\S~\ref{sec:fence-inference}), which extends the
typing judgment from above with (1) a set of \New{implicitly} protected variables \ensuremath{\protectedSet} (the \emph{cut-set}), and (2)  a set of type-constraints \ensuremath{\Varid{k}}
(the \emph{def-use graph}).
\New{Intuitively, the set \ensuremath{\protectedSet} identifies the variables of program \ensuremath{\Varid{c}} that
contain \sourceColorOf{\emph{transient}} data and that must be protected (\ie
they must be made \sinkColorOf{\emph{stable}} using \ensuremath{\mathbf{protect}}) to avoid
leaks.}
%
%
At a high level, type inference has 3 steps: \emph{(i)} generate a set
of constraints under an \emph{initial} typing environment and
protected set that allow any program to type-check, \emph{(ii)}
construct the def-use graph from the constraints and find a cut-set
(the \emph{final} protected set), and \emph{(iii)} compute the
\emph{final} typing environment which types the variables in the
cut-set as stable.
%
%
To characterize the security of a still \emph{unrepaired} program
after type inference, we define a typing judgment \ensuremath{\Gamma,\protectedSet\vdash\Varid{c}},
where unprotected variables are explicitly accounted for in the \ensuremath{\protectedSet}
set.\footnote{The judgment \ensuremath{\Gamma\vdash\Varid{c}} is just a short-hand for
  \ensuremath{\Gamma,\varnothing\vdash\Varid{c}}.}
Intuitively, the program is secure if we promise to insert a
\protectMS{} statement for each variable in \ensuremath{\protectedSet}.
To repair programs, we simply honor this promise and insert a protect
statement for each variable in the protected set of the typing
judgment obtained above.
Once repaired, the program type checks under an \emph{empty} protected set.\looseness=-1
%

\subsection{Attacker Model}
%
We delineate the extents of the security guarantees of our type system and repair algorithm by discussing the attacker model considered in this work.
We assume an attacker model where the attacker runs cryptographic code, written
in Wasm, on a speculative out-of-order processor; the attacker can influence
how programs are speculatively executed using the branch predictor and choose
the instruction execution order in the processor pipeline.
\New{
The attacker can make predictions based on control-flow history and
memory-access patterns similar to real, adaptive predictors.
Though these predictions \emph{can} depend on secret information in general,
we assume that they are secret-independent for the constant-time programs
repaired by \tool.
%
}
%
%

We assume that the attacker can observe the effects of their actions on the
cache, even if these effects are otherwise invisible at the ISA level.
In particular, while programs run, the attacker can take precise timing
measurements of the data- and instruction-cache with a cache-line
granularity, and thus infer the value of secret data.
These features allow the attacker to mount Spectre-PHT
attacks~\cite{Kocher2018spectre,Kiriansky18} and exfiltrate data through
FLUSH+RELOAD~\cite{FLUSH+RELOAD} and PRIME+PROBE~\cite{PRIME+PROBE} cache
side-channel attacks.
We do not consider speculative attacks that rely on the Return Stack Buffer
(\eg Spectre-RSB~\cite{Maisuradze:2018,Koruyeh:2018}), Branch Target Buffer
(Spectre-BTB~\cite{Kocher2018spectre}), or Store-to-Load forwarding
misprediction (Spectre-STL~\cite{spectrev4}, recently reclassified as a Meltdown attack~\cite{medusa}).
We similarly do not consider Meltdown attacks~\cite{Lipp2018meltdown} or
attacks that do not use the cache to exfiltrate data, \eg port contention
(SMoTherSpectre~\cite{Bhattacharyya:2019}).
%


\section{A JIT-Step Semantics for Speculation}
\label{sec:semantics}

\begin{figure}[t]
  \centering
  \rulesize
\begin{subfigure}[b]{.42\textwidth}
\setlength\mathindent{0pt}%
\begin{align*}
\text{Values }  \ensuremath{\Varid{v}}  &\ensuremath{\Coloneqq}  \ensuremath{\Varid{n}\;\;|\;\;\Varid{b}\;\;|\;\;\Varid{a}} \\
\text{Expr. } \ensuremath{\Varid{e}}  &\ensuremath{\Coloneqq}  \ensuremath{\Varid{v}\;\;|\;\;\Varid{x}\;\;|\;\;\Varid{e}\mathbin{+}\Varid{e}\;\;|\;\;\Varid{e}<\Varid{e}} \\
                   &\ \ensuremath{\;|\;\;\Varid{e}\;\otimes\;\Varid{e}\;\;|\;\;\Varid{e}\mathbin{?}\Varid{e}\mathbin{:}\Varid{e}} \\
                   &\ \ensuremath{\;|\;\;\mathit{length}(\Varid{e})\;\;|\;\;\Varid{base}(\Varid{e})} \\
\text{Rhs. }   \ensuremath{\Varid{r}}  & \ensuremath{\Coloneqq}  \ensuremath{\Varid{e}\;\;|\;\;\ast\Varid{e}\;\;|\;\;\Varid{a}{[}\Varid{e}{]}} \\
\text{Cmd. }   \ensuremath{\Varid{c}} & \ensuremath{\Coloneqq}  \ensuremath{\mathbf{skip}\;\;|\;\;\Varid{x}\mathbin{:=}\Varid{r}\;\;|\;\;\ast\Varid{e}\mathrel{=}\Varid{e}} \\
                   &\ \ensuremath{\;|\;\;\Varid{a}{[}\Varid{e}{]}\mathbin{:=}\Varid{e}\;\;|\;\;\mathbf{fail}\;\;|\;\;\Varid{c};\Varid{c}} \\
                   &\ \ensuremath{\;|\;\;\mathbf{if}\;\Varid{e}\;\mathbf{then}\;\Varid{c}\;\mathbf{else}\;\Varid{c}} \\
                   &\ \ensuremath{\;|\;\;\mathbf{while}\;\Varid{e}\;\mathbf{do}\;\Varid{c}} \\
                   &\ \ensuremath{\;|\;\;\Varid{x}\mathbin{:=}\mathbf{protect}(\Varid{r})}
\end{align*}
\caption{Source syntax.}
\label{fig:surface-syntax}
\end{subfigure}%
\hfill
\begin{subfigure}[b]{.58\textwidth}
\setlength\mathindent{0pt}%
\begin{align*}
\text{Instructions }  \ensuremath{\Varid{i}}  &\ensuremath{\Coloneqq} \ensuremath{\mathbf{nop}\;\;|\;\;\Varid{x}\mathbin{:=}\Varid{e}\;\;|\;\;\Varid{x}\mathbin{:=}\mathbf{load}(\Varid{e})} \\
                    &\ \ensuremath{\;|\;\;\mathbf{store}(\Varid{e},\Varid{e})\;\;|\;\;\Varid{x}\mathbin{:=}\mathbf{protect}(\Varid{e})} \\
                  &\ \ensuremath{\;|\;\;\mathbf{guard}(\Varid{e}^{\Varid{b}},\Varid{cs},\Varid{p})\;\;|\;\;\mathbf{fail}(\Varid{p})} \\
\text{Directives }   \ensuremath{\Varid{d}}  &\ensuremath{\Coloneqq} \ensuremath{\mathbf{fetch}\;\;|\;\;\mathbf{fetch}\;\Varid{b}\;\;|\;\;\mathbf{exec}\;\Varid{n}\;\;|\;\;\mathbf{retire}} \\
\text{Observations }   \ensuremath{\Varid{o}}  &\ensuremath{\Coloneqq} \ensuremath{\epsilon\;\;|\;\;\mathbf{read}(\Varid{n},\Varid{ps})\;\;|\;\;\mathbf{write}(\Varid{n},\Varid{ps})} \\
                  &\ \ensuremath{\;|\;\;\mathbf{fail}(\Varid{p})\;\;|\;\;\mathbf{rollback}(\Varid{p})} \\
\text{Predictions } \ensuremath{\Varid{b}}  & \ensuremath{\in} \ensuremath{\{\mskip1.5mu \mathbf{true},\mathbf{false}\mskip1.5mu\}} \\
\text{Guard Fail Ids. } \ensuremath{\Varid{p}}  & \ensuremath{\in} \ensuremath{\mathbb{N}} \\
\text{Instr. Buffers }     \ensuremath{\Varid{is}}  &\ensuremath{\Coloneqq} \ensuremath{\Varid{i}\mathbin{:}\Varid{is}\;\;|\;\;[\mskip1.5mu \mskip1.5mu]} \\
\text{Cmd Stacks }     \ensuremath{\Varid{cs}} & \ensuremath{\Coloneqq} \ensuremath{\Varid{c}\mathbin{:}\Varid{cs}\;\;|\;\;[\mskip1.5mu \mskip1.5mu]} \\
\text{Stores }     \ensuremath{\mu} & \ensuremath{\in} \ensuremath{\mathbb{N}\rightharpoonup\Conid{Value}} \\
\text{Var. Maps }    \ensuremath{\rho} & \ensuremath{\in} \ensuremath{\Conid{Var}\to \Conid{Value}} \\
\text{Config. }    \ensuremath{\Conid{C}} & \ensuremath{\Coloneqq} \ensuremath{\langle\Varid{is},\Varid{cs},\mu,\rho\rangle}
\end{align*}
\caption{Processor syntax.}
\label{fig:processor-syntax}
\end{subfigure}
\caption{Syntax of the calculus.}
\end{figure}
We now formalize the concepts discussed in the overview, presenting a semantics
in this section and a type system in \Cref{sec:type-system}.
\New{Our semantics are explicitly conservative and abstract, i.e., they only
model the \emph{essential} features required to capture Spectre-PHT attacks
on modern microarchitectures---namely, speculative and out-of-order execution.
We do not model microarchitectural features exploited by other attacks (e.g.,
Spectre-BTB and Spectre-RSB), which require a more complex semantics
model~\cite{pitchfork,balliu2019inspectre}.
}

\mypara{Language}
We start by giving a formal just-in-time step semantics for a \while language with
speculative execution.
We present the language's source syntax in Figure \ref{fig:surface-syntax}.
Its values consist of Booleans \ensuremath{\Varid{b}}, pointers \ensuremath{\Varid{n}} represented as natural
numbers, and arrays \ensuremath{\Varid{a}}.
\New{The calculus has standard expression constructs: values \ensuremath{\Varid{v}}, variables \ensuremath{\Varid{x}}, addition \ensuremath{\Varid{e}\mathbin{+}\Varid{e}}, and a comparison operator \ensuremath{\Varid{e}<\Varid{e}}.
To formalize the semantics of the software (SLH) implementation of \ensuremath{\mathbf{protect}(\cdot )} we
rely on several additional constructs: the bitwise AND operator \ensuremath{\Varid{e}\;\otimes\;\Varid{e}},
the \emph{non-speculative} conditional ternary operator \ensuremath{\Varid{e}\mathbin{?}\Varid{e}\mathbin{:}\Varid{e}}, and
primitives for getting the length ($\mathit{length}(\cdot)$) and base
address ($\mathit{base}(\cdot)$) of an array.
We do not specify the size and bit-representation of values and, instead, they
satisfy standard properties.
In particular, we assume that \ensuremath{\Varid{n}\;\otimes\;\bm{\mathrm{0}}\mathrel{=}\bm{\mathrm{0}}\mathrel{=}\bm{\mathrm{0}}\;\otimes\;\Varid{n}}, i.e., the bitmask \ensuremath{\bm{\mathrm{0}}} consisting of all 0s is the zero element for \ensuremath{\otimes} and corresponds to number \ensuremath{\mathrm{0}}, and \ensuremath{\Varid{n}\;\otimes\;\bm{\mathrm{1}}\mathrel{=}\Varid{n}\mathrel{=}\bm{\mathrm{1}}\;\otimes\;\Varid{n}}, i.e., the bitmask \ensuremath{\bm{\mathrm{1}}} consisting of all 1s is the unit element for \ensuremath{\otimes} and corresponds to some number \ensuremath{\Varid{n'}\neq\mathrm{0}}.
Commands include variable assignments, pointer dereferences, array
stores,\footnote{\New{The syntax for reading and writing arrays
    requires the array to be a constant value (i.e., \ensuremath{\Varid{a}{[}\Varid{e}{]}}) to
    simplify our security analysis. This simplification does not restrict the power of the attacker, who can
    still access memory at arbitrary addresses speculatively through the index expression \ensuremath{\Varid{e}}.
}} conditionals, and loops.\footnote{\New{
  The branch constructs of our source language can be used to directly model
  Wasm's branch constructs (e.g., conditionals and while loops can map to Wasm's
  \textbf{br\_if} and \textbf{br}) or could be easily extended (e.g., we can
  add \textbf{switch} statements to support Wasm's branch table
  \textbf{br\_table} instruction)~\cite{WASM}.
To keep our calculus small, we do not replicate Wasm's (less traditional)
low-level branch constructs which rely on block labels (e.g.,
\textbf{br}~\emph{label}).
Since Wasm programs are statically typed, these branch instructions switch
execution to statically known, well-specified program points marked with
labels.
Our calculus simply avoids modeling labels explicitly and replaces them with
the code itself.
}}
Beyond these standard constructs, we expose a special command that is used to
prevent transient execution attacks: the command \ensuremath{\Varid{x}\mathbin{:=}\mathbf{protect}(\Varid{r})} evaluates
\ensuremath{\Varid{r}} and assigns its value to \ensuremath{\Varid{x}}, but only after the value is \emph{stable}
(i.e., non-transient).
}
%
%
Lastly, \ensuremath{\mathbf{fail}} triggers a memory violation error (caused by trying to read
or write an array out-of-bounds) and aborts the program.
%


\mypara{JIT-Step Semantics}
Our operational semantics formalizes the execution of source programs on a
pipelined processor and thus enables \emph{source}-level reasoning about
speculation-based information leaks.
%
%
In contrast to previous semantics for speculative
execution~\cite{Guarnieri20,cheang-csf19,McIlroy19,pitchfork}, our
processor abstract machine does not operate directly on fully compiled
assembly programs.
Instead, our processor translates high-level commands into low-level instructions \emph{just in time}, by converting individual commands into
corresponding instructions in the first stage of the processor
pipeline.
To support this execution model, the processor \emph{progressively} flattens structured commands (e.g., \ensuremath{\mathbf{if}}-statements and \ensuremath{\mathbf{while}} loops) into predicted straight-line code and maintains a \emph{stack} of (partially flattened) commands to keep track of the program execution path.
In this model, when the processor detects a misspeculation, it only
needs to replace the command stack with the sequence of commands that
should have been executed instead to restart the execution on the
correct path.
%

%
%
%

\mypara{Processor Instructions}
%
%
Our semantics translates
source commands into an abstract set of
processor instructions shown in \Cref{fig:processor-syntax}.
Most of the processor instructions correspond directly to the basic source commands.
Notably, the processor instructions do not include an explicit jump
instruction for branching.
Instead, a sequence of \emph{guard} instructions represents a series
of \emph{pending} branch points along a \emph{single} predicted path.
Guard instructions have the form \ensuremath{\mathbf{guard}(\Varid{e}^{\Varid{b}},\Varid{cs},\Varid{p})}, which records the
branch condition \ensuremath{\Varid{e}}, its predicted truth value \ensuremath{\Varid{b}}, and a unique
guard identifier \ensuremath{\Varid{p}}, used in our security analysis
(\Cref{sec:correctness}).
Each guard attests to the fact that the current execution is valid
only if the branch condition gets resolved as predicted.
In order to enable a roll-back in case of a missprediction, guards
additionally record the sequence of commands \ensuremath{\Varid{cs}} along the
alternative branch.\looseness=-1

\mypara{Directives and Observations}
Instructions do not have to be executed in sequence: they can be
executed in any order, enabling out-of-order execution.
We use a simple three stage processor pipeline: the execution of each
instruction is split into \ensuremath{\mathbf{fetch}}, \ensuremath{\mathbf{exec}}, and \ensuremath{\mathbf{retire}}.
We do not fix the order in which instructions and their individual
stages are executed, nor do we supply a model of the branch predictor
to decide which control flow path to follow.
Instead, we let the attacker supply those decisions through a set of
\emph{directives} \cite{pitchfork} shown in
Fig.\ \ref{fig:processor-syntax}.
For example, directive \ensuremath{\mathbf{fetch}\;\mathbf{true}} fetches the \ensuremath{\mathbf{true}} branch of a
conditional and \ensuremath{\mathbf{exec}\;\Varid{n}} executes the \ensuremath{\Varid{n}}-th instruction in the
reorder buffer.
Executing an instruction generates an
\emph{observation} (Fig.\ \ref{fig:processor-syntax}) which records
attacker observable behavior.
Observations include \emph{speculative} memory reads and writes \New{(i.e., \ensuremath{\mathbf{read}(\Varid{n},\Varid{ps})} and \ensuremath{\mathbf{write}(\Varid{n},\Varid{ps})} issued while the guard and fail instructions identified by \ensuremath{\Varid{ps}} are
pending), rollbacks (i.e., \ensuremath{\mathbf{rollback}(\Varid{p})} due to misspeculation of guard
\ensuremath{\Varid{p}}), and memory violations (i.e., \ensuremath{\mathbf{fail}(\Varid{p})} due to instruction \ensuremath{\mathbf{fail}(\Varid{p})}).}
Most instructions generate the \emph{silent} observation \ensuremath{\epsilon}.
\New{Like~\cite{pitchfork}, we do not include observations about branch
directions.
Doing so would \emph{not} increase the power of the attacker:
a constant-time program that leaks through the branch direction would also leak
through the presence or absence of rollback events, i.e., different predictions
produce different rollback events, capturing leaks due to branch direction.}

\mypara{Configurations and Reduction Relation}
We formally specify our semantics as a reduction relation between processor
configurations.
A configuration \ensuremath{\langle\Varid{is},\Varid{cs},\mu,\rho\rangle} consists of a queue of in-flight
instructions \ensuremath{\Varid{is}} called the \emph{reorder buffer}, a stack of commands
\ensuremath{\Varid{cs}} representing the current \emph{execution path}, a memory \ensuremath{\mu}, and a map \ensuremath{\rho} from variables to values.
A reduction step \ensuremath{\Conid{C}\;\stepT{\Varid{d}}{\Varid{o}}\;\Conid{C}'} denotes that, under directive
\ensuremath{\Varid{d}}, configuration \ensuremath{\Conid{C}} is transformed into \ensuremath{\Conid{C}'} and generates
observation \ensuremath{\Varid{o}}.
To execute a program \ensuremath{\Varid{c}} with initial memory \ensuremath{\mu} and variable map
\ensuremath{\rho}, the processor initializes the configuration with an empty
reorder buffer and inserts the program into the command stack, i.e.,
\ensuremath{\langle[\mskip1.5mu \mskip1.5mu],[\mskip1.5mu \Varid{c}\mskip1.5mu],\mu,\rho\rangle}.
Then, the execution proceeds until both the reorder buffer and the
stack in the configuration are empty, i.e., we reach a configuration
of the form \ensuremath{\langle[\mskip1.5mu \mskip1.5mu],[\mskip1.5mu \mskip1.5mu],\mu',\rho'\rangle},
for some final memory store \ensuremath{\mu'} and variable map \ensuremath{\rho'}.

We now discuss the semantics rules of each execution stage and then
those for our security primitive.
\subsection{Fetch Stage}
\begin{figure}[t]
\rulesize
\begin{mathpar}
\inferrule[Fetch-Seq]
{}
{\ensuremath{\langle\Varid{is},(\Varid{c}_{1};\Varid{c}_{2})\mathbin{:}\Varid{cs},\mu,\rho\rangle\stepT{\mathbf{fetch}}{\epsilon}\langle\Varid{is},\Varid{c}_{1}\mathbin{:}\Varid{c}_{2}\mathbin{:}\Varid{cs},\mu,\rho\rangle}}
\and
\inferrule[Fetch-Asgn]
{}
{\ensuremath{\langle\Varid{is},\Varid{x}\mathbin{:=}\Varid{e}\mathbin{:}\Varid{cs},\mu,\rho\rangle\stepT{\mathbf{fetch}}{\epsilon}\langle\Varid{is}+{\mkern-9mu+}\ [\mskip1.5mu \Varid{x}\mathbin{:=}\Varid{e}\mskip1.5mu],\Varid{cs},\mu,\rho\rangle}}
\and
\inferrule[Fetch-Ptr-Load]
{}
{\ensuremath{\langle\Varid{is},\Varid{x}\mathbin{:=}\ast\Varid{e}\mathbin{:}\Varid{cs},\mu,\rho\rangle\stepT{\mathbf{fetch}}{\epsilon}\langle\Varid{is}+{\mkern-9mu+}\ [\mskip1.5mu \Varid{x}\mathbin{:=}\mathbf{load}(\Varid{e})\mskip1.5mu],\Varid{cs},\mu,\rho\rangle}}
\and
\inferrule[Fetch-Array-Load]
{\ensuremath{\Varid{c}\mathrel{=}\Varid{x}\mathbin{:=}\Varid{a}{[}\Varid{e}{]}} \\ \ensuremath{\Varid{e}_{1}\mathrel{=}\Varid{e}<\mathit{length}(\Varid{a})} \\ \ensuremath{\Varid{e}_{2}\mathrel{=}\Varid{base}(\Varid{a})\mathbin{+}\Varid{e}} \\
 \ensuremath{\Varid{c'}\mathrel{=}\mathbf{if}\;\Varid{e}_{1}\;\mathbf{then}\;\Varid{x}\mathbin{:=}\mathbin{*}\Varid{e}_{2}\;\mathbf{else}\;\mathbf{fail}}}
{\ensuremath{\langle\Varid{is},\Varid{c}\mathbin{:}\Varid{cs},\mu,\rho\rangle\stepT{\mathbf{fetch}}{\epsilon}\langle\Varid{is},\Varid{c'}\mathbin{:}\Varid{cs},\mu,\rho\rangle}}
\and
\inferrule[Fetch-If-True]
{\ensuremath{\Varid{c}\mathrel{=}\mathbf{if}\;\Varid{e}\;\mathbf{then}\;\Varid{c}_{1}\;\mathbf{else}\;\Varid{c}_{2}} \\ \ensuremath{\fresh{\Varid{p}}} \\ \ensuremath{\Varid{i}\mathrel{=}\mathbf{guard}(\Varid{e}^{\mathbf{true}},\Varid{c}_{2}\mathbin{:}\Varid{cs},\Varid{p})} }
{\ensuremath{\langle\Varid{is},\Varid{c}\mathbin{:}\Varid{cs},\mu,\rho\rangle\stepT{\mathbf{fetch}\;\mathbf{true}}{\epsilon}\langle\Varid{is}+{\mkern-9mu+}\ [\mskip1.5mu \Varid{i}\mskip1.5mu],\Varid{c}_{1}\mathbin{:}\Varid{cs},\mu,\rho\rangle}}
\and
\end{mathpar}
\caption{Fetch stage (selected rules).}
\label{fig:semantics-fetch}
\end{figure}
The fetch stage flattens the input commands into a sequence of
instructions which it stores in the reorder buffer.
Figure \ref{fig:semantics-fetch} presents selected rules; the
remaining rules are in
\ifextended
Appendix~\ref{app:full-calculus}.
\else
\New{the extended version of this paper~\cite{vassena2020automatically}.}
\fi
Rule \srule{Fetch-Seq} pops command \ensuremath{\Varid{c}_{1};\Varid{c}_{2}} from the commands stack
and pushes the two sub-commands for further processing.
%
%
\srule{Fetch-Asgn} pops an assignment from the commands stack and
appends the corresponding processor instruction (\ensuremath{\Varid{x}\mathbin{:=}\Varid{e}}) at the end
of the reorder buffer.\footnote{Notation \ensuremath{[\mskip1.5mu \Varid{i}_{1},\mathbin{...},\Varid{i}_{\Varid{n}}\mskip1.5mu]}
  represents a list of \ensuremath{\Varid{n}} elements, \ensuremath{\Varid{is}_{1}+{\mkern-9mu+}\ \Varid{is}_{2}} denotes list
  concatenation, and \ensuremath{{\vert}\Varid{is}{\vert}} is the length of list \ensuremath{\Varid{is}}.}
Rule \srule{Fetch-Ptr-Load} is 
similar and simply translates pointer dereferences to the
corresponding load instruction.
Arrays provide a memory-safe interface to read and write memory: the
processor injects bounds-checks when fetching commands that read and
write arrays.
%
%
%
For example, rule \srule{Fetch-Array-Load} expands command \ensuremath{\Varid{x}\mathbin{:=}\Varid{a}{[}\Varid{e}{]}} into the corresponding pointer dereference, but guards the
command with a bounds-check condition.
First, the rule generates the condition \ensuremath{\Varid{e}_{1}\mathrel{=}\Varid{e}<\mathit{length}(\Varid{a})} and
calculates the address of the indexed element \ensuremath{\Varid{e}_{1}\mathrel{=}\Varid{base}(\Varid{a})\mathbin{+}\Varid{e}}.
Then, it replaces the array read on the stack with command \ensuremath{\mathbf{if}\;\Varid{e}_{1}\;\mathbf{then}\;\Varid{x}\mathbin{:=}\mathbin{*}\Varid{e}_{2}\;\mathbf{else}\;\mathbf{fail}} to abort the program and prevent the buffer
overrun if the bounds check fails.
Later, we show that speculative out-of-order execution can simply
ignore the bounds check guard and cause the processor to transiently
read memory at an invalid address.
Rule \srule{Fetch-If-True} fetches a conditional branch from the stack
and, following the prediction provided in directive \ensuremath{\mathbf{fetch}\;\mathbf{true}},
speculates that the condition \ensuremath{\Varid{e}} will evaluate to \ensuremath{\mathbf{true}}.
Thus, the processor inserts the corresponding instruction \ensuremath{\mathbf{guard}(\Varid{e}^{\mathbf{true}},\Varid{c}_{2}\mathbin{:}\Varid{cs},\Varid{p})} with a fresh guard identifier \ensuremath{\Varid{p}} in the
reorder buffer and pushes the then-branch \ensuremath{\Varid{c}_{1}} onto the stack \ensuremath{\Varid{cs}}.
Importantly, the guard instruction stores the else-branch together
with a copy of the current commands stack (\ie \ensuremath{\Varid{c}_{2}\mathbin{:}\Varid{cs}}) as a
rollback stack to restart the execution in case of
misprediction.\looseness=-1

\subsection{Execute Stage}
\label{subsec:execute-stage}
In the execute stage, the processor evaluates the operands of
instructions in the reorder buffer and rolls back the program state
whenever it detects a misprediction.

\mypara{Transient Variable Map}
\label{par:transient-map}
Since instructions can be executed out-of-order,
when we evaluate operands we need to take into account how previous, possibly unresolved
assignments in the reorder buffer affect the variable map.
%
%
In particular, we need to ensure that an instruction cannot execute if
it depends on a preceding assignment whose value is still unknown.
To this end, we define a function \ensuremath{\phi(\rho,\Varid{is})}, called the
\emph{transient variable map} (Fig.~\ref{fig:transient-map}), which
updates variable map \ensuremath{\rho} with the pending assignments in reorder
buffer \ensuremath{\Varid{is}}.
%
%
The function walks through the reorder buffer, registers each resolved
assignment instruction (\ensuremath{\Varid{x}\mathbin{:=}\Varid{v}}) in the variable map through
function update \ensuremath{\update{\rho}{\Varid{x}}{\Varid{v}}},
and marks variables from pending assignments (\ie \ensuremath{\Varid{x}\mathbin{:=}\Varid{e}}, \ensuremath{\Varid{x}\mathbin{:=}\mathbf{load}(\Varid{e})}, and \ensuremath{\Varid{x}\mathbin{:=}\mathbf{protect}(\Varid{r})}) as \emph{undefined} (\ensuremath{\update{\rho}{\Varid{x}}{\bot}}), making their respective values unavailable to following
instructions.

\begin{figure}[t]
    \rulesize
\begin{subfigure}{\columnwidth}
\centering
\setlength{\abovedisplayskip}{0pt}%
\setlength{\belowdisplayshortskip}{0pt}%
\setlength{\abovedisplayshortskip}{0pt}%
\setlength{\belowdisplayskip}{0pt}%
\begin{align*}
&\ensuremath{\phi(\rho,[\mskip1.5mu \mskip1.5mu])\mathrel{=}\rho} \\
&\ensuremath{\phi(\rho,(\Varid{x}\mathbin{:=}\Varid{v})\mathbin{:}\Varid{is})\mathrel{=}\phi(\update{\rho}{\Varid{x}}{\Varid{v}},\Varid{is})} \\
&\ensuremath{\phi(\rho,(\Varid{x}\mathbin{:=}\Varid{e})\mathbin{:}\Varid{is})\mathrel{=}\phi(\update{\rho}{\Varid{x}}{\bot},\Varid{is})} \\
&\ensuremath{\phi(\rho,(\Varid{x}\mathbin{:=}\mathbf{load}(\Varid{e}))\mathbin{:}\Varid{is})\mathrel{=}\phi(\update{\rho}{\Varid{x}}{\bot},\Varid{is})} \\
&\ensuremath{\phi(\rho,(\Varid{x}\mathbin{:=}\mathbf{protect}(\Varid{e}))\mathbin{:}\Varid{is})\mathrel{=}\phi(\update{\rho}{\Varid{x}}{\bot},\Varid{is})} \\
&\ensuremath{\phi(\rho,\Varid{i}\mathbin{:}\Varid{is})\mathrel{=}\phi(\rho,\Varid{is})}
\end{align*}
\captionsetup{format=caption-with-line}
\caption{Transient variable map.}
\label{fig:transient-map}
\end{subfigure}%
\\
\begin{subfigure}{\columnwidth}
\centering
\begin{mathpar}
\inferrule[Execute]
{\ensuremath{{\vert}\Varid{is}_{1}{\vert}\mathrel{=}\Varid{n}\mathbin{-}\mathrm{1}} \\ \ensuremath{\rho'\mathrel{=}\phi(\Varid{is}_{1},\rho)} \\ \ensuremath{\langle\Varid{is}_{1},\Varid{i},\Varid{is}_{2},\Varid{cs}\rangle\xrsquigarrow{(\mu,\rho',\Varid{o})}\langle\Varid{is'},\Varid{cs'}\rangle} }
{\ensuremath{\langle\Varid{is}_{1}+{\mkern-9mu+}\ [\mskip1.5mu \Varid{i}\mskip1.5mu]+{\mkern-9mu+}\ \Varid{is}_{2},\Varid{cs},\mu,\rho\rangle\stepT{\mathbf{exec}\;\Varid{n}}{\Varid{o}}\langle\Varid{is'},\Varid{cs'},\mu,\rho\rangle}}
\end{mathpar}
\captionsetup{format=caption-with-line}
\caption{Execute rule.}
\label{fig:execute-phase1}
\end{subfigure}%
\\
\begin{subfigure}[b]{\columnwidth}
\begin{mathpar}
\inferrule[Exec-Asgn]
{\ensuremath{\Varid{i}\mathrel{=}(\Varid{x}\mathbin{:=}\Varid{e})} \\ \ensuremath{\Varid{v}\mathrel{=}\llbracket \Varid{e}\rrbracket^{\rho}} \\ \ensuremath{\Varid{i'}\mathrel{=}(\Varid{x}\mathbin{:=}\Varid{v})} }
{\ensuremath{\langle\Varid{is}_{1},\Varid{i},\Varid{is}_{2},\Varid{cs}\rangle\xrsquigarrow{(\mu,\rho,\epsilon)}}  \ensuremath{\langle\Varid{is}_{1}+{\mkern-9mu+}\ [\mskip1.5mu \Varid{i'}\mskip1.5mu]+{\mkern-9mu+}\ \Varid{is}_{2},\Varid{cs}\rangle}}
\and
\inferrule[Exec-Load]
{\ensuremath{\Varid{i}\mathrel{=}(\Varid{x}\mathbin{:=}\mathbf{load}(\Varid{e}))} \\ \ensuremath{\mathbf{store}(\anonymous ,\anonymous )\;\not\in\;\Varid{is}_{1}} \\
 \ensuremath{\Varid{n}\mathrel{=}\llbracket \Varid{e}\rrbracket^{\rho}} \\ \ensuremath{\Varid{ps}\mathrel{=}\Moon{\Varid{is}_{1}}} \\
 \ensuremath{\Varid{i'}\mathrel{=}(\Varid{x}\mathbin{:=}\mu(\Varid{n}))} }
{\ensuremath{\langle\Varid{is}_{1},\Varid{i},\Varid{is}_{2},\Varid{cs}\rangle\xrsquigarrow{(\mu,\rho,\mathbf{read}(\Varid{n},\Varid{ps}))}} \ensuremath{\langle\Varid{is}_{1}+{\mkern-9mu+}\ [\mskip1.5mu \Varid{i'}\mskip1.5mu]+{\mkern-9mu+}\ \Varid{is}_{2},\Varid{cs}\rangle}}
\and
\inferrule[Exec-Branch-Ok]
{\ensuremath{\Varid{i}\mathrel{=}\mathbf{guard}(\Varid{e}^{\Varid{b}},\Varid{cs'},\Varid{p})} \\ \ensuremath{\llbracket \Varid{e}\rrbracket^{\rho}\mathrel{=}\Varid{b}} } 
{\ensuremath{\langle\Varid{is}_{1},\Varid{i},\Varid{is}_{2},\Varid{cs}\rangle\xrsquigarrow{(\mu,\rho,\epsilon)}\langle\Varid{is}_{1}+{\mkern-9mu+}\ [\mskip1.5mu \mathbf{nop}\mskip1.5mu]+{\mkern-9mu+}\ \Varid{is}_{2},\Varid{cs}\rangle}}
\and
\inferrule[Exec-Branch-Mispredict]
{\ensuremath{\Varid{i}\mathrel{=}\mathbf{guard}(\Varid{e}^{\Varid{b}},\Varid{cs'},\Varid{p})} \\ \ensuremath{\Varid{b'}\mathrel{=}\llbracket \Varid{e}\rrbracket^{\rho}} \\ \ensuremath{\Varid{b'}\neq\Varid{b}} } 
{\ensuremath{\langle\Varid{is}_{1},\Varid{i},\Varid{is}_{2},\Varid{cs}\rangle\xrsquigarrow{(\mu,\rho,\mathbf{rollback}(\Varid{p}))}} \ensuremath{\langle\Varid{is}_{1}+{\mkern-9mu+}\ [\mskip1.5mu \mathbf{nop}\mskip1.5mu],\Varid{cs'}\rangle}}
\end{mathpar}
\captionsetup{format=caption-with-line}
\caption{Auxiliary relation (selected rules).}
\label{fig:execute-phase2}
\end{subfigure}
\caption{Execute stage.}
\label{fig:semantics-execute}
\end{figure}

\mypara{Execute Rule and Auxiliary Relation}
\Cref{fig:semantics-execute} shows selected rules for the execute stage.
Rule \srule{Execute} executes the \ensuremath{\Varid{n}}th instruction in the reorder
buffer, following the directive \ensuremath{\mathbf{exec}\;\Varid{n}}. For this, the rule splits the reorder
buffer into prefix \ensuremath{\Varid{is}_{1}}, $n$th instruction \ensuremath{\Varid{i}}, and suffix \ensuremath{\Varid{is}_{2}}.
Next, it computes the transient variable map \ensuremath{\phi(\Varid{is}_{1},\rho)} and
executes a transition step under the new map using an auxiliary
relation \ensuremath{\rightsquigarrow}.
Notice that \srule{Execute} does not update the store or the variable
map---the transient map is simply discarded. These changes are performed
later in the retire stage.

The rules for the auxiliary relation are shown in
Fig.~\ref{fig:execute-phase2}.
The relation transforms a tuple \ensuremath{\langle\Varid{is}_{1},\Varid{i},\Varid{is}_{2},\Varid{cs}\rangle} consisting of
prefix, suffix and current instruction \ensuremath{\Varid{i}} into a tuple \ensuremath{\langle\Varid{is'},\Varid{cs'}\rangle}
specifying the reorder buffer and command stack obtained by executing
\ensuremath{\Varid{i}}.
For example, rule \srule{Exec-Asgn} evaluates the right-hand side of
the assignment \ensuremath{\Varid{x}\mathbin{:=}\Varid{e}} where \ensuremath{\llbracket \Varid{e}\rrbracket^{\rho}} denotes the value of \ensuremath{\Varid{e}}
under \ensuremath{\rho}.
The premise \ensuremath{\Varid{v}\mathrel{=}\llbracket \Varid{e}\rrbracket^{\rho}} ensures that the expression is defined,
\ie it does not evaluate to \ensuremath{\bot}.
Then, the rule substitutes the computed value into the assignment
(\ensuremath{\Varid{x}\mathbin{:=}\Varid{v}}), and reinserts the instruction back into its original
position in the reorder buffer.

\mypara{Loads}
Rule \srule{Exec-Load} executes a memory load.
The rule computes the address (\ensuremath{\Varid{n}\mathrel{=}\llbracket \Varid{e}\rrbracket^{\rho}}), retrieves the value
at that address from memory (\ensuremath{\mu(\Varid{n})}) and rewrites the load into
an assignment (\ensuremath{\Varid{x}\mathbin{:=}\mu(\Varid{n})}).
By inserting the resolved assignment into the reorder buffer, the rule allows the processor
to transiently forward the loaded value to later
instructions through the transient variable map.
To record that the load is issued \emph{speculatively}, the
observation \ensuremath{\mathbf{read}(\Varid{n},\Varid{ps})} stores list \ensuremath{\Varid{ps}} containing the identifiers
of the \New{guard and fail instructions still pending in the
reorder buffer.}
%
Function \ensuremath{\Moon{\Varid{is}_{1}}} simply extracts these identifiers from the guard \New{and fail instructions} in prefix \ensuremath{\Varid{is}_{1}}.
%
%
In the rule, premise \ensuremath{\mathbf{store}(\anonymous ,\anonymous )\;\not\in\;\Varid{is}_{1}} prevents the processor from reading
potentially stale data from memory: if the load aliases with a preceding (but
pending) store, ignoring the store could produce a stale read.

\mypara{Store Forwarding}
Alternatively, instead of exclusively reading fresh values from memory, it would be also possible to \emph{forward} values waiting to be written to memory at the same address by a pending, aliasing store.
For example, if buffer \ensuremath{\Varid{is}_{1}} contained a resolved instruction \ensuremath{\mathbf{store}(\Varid{n},\Varid{v})},  rule \srule{Exec-Load} could \emph{directly} propagate the value \ensuremath{\Varid{v}} to the load, which would be resolved to \ensuremath{\Varid{x}\mathbin{:=}\Varid{v}}, without reading memory and thus generating silent event \ensuremath{\epsilon} instead of \ensuremath{\mathbf{read}(\Varid{n},\Varid{ps})}.\footnote{Capturing the semantics of \emph{store-forwarding} would additionally require some bookkeeping about the freshness of (forwarded) values in order to detect reading stale data. We refer the interested reader to \cite{pitchfork} for full
 details.}
The Spectre v1.1 variant relies on this particular microarchitectural optimization to leak data---for example by speculatively writing transient data out of the bounds of an array, forwarding that data to an aliasing load.\looseness=-1
%
%
%

\mypara{Guards and Rollback}
Rules \srule{Exec-Branch-Ok} and \srule{Exec-Branch-Mispredict}
resolve guard instructions.
In rule \srule{Exec-Branch-Ok}, the predicted and computed value of
the guard expression match (\ensuremath{\llbracket \Varid{e}\rrbracket^{\rho}\mathrel{=}\Varid{b}}), and thus the processor
only replaces the guard with \ensuremath{\mathbf{nop}}.
In contrast, in rule \srule{Exec-Branch-Mispredict} the predicted and
computed value differ (\ensuremath{\llbracket \Varid{e}\rrbracket^{\rho}\mathrel{=}\Varid{b'}} and \ensuremath{\Varid{b'}\neq\Varid{b}}).
This causes the processor to revert the program state and issue a
rollback observation (\ensuremath{\mathbf{rollback}(\Varid{p})}).
%
For the rollback, the processor discards the instructions \emph{past}
the guard (\ie \ensuremath{\Varid{is}_{2}}) and substitutes the current commands stack \ensuremath{\Varid{cs}}
with the rollback stack \ensuremath{\Varid{cs'}} which causes execution to revert to the
alternative branch.

\subsection{Retire Stage}
\begin{figure}[t]
\rulesize
\centering
\begin{mathpar}
\inferrule[Retire-Nop]
{}
{\ensuremath{\langle\mathbf{nop}\mathbin{:}\Varid{is},\Varid{cs},\mu,\rho\rangle\stepT{\mathbf{retire}}{\epsilon}\langle\Varid{is},\Varid{cs},\mu,\rho\rangle}}
\and
\inferrule[Retire-Asgn]
{}
{\ensuremath{\langle\Varid{x}\mathbin{:=}\Varid{v}\mathbin{:}\Varid{is},\Varid{cs},\mu,\rho\rangle\stepT{\mathbf{retire}}{\epsilon}\langle\Varid{is},\Varid{cs},\mu,\update{\rho}{\Varid{x}}{\Varid{v}}\rangle}}
\and
\inferrule[Retire-Store]
{} 
{\ensuremath{\langle\mathbf{store}(\Varid{n},\Varid{v})\mathbin{:}\Varid{is},\Varid{cs},\mu,\rho\rangle\stepT{\mathbf{retire}}{\epsilon}\langle\Varid{is},\Varid{cs},\update{\mu}{\Varid{n}}{\Varid{v}},\rho\rangle}}
\and
\inferrule[Retire-Fail]
{}
{\ensuremath{\langle\mathbf{fail}(\Varid{p})\mathbin{:}\Varid{is},\Varid{cs},\mu,\rho\rangle\stepT{\mathbf{retire}}{\mathbf{fail}(\Varid{p})}\langle[\mskip1.5mu \mskip1.5mu],[\mskip1.5mu \mskip1.5mu],\mu,\rho\rangle}}
\end{mathpar}
\caption{Retire stage.}
\label{fig:semantics-retire}
\end{figure}
The retire stage removes completed instructions from
the reorder buffer and propagates their changes to the variable map
and memory store.
While instructions are executed out-of-order, they are retired
in-order to preserve the illusion of sequential execution to the user.
For this reason, the rules for the retire stage in
\Cref{fig:semantics-retire} always remove the \emph{first} instruction
in the reorder buffer.
For example, rule \srule{Retire-Nop} removes \ensuremath{\mathbf{nop}} from the front of
the reorder buffer.
Rules \srule{Retire-Asgn} and \srule{Retire-Store} remove the resolved
assignment \ensuremath{\Varid{x}\mathbin{:=}\Varid{v}} and instruction \ensuremath{\mathbf{store}(\Varid{n},\Varid{v})} from the reorder
buffer and update the variable map (\ensuremath{\update{\rho}{\Varid{x}}{\Varid{v}}}) and the memory
store (\ensuremath{\update{\mu}{\Varid{n}}{\Varid{v}}}) respectively.
%
Rule \srule{Retire-Fail} aborts the program by emptying reorder buffer
and command stack and generates a \ensuremath{\mathbf{fail}(\Varid{p})} observation, simulating a
processor raising an exception (e.g., a segmentation fault).\looseness=-1

\begin{figure}[t]
%
\examplesize
\centering
\hfill
\begin{subfigure}[t]{.3\textwidth}
   \begin{tabular}{l | l}
     \multicolumn{2}{c}{\textbf{Memory Layout}} \\[.5ex]
  \ensuremath{\mu(\mathrm{0})\mathrel{=}\mathrm{0}}         & \ensuremath{\Varid{b}{[}\mathrm{0}{]}} \\
  \ensuremath{\mu(\mathrm{1})\mathrel{=}\mathrm{0}}         & \ensuremath{\Varid{a}{[}\mathrm{0}{]}} \\
  \ensuremath{\mu(\mathrm{2})\mathrel{=}\mathrm{0}}         & \ensuremath{\Varid{a}{[}\mathrm{1}{]}} \\
  \ensuremath{\mu(\mathrm{3})\mathrel{=}\sourceColorOf{\mathrm{42}}}    & \ensuremath{ s{[}\mathrm{0}{]}} \\
  \ensuremath{\cdots}                 & \ensuremath{\cdots} \\
   \end{tabular} \\
\end{subfigure}%
\hfill
\begin{subfigure}[t]{.3\textwidth}
\begin{tabular}{c}
\textbf{Variable Map} \\[.5ex]
\ensuremath{\rho(\Varid{i}_{1})\mathrel{=}\mathrm{1}} \\
\ensuremath{\rho(\Varid{i}_{2})\mathrel{=}\mathrm{2}} \\
\ensuremath{\cdots} \\
\end{tabular}
\end{subfigure}%
\hfill
\\[1ex]
\begin{subfigure}{\textwidth}
\centering
\renewcommand{\arraystretch}{1.1}
\examplesize
\begin{tabular}{c l | c | c | c | c }
  \multicolumn{2}{c|}{\textbf{Reorder Buffer}} &$\mathbf{exec} \ 2$ &
                                                                       $\mathbf{exec}
                                                                       \
                                                                       4$
  & $\mathbf{exec} \ 5$ & $\mathbf{exec} \ 7$ \\
\cline{1-2}
1 & \ensuremath{\mathbf{guard}((\Varid{i}_{1}<\mathit{length}(\Varid{a}))^{\mathbf{true}},[\mskip1.5mu \mathbf{fail}\mskip1.5mu],\mathrm{1})} & & & & \\
2 & \ensuremath{\Varid{x}\mathbin{:=}\mathbf{load}(\Varid{base}(\Varid{a})\mathbin{+}\Varid{i}_{1})}                       & \ensuremath{\Varid{x}\mathbin{:=}\mu(\mathrm{2})} & & & \\
3 & \ensuremath{\mathbf{guard}((\Varid{i}_{2}<\mathit{length}(\Varid{a}))^{\mathbf{true}},[\mskip1.5mu \mathbf{fail}\mskip1.5mu],\mathrm{2})}  & & & & \\
4 & \ensuremath{\Varid{y}\mathbin{:=}\mathbf{load}(\Varid{base}(\Varid{a})\mathbin{+}\Varid{i}_{2})}                       & & \ensuremath{\Varid{y}\mathbin{:=}\sourceColorOf{\mu(\mathrm{3})}}   &                 & \\
5 & \ensuremath{\Varid{z}\mathbin{:=}\Varid{x}\mathbin{+}\Varid{y}} & &                                                                             & \ensuremath{\Varid{z}\mathbin{:=}\sourceColorOf{\mathrm{42}}} & \\
6 & \ensuremath{\mathbf{guard}((\Varid{z}<\mathit{length}(\Varid{b}))^{\mathbf{true}},[\mskip1.5mu \mathbf{fail}\mskip1.5mu],\mathrm{3})}  & & & & \\
7 & \ensuremath{\Varid{w}\mathbin{:=}\mathbf{load}(\Varid{base}(\Varid{b})\mathbin{+}\Varid{z})} & & &                                                                            & \ensuremath{\Varid{w}\mathbin{:=}\mu(\sourceColorOf{\mathrm{42}})} \\
\hline
\multicolumn{2}{c|}{Observations:}                    & \ensuremath{\mathbf{read}(\mathrm{2},[\mskip1.5mu \mathrm{1}\mskip1.5mu])}  &  \ensuremath{\mathbf{read}(\mathrm{3},[\mskip1.5mu \mathrm{1},\mathrm{2}\mskip1.5mu])} & \ensuremath{\epsilon}  & \ensuremath{\mathbf{read}(\sourceColorOf{\mathrm{42}},[\mskip1.5mu \mathrm{1},\mathrm{2},\mathrm{3}\mskip1.5mu])} \\
\end{tabular}
\end{subfigure}
\caption{Leaking execution of running program \exNum{1}.}
\label{fig:ex1-execution}
\end{figure}

\mypara{Example}
We demonstrate how the attacker can leak a secret from program
\exNum{1} (Fig.~\ref{fig:ex1-code}) in our model.
First, the attacker instructs the processor to fetch all the
instructions, supplying prediction \ensuremath{\mathbf{true}} for all bounds-check
conditions.
\Cref{fig:ex1-execution} shows the resulting buffer and how it evolves
after each attacker directive; the memory \ensuremath{\mu} and variable map \ensuremath{\rho} are
shown above.
The attacker directives instruct the processor to speculatively execute the
load instructions and the assignment (but not the guard instructions).
Directive \ensuremath{\mathbf{exec}\;\mathrm{2}} executes the first load instruction by computing
the memory address \ensuremath{\mathrm{2}\mathrel{=}\llbracket \Varid{base}(\Varid{a})\mathbin{+}\Varid{i}_{1}\rrbracket^{\rho}} and replacing
the instruction with the assignment \ensuremath{\Varid{x}\mathbin{:=}\mu(\mathrm{2})} containing the
loaded value.
Directive \ensuremath{\mathbf{exec}\;\mathrm{4}} transiently reads \emph{public} array \ensuremath{\Varid{a}} past its bound, at
index $2$, reading into the memory (\ensuremath{\sourceColorOf{\mu(\mathrm{3})}\mathrel{=}\sourceColorOf{\mathrm{42}}}) of \emph{secret} array \ensuremath{ s{[}\mathrm{0}{]}} and generates the corresponding
observation.
Finally, the processor
forwards the values of \ensuremath{\Varid{x}} and \ensuremath{\Varid{y}} through the
\emph{transient} variable map \ensuremath{\rho{[}\Varid{x}\;\mapsto\;\mu(\mathrm{2}),\Varid{y}\;\mapsto\;\sourceColorOf{\mu(\mathrm{3})}{]}}
to compute their sum in the
fifth instruction, (\ensuremath{\Varid{z}\mathbin{:=}\sourceColorOf{\mathrm{42}}}), which is then used as an index
in the last instruction and leaked to the attacker via observation
\ensuremath{\mathbf{read}(\sourceColorOf{\mathrm{42}},[\mskip1.5mu \mathrm{1},\mathrm{2},\mathrm{3}\mskip1.5mu])}.
%

\subsection{Protect}
\label{subsec:security-rules}
Next, we turn to the rules that formalize the semantics of \ensuremath{\mathbf{protect}}
as an ideal hardware primitive and then its software implementation
via speculative-load-hardening (SLH).

\begin{figure}[t]
\rulesize
\begin{subfigure}[b]{\columnwidth}
\begin{mathpar}
\inferrule[Fetch-Protect-Array]
{\ensuremath{\Varid{c}\mathrel{=}(\Varid{x}\mathbin{:=}\mathbf{protect}(\Varid{a}{[}\Varid{e}{]}))} \\
 \ensuremath{\Varid{c}_{1}\mathrel{=}(\Varid{x}\mathbin{:=}\Varid{a}{[}\Varid{e}{]})} \\ \ensuremath{\Varid{c}_{2}\mathrel{=}(\Varid{x}\mathbin{:=}\mathbf{protect}(\Varid{x}))}}
{\ensuremath{\langle\Varid{is},\Varid{c}\mathbin{:}\Varid{cs},\mu,\rho\rangle\stepT{\mathbf{fetch}}{\epsilon}\langle\Varid{is},\Varid{c}_{1}\mathbin{:}\Varid{c}_{2}\mathbin{:}\Varid{cs},\mu,\rho\rangle}}
\and
\inferrule[Fetch-Protect-Expr]
{\ensuremath{\Varid{c}\mathrel{=}(\Varid{x}\mathbin{:=}\mathbf{protect}(\Varid{e}))} \\ \ensuremath{\Varid{i}\mathrel{=}(\Varid{x}\mathbin{:=}\mathbf{protect}(\Varid{e}))}}
{\ensuremath{\langle\Varid{is},\Varid{c}\mathbin{:}\Varid{cs},\mu,\rho\rangle\stepT{\mathbf{fetch}}{\epsilon}\langle\Varid{is}+{\mkern-9mu+}\ [\mskip1.5mu \Varid{i}\mskip1.5mu],\Varid{cs},\mu,\rho\rangle}}
\and
\inferrule[Exec-Protect$_1$]
{\ensuremath{\Varid{i}\mathrel{=}(\Varid{x}\mathbin{:=}\mathbf{protect}(\Varid{e}))} \\ \ensuremath{\Varid{v}\mathrel{=}\llbracket \Varid{e}\rrbracket^{\rho}} \\ \ensuremath{\Varid{i'}\mathrel{=}(\Varid{x}\mathbin{:=}\mathbf{protect}(\Varid{v}))} }
{\ensuremath{\langle\Varid{is}_{1},\Varid{i},\Varid{is}_{2},\Varid{cs}\rangle\xrsquigarrow{(\mu,\rho,\epsilon)}\langle\Varid{is}_{1}+{\mkern-9mu+}\ [\mskip1.5mu \Varid{i'}\mskip1.5mu]+{\mkern-9mu+}\ \Varid{is}_{2},\Varid{cs}\rangle}}
\and
\inferrule[Exec-Protect$_2$]
{\ensuremath{\Varid{i}\mathrel{=}(\Varid{x}\mathbin{:=}\mathbf{protect}(\Varid{v}))} \\
 \ensuremath{\mathbf{guard}(\anonymous ,\anonymous ,\anonymous )\;\not\in\;\Varid{is}_{1}} \\
 \ensuremath{\Varid{i'}\mathrel{=}(\Varid{x}\mathbin{:=}\Varid{v})}}
{\ensuremath{\langle\Varid{is}_{1},\Varid{i},\Varid{is}_{2},\Varid{cs}\rangle\xrsquigarrow{(\mu,\rho,\epsilon)}\langle\Varid{is}_{1}+{\mkern-9mu+}\ [\mskip1.5mu \Varid{i'}\mskip1.5mu]+{\mkern-9mu+}\ \Varid{is}_{2},\Varid{cs}\rangle}}
\end{mathpar}
\captionsetup{format=caption-with-line}
\caption{Semantics of \ensuremath{\mathbf{protect}} as a hardware primitive (selected rules). }
\label{fig:semantics-protect}
\end{subfigure}
\begin{subfigure}[b]{\columnwidth}
\rulesize
\begin{mathpar}
\inferrule[Fetch-Protect-SLH]
{\ensuremath{\Varid{c}\mathrel{=}\Varid{x}\mathbin{:=}\mathbf{protect}(\Varid{a}{[}\Varid{e}{]})} \\ \ensuremath{\Varid{e}_{1}\mathrel{=}\Varid{e}<\mathit{length}(\Varid{a})} \\
 \ensuremath{\Varid{e}_{2}\mathrel{=}\Varid{base}(\Varid{a})\mathbin{+}\Varid{e}} \\
 \ensuremath{\Varid{c}_{1}\mathrel{=}\Varid{m}\mathbin{:=}\Varid{e}_{1}} \\ \ensuremath{\Varid{c}_{2}\mathrel{=}\Varid{m}\mathbin{:=}\Varid{m}\mathbin{?}\bm{\mathrm{1}}\mathbin{:}\bm{\mathrm{0}}} \\
 \ensuremath{\Varid{c}_{3}\mathrel{=}\Varid{x}\mathbin{:=}\ast(\Varid{e}_{2}\;\otimes\;\Varid{m})} \\
 \ensuremath{\Varid{c'}\mathrel{=}\Varid{c}_{1};\mathbf{if}\;\Varid{m}\;\mathbf{then}\;\Varid{c}_{2};\Varid{c}_{3}\;\mathbf{else}\;\mathbf{fail}}}
{\ensuremath{\langle\Varid{is},\Varid{c}\mathbin{:}\Varid{cs},\mu,\rho\rangle\stepT{\mathbf{fetch}}{\epsilon}\langle\Varid{is},\Varid{c'}\mathbin{:}\Varid{cs},\mu,\rho\rangle}}
\end{mathpar}
\captionsetup{format=caption-with-line}
\caption{Software implementation of \ensuremath{\mathbf{protect}(\Varid{a}{[}\Varid{e}{]})}.}
\label{fig:fetch-slh}
\end{subfigure}
\caption{Semantics of \ensuremath{\mathbf{protect}}.}
\end{figure}

\mypara{Protect in Hardware}
Instruction \ensuremath{\Varid{x}\mathbin{:=}\mathbf{protect}(\Varid{r})} assigns the value of $r$, only after all
previous $\mathbf{guard}$ instructions have been executed, \ie when
the value has become stable and no more rollbacks are possible.
\Cref{fig:semantics-protect} formalizes this intuition.
Rule \srule{Fetch-Protect-Expr} fetches protect commands involving
simple expressions (\ensuremath{\Varid{x}\mathbin{:=}\mathbf{protect}(\Varid{e})}) and inserts the corresponding
protect instruction in the reorder buffer.
Rule \srule{Fetch-Protect-Array} piggy-backs on the previous rule by
splitting a protect of an array read (\ensuremath{\Varid{x}\mathbin{:=}\mathbf{protect}(\Varid{a}{[}\Varid{e}{]})})
into a separate assignment of the array value (\ensuremath{\Varid{x}\mathbin{:=}\Varid{a}{[}\Varid{e}{]}}) and
protect of the variable (\ensuremath{\Varid{x}\mathbin{:=}\mathbf{protect}(\Varid{x})}).
Rules \srule{Exec-Protect$_1$} and \srule{Exec-Protect$_2$} extend the
auxiliary relation \ensuremath{\rightsquigarrow}.
Rule \srule{Exec-Protect$_1$} evaluates the expression (\ensuremath{\Varid{v}\mathrel{=}\llbracket \Varid{e}\rrbracket^{\rho}}) and reinserts the instruction in the reorder buffer as if it were
a normal assignment.
However, the processor leaves the value wrapped inside the protect
instruction in the reorder buffer, i.e., \ensuremath{\Varid{x}\mathbin{:=}\mathbf{protect}(\Varid{v})}, to prevent
forwarding the value to the later instructions via the the transient
variable map.
When no guards are pending in the reorder buffer (\ensuremath{\mathbf{guard}(\anonymous ,\anonymous ,\anonymous )\;\not\in\;\Varid{is}_{1}}), rule \srule{Exec-Protect$_2$} transforms the instruction into a
normal assignment, so that the processor can propagate and commit its
value.

\mypara{Example}
Consider again \exNum{1} and the execution shown in
\Cref{fig:ex1-execution}. In the repaired program, \ensuremath{\Varid{x}\mathbin{+}\Varid{y}} is wrapped in
a $\mathbf{protect}$ statement.  As a result, directive \ensuremath{\mathbf{exec}\;\mathrm{5}}
produces value \ensuremath{\Varid{z}\mathbin{:=}\mathbf{protect}(\sourceColorOf{\mathrm{42}})}, instead of \ensuremath{\Varid{z}\mathbin{:=}\sourceColorOf{\mathrm{42}}}
which prevents instruction \ensuremath{\mathrm{7}} from executing (as its target address
is undefined), until all guards are resolved. This in turn
prevents leaking the transient value.

\mypara{Protect in Software}
%
%
%
The software implementation of \ensuremath{\mathbf{protect}} applies SLH to array reads.
Intuitively, we rewrite array reads by injecting \emph{artificial}
data-dependencies between bounds-check conditions and the
corresponding addresses in load instructions, thus transforming
control-flow dependencies into data-flow dependencies.\footnote{\New{Technically, applying SLH to array reads is a program transformation.
%
%
In our implementation (\S~\ref{sec:impl}), the SLH version of \ensuremath{\mathbf{protect}} inserts additional instructions into the program to perform the conditional update and mask operation described in rule \srule{Fetch-Protect-SLH}.}}
These data-dependencies validate control-flow decisions at runtime by
stalling speculative loads until the processor resolves their bounds
check conditions.\footnote{A fully-fledged security tool could apply static analysis techniques to infer array bounds.
Our implementation of \tool merely simulates SLH using a static constant instead of the actual lengths, as discussed in \Cref{sec:impl}. }
Formally, we replace rule \srule{Fetch-Protect-Array} with rule
\srule{Fetch-Protect-SLH} in Figure \ref{fig:fetch-slh}.
The rule computes the bounds check condition \ensuremath{\Varid{e}_{1}\mathrel{=}\Varid{e}<\mathit{length}(\Varid{a})},
the target address \ensuremath{\Varid{e}_{2}\mathrel{=}\Varid{base}(\Varid{a})\mathbin{+}\Varid{e}}, and generates commands that
abort the execution if the check fails, like for regular array reads.
%
%
Additionally, the rule generates \emph{regular commands} that (i)
assign the result of the bounds check to a reserved variable \ensuremath{\Varid{m}} (\ensuremath{\Varid{c}_{1}\mathrel{=}\Varid{m}\mathbin{:=}\Varid{e}_{1}}), (ii) conditionally update the variable with a bitmask
consisting of all 1s or 0s
 (\ensuremath{\Varid{c}_{2}\mathrel{=}\Varid{m}\mathbin{:=}\Varid{m}\mathbin{?}\bm{\mathrm{1}}\mathbin{:}\bm{\mathrm{0}}}), and (iii) mask off the target address with the bitmask (\ensuremath{\Varid{c}_{3}\mathrel{=}\Varid{x}\mathbin{:=}\mathbin{*}(\Varid{e}_{2}\;\otimes\;\Varid{m})}).\footnote{Alternatively, it would be also possible to mask the loaded value, i.e., \ensuremath{\Varid{c}_{3}\mathrel{=}\Varid{x}\mathbin{:=}(\mathbin{*}\Varid{e}_{2})\;\otimes\;\Varid{m}}.
However, this alternative mitigation would still introduce transient data in various processor internal buffers, where it could be leaked.
In contrast, we conservatively mask the address, which has the effect of stalling the load and preventing transient data from even entering the processor, thus avoiding the risk of leaking it altogether.
}
Since the target address in command \ensuremath{\Varid{c}_{3}} \emph{depends} on variable
\ensuremath{\Varid{m}}, the processor cannot read memory until the bounds check is
resolved.
If the check succeeds, the bitmask \ensuremath{\Varid{m}\mathrel{=}\bm{\mathrm{1}}} leaves the target
address unchanged (\ensuremath{\llbracket \Varid{e'}\rrbracket^{\rho}\mathrel{=}\llbracket \Varid{e'}\;\otimes\;\bm{\mathrm{1}}\rrbracket^{\rho}}) and
the processor reads the correct address normally.
Otherwise, the bitmask \ensuremath{\Varid{m}\mathrel{=}\bm{\mathrm{0}}} zeros out the target address
and the processor loads speculatively only from the constant address
\ensuremath{\mathrm{0}\mathrel{=}\llbracket \Varid{e'}\;\otimes\;\bm{\mathrm{0}}\rrbracket^{\rho}}.
(We assume that the processor reserves the first memory cell and
initializes it with a dummy value, e.g., \ensuremath{\mu(\mathrm{0})\mathrel{=}\mathrm{0}}.)
Notice that this solution works under the assumption that the
processor does not evaluate the conditional update \ensuremath{\Varid{m}\mathbin{:=}\Varid{m}\mathbin{?}\bm{\mathrm{1}}\mathbin{:}\bm{\mathrm{0}}} speculatively.
We can easily enforce that by compiling conditional updates to
non-speculative instructions available on commodity processors (e.g.,
the conditional move instruction CMOV on x86).
%
%

\mypara{Example}
Consider again \exNum{1}.
The optimal patch \ensuremath{\mathbf{protect}(\Varid{x}\mathbin{+}\Varid{y})} cannot be executed on existing
processors without support for a generic \ensuremath{\mathbf{protect}} primitive.
%
%
Nevertheless, we can repair the program by applying SLH to the
individual array reads, i.e., \ensuremath{\Varid{x}\mathbin{:=}\mathbf{protect}(\Varid{a}{[}\Varid{i}_{1}{]})} and \ensuremath{\Varid{y}\mathbin{:=}\mathbf{protect}(\Varid{a}{[}\Varid{i}_{2}{]})}.

\section{Type System and Inference}
\label{sec:type-system}
%
%
In \Cref{sec:types}, we present a transient-flow type system which
statically rejects programs that can potentially leak through transient execution attacks.
These speculative leaks arise in programs as transient data flows from source to sink expressions.
\New{
Our type system does \emph{not} rely on user annotations to identify secret.
Indeed, our typing rules simply ignore security annotations and instead,
 conservatively reject programs that exhibit any source-to-sink
data flows.
Intuitively, this is because security annotations are not trustworthy when
programs are executed speculatively: public variables could contain transient
secrets and secret variables could flow speculatively into public sinks.
}
%
\New{
Additionally, our soundness theorem, which states that well-typed
programs are \emph{speculatively} constant-time, assumes that programs are
sequentially constant-time.
We do \emph{not} enforce a \emph{sequential} constant-time discipline; this can
be done using existing constant-time type systems or interfaces (e.g.,~\cite{Watt:2019, HACL}).
}
Given an unannotated program, we apply constraint-based type
inference~\cite{Aiken96,Nielson98} to generate its use-def graph and
reconstruct type information~(\Cref{sec:fence-inference}).
Then, reusing off-the-shelf Max-Flow/Min-Cut algorithms, we analyze
the graph and locate potential speculative vulnerabilities in the form
of a variable min-cut set.
Finally, using a simple program repair algorithm we patch the program
by inserting a minimum number of \ensuremath{\mathbf{protect}} so that it cannot leak
speculatively anymore~(\Cref{sec:repair}).

\begin{figure}[t]
\centering
\begin{subfigure}{\columnwidth}
\rulesize
\begin{mathpar}
\inferrule[Value]
{}
{\ensuremath{\Gamma\vdash\Varid{v}\mathbin{:}\tau} \ \shadedCons{\cnsSep\ \emptyset}}
\and
\inferrule[Var]
{\ensuremath{\Gamma(\Varid{x})\mathrel{=}\tau}}
{\ensuremath{\Gamma\vdash\Varid{x}\mathbin{:}\tau} \; \shadedCons{\cnsSep\ \ x\ \flows\ \alpha_x} }
\and
\inferrule[Bop]
{ \ensuremath{\Gamma\vdash\Varid{e}_{\mathrm{1}}\mathbin{:}\tau_{1}\;\shadedCons{\cnsSep\;\Varid{k}_{\mathrm{1}}}} \\
  \ensuremath{\Gamma\vdash\Varid{e}_{\mathrm{2}}\mathbin{:}\tau_{2}\;\shadedCons{\cnsSep\;\Varid{k}_{\mathrm{2}}}} \\
  \ensuremath{\tau_{\mathrm{1}}\;\flows\;\tau} \\
  \ensuremath{\tau_{\mathrm{2}}\;\flows\;\tau}
}
{\ensuremath{\Gamma\vdash\Varid{e}_{1}\;\oplus\;\Varid{e}_{2}\mathbin{:}\tau}  \; \shadedCons{\cnsSep\ k_1\  \cup\ k_2\
    \cup\ (e_1\ \flows\ e_1 \oplus e_2 )\
     \cup\ (e_2\ \flows\ e_1 \oplus e_2 )}
}
\and
\inferrule[Array-Read]
{  \ensuremath{\Gamma\vdash\Varid{e}\mathbin{:}\ensuremath{\Concrete}} \; \shadedCons{\cnsSep\ k}}
{\ensuremath{\Gamma\vdash\Varid{a}{[}\Varid{e}{]}\mathbin{:}\ensuremath{\Transient}} \; \shadedCons{\cnsSep\ k
    \cup\ ( e\ \flows\ \Concrete )\ \cup\  (\Transient\ \flows\ a[e])}
}
\end{mathpar}
\captionsetup{format=caption-with-line}
\caption{Typing rules for expressions and arrays.}
\end{subfigure}%
\\
\begin{subfigure}{\columnwidth}
\rulesize
\begin{mathpar}
\inferrule[Asgn]
{\ensuremath{\Gamma\vdash\Varid{r}\mathbin{:}\tau} \ \shadedCons{\cnsSep\ k} \qquad
 \ensuremath{\tau\;\flows\;\Gamma(\Varid{x})}}
{\ensuremath{\Gamma,\protectedSet\vdash\Varid{x}\mathbin{:=}\Varid{r}} \ \shadedCons{\cnsSep\ k\ \cup\ (r\ \flows\ x)}}
\and
\inferrule[Protect]
{\ensuremath{\Gamma\vdash\Varid{r}\mathbin{:}\tau} \ \shadedCons{\cnsSep\ k}}
{\ensuremath{\Gamma,\protectedSet\vdash\Varid{x}\mathbin{:=}\mathbf{protect}(\Varid{r})} \ \shadedCons{\cnsSep\ k}}
\and
\inferrule[Asgn-Prot]
{\ensuremath{\Gamma\vdash\Varid{r}\mathbin{:}\tau} \ \shadedCons{\cnsSep\ k} \qquad
\ensuremath{\Varid{x}\;\in\;\protectedSet} }
{\ensuremath{\Gamma,\protectedSet\vdash\Varid{x}\mathbin{:=}\Varid{r}} \ \shadedCons{\cnsSep\ k\ \cup\ (r \flows x)}}
\and
\inferrule[Array-Write]
{\ensuremath{\Gamma\vdash\Varid{e}_{1}\mathbin{:}\ensuremath{\Concrete}} \ \shadedCons{\cnsSep k_1} \\
 \ensuremath{\Gamma\vdash\Varid{e}_{2}\mathbin{:}\tau} \ \shadedCons{\cnsSep k_2} }
{\ensuremath{\Gamma,\protectedSet\vdash\Varid{a}{[}\Varid{e}_{1}{]}\mathbin{:=}\Varid{e}_{2}} \ \shadedCons{\cnsSep k_1 \cup k_2 \cup (e_1 \flows \Concrete)}}
\and
\inferrule[If-Then-Else]
{\ensuremath{\Gamma\vdash\Varid{e}\mathbin{:}\ensuremath{\Concrete}} \; \shadedCons{\cnsSep\ k} \\ \ensuremath{\Gamma,\protectedSet\vdash\Varid{c}_{\mathrm{1}}} \; \shadedCons{\cnsSep\ k_1} \\ \ensuremath{\Gamma,\protectedSet\vdash\Varid{c}_{\mathrm{2}}} \; \shadedCons{\cnsSep\ k_2}}
{\ensuremath{\Gamma,\protectedSet\vdash\mathbf{if}\;\Varid{e}\;\mathbf{then}\;\Varid{c}_{1}\;\mathbf{else}\;\Varid{c}_{2}} \; \shadedCons{\cnsSep\ k\ \cup\
  k_1\ \cup\ k_2\ \cup\ (e \flows \Concrete)}}
\end{mathpar}
\captionsetup{format=caption-with-line}
\caption{Typing rules for commands.\label{fig:transient-ts-lower}}
\end{subfigure}
\caption{Transient-flow type system and \colorbox{constColor}{constraints generation}.\label{fig:transient-ts}}
\end{figure}

\subsection{Type System}
\label{sec:types}
%
%
%
Our type system assigns a \emph{transient-flow type} to expressions
and tracks how transient values propagate within programs, rejecting
programs in which transient values reach commands which may leak them.
An expression can either be typed as \sinkColorOf{\emph{stable}}
(\ensuremath{\ensuremath{\Concrete}}) indicating that it cannot contain transient values during
execution, or as \sourceColorOf{\emph{transient}} (\ensuremath{\ensuremath{\Transient}}) indicating
that it can. 
These types form a 2-point lattice~\cite{Landauer-Lattice}, which
allows stable expressions to be typed as transient, but not vice
versa, i.e., we define a can-flow-to relation \ensuremath{\flows} such that
\ensuremath{\ensuremath{\Concrete}\;\flows\;\ensuremath{\Transient}}, but \ensuremath{\ensuremath{\Transient}\;\not\flows\;\ensuremath{\Concrete}}.
\mypara{Typing Expressions}
%
Given a typing environment for variables \ensuremath{\Gamma\;\in\;\Conid{Var}\to \{\mskip1.5mu \ensuremath{\Concrete},\ensuremath{\Transient}\mskip1.5mu\}}, the typing judgment \ensuremath{\Gamma\vdash\Varid{r}\mathbin{:}\tau} assigns a
transient-flow type \ensuremath{\tau} to \ensuremath{\Varid{r}}.
%
%
Figure \ref{fig:transient-ts} presents selected rules \ifextended
(see Appendix~\ref{app:type-inference} for the rest).
\else
\New{(see~\cite{vassena2020automatically} for the rest).}
\fi
The \colorbox{constColor}{shaded} part of the rules generates type
constraints during type inference and are explained later.
Values can assume any type in rule \srule{Value} and variables are assigned
their respective type from the environment in rule \srule{Var}.
Rule \srule{Bop} propagates the type of the operands to the result of
binary operators \ensuremath{\oplus\;\in\;\{\mskip1.5mu \mathbin{+},<,\otimes\mskip1.5mu\}}.
%
%
Finally, rule \srule{Array-Read} assigns the \sourceColorOf{\emph{transient}}
type \ensuremath{\ensuremath{\Transient}} to array reads as the array may potentially be indexed out of
bounds during speculation.
Importantly, the rule requires the index expression to be typed \sinkColorOf{\emph{stable}} (\ensuremath{\ensuremath{\Concrete}}) to prevent programs from leaking through the
cache.
%


\mypara{Typing Commands}
Given a set of \New{implicitly} protected variables \ensuremath{\protectedSet}, we define a typing judgment
\ensuremath{\Gamma,\protectedSet\vdash\Varid{c}} for commands.
Intuitively, a command \ensuremath{\Varid{c}} is well-typed under environment \ensuremath{\Gamma} and
set \ensuremath{\protectedSet}, if \ensuremath{\Varid{c}} does not leak, under the assumption that the
expressions assigned to all variables in \ensuremath{\protectedSet} are protected using
the \ensuremath{\mathbf{protect}} primitive.
\Cref{fig:transient-ts-lower} shows our typing rules.
Rule \srule{Asgn} disallows assignments from
\sourceColorOf{\emph{transient}} to \sinkColorOf{\emph{stable}} variables (as \ensuremath{\ensuremath{\Transient}\;\not\flows\;\ensuremath{\Concrete}}).
Rule \srule{Protect} relaxes this policy as long as the right-hand
side is explicitly protected.\footnote{Readers familiar with
  information-flow control may see an analogy between \ensuremath{\mathbf{protect}} and
  the \textbf{declassify} primitive of some IFC
  languages~\cite{declassify}.}
Intuitively, the result of \ensuremath{\mathbf{protect}} is \sinkColorOf{\emph{stable}} and it
can thus flow securely to variables of any type.
%
%
Rule \srule{Asgn-Prot} is similar, but instead of requiring an
explicit \ensuremath{\mathbf{protect}} statement, it demands that the variable is
accounted for in the protected set \ensuremath{\protectedSet}.
This is secure because all assignments to variables in \ensuremath{\protectedSet} will
eventually be protected through the repair function discussed later in
this section.
%
%
%
Rule \srule{Array-Write} requires the index expression used in array writes to be typed \sinkColorOf{\emph{stable}} (\ensuremath{\ensuremath{\Concrete}}) to avoid leaking through the cache, similarly to rule \srule{Array-Read}.
Notice that the rule does not prevent storing transient data, i.e., the stored value can have any type.
While this is sufficient to mitigate Spectre v1 attacks, it is inadequate to defend against Spectre v1.1.
Intuitively, the store-to-load forward optimization  enables additional \emph{implicit} data-flows between stored data to aliasing load instructions, thus enabling Spectre v1.1 attacks.
Luckily, to protect against these attacks we need only modify one clause of rule \srule{Array-Write}: we change the type system to conservatively treat stored values as \emph{sinks}  and therefore require them
to be typed \sinkColorOf{\emph{stable}} (\ensuremath{\Gamma\vdash\Varid{e}_{2}\mathbin{:}\ensuremath{\Concrete}}).\looseness=-1\footnote{\New{Both variants of rule \srule{Array-Write} are implemented in \tool (\S~\ref{sec:impl}) and evaluated (\S~\ref{sec:eval}).}}

\mypara{Implicit Flows}
To prevent programs from leaking data \emph{implicitly} through their
control flow, rule \srule{If-Then-Else} requires the branch condition
to be \sinkColorOf{\emph{stable}}.
This might seem overly restrictive, at first: why can't we accept a
program that branches on transient data, as long as it does not
perform any attacker-observable operations (e.g., memory reads and
writes) along the branches?
Indeed, classic information-flow control (IFC) type systems (e.g.,
\cite{Volpano:Smith:Irvine:Sound}) take this approach by keeping track
of an explicit program counter label.
Unfortunately, such permissiveness is \emph{unsound} under
speculation.
Even if a branch does not contain observable behavior, the value of
the branch condition can be leaked by the instructions that
\emph{follow} a mispredicted branch.
\New{In particular, upon a rollback, the processor may
\emph{repeat} some load and store instructions after the mispredicted branch and thus generate additional observations, which can implicitly reveal the value of the
branch condition.}\looseness=-1
%
%

%
%
%

\mypara{Example}
Consider the program
%
\ensuremath{\{\mskip1.5mu \mathbf{if}\;\sourceColorOf{\Varid{tr}}}\ \ensuremath{\mathbf{then}}\ \ensuremath{\Varid{x}\mathbin{:=}\mathrm{0}}\ \ensuremath{\mathbf{else}}\ \ensuremath{\mathbf{skip}\mskip1.5mu\};}\ \ensuremath{\Varid{y}\mathbin{:=}\Varid{a}{[}\mathrm{0}{]}}.
%
The program can leak the transient value of \ensuremath{\sourceColorOf{\Varid{tr}}} during speculative execution.
To see that, assume that the processor predicts that \ensuremath{\sourceColorOf{\Varid{tr}}} will
evaluate to \ensuremath{\mathbf{true}}.
Then, the processor speculatively executes the then-branch (\ensuremath{\Varid{x}\mathbin{:=}\mathrm{0}}) and
the load instruction (\ensuremath{\Varid{y}\mathbin{:=}\Varid{a}{[}\mathrm{0}{]}}), before resolving the
condition.
If \ensuremath{\sourceColorOf{\Varid{tr}}} is \ensuremath{\mathbf{true}}, the observation trace of the program contains a
single read observation.
However, if \ensuremath{\sourceColorOf{\Varid{tr}}} is \ensuremath{\mathbf{false}}, the processor detects a
misprediction, restarts the execution from the other branch (\ensuremath{\mathbf{skip}})
and executes the array read \emph{again}, producing a rollback and
\emph{two} read observations.
From these observations, an attacker could potentially make inferences
about the value of \ensuremath{\sourceColorOf{\Varid{tr}}}.
Consequently, if \ensuremath{\sourceColorOf{\Varid{tr}}} is typed as \ensuremath{\ensuremath{\Transient}}, our type system rejects
the program as unsafe.


\subsection{Type Inference}
\label{sec:fence-inference}
%
Our type-inference approach is based on type-constraints
satisfaction~\cite{Aiken96,Nielson98}.
Intuitively, type constraints restrict the types that variables and
expressions may assume in a program.
In the constraints, the possible types of variables and expressions
are represented by \emph{atoms} consisting of unknown types of
expressions and variables.
%
Solving these constrains requires finding a \emph{substitution}, i.e.,
a mapping from atoms to concrete transient-flow types, such that all
constraints are \emph{satisfied} if we instantiate the atoms with
their type.
Our type inference algorithm consists of 3 steps: (i) generate a set
of type constraints under an initial typing environment and protected
set that under-approximates the solution of the constraints, (ii)
construct the def-use graph from the constraints and find a cut-set,
and (iii) cut the transient-to-stable dataflows in the graph and
compute the resulting typing environment.
We start by describing the generation of constraints through the typing
judgment from \Cref{fig:transient-ts}.

\mypara{Type Constraints}
%
%
Given a typing environment \ensuremath{\Gamma}, a protected set \ensuremath{\protectedSet}, the
judgment \ensuremath{\Gamma,\protectedSet\vdash\Varid{r}\cnsSep\Varid{k}} type checks \ensuremath{\Varid{r}} and generates type
constraints \ensuremath{\Varid{k}}.
The syntax for constraints is shown in \Cref{fig:typing-constraints}.
Constraints are sets of can-flow-to relations involving concrete types
(\ensuremath{\ensuremath{\Concrete}} and \ensuremath{\ensuremath{\Transient}}) and \emph{atoms}, i.e., type variables corresponding
to program variables (e.g., \ensuremath{\alpha_{\Varid{x}}} for \ensuremath{\Varid{x}}) and unknown types
for expressions (e.g., \ensuremath{\Varid{r}}).
%
%
%
In rule \srule{Var}, constraint \ensuremath{\Varid{x}\;\flows\;\alpha_{\Varid{x}}} indicates
that the type variable of \ensuremath{\Varid{x}} should be at least as transient as the
unknown type of \ensuremath{\Varid{x}}.
This ensures that, if variable \ensuremath{\Varid{x}} is transient, then \ensuremath{\alpha_{\Varid{x}}}
can only be instantiated with type \ensuremath{\ensuremath{\Transient}}.
Rule \srule{Bop} generates constraints \ensuremath{\Varid{e}_{1}\;\flows\;\Varid{e}_{1}\;\oplus\;\Varid{e}_{2}} and
\ensuremath{\Varid{e}_{2}\;\flows\;\Varid{e}_{1}\;\oplus\;\Varid{e}_{2}} to reflect the fact that the unknown type
of \ensuremath{\Varid{e}_{1}\;\oplus\;\Varid{e}_{2}} should be at least as transient as the (unknown) type
of \ensuremath{\Varid{e}_{1}} and \ensuremath{\Varid{e}_{2}}.
Notice that these constraints correspond exactly to the premises \ensuremath{\tau_{1}\;\flows\;\tau} and \ensuremath{\tau_{2}\;\flows\;\tau} of the same rule.
Similarly, rule \srule{Array-Read} generates constraint \ensuremath{\Varid{e}\;\flows\;\ensuremath{\Concrete}} for the unknown type of the array index, thus forcing it to be typed \ensuremath{\ensuremath{\Concrete}}.\footnote{No constraints are generated for the type of the array because
our syntax forces the array to be a \emph{static}, constant value.
If we allowed arbitrary expressions for arrays the rule would require them to by typed as stable.}
In addition to these, the rule generates also the constraint \ensuremath{\ensuremath{\Transient}\;\flows\;\Varid{a}{[}\Varid{e}{]}}, which forces the type of \ensuremath{\Varid{a}{[}\Varid{e}{]}} to be
\ensuremath{\ensuremath{\Transient}}.
Rule \srule{Asgn} generates the constraint \ensuremath{\Varid{r}\;\flows\;\Varid{x}} disallowing
transient to stable assignments.
%
%
In contrast, rule \srule{Protect} does \emph{not} generate the
constraint \ensuremath{\Varid{r}\;\flows\;\Varid{x}} because \ensuremath{\Varid{r}} is explicitly protected.
Rule \srule{Asgn-Prot} generates the same constraint as rule
\srule{Asgn}, because type inference ignores the protected set, which
is computed in the next step of the algorithm.
The constraints generated by the other rules follow the same
intuition.
In the following we describe the inference algorithm in more detail.\looseness=-1
%
%

%
\mypara{Generating Constraints}
We start by collecting a set of constraints~$\constrColorOf{k}$ via
typing judgement~$\Gamma, \mathsf{Prot} \vdash s \ \cnsSep
k$.
For this, we define a dummy environment~$\Gamma^ \ast$ and protected
set~$\protectedSet^\ast$, such that $\Gamma^\ast, \protectedSet^\ast
\vdash c \ \cnsSep k$ holds for any command~$c$, (\ie we let \ensuremath{\Gamma^{\ast}\mathrel{=}\lambda \Varid{x}.\ensuremath{\Concrete}} and include \emph{all} variables in the cut-set)
and use it to extract the set of constraints~$\constrColorOf{k}$.
%
%


\mypara{Solutions and Satisfiability}
We define the solution to a set of constraints as a function~$\sigma$
from atoms to flow types, \ie $\sigma \in \textsc{Atoms} \mapsto
\theSet{\Transient, \Concrete}$, and extend solutions to map
$\Transient$ and $\Concrete$ to themselves. For
a set of constraints~$k$ and a solution function~$\sigma$,
we write $\sigma \vdash k$ to say that the constraints $k$ are
satisfied under solution~$\sigma$.
A solution $\sigma$ satisfies $k$, if all can-flow-to constraints
hold, when the atoms are replaced by their values under $\sigma$ (Fig.~\ref{fig:typing-constraints}).
We say that a set of constraints~$k$ is satisfiable, if there is a
solution~$\sigma$ such that~$\sigma \vdash k$.

\mypara{Def-Use Graph \& Paths}
The constraints generated by our type
system give rise to the def-use graph of the type-checked program.
For a set of constraints~$k$, we call a sequence of atoms~$a_1 \dots
a_n$ a \emph{path} in $k$, if $a_i \flows a_{i+1} \in k$ for $i \in \{1,
\dots, n-1\}$ and say that $a_1$ is the path's entry and $a_n$ its exit. A
$\Transient$-$\Concrete$ path is a path with entry~$\Transient$ and
exit $\Concrete$. A set of constraints~$k$ is satisfiable if and only
if there is no $\Transient$-$\Concrete$ path in $k$, as such a path
would correspond to a derivation of \emph{false}.
If $k$ is satisfiable, we can compute a solution $\sigma(k)$ by
letting $\sigma(k)(a) = \Transient$, if there is a path with
entry~$\Transient$ and exit $a$, and $\Concrete$ otherwise.\looseness=-1

\mypara{Cuts}
If a set of constraints is unsatisfiable, we can make it satisfiable
by removing some of the nodes in its graph or equivalently protecting
some of the variables.
A set of atoms~$A$ \emph{cuts} a path $a_1 \dots a_n$, if some $a \in
A$ occurs along the path, \ie there exists $a \in A$ and $i \in \{1,
\dots, n\}$ such that $a_i=a$.
We call~$A$ a cut-set for a set of constraints~$k$, if $A$ cuts
all~$\Transient$-$\Concrete$ paths in~$k$.
A cut-set $A$ is minimal for $k$, if all other cut-sets~$A'$ contain
as many or more atoms than $A$, \ie $\#A \leq \#A'$.
%

\mypara{Extracting Types From Cuts}
From a set of variables~$A$ such that $A$ is a cut-set of
constraints~$k$, we can extract a typing environment~$\Gamma(k,A)$ as
follows: for an atom~$\alpha_x$, we define $\Gamma (k,A)(x) =
\Transient$, if there is a path with entry $\Transient$ and
exit~$\alpha_x$ in~$k$ that is not cut by $A$, and let $\Gamma
(k,A)(x) = \Concrete$ otherwise.
%
\begin{prop}[Type Inference]
If $\ \Gamma^\ast, \protectedSet^\ast \vdash c \
\cnsSep k$ and $A$ is a set of variables that cut $k$,
then $\Gamma (k,A), A \vdash c$.
\label{prop:typeInfer}
\end{prop}
%
\mypara{Remark}
To infer a repair using exclusively SLH-based \ensuremath{\mathbf{protect}} statements,
we simply restrict our cut-set to only include variables that are
assigned from an array read.

\mypara{Example}
Consider again \exNum{1} in \Cref{fig:ex1-code}.
The graph defined by the constraints $k$, given by $\Gamma^\ast,
\protectedSet^\ast \vdash \exNum{1} \ \cnsSep k$ is shown in
\Cref{fig:ex1:data-flow}, where we have omitted $\alpha$-nodes.
The constraints are not satisfiable, since there are
$\Transient$-$\Concrete$ paths.
Both $\theSet{x,y}$ and $\theSet{z}$ are cut-sets, since they cut each
$\Transient$-$\Concrete$ path, however, the set $\theSet{z}$ contains
only one element and is therefore minimal.
The typing environment $\Gamma(k, \theSet{x,y})$ extracted from the
sub-obptimal cut $\theSet{x,y}$ types all variables as $\Concrete$,
while the typing extracted from the optimal cut, i.e., $\Gamma(k,
\theSet{z})$ types $x$ and $y$ as $\Transient$ and $z$, $i_1$ and
$i_2$ as $\Concrete$.
By \Cref{prop:typeInfer} both $\Gamma(k, \theSet{x,y}), \theSet{x,y}
\vdash \exNum{1}$ and $\Gamma(k, \theSet{z}), \theSet{z} \vdash
\exNum{1}$ hold.

\begin{figure}[t]
\centering
\begin{subfigure}{.45\textwidth}
\rulesize
\begin{align*}
\text{Atoms }&   a\ \ensuremath{\Coloneqq} \  \alpha_x \mid\ r \\
\text{Constraints }&  k\ \ensuremath{\Coloneqq} \ a\ \flows\ \Concrete\ \mid\ \Transient\ \flows\ a\ \mid\ a\ \flows\ a \\
                   &  \mid\ k\ \cup\ k\ \mid\ \emptyset \\
\text{Solutions }&  \sigma\ \in\  \textsc{Atoms}\ \uplus\ \ensuremath{\{\mskip1.5mu \ensuremath{\Concrete},\ensuremath{\Transient}\mskip1.5mu\}}\ \mapsto\ \ensuremath{\{\mskip1.5mu \ensuremath{\Concrete},\ensuremath{\Transient}\mskip1.5mu\}} \\
 & \qquad \text{ where } \ensuremath{\sigma(\ensuremath{\Concrete})\mathrel{=}\ensuremath{\Concrete}} \text{ and } \ensuremath{\sigma(\ensuremath{\Transient})\mathrel{=}\ensuremath{\Transient}}
\end{align*}
\end{subfigure}%
\begin{subfigure}{.55\textwidth}
\rulesize
\begin{mathpar}
\inferrule[Sol-Empty]
{}
{\sigma\ \vdash\ \emptyset}
\and
\inferrule[Sol-Set]
{\sigma\ \vdash\ k_1 \quad \sigma\ \vdash\ k_2}
{\sigma\ \vdash\ k_1\ \cup\ k_2}
\and
\inferrule[Sol-Trans]
{\Transient\ \flows\ \sigma(a_2)}
{\sigma\ \vdash\ \Transient\ \flows\ a_2}
\and
\inferrule[Sol-Stable]
{\sigma(a_1)\ \flows\ \Concrete}
{\sigma\ \vdash\ a_1\ \flows\ \Concrete}
\and
\inferrule[Sol-Flow]
{\sigma(a_1)\ \flows\ \sigma(a_2)}
{\sigma\ \vdash\ a_1\ \flows\ a_2}
\end{mathpar}%
\end{subfigure}
\caption{Type constraints and satisfiability.}
\label{fig:typing-constraints}
\end{figure}

\subsection{Program Repair}
\label{sec:repair}
As a final step, our repair algorithm \ensuremath{ repair(\Varid{c},\protectedSet) }
traverses program \ensuremath{\Varid{c}} and inserts a \ensuremath{\mathbf{protect}} statement for each
variable in the cut-set \ensuremath{\protectedSet}.
For simplicity, we assume that programs are in static single assignment (SSA)
form.\footnote{\New{We make this assumption to simplify our analysis and security proof. We also omit phi-nodes from our calculus to avoid cluttering the semantics and remark that this simplification does not affect our \emph{flow-insensitive} analysis. Our implementation operates on Cranelift’s SSA form (\S~\ref{sec:impl}).
}}
Therefore, for each variable \ensuremath{\Varid{x}\;\in\;\protectedSet} there is a single assignment~\ensuremath{\Varid{x}\mathbin{:=}\Varid{r}}, and our repair algorithm simply replaces it with \ensuremath{\Varid{x}\mathbin{:=}\mathbf{protect}(\Varid{r})}.
%


%
%

\section{Consistency and Security}
\label{sec:correctness}
%
We now present two formal results about our speculative semantics and
the security of our type system.
First, we prove that the semantics from
Section \ref{sec:semantics} is \emph{consistent} with sequential
program execution (Theorem \ref{thm:consistent}).
Intuitively, programs running on our processor produce the same
results (with respect to the memory store and variables) as if their
commands were executed in-order and without speculation.
The second result establishes that our type system is \emph{sound}
(Theorem \ref{thm:ts-sound}).
We prove that our transient-flow type system in combination with a
standard constant-time type system
(e.g.,~\cite{Watt:2019,libsignalSP}) enforces constant time under
speculative execution~\cite{pitchfork}.
%
%
We provide full definitions and proofs in
\ifextended
Appendix \ref{app:proofs}.
\else
\New{the extended version of this paper~\cite{vassena2020automatically}.}
\fi

%

%

\mypara{Consistency}
We write \ensuremath{\Conid{C}\;\Downarrow_{\Conid{O}}^{\Conid{D}}\;\Conid{C}'} for the complete \emph{speculative}
execution of configuration \ensuremath{\Conid{C}} to final configuration \ensuremath{\Conid{C}'}, which
generates a trace of observations \ensuremath{\Conid{O}} under list of directives \ensuremath{\Conid{D}}.
Similarly, we write \ensuremath{\langle\mu,\rho\rangle\Downarrow_{\Conid{O}}^{\Varid{c}}\langle\mu',\rho'\rangle} for the
\emph{sequential} execution of program \ensuremath{\Varid{c}} with initial memory \ensuremath{\mu}
and variable map \ensuremath{\rho} resulting in final memory \ensuremath{\mu'} and variable
map \ensuremath{\rho'}.
%
%
%
To relate speculative and sequential observations, we define a
projection function, written \ensuremath{\Conid{O}{\downarrow}}, which removes prediction
identifiers, rollbacks, and misspeculated loads and stores.
%


%
\begin{thm}[Consistency]
\label{thm:consistent}
For all programs \ensuremath{\Varid{c}}, initial memory stores \ensuremath{\mu}, variable maps \ensuremath{\rho},
and directives \ensuremath{\Conid{D}}, if \ensuremath{\langle\mu,\rho\rangle\Downarrow_{\Conid{O}}^{\Varid{c}}\langle\mu',\rho'\rangle} and \ensuremath{\langle[\mskip1.5mu \mskip1.5mu],[\mskip1.5mu \Varid{c}\mskip1.5mu],\mu,\rho\rangle\Downarrow_{\Conid{O}'}^{\Conid{D}}\langle[\mskip1.5mu \mskip1.5mu],[\mskip1.5mu \mskip1.5mu],\mu'',\rho''\rangle}, then \ensuremath{\mu'\mathrel{=}\mu''}, \ensuremath{\rho'\mathrel{=}\rho''}, and \ensuremath{\Conid{O}\cong\Conid{O}'{\downarrow}}.
\end{thm}
%
%
The theorem ensures equivalence of the final memory stores,
variable maps, and observation traces from the sequential and the
speculative execution.
Notice that trace equivalence is up to \emph{permutation}, i.e., \ensuremath{\Conid{O}\cong\Conid{O}'{\downarrow}}, because the processor can execute load and store
instructions out-of-order.
%
%

\mypara{Speculative Constant Time}
%
%
In our model, an attacker can leak information through the
architectural state (i.e., the variable map and the memory store) and
through the cache by supplying directives that force the execution of
an otherwise constant-time cryptographic program to generate different
traces.
\New{In the following, the relation \ensuremath{\approx_{\Blue{\Conid{L}}}} denotes \ensuremath{\Blue{\Conid{L}}}-equivalence, i.e.,
equivalence of configurations
with respect to a security policy \ensuremath{\Blue{\Conid{L}}} that specifies which variables and arrays are public (\ensuremath{\Blue{\Conid{L}}}) and attacker observable.
Initial and final configurations are \ensuremath{\Blue{\Conid{L}}}-equivalent (\ensuremath{\Conid{C}_{1}\approx_{\Blue{\Conid{L}}}\Conid{C}_{2}}) if the values of public variables in the variable maps and the content of public arrays in the memories coincide, i.e., \ensuremath{\forall\;\Varid{x}\;\in\;\Blue{\Conid{L}}}\ \ensuremath{.}\ \ensuremath{\rho_{1}(\Varid{x})\mathrel{=}\rho_{2}(\Varid{x})} and \ensuremath{\forall\;\Varid{a}\;\in\;\Blue{\Conid{L}}} and addresses \ensuremath{\Varid{n}\;\in\;\{\mskip1.5mu \Varid{base}(\Varid{a}),\mathbin{...},\Varid{base}(\Varid{a})\mathbin{+}\mathit{length}(\Varid{a})\mathbin{-}\mathrm{1}\mskip1.5mu\}}, \ensuremath{\mu_{1}(\Varid{n})\mathrel{=}\mu_{2}(\Varid{n})}, respectively.
}
%

%

%
\begin{definition}[Speculative Constant Time]
A program \ensuremath{\Varid{c}} is speculative constant time with respect to a security policy
\ensuremath{\Blue{\Conid{L}}}, written \ensuremath{SCT_{\mkern-1mu\Blue{\Conid{L}}}(\Varid{c})}, iff for all directives \ensuremath{\Conid{D}} and
initial configurations
\ensuremath{\Conid{C}_{\Varid{i}}\mathrel{=}\langle[\mskip1.5mu \mskip1.5mu],[\mskip1.5mu \Varid{c}\mskip1.5mu],\mu_{\Varid{i}},\rho_{\Varid{i}}\rangle} for \ensuremath{\Varid{i}\;\in\;\{\mskip1.5mu \mathrm{1},\mathrm{2}\mskip1.5mu\}}, if
\ensuremath{\Conid{C}_{1}\approx_{\Blue{\Conid{L}}}\Conid{C}_{2}},
\ensuremath{\Conid{C}_{1}\;\Downarrow_{\Conid{O}_{1}}^{\Conid{D}}\;\Conid{C}_{1}'}, and \ensuremath{\Conid{C}_{2}\;\Downarrow_{\Conid{O}_{2}}^{\Conid{D}}\;\Conid{C}_{2}'}, then \ensuremath{\Conid{O}_{1}\mathrel{=}\Conid{O}_{2}} and \ensuremath{\Conid{C}_{1}'\approx_{\Blue{\Conid{L}}}\Conid{C}_{2}'}.
\label{def:STC}
\end{definition}

%
%
In the definition above, we consider syntactic equivalence of traces
because both executions follow the same list of directives.
%
%
%
We now present our soundness theorem: well-typed programs satisfy
speculative constant-time.
Our approach focuses on side-channel attacks through the observation
trace and therefore relies on a separate, but standard, type system to
control leaks through the program control-flow and architectural
state.
In particular, we write \ensuremath{CT_{\mkern-1mu\Blue{\Conid{L}}}(\Varid{c})} if \ensuremath{\Varid{c}} follows the (sequential)
constant time discipline from \cite{Watt:2019,libsignalSP}, i.e., it
is free of secret-dependent branches and memory accesses.

\begin{thm}[Soundness]
  For all programs \ensuremath{\Varid{c}} and security policies \ensuremath{\Blue{\Conid{L}}}, if \ensuremath{CT_{\mkern-1mu\Blue{\Conid{L}}}(\Varid{c})} and \ensuremath{\Gamma\vdash\Varid{c}},
  then \ensuremath{SCT_{\mkern-1mu\Blue{\Conid{L}}}(\Varid{c})}.
\label{thm:ts-sound}
\end{thm}
As mentioned in \Cref{sec:type-system}, our transient flow-type system is oblivious to
the security policy \ensuremath{\Blue{\Conid{L}}}, which is only required by the constant-time type system and the definition of speculative constant~time.
%

We conclude with a corollary that combines all the components of our
protection chain (type inference, type checking and automatic repair)
and shows that repaired programs satisfy speculative constant time.
\begin{cor}
  \label{alltogether-corollary}
  For all sequential constant-time programs \ensuremath{CT_{\mkern-1mu\Blue{\Conid{L}}}(\Varid{c})}, there exists a
  set of constraints \ensuremath{\Varid{k}} such that \ensuremath{\Gamma^{\ast},\protectedSet^{\ast}\vdash\Varid{c}\cnsSep\Varid{k}}. Let \ensuremath{\Conid{A}} be a set of variables that cut \ensuremath{\Varid{k}}. Then, it
  follows that \ensuremath{SCT_{\mkern-1mu\Blue{\Conid{L}}}( repair(\Varid{c},\Conid{A}) )}.
\end{cor}

\section{Implementation}
\label{sec:impl}

We implement \tool as a compilation pass in the Cranelift~\cite{cranelift}
Wasm code-generator, which is used by the Lucet compiler and
runtime~\cite{lucet-talk}.
\tool first identifies all sources and sinks.
Then, it finds the cut points using the Max-Flow/Min-Cut
algorithm~(\S\ref{sec:fence-inference}), and either inserts fences at the cut
points, or applies SLH to all of the loads which feed the cut point in the
graph.  This difference is why SLH sometimes requires code insertions
in more locations.

Our SLH prototype implementation does not track the length of arrays, and
instead uses a static constant for all array lengths when applying masking.
Once compilers like Clang add support for conveying array length information to
Wasm (e.g., via Wasm's custom section), our compilation pass would be able to
take this information into account.
This simplification in our experiments does not affect the sequence of instructions
emitted for the SLH
masks and thus \tool's performance overhead is accurately measured.

Our Cranelift \tool pass runs after the control-flow graph has been finalized
and right before register allocation.\footnote{More precisely: The Cranelift
register allocation pass modifies the control-flow graph as an initial step; we
insert our pass after this initial step but before register allocation proper.}
Placing \tool before register allocation allows our implementation to
remain oblivious of low-level details such as register pressure and
stack spills and fills.
%
%
Ignoring the memory operations incurred by spills and fills simplifies
\tool's analysis and reduces the required number of \ensuremath{\mathbf{protect}} statements.
This, importantly, does not compromise the security of its mitigations:
In Cranelift, spills and fills are always to constant addresses which are
inaccessible to ordinary Wasm loads and stores, even speculatively.
(Cranelift uses guard pages---not conditional bounds checks---to ensure that
Wasm memory accesses cannot access anything outside the linear memory, such
as the stack used for spills and fills.)
%
%
%
%
As a result, we can treat stack spill slots like registers.
Indeed, since \tool runs before register allocation, it already traces
def-use chains across operations that will become spills and fills.
Even if a particular spill-fill sequence would handle potentially sensitive
transient data, \tool would insert a \ensuremath{\mathbf{protect}} between the original transient
source and the final transient sink (and thus mitigate the attack).

Our implementation implements a single optimization: we do not mark
constant-address loads as transient sources.
We assume that the program contains no loads from out-of-bounds constant
addresses, and therefore that loads from constant (Wasm linear memory)
addresses can never speculatively produce invalid data.
As we describe below, however, we omit this optimization when considering
Spectre v1.1.

At its core, our repair algorithm addresses Spectre v1 attacks based on
PHT mispredictions.
To also protect against Spectre variant 1.1 attacks, which exploit store forwarding in
the presence of PHT mispredictions,\footnote{
  Spectre v1 and Spectre v1.1 attacks are both classified as Spectre-PHT
  attacks~\cite{Canella:2019}.}
we perform two additional mitigations.
First, we mark constant-address loads as transient sources (and thus omit the
above optimization).
Under Spectre v1.1, a load from a constant address may speculatively
produce transient data, if a previous speculative store wrote transient data
to that constant address---and, thus, \tool must account for this.
Second, our SLH implementation marks all stored \emph{values} as sinks,
essentially preventing any transient data from being stored to memory.
This is necessary when considering Spectre v1.1 because otherwise, ensuring
that a load is in-bounds using SLH is insufficient to guarantee that the
produced data is not transient---again, a previous speculative store may have
written transient data to that in-bounds address.\looseness=-1

\section{Evaluation}
\label{sec:eval}

We evaluate \tool by answering two questions:
\emph{\textbf{(Q1)}} How many \ensuremath{\mathbf{protect}}s does \tool insert when
repairing existing programs?
\emph{\textbf{(Q2)}} What is the runtime performance overhead of eliminating
speculative leaks with \tool on existing hardware?

\newcommand{\X}{\Red{XXX}}

\begin{table}[t]
  \caption{
    \textbf{Ref}: Reference implementation with no Spectre mitigations;
    \textbf{Baseline-F}: Baseline mitigation inserting fences;
    \textbf{\tool-F}: \tool using fences as \ensuremath{\mathbf{protect}};
    \textbf{Baseline-S}: Baseline mitigation using SLH;
    \textbf{\tool-S}: \tool using SLH;
    \textbf{Overhead}: Runtime overhead compared to Ref;
    \textbf{Defs}: number of fences inserted (Baseline-F and \tool-F), or number of loads protected with SLH (Baseline-S and \tool-S)
  }
  \label{tab:evaluation}



\centering
\footnotesize
\begin{tabular}{llcccccc}

\toprule
& & \multicolumn{3}{c}{Without v1.1 protections} & \multicolumn{3}{c}{With v1.1 protections} \\
\textbf{Benchmark} & \textbf{Defense} & \textbf{Time} & \textbf{Overhead} & \textbf{Defs}
                                      & \textbf{Time} & \textbf{Overhead} & \textbf{Defs}
\\

\midrule
\multirow{5}{*}{Salsa20 (CT-Wasm), 64 bytes}    & Ref         &    4.3 us & -        & -    &    4.3 us & -        & -    \\
                                                & Baseline-F  &    4.6 us & 7.2\%    & 3    &    8.6 us & 101.7\%  & 99   \\
                                                & \tool-F     &    4.4 us & 1.9\%    & 0    &    4.3 us & 1.7\%    & 0    \\
                                                & Baseline-S  &    4.4 us & 2.7\%    & 3    &    5.3 us & 24.3\%   & 99   \\
                                                & \tool-S     &    4.3 us & 0.5\%    & 0    &    5.4 us & 26.4\%   & 99   \\
\midrule
\multirow{5}{*}{SHA-256 (CT-Wasm), 64 bytes}    & Ref         &   13.7 us & -        & -    &   13.7 us & -        & -    \\
                                                & Baseline-F  &   19.8 us & 43.8\%   & 23   &   20.3 us & 48.0\%   & 54   \\
                                                & \tool-F     &   13.8 us & 0.2\%    & 0    &   14.5 us & 5.4\%    & 3    \\
                                                & Baseline-S  &   15.0 us & 9.1\%    & 23   &   15.1 us & 10.0\%   & 54   \\
                                                & \tool-S     &   13.9 us & 0.8\%    & 0    &   15.2 us & 10.9\%   & 54   \\
\midrule
\multirow{5}{*}{SHA-256 (CT-Wasm), 8192 bytes}  & Ref         &  114.6 us & -        & -    &  114.6 us & -        & -    \\
                                                & Baseline-F  &  516.6 us & 350.6\%  & 23   &  632.6 us & 451.8\%  & 54   \\
                                                & \tool-F     &  113.7 us & -0.8\%   & 0    &  193.3 us & 68.6\%   & 3    \\
                                                & Baseline-S  &  187.4 us & 63.4\%   & 23   &  208.0 us & 81.5\%   & 54   \\
                                                & \tool-S     &  115.2 us & 0.5\%    & 0    &  216.5 us & 88.9\%   & 54   \\
\midrule
\multirow{5}{*}{ChaCha20 (HACL*), 8192 bytes}   & Ref         &   43.7 us & -        & -    &   43.7 us & -        & -    \\
                                                & Baseline-F  &   85.2 us & 94.8\%   & 136  &   85.4 us & 95.3\%   & 142  \\
                                                & \tool-F     &   44.4 us & 1.5\%    & 3    &   45.4 us & 3.8\%    & 7    \\
                                                & Baseline-S  &   52.8 us & 20.8\%   & 136  &   53.3 us & 21.9\%   & 142  \\
                                                & \tool-S     &   43.6 us & -0.3\%   & 3    &   53.8 us & 22.9\%   & 142  \\
\midrule
\multirow{5}{*}{Poly1305 (HACL*), 1024 bytes}   & Ref         &    5.5 us & -        & -    &    5.5 us & -        & -    \\
                                                & Baseline-F  &    6.3 us & 15.9\%   & 133  &    6.4 us & 17.2\%   & 139  \\
                                                & \tool-F     &    5.5 us & 1.4\%    & 3    &    5.6 us & 2.2\%    & 9    \\
                                                & Baseline-S  &    5.6 us & 1.8\%    & 133  &    5.7 us & 4.4\%    & 139  \\
                                                & \tool-S     &    5.5 us & 1.0\%    & 3    &    5.6 us & 2.5\%    & 139  \\
\midrule
\multirow{5}{*}{Poly1305 (HACL*), 8192 bytes}   & Ref         &   15.1 us & -        & -    &   15.1 us & -        & -    \\
                                                & Baseline-F  &   21.3 us & 41.1\%   & 133  &   21.4 us & 41.2\%   & 139  \\
                                                & \tool-F     &   15.1 us & -0.0\%   & 3    &   15.2 us & 0.8\%    & 9    \\
                                                & Baseline-S  &   16.2 us & 7.2\%    & 133  &   16.3 us & 7.6\%    & 139  \\
                                                & \tool-S     &   15.2 us & 0.7\%    & 3    &   16.2 us & 7.1\%    & 139  \\
\midrule
\multirow{5}{*}{ECDH Curve25519 (HACL*)}        & Ref         &  354.3 us & -        & -    &  354.3 us & -        & -    \\
                                                & Baseline-F  &  989.8 us & 179.3\%  & 1862 & 1006.4 us & 184.0\%  & 1887 \\
                                                & \tool-F     &  479.9 us & 35.4\%   & 235  &  497.8 us & 40.5\%   & 256  \\
                                                & Baseline-S  &  507.0 us & 43.1\%   & 1862 &  520.4 us & 46.9\%   & 1887 \\
                                                & \tool-S     &  386.1 us & 9.0\%    & 1419 &  516.8 us & 45.9\%   & 1887 \\

\bottomrule

\multirow{5}{*}{Geometric means}                & Ref         &           & -        &      &           & -        &      \\
                                                & Baseline-F  &           & 80.2\%   &      &           & 104.8\%  &      \\
                                                & \tool-F     &           & 5.0\%    &      &           & 15.3\%   &      \\
                                                & Baseline-S  &           & 19.4\%   &      &           & 25.8\%   &      \\
                                                & \tool-S     &           & 1.7\%    &      &           & 26.6\%   &      \\

\bottomrule
\end{tabular}
\end{table}

\mypara{Benchmarks}
We evaluate \tool on existing cryptographic code taken from two sources.
First, we consider two cryptographic primitives from CT-Wasm~\cite{Watt:2019}:
\begin{CompactItemize}
\item The Salsa20 stream cipher, with a workload of 64 bytes.
\item The SHA-256 hash function, with workloads of 64 bytes (one block) or
8192 bytes (128 blocks).
\end{CompactItemize}
Second, we consider automatically generated cryptographic primitives and
protocols from the \Hacl~\cite{HACL} library. We compile the automatically
generated C code to Wasm using Clang's Wasm backend. (We do not use \Hacl's
Wasm backend since it relies on a JavaScript embedding environment and is not
well suited for Lucet.)
Specifically, from \Hacl we consider:
\begin{CompactItemize}
\item The ChaCha20 stream cipher, with a workload of 8192 bytes.
\item The Poly1305 message authentication code, with workloads of 1024 or 8192 bytes.
\item ECDH key agreement using Curve25519.
\end{CompactItemize}
We selected these primitives to cover different kinds of modern crypto
workloads (including hash functions, MACs, encryption ciphers, and public key
exchange algorithms).
We omitted primitives that had inline assembly or SIMD since Lucet does not
yet support either; we also omitted the AES from \Hacl and TEA from
CT-Wasm---modern processors implement AES in hardware (largely
because efficient software implementations of AES are generally not
constant-time~\cite{Osvik:2006:CAC}), while TEA is not used in practice.
All the primitives we consider have been verified to be constant-time---free of cache
and timing side-channels.
However, the proofs assume a sequential execution model and do not account
for speculative leaks as addressed in this work.

\mypara{Experimental Setup}
We conduct our experiments on an Intel Xeon Platinum 8160 (Skylake) with 1TB
of RAM.
The machine runs Arch Linux with kernel \texttt{5.8.14}, and we use the Lucet
runtime version \texttt{0.7.0-dev} (Cranelift version \texttt{0.62.0} with
our modifications) compiled with \texttt{rustc} version \texttt{1.46.0}.
We collect benchmarks using the Rust \texttt{criterion} crate version
\texttt{0.3.3}~\cite{criterion} and report the point estimate for the mean
runtime of each benchmark.

\mypara{Reference and Baseline Comparisons}
We compare \tool to a reference (unsafe) implementation and a baseline (safe)
implementation which simply \ensuremath{\mathbf{protect}}s every Wasm memory load instruction.
We consider two baseline variants: The baseline solution with Spectre v1.1
mitigation \ensuremath{\mathbf{protect}}s every Wasm load instruction, while the baseline
solution with only Spectre v1 mitigation \ensuremath{\mathbf{protect}}s only Wasm load
instructions with non-constant addresses.
The latter is similar to Clang's Spectre mitigation pass, which applies SLH to
each non-constant array read~\cite{SLH}.
We evaluate both \tool and the baseline implementation with Spectre v1 protection
and with both v1 and v1.1 protections combined.  We consider both fence-based and SLH-based
implementations of the \ensuremath{\mathbf{protect}} primitive.
In the rest of this section, we use Baseline-F and \tool-F to refer to
fence-based implementations of their respective mitigations and Baseline-S and
\tool-S to refer to the SLH-based implementations.\looseness=-1

\mypara{Results}
\Cref{tab:evaluation} summarizes our results.
With Spectre v1 protections, both \tool-F and \tool-S insert very few \ensuremath{\mathbf{protect}}s
and have negligible performance overhead on most of our benchmarks---the
geometric mean overheads imposed by \tool-F and \tool-S are 5.0\% and
1.7\%, respectively.
In contrast, the baseline passes insert between 3 and 1862
protections and incur significantly higher overheads than \tool---the geometric mean
overheads imposed by Baseline-F and Baseline-S are 80.2\% and 19.4\%,
respectively.\looseness=-1

With both v1 and v1.1 protections, \tool-F inserts an order of magnitude fewer protections
than Baseline-F, and has correspondingly low performance overhead---the
geometric mean overhead of \tool-F is 15.3\%, whereas Baseline-F's is
104.8\%.
The geometric mean overhead of both \tool-S and Baseline-S, on the other hand, is
roughly 26\%.
Unlike \tool-F, \tool-S must mark all stored values as sinks in order to
eliminate Spectre v1.1 attacks; for these benchmarks, this countermeasure
requires \tool-S to apply protections to every Wasm load, just like
Baseline-S.  Indeed, we see in the table that Baseline-S and \tool-S
make the exact same number of additions to the code.

We make \New{three} observations from our measurements.
First, and somewhat surprisingly, \tool does not insert any \ensuremath{\mathbf{protect}}s
for Spectre v1 on any of the CT-Wasm benchmarks.
We attribute this to the style of code: the CT-Wasm primitives are hand-written
and, moreover, statically allocate variables and arrays in the Wasm linear
memory---which, in turn, results in many constant-address loads.
This is unlike the \Hacl primitives which are written in F*, compiled to C and
then Wasm---and thus require between 3 and 235 \ensuremath{\mathbf{protect}}s.

\New{
Second, the benchmarks with short reference runtimes tend to have overall
lower overheads, particularly for the baseline schemes.
This is because for short workloads, the overall runtime is dominated by
sandbox setup and teardown, which \tool does not introduce much overhead for.
In contrast, for longer workloads, the execution of the Wasm code becomes the
dominant portion of the benchmark---and exposes the overhead imposed by the
different mitigations.
We explore the relationship between workload size and performance overhead in
more detail
\ifextended
in~\Cref{app:perf-overheads}.
\else
in~\cite{vassena2020automatically}.\looseness=-1
\fi
}

And third, we observe for the Spectre v1 version that SLH gives overall better
performance than fences, as expected.
This is true even in the case of Curve25519, where implementing
\ensuremath{\mathbf{protect}} using SLH (\tool-S) results in a significant increase in the number of
protections versus the fence-based implementation (\tool-F).
Even in this case, the more targeted
restriction of speculation, and the less heavyweight impact on the pipeline,
allows SLH to still prevail over the fewer fences.
\New{
However, this advantage is lost when considering both v1 and v1.1 mitigation:
The sharp increase in the number of \ensuremath{\mathbf{protect}}s required for v1.1
ends up being slower than using (fewer) fences.
%
A hybrid approach that uses both fences and SLH could potentially
outperform both \tool-F and \tool-S.
}

In reality, though, both versions are inadequate software emulations of what the
\ensuremath{\mathbf{protect}} primitive should be.  Fences take a heavy toll on the pipeline and are
far too restrictive of speculation, while SLH pays a heavy instruction overhead for
each instance, and can only be applied directly to loads, not to arbitrary cut points.
A hardware implementation of the \ensuremath{\mathbf{protect}} primitive could combine the best of
\tool-F and \tool-S: targeted restriction of speculation, minimal instruction
overhead, and only as many defenses as \tool-F, without the inflation in insertion
count required by \tool-S.

However, even without any hardware assistance, both versions of the \tool tool
provide significant performance gains over the current state of the art in mitigating
Spectre v1, and over existing fence-based solutions when targeting v1 or both v1 and
v1.1.


\section{Related Work}
\label{sec:related}

\mypara{Speculative Execution Semantics}
Several semantics models for speculative execution have been proposed
recently~\cite{Guarnieri20,cheang-csf19,McIlroy19,pitchfork,DisselkoenJJR19,balliu2019inspectre}.
Of those, \cite{pitchfork} is closest to ours, and inspired our semantics
(e.g., we share the 3-stages pipeline, attacker-supplied directives and the
instruction reorder buffer).
%
\New{However, their semantics---and, indeed, the semantics of most of the other works---are
exclusively for low-level assembly-like languages.\footnote{
  The one exception, \citet{DisselkoenJJR19}, present a Spectre-aware relaxed
  memory model based on pomsets, which is even further abstracted from the
  microarchitectural features of real processors.
}
In contrast, our JIT semantics bridges the gap between high-level commands and low-level instructions, which
allows us to reason about speculative execution of source-level imperative programs through straightforward typing rules, while being faithful to low-level microarchitectural details.
Moreover, the idea of modeling speculative and out-of-order execution using stacks of progressively flattened commands is novel and key to enable source-level reasoning about the low-level effects of speculation.
}

\mypara{Detection and Repair}
\citet{Wu:2019} detect cache side channels via abstract interpretation
by augmenting the program control-flow to accommodate for speculation.
\textsc{Spectector}~\cite{Guarnieri20} and \textsc{Pitchfork}~\cite{pitchfork} use
symbolic execution on x86 binaries to detect speculative
vulnerabilities.
%
\citet{cheang-csf19} and \citet{Roderick19} apply bounded model
checking to detect potential speculative vulnerabilities respectively
via 4-ways self-composition and taint-tracking.
These efforts assume a \emph{fixed} speculation bound, and they focus
on vulnerability detection rather than proposing techniques to \emph{repair}
vulnerable programs.
Furthermore, many of these works consider only \emph{in-order} execution.
%
In contrast, our type system enforces \emph{speculative constant-time}
when program instructions are executed \emph{out-of-order} with
\emph{unbounded} speculation---and our tool \tool automatically synthesizes
repairs.
Separately, \textsc{oo7}~\cite{oo7} statically analyzes a binary from a set of
untrusted input sources, detecting vulnerable patterns and inserting fences
in turn.
Our tool, \tool, not only repairs vulnerable programs without user annotation,
but ensures that program patches contain a minimum number of fences.
Furthermore, \tool formally guarantees that repaired programs are free
from speculation-based attacks.

Concurrent to our work, Intel proposed a mitigation for a new class of
LVI attacks~\cite{vanbulck2020lvi,Intel-lvi}.
Like \tool, they implement a compiler pass that analyzes the program to
determine an optimal placement of fences to cut source-to-sink data flows.
While we consider an abstract, ideal \ensuremath{\mathbf{protect}} primitive, they
focus on the optimal placement of fences in particular.
This means that they optimize the fence placement by taking into account the
coarse-grained effects of fences---e.g., one fence providing a speculation
barrier for multiple independent data-dependency chains.\footnote{
Unlike our approach, their resulting optimization problem is NP-hard---and only
sub-optimal solutions may be found through heuristics.
}
This also means, however, their approach does not easily transfer to
using SLH for cases where SLH would be faster.

\mypara{Hardware-based Mitigations}
%
%
%
%
To eliminate speculative attacks, several secure hardware designs have been
proposed.
%
%
%
\citet{tullsen-asplos} propose context-sensitive fencing, a
hardware-based mitigation that dynamically inserts fences in the
instruction stream when dangerous conditions arise.
\textsc{InvisiSpec} \cite{Yan:2018} 
features a special \emph{speculative buffer} to prevent speculative
loads from polluting the cache.
\textsc{STT} \cite{STT-Yu-19} tracks speculative taints \emph{dynamically}
inside the processor micro-architecture and stalls instructions to
prevent speculative leaks.
%
%
\citet{ConTExT20} propose \textsc{ConTExT}, a whole architecture change
(applications, compilers, operating systems, and hardware) to
eliminate \emph{all} Spectre attacks.
Though \tool can benefit from a hardware implementation of \ensuremath{\mathbf{protect}}, this
work also shows that Spectre-PHT on existing hardware can be automatically
eliminated in pure software with modest performance overheads.
%



\section{Limitations and Future Work}
\tool only addresses Spectre-PHT attacks and does so at the Wasm-layer.
Extending \tool to tackle other Spectre variants and the limitations of
operating on Wasm is future work.
%

%
\mypara{Other Spectre Variants}
The Spectre-BTB variant~\cite{Kocher2018spectre} mistrains the Branch Target
Buffer (BTB), which is used to predict indirect jump targets, to hijack the
(speculative) control-flow of the program.
Although Wasm does not provide an unrestricted indirect jump instruction,
the indirect function call instruction---which is used to call functions
registered in a function table---can be abused by an attacker.
To address (in-process) Spectre-BTB, we could extend our type system to
restrict the values used as indices into the function table to be typed
as \emph{stable}.


The other Spectre variant, Spectre-RSB~\cite{Koruyeh:2018,Maisuradze:2018},
abuses the return stack buffer.
To mitigate these attacks, we could analyze Wasm code to identify potential RSB
over/underflows and insert fences in response, or use mitigation strategies
like RSB stuffing~\cite{retpoline-whitepaper}.
A more promising approach, however, is to use Intel's recent shadow stack,
which ensures that returns cannot be speculatively
hijacked~\cite{shanbhogue2019security}.
%



\mypara{Detecting Spectre Gadgets at the Binary Level}
\tool operates on Wasm code---or more precisely, on the Cranelift compiler's
IR---and can thus
miss leaks inserted by the compiler passes that run after \tool---namely,
register allocation and instruction selection.
Though these passes are unlikely to introduce such leaks, we leave the
validation of the generated binary code to future work.

\mypara{Spectre Resistant Compilation}
An alternative to repairing existing programs is to ensure they are compiled
securely from the start.
Recent works have developed verified constant-time-preserving
optimizing compilers for generating correct, efficient, and secure
cryptographic code \cite{jasmin,ct-c-compiler}.
Doing this for speculative constant-time, and understanding which optimizations
break the SCT notion, is an interesting direction for future
work.
%



\New{
\mypara{Bounds Information}
\tool-S relies on array bounds information to implement the speculative load
hardening.
For a memory safe language, this information can be made available to \tool
when compiling to Wasm (e.g., as a custom section).
When compiling languages like C, where arrays bounds information is not
explicit, this is harder---and we would need to use program analysis to track
array lengths statically~\cite{Venet04}.
Although such an analysis may be feasible for cryptographic code, it is likely
to fall short for other application domains (e.g., due to dynamic memory
allocation and pointer chasing).
In these cases, we could track array lengths at runtime (e.g., by
instrumenting programs~\cite{soft-bound}) or, more simply, fall back to fences
(especially since the overhead of tracking bounds information at runtime is
typically high).
}

\section{Conclusion}
We presented \tool, a fully automatic
approach to provably and efficiently eliminate speculation-based
leakage in unannotated cryptographic code.
\tool statically detects data flows from transient sources to stable sinks
and synthesizes a minimal number of fence-based or SLH-based \ensuremath{\mathbf{protect}} calls
to eliminate potential leaks.
Our evaluation shows that \tool inserts an order of magnitude fewer protections
than would be added by today's compilers, and that existing crypto
primitives repaired with \tool impose modest overheads when using both fences
and SLH for \ensuremath{\mathbf{protect}}.


\section*{Acknowledgements}
We thank the reviewers and our shepherd Aseem Rastogi for their suggestions and
insightful comments.
Many thanks to Shravan Narayan, Ravi Sahita, and Anjo Vahldiek-Oberwagner for
fruitful discussions.
This work was supported in part by gifts from Fastly, Fujitsu, and Cisco; by the
NSF under Grant Number CNS-1514435 and CCF-1918573; by ONR Grant N000141512750;
by the German Federal Ministry of Education and Research (BMBF) through funding
for the CISPA-Stanford Center for Cybersecurity;
and, by the CONIX Research Center, one of six centers in JUMP, a Semiconductor
Research Corporation (SRC) program sponsored by DARPA.

\bibliography{local}

\ifextended
\clearpage
\newpage
\appendix

\section{Full Calculus}
\label{app:full-calculus}



\begin{figure}[h]
\rulesize
\centering
\begin{align*}
\text{Arrays: } & \ensuremath{\Varid{a}}\ \ensuremath{\Coloneqq}\ \ensuremath{\{\mskip1.5mu \Varid{base}\;\Varid{n},}\ \ensuremath{\mathit{length}\;\Varid{n},}\ \ensuremath{\Varid{label}\;\ell\mskip1.5mu\}}
\\
\text{Values: } & \ensuremath{\Varid{v}}\ \ensuremath{\Coloneqq}\ \ensuremath{\Varid{n}\;\;|\;\;\Varid{b}\;\;|\;\;\Varid{a}}
\\
\text{Expressions: } & \ensuremath{\Varid{e}}\ \ensuremath{\Coloneqq}\ \ensuremath{\Varid{v}\;\;|\;\;\Varid{x}\;\;|\;\;\Varid{e}_{1}\mathbin{+}\Varid{e}_{2}\;\;|\;\;\Varid{e}_{1}\leq \Varid{e}_{2}\;\;|\;\;\Varid{e}_{1}}\ \ensuremath{\mathbin{?}}\ \ensuremath{\Varid{e}_{2}}\ \ensuremath{\mathbin{:}}\ \ensuremath{\Varid{e}_{3}\;\;|\;\;\Varid{e}\;\otimes\;\Varid{e}}
\\
\text{Right-hand Sides: } & \ensuremath{\Varid{r}}\ \ensuremath{\Coloneqq}\ \ensuremath{\Varid{e}\;\;|\;\;(\mathbin{*}_{\ell}\;\Varid{e})\;\;|\;\;\Varid{a}{[}\Varid{e}{]}}
\\
\text{Commands: } & \ensuremath{\Varid{c}}\ \ensuremath{\Coloneqq}\ \ensuremath{\mathbf{skip}\;\;|\;\;\Varid{x}\mathbin{:=}\Varid{r}\;\;|\;\;(\mathbin{*}_{\ell}\;\Varid{e})\mathrel{=}\Varid{e}\;\;|\;\;\Varid{a}{[}\Varid{e}_{1}{]}\mathbin{:=}\Varid{e}_{2}\;\;|\;\;\Varid{x}\mathbin{:=}\mathbf{protect}(\Varid{r})}
\\ & \ \ \ \ \ \ \ \ \ \ensuremath{\;|\;\;\mathbf{if}\;\Varid{e}\;\mathbf{then}\;\Varid{c}_{1}\;\mathbf{else}\;\Varid{c}_{2}\;\;|\;\;\mathbf{while}\;\Varid{e}\;\mathbf{do}\;\Varid{c}\;\;|\;\;\mathbf{fail}\;\;|\;\;\Varid{c}_{1};\Varid{c}_{2}} \\
\text{Instructions: } & \ensuremath{\Varid{i}}\ \ensuremath{\Coloneqq}\
\ensuremath{\mathbf{nop}\;\;|\;\;\mathbf{fail}(\Varid{p})\;\;|\;\;\Varid{x}\mathbin{:=}\Varid{e}\;\;|\;\;\Varid{x}\mathbin{:=}\mathbf{load}_{\ell}(\Varid{e})\;\;|\;\;\mathbf{store}_{\ell}(\Varid{e}_{1},\Varid{e}_{2})}
\\ & \ \ \ \ \ \ \ \ \ \ensuremath{\;|\;\;\Varid{x}\mathbin{:=}\mathbf{protect}(\Varid{e})\;\;|\;\;\mathbf{guard}(\Varid{e}^{\Varid{b}},\Varid{cs},\Varid{p})}
\\
\text{Predictions: } & \ensuremath{\Varid{b}}\ \ensuremath{\in}\ \ensuremath{\{\mskip1.5mu \mathbf{true},}\ \ensuremath{\mathbf{false}\mskip1.5mu\}}
\\
\text{Guard and Fail Identifiers: } & \ensuremath{\Varid{p}}\ \ensuremath{\in}\ \ensuremath{\mathbb{N}}
\\
\text{Directives: } & \ensuremath{\Varid{d}}\ \ensuremath{\Coloneqq}\ \ensuremath{\mathbf{fetch}\;\;|\;\;\mathbf{fetch}\;\Varid{b}\;\;|\;\;\mathbf{exec}\;\Varid{n}\;\;|\;\;\mathbf{retire}}
\\
\text{Schedules: } & \ensuremath{\Conid{D}}\ \ensuremath{\Coloneqq}\ \ensuremath{[\mskip1.5mu \mskip1.5mu]\;\;|\;\;\Varid{d}\mathbin{:}\Conid{D}}
\\
\text{Observations: } & \ensuremath{\Varid{o}}\ \ensuremath{\Coloneqq}\ \ensuremath{\epsilon\;\;|\;\;\mathbf{fail}(\Varid{p})\;\;|\;\;\mathbf{read}(\Varid{n},\Varid{ps})\;\;|\;\;\mathbf{write}(\Varid{n},\Varid{ps})\;\;|\;\;\mathbf{rollback}(\Varid{p})}
\\
\text{Observation Traces: } & \ensuremath{\Conid{O}}\ \ensuremath{\Coloneqq}\ \ensuremath{\epsilon\;\;|\;\;\Varid{o}\;{{\mkern-2mu\cdot\mkern-2mu}}\;\Conid{O}}
\\
\text{Reorder Buffers: } & \ensuremath{\Varid{is}}\ \ensuremath{\Coloneqq}\ \ensuremath{\Varid{i}\mathbin{:}\Varid{is}\;\;|\;\;[\mskip1.5mu \mskip1.5mu]}
\\
\text{Command Stacks: } & \ensuremath{\Varid{cs}}\ \ensuremath{\Coloneqq}\ \ensuremath{\Varid{c}\mathbin{:}\Varid{cs}\;\;|\;\;[\mskip1.5mu \mskip1.5mu]}
\\
\text{Memory Stores: } &  \ensuremath{\mu}\ \ensuremath{\in}\ \ensuremath{\mathbb{N}\rightharpoonup\Conid{Value}} \\
\text{Variable Maps: } & \ensuremath{\rho}\ \ensuremath{\in}\ \ensuremath{\Conid{Var}\to \Conid{Value}} \\
\text{Configurations: } & \ensuremath{\Conid{C}}\ \ensuremath{\Coloneqq}\ \ensuremath{\langle\Varid{is},\Varid{cs},\mu,\rho\rangle}
\end{align*}
\caption{Decorated syntax. Pointer operators (e.g., \ensuremath{\mathbin{*}_{\ell}\;\Varid{e}}) and memory instructions (e.g., \ensuremath{\mathbf{load}_{\ell}(\Varid{e})}) are annotated with a security label \ensuremath{\ell}, which represents the sensitivity of the data stored at the corresponding memory accesses.}
\label{fig:decorated-syntax}
\end{figure}

\begin{figure}[t]
\begin{subfigure}{0.9\textwidth}
\begin{mathpar}
\inferrule[Fetch-Skip]
{}
{\ensuremath{\langle\Varid{is},\mathbf{skip}\mathbin{:}\Varid{cs},\mu,\rho\rangle\stepT{\mathbf{fetch}}{\epsilon}\langle\Varid{is}+{\mkern-9mu+}\ [\mskip1.5mu \mathbf{nop}\mskip1.5mu],\Varid{cs},\mu,\rho\rangle}}
\and
\inferrule[Fetch-Fail]
{\ensuremath{\fresh{\Varid{p}}} }
{\ensuremath{\langle\Varid{is},\mathbf{fail}\mathbin{:}\Varid{cs},\mu,\rho\rangle\stepT{\mathbf{fetch}}{\epsilon}\langle\Varid{is}+{\mkern-9mu+}\ [\mskip1.5mu \mathbf{fail}(\Varid{p})\mskip1.5mu],\Varid{cs},\mu,\rho\rangle}}
\and
\inferrule[Fetch-Asgn]
{}
{\ensuremath{\langle\Varid{is},\Varid{x}\mathbin{:=}\Varid{e}\mathbin{:}\Varid{cs},\mu,\rho\rangle\stepT{\mathbf{fetch}}{\epsilon}\langle\Varid{is}+{\mkern-9mu+}\ [\mskip1.5mu \Varid{x}\mathbin{:=}\Varid{e}\mskip1.5mu],\Varid{cs},\mu,\rho\rangle}}
\and
\inferrule[Fetch-Seq]
{}
{\ensuremath{\langle\Varid{is},\Varid{c}_{1};\Varid{c}_{2}\mathbin{:}\Varid{cs},\mu,\rho\rangle\stepT{\mathbf{fetch}}{\epsilon}\langle\Varid{is},\Varid{c}_{1}\mathbin{:}\Varid{c}_{2}\mathbin{:}\Varid{cs},\mu,\rho\rangle}}
\and
\inferrule[Fetch-Ptr-Load]
{\ensuremath{\Varid{c}}\ \ensuremath{\mathrel{=}}\ \ensuremath{\Varid{x}\mathbin{:=}\mathbin{*}_{\ell}\;\Varid{e}} \\ \ensuremath{\Varid{i}\mathrel{=}\Varid{x}\mathbin{:=}\mathbf{load}_{\ell}(\Varid{e})} }
{\ensuremath{\langle\Varid{is},\Varid{c}\mathbin{:}\Varid{cs},\mu,\rho\rangle\stepT{\mathbf{fetch}}{\epsilon}\langle\Varid{is}+{\mkern-9mu+}\ [\mskip1.5mu \Varid{i}\mskip1.5mu],\Varid{cs},\mu,\rho\rangle}}
\and
\inferrule[Fetch-Ptr-Store]
{\ensuremath{\Varid{c}}\ \ensuremath{\mathrel{=}}\ \ensuremath{\mathbin{*}_{\ell}\;\Varid{e}_{1}\mathbin{:=}\Varid{e}_{2}} \\ \ensuremath{\Varid{i}\mathrel{=}\mathbf{store}_{\ell}(\Varid{e}_{1},\Varid{e}_{2})} }
{\ensuremath{\langle\Varid{is},\Varid{c}\mathbin{:}\Varid{cs},\mu,\rho\rangle\stepT{\mathbf{fetch}}{\epsilon}\langle\Varid{is}+{\mkern-9mu+}\ [\mskip1.5mu \Varid{i}\mskip1.5mu],\Varid{cs},\mu,\rho\rangle}}
\and
\inferrule[Fetch-Array-Load]
{\ensuremath{\Varid{c}}\ \ensuremath{\mathrel{=}}\ \ensuremath{\Varid{x}\mathbin{:=}\Varid{a}{[}\Varid{e}_{1}{]}} \quad \ensuremath{\Varid{e}\mathrel{=}\Varid{e}_{1}<\mathit{length}(\Varid{a})} \quad \ensuremath{\Varid{e'}\mathrel{=}\Varid{base}(\Varid{a})\mathbin{+}\Varid{e}_{1}} \\\\
\ensuremath{\ell\mathrel{=}\Varid{label}(\Varid{a})} \qquad \ensuremath{\Varid{c'}}\ \ensuremath{\mathrel{=}}\ \ensuremath{\mathbf{if}}\ \ensuremath{\Varid{e}}\ \ensuremath{\mathbf{then}}\ \ensuremath{\Varid{x}\mathbin{:=}\mathbin{*}_{\ell}\;\Varid{e'}}\ \ensuremath{\mathbf{else}}\ \ensuremath{\mathbf{fail}}}
{\ensuremath{\langle\Varid{is},\Varid{c}\mathbin{:}\Varid{cs},\mu,\rho\rangle\stepT{\mathbf{fetch}}{\epsilon}\langle\Varid{is},\Varid{c'}\mathbin{:}\Varid{cs},\mu,\rho\rangle}}
\and
\inferrule[Fetch-Array-Store]
{\ensuremath{\Varid{c}}\ \ensuremath{\mathrel{=}}\ \ensuremath{\Varid{a}{[}\Varid{e}_{1}{]}\mathbin{:=}\Varid{e}_{2}} \quad \ensuremath{\Varid{e}\mathrel{=}\Varid{e}_{1}<\mathit{length}(\Varid{a})} \quad \ensuremath{\Varid{e'}\mathrel{=}\Varid{base}(\Varid{a})\mathbin{+}\Varid{e}_{1}} \\\\ \ensuremath{\ell\mathrel{=}\Varid{label}(\Varid{a})} \qquad
 \ensuremath{\Varid{c'}}\ \ensuremath{\mathrel{=}}\ \ensuremath{\mathbf{if}}\ \ensuremath{\Varid{e}}\ \ensuremath{\mathbf{then}}\ \ensuremath{\mathbin{*}_{\ell}\;\Varid{e'}\mathbin{:=}\Varid{e}}\ \ensuremath{\mathbf{else}}\ \ensuremath{\mathbf{fail}}}
{\ensuremath{\langle\Varid{is},\Varid{c}\mathbin{:}\Varid{cs},\mu,\rho\rangle\stepT{\mathbf{fetch}}{\epsilon}\langle\Varid{is},\Varid{c'}\mathbin{:}\Varid{cs},\mu,\rho\rangle}}
\and
\inferrule[Fetch-If-True]
{\ensuremath{\Varid{c}}\ \ensuremath{\mathrel{=}}\ \ensuremath{\mathbf{if}\;\Varid{e}\;\mathbf{then}\;\Varid{c}_{1}\;\mathbf{else}\;\Varid{c}_{2}} \qquad \ensuremath{\Varid{b}\mathrel{=}\mathbf{true}} \\\\ \ensuremath{\fresh{\Varid{p}}} \qquad \ensuremath{\Varid{i}\mathrel{=}\mathbf{guard}(\Varid{e}^{\Varid{b}},\Varid{c}_{2}\mathbin{:}\Varid{cs},\Varid{p})}}
{\ensuremath{\langle\Varid{is},\Varid{c}\mathbin{:}\Varid{cs},\mu,\rho\rangle\stepT{\mathbf{fetch}\;\Varid{b}}{\epsilon}\langle\Varid{is}+{\mkern-9mu+}\ [\mskip1.5mu \Varid{i}\mskip1.5mu],\Varid{c}_{1}\mathbin{:}\Varid{cs},\mu,\rho\rangle}}
\and
\inferrule[Fetch-If-False]
{\ensuremath{\Varid{c}}\ \ensuremath{\mathrel{=}}\ \ensuremath{\mathbf{if}\;\Varid{e}\;\mathbf{then}\;\Varid{c}_{1}\;\mathbf{else}\;\Varid{c}_{2}} \qquad \ensuremath{\Varid{b}\mathrel{=}\mathbf{false}} \\\\ \ensuremath{\fresh{\Varid{p}}} \qquad \ensuremath{\Varid{i}\mathrel{=}\mathbf{guard}(\Varid{e}^{\Varid{b}},\Varid{c}_{1}\mathbin{:}\Varid{cs},\Varid{p})}}
{\ensuremath{\langle\Varid{is},\Varid{c}\mathbin{:}\Varid{cs},\mu,\rho\rangle\stepT{\mathbf{fetch}\;\Varid{b}}{\epsilon}\langle\Varid{is}+{\mkern-9mu+}\ [\mskip1.5mu \Varid{i}\mskip1.5mu],\Varid{c}_{2}\mathbin{:}\Varid{cs},\mu,\rho\rangle}}
\and
\inferrule[Fetch-While]
{\ensuremath{\Varid{c}_{1}\mathrel{=}\Varid{c};\mathbf{while}\;\Varid{e}\;\mathbf{do}\;\Varid{c}} \qquad \ensuremath{\Varid{c}_{2}\mathrel{=}\mathbf{if}\;\Varid{e}\;\mathbf{then}\;\Varid{c}_{1}\;\mathbf{else}}\ \ensuremath{\mathbf{skip}} }
{\ensuremath{\langle\Varid{is},\mathbf{while}\;\Varid{e}\;\mathbf{do}\;\Varid{c}\mathbin{:}\Varid{cs},\mu,\rho\rangle\stepT{\mathbf{fetch}}{\epsilon}\langle\Varid{is},\Varid{c}_{2}\mathbin{:}\Varid{cs},\mu,\rho\rangle}}
\end{mathpar}
\captionsetup{format=caption-with-line}
\caption{Fetch stage.}
\label{fig:full-fetch-stage}
\end{subfigure}%

\caption{Full semantics.}
\end{figure}

\begin{figure}[p]
\ContinuedFloat
\centering
\begin{subfigure}{.90\textwidth}
\begin{mathpar}
\inferrule[Execute]
{\ensuremath{{\vert}\Varid{is}_{1}{\vert}\mathrel{=}\Varid{n}\mathbin{-}\mathrm{1}} \\
 \ensuremath{\rho'\mathrel{=}\phi(\Varid{is}_{1},\rho)} \\
 \ensuremath{\langle\Varid{is}_{1},\Varid{i},\Varid{is}_{2},\Varid{cs}\rangle\xrsquigarrow{(\mu,\rho',\Varid{o})}\langle\Varid{is'},\Varid{cs'}\rangle} }
{\ensuremath{\langle\Varid{is}_{1}+{\mkern-9mu+}\ [\mskip1.5mu \Varid{i}\mskip1.5mu]+{\mkern-9mu+}\ \Varid{is}_{2},\Varid{cs},\mu,\rho\rangle\stepT{\mathbf{exec}\;\Varid{n}}{\Varid{o}}\langle\Varid{is'},\Varid{cs'},\mu,\rho\rangle}}
\and
\inferrule[Exec-Asgn]
{\ensuremath{\Varid{i}}\ \ensuremath{\mathrel{=}}\ \ensuremath{\Varid{x}\mathbin{:=}\Varid{e}} \\
 \ensuremath{\Varid{v}\mathrel{=}\llbracket \Varid{e}\rrbracket^{\rho}} \\
 \ensuremath{\Varid{i'}}\ \ensuremath{\mathrel{=}}\ \ensuremath{\Varid{x}\mathbin{:=}\Varid{v}} }
{\ensuremath{\langle\Varid{is}_{1},\Varid{i},\Varid{is}_{2},\Varid{cs}\rangle\xrsquigarrow{(\mu,\rho,\epsilon)}\langle\Varid{is}_{1}+{\mkern-9mu+}\ [\mskip1.5mu \Varid{i'}\mskip1.5mu]+{\mkern-9mu+}\ \Varid{is}_{2},\Varid{cs}\rangle}}
\and
\inferrule[Exec-Branch-Ok]
{\ensuremath{\Varid{i}\mathrel{=}\mathbf{guard}(\Varid{e}^{\Varid{b}},\Varid{cs'},\Varid{p})} \\ \ensuremath{\llbracket \Varid{e}\rrbracket^{\rho}\mathrel{=}\Varid{b}} \\ \ensuremath{\Varid{ps}\mathrel{=}\Moon{\Varid{is}_{1}}} }
{\ensuremath{\langle\Varid{is}_{1},\Varid{i},\Varid{is}_{2},\Varid{cs}\rangle\xrsquigarrow{(\mu,\rho,\epsilon)}\langle\Varid{is}_{1}+{\mkern-9mu+}\ [\mskip1.5mu \mathbf{nop}\mskip1.5mu]+{\mkern-9mu+}\ \Varid{is}_{2},\Varid{cs}\rangle}}
\and
\inferrule[Exec-Branch-Mispredict]
{\ensuremath{\Varid{i}\mathrel{=}\mathbf{guard}(\Varid{e}^{\Varid{b}},\Varid{cs'},\Varid{p})} \\ \ensuremath{\Varid{b'}\mathrel{=}\llbracket \Varid{e}\rrbracket^{\rho}} \\ \ensuremath{\Varid{b'}\neq\Varid{b}} \\ \ensuremath{\Varid{ps}\mathrel{=}\Moon{\Varid{is}_{1}}} }
{\ensuremath{\langle\Varid{is}_{1},\Varid{i},\Varid{is}_{2},\Varid{cs}\rangle\xrsquigarrow{(\mu,\rho,\mathbf{rollback}(\Varid{p}))}\langle\Varid{is}_{1}+{\mkern-9mu+}\ [\mskip1.5mu \mathbf{nop}\mskip1.5mu],\Varid{cs'}\rangle}}
\and
\inferrule[Exec-Load]
{\ensuremath{\Varid{i}}\ \ensuremath{\mathrel{=}}\ \ensuremath{\Varid{x}\mathbin{:=}\mathbf{load}_{\ell}(\Varid{e})} \\
 \ensuremath{\mathbf{store}_{\ell'}(\anonymous ,\anonymous )\;\not\in\;\Varid{is}_{1}} \\
 \ensuremath{\Varid{n}\mathrel{=}\llbracket \Varid{e}\rrbracket^{\rho}} \\\\
 \ensuremath{\Varid{ps}\mathrel{=}\Moon{\Varid{is}_{1}}} \\
 \ensuremath{\Varid{i'}}\ \ensuremath{\mathrel{=}}\ \ensuremath{\Varid{x}\mathbin{:=}\mu(\Varid{n})} }
{\ensuremath{\langle\Varid{is}_{1},\Varid{i},\Varid{is}_{2},\Varid{cs}\rangle\xrsquigarrow{(\mu,\rho,\mathbf{read}(\Varid{n},\Varid{ps}))}\langle\Varid{is}_{1}+{\mkern-9mu+}\ [\mskip1.5mu \Varid{i'}\mskip1.5mu]+{\mkern-9mu+}\ \Varid{is}_{2},\Varid{cs}\rangle}}
\and
\inferrule[Exec-Store]
{\ensuremath{\Varid{i}\mathrel{=}\mathbf{store}_{\ell}(\Varid{e}_{1},\Varid{e}_{2})} \\
\ensuremath{\Varid{n}\mathrel{=}\llbracket \Varid{e}_{1}\rrbracket^{\rho}} \\
\ensuremath{\Varid{v}\mathrel{=}\llbracket \Varid{e}_{2}\rrbracket^{\rho}} \\\\
\ensuremath{\Varid{ps}\mathrel{=}\Moon{\Varid{is}_{1}}} \\
\ensuremath{\Varid{i'}\mathrel{=}\mathbf{store}(\Varid{n},\Varid{v})}}
{\ensuremath{\langle\Varid{is}_{1},\Varid{i},\Varid{is}_{2},\Varid{cs}\rangle\xrsquigarrow{(\mu,\rho,\mathbf{write}(\Varid{n},\Varid{ps}))}\langle\Varid{is}_{1}+{\mkern-9mu+}\ [\mskip1.5mu \Varid{i'}\mskip1.5mu]+{\mkern-9mu+}\ \Varid{is}_{2},\Varid{cs}\rangle}}
%
\end{mathpar}
\captionsetup{format=caption-with-line}
\caption{Execute stage.}
\label{fig:full-execute-stage}
\end{subfigure}%

\begin{subfigure}{.90\textwidth}
\vspace{\baselineskip}
\begin{mathpar}
\inferrule[Retire-Nop]
{}
{\ensuremath{\langle\mathbf{nop}\mathbin{:}\Varid{is},\Varid{cs},\mu,\rho\rangle\stepT{\mathbf{retire}}{\epsilon}\langle\Varid{is},\Varid{cs},\mu,\rho\rangle}}
\and
\inferrule[Retire-Asgn]
{}
{\ensuremath{\langle\Varid{x}\mathbin{:=}\Varid{v}\mathbin{:}\Varid{is},\Varid{cs},\mu,\rho\rangle\stepT{\mathbf{retire}}{\epsilon}\langle\Varid{is},\Varid{cs},\mu,\update{\rho}{\Varid{x}}{\Varid{v}}\rangle}}
\and
\inferrule[Retire-Store]
{\ensuremath{\Varid{i}\mathrel{=}\mathbf{store}_{\ell}(\Varid{n},\Varid{v})}} 
{\ensuremath{\langle\Varid{i}\mathbin{:}\Varid{is},\Varid{cs},\mu,\rho\rangle\stepT{\mathbf{retire}}{\epsilon}\langle\Varid{is},\Varid{cs},\update{\mu}{\Varid{n}}{\Varid{v}},\rho\rangle}}
\and
\inferrule[Retire-Fail]
{}
{\ensuremath{\langle\mathbf{fail}(\Varid{p})\mathbin{:}\Varid{is},\Varid{cs},\mu,\rho\rangle\stepT{\mathbf{retire}}{\mathbf{fail}(\Varid{p})}\langle[\mskip1.5mu \mskip1.5mu],[\mskip1.5mu \mskip1.5mu],\mu,\rho\rangle}}
\end{mathpar}
\captionsetup{format=caption-with-line}
\caption{Retire stage.}
\label{fig:full-retire-stage}
\end{subfigure}%
\caption{Full semantics (continued).}
\end{figure}

\begin{figure}[p]
\small
\ContinuedFloat
\begin{subfigure}{0.9\textwidth}
\centering
\begin{minipage}{0.4\textwidth}
\centering
\begin{subfigure}{\textwidth}
\begin{align*}
&\ensuremath{\phi(\rho,[\mskip1.5mu \mskip1.5mu])\mathrel{=}\rho} \\
&\ensuremath{\phi(\rho,(\Varid{x}\mathbin{:=}\Varid{v})\mathbin{:}\Varid{is})\mathrel{=}\phi(\update{\rho}{\Varid{x}}{\Varid{v}},\Varid{is})} \\
&\ensuremath{\phi(\rho,(\Varid{x}\mathbin{:=}\Varid{e})\mathbin{:}\Varid{is})\mathrel{=}\phi(\update{\rho}{\Varid{x}}{\bot},\Varid{is})} \\
&\ensuremath{\phi(\rho,(\Varid{x}\mathbin{:=}\mathbf{load}_{\ell}(\Varid{e}))\mathbin{:}\Varid{is})\mathrel{=}\phi(\update{\rho}{\Varid{x}}{\bot},\Varid{is})} \\
&\ensuremath{\phi(\rho,(\Varid{x}\mathbin{:=}\mathbf{protect}(\Varid{e}))\mathbin{:}\Varid{is})\mathrel{=}\phi(\update{\rho}{\Varid{x}}{\bot},\Varid{is})} \\
&\ensuremath{\phi(\rho,\Varid{i}\mathbin{:}\Varid{is})\mathrel{=}\phi(\rho,\Varid{is})}
\end{align*}
\vspace{-\baselineskip}
\captionsetup{format=caption-with-line}
\captionof{figure}{Transient Variable Map.}
\label{fig:full-transient-var-map}
\end{subfigure}%
\\
\begin{subfigure}{\textwidth}
\begin{align*}
&\ensuremath{\Moon{[\mskip1.5mu \mskip1.5mu]}\mathrel{=}[\mskip1.5mu \mskip1.5mu]} \\
&\ensuremath{\Moon{\mathbf{guard}(\Varid{e}^{\Varid{b}},\Varid{cs},\Varid{p})\mathbin{:}\Varid{is}}\mathrel{=}\Varid{p}\mathbin{:}\Moon{\Varid{is}}} \\
&\ensuremath{\Moon{\mathbf{fail}(\Varid{p})\mathbin{:}\Varid{is}}\mathrel{=}\Varid{p}\mathbin{:}\Moon{\Varid{is}}} \\
&\ensuremath{\Moon{\Varid{i}\mathbin{:}\Varid{is}}\mathrel{=}\Moon{\Varid{is}}}
\end{align*}
\vspace{-\baselineskip}
\captionsetup{format=caption-with-line}
\captionof{figure}{Pending Fail and Guard Identifiers.}
\end{subfigure}%
\end{minipage}\hspace{.1\textwidth}
\centering
\begin{minipage}{0.4\textwidth}
\begin{align*}
&\ensuremath{\llbracket \Varid{v}\rrbracket^{\rho}\mathrel{=}\Varid{v}} \\
&\ensuremath{\llbracket \Varid{x}\rrbracket^{\rho}\mathrel{=}\rho(\Varid{x})} \\
&\ensuremath{\llbracket \Varid{e}_{1}\mathbin{+}\Varid{e}_{2}\rrbracket^{\rho}\mathrel{=}\llbracket \Varid{e}_{1}\rrbracket^{\rho}\mathbin{+}\llbracket \Varid{e}_{2}\rrbracket^{\rho}} \\
&\ensuremath{\llbracket \Varid{e}_{1}<\Varid{e}_{2}\rrbracket^{\rho}\mathrel{=}\llbracket \Varid{e}_{1}\rrbracket^{\rho}<\llbracket \Varid{e}_{2}\rrbracket^{\rho}} \\
&\ensuremath{\llbracket \Varid{e}_{1}\;\otimes\;\Varid{e}_{2}\rrbracket^{\rho}\mathrel{=}\llbracket \Varid{e}_{1}\rrbracket^{\rho}\;\otimes\;\llbracket \Varid{e}_{2}\rrbracket^{\rho}} \\
&\ensuremath{\llbracket \Varid{e}_{1}\mathbin{?}\Varid{e}_{2}\mathbin{:}\Varid{e}_{3}\rrbracket^{\rho}\mathrel{=}}
\begin{cases}
\ensuremath{\llbracket \Varid{e}_{2}\rrbracket^{\rho}} & \text{ if } \ensuremath{\llbracket \Varid{e}_{1}\rrbracket^{\rho}\mathrel{=}\mathbf{true}} \\
\ensuremath{\llbracket \Varid{e}_{3}\rrbracket^{\rho}} & \text{ if } \ensuremath{\llbracket \Varid{e}_{1}\rrbracket^{\rho}\mathrel{=}\mathbf{false}} \\
\ensuremath{\bot} & \text{ if } \ensuremath{\llbracket \Varid{e}_{1}\rrbracket^{\rho}\mathrel{=}\bot}
\end{cases}\\
&\ensuremath{\llbracket \Varid{base}(\Varid{e})\rrbracket^{\rho}\mathrel{=}}
\begin{cases}
\ensuremath{\Varid{n}} & \text{ if } \ensuremath{\llbracket \Varid{e}\rrbracket^{\rho}\mathrel{=}\{\mskip1.5mu \Varid{n},\anonymous ,\anonymous \mskip1.5mu\}} \\
\ensuremath{\bot} & \text{ if } \ensuremath{\llbracket \Varid{e}\rrbracket^{\rho}\mathrel{=}\bot}
\end{cases} \\
&\ensuremath{\llbracket \mathit{length}(\Varid{e})\rrbracket^{\rho}\mathrel{=}}
\begin{cases}
\ensuremath{\Varid{n}} & \text{ if } \ensuremath{\llbracket \Varid{e}\rrbracket^{\rho}\mathrel{=}\{\mskip1.5mu \anonymous ,\Varid{n},\anonymous \mskip1.5mu\}} \\
\ensuremath{\bot} & \text{ if } \ensuremath{\llbracket \Varid{e}\rrbracket^{\rho}\mathrel{=}\bot}
\end{cases}
\end{align*}
\captionsetup{format=caption-with-line}
\captionof{figure}{Evaluation Function.}
\end{minipage}
\vspace{.5\baselineskip}
\captionsetup{format=caption-with-line}
\caption{Helper functions.}
\end{subfigure}%
\\
\caption{Full semantics (continued).}
\end{figure}

\begin{figure}[p]
\ContinuedFloat
\centering
\begin{subfigure}{.9\textwidth}
\begin{mathpar}
\inferrule[Fetch-Protect-Ptr]
{\ensuremath{\Varid{c}}\ \ensuremath{\mathrel{=}}\ \ensuremath{\Varid{x}\mathbin{:=}\mathbf{protect}(\mathbin{*}_{\ell}\;\Varid{e})} \\\\
 \ensuremath{\Varid{c}_{1}}\ \ensuremath{\mathrel{=}}\ \ensuremath{\Varid{x'}\mathbin{:=}\mathbin{*}_{\ell}\;\Varid{e}} \\
 \ensuremath{\Varid{c}_{2}}\ \ensuremath{\mathrel{=}}\ \ensuremath{\Varid{x}\mathbin{:=}\mathbf{protect}(\Varid{x'})}}
{\ensuremath{\langle\Varid{is},\Varid{c}\mathbin{:}\Varid{cs},\mu,\rho\rangle\stepT{\mathbf{fetch}}{\epsilon}\langle\Varid{is},\Varid{c}_{1}\mathbin{:}\Varid{c}_{2}\mathbin{:}\Varid{cs},\mu,\rho\rangle}}
\and
\inferrule[Fetch-Protect-Array]
{\ensuremath{\Varid{c}}\ \ensuremath{\mathrel{=}}\ \ensuremath{\Varid{x}\mathbin{:=}\mathbf{protect}(\Varid{a}{[}\Varid{e}{]})} \\\\
 \ensuremath{\Varid{c}_{1}}\ \ensuremath{\mathrel{=}}\ \ensuremath{\Varid{x'}\mathbin{:=}\Varid{a}{[}\Varid{e}{]}} \\
 \ensuremath{\Varid{c}_{2}}\ \ensuremath{\mathrel{=}}\ \ensuremath{\Varid{x}\mathbin{:=}\mathbf{protect}(\Varid{x'})}}
{\ensuremath{\langle\Varid{is},\Varid{c}\mathbin{:}\Varid{cs},\mu,\rho\rangle\stepT{\mathbf{fetch}}{\epsilon}\langle\Varid{is},\Varid{c}_{1}\mathbin{:}\Varid{c}_{2}\mathbin{:}\Varid{cs},\mu,\rho\rangle}}
\and
\inferrule[Fetch-Protect-Expr]
{\ensuremath{\Varid{c}}\ \ensuremath{\mathrel{=}}\ \ensuremath{\Varid{x}\mathbin{:=}\mathbf{protect}(\Varid{e})} \\
 \ensuremath{\Varid{i}}\ \ensuremath{\mathrel{=}}\ \ensuremath{\Varid{x}\mathbin{:=}\mathbf{protect}(\Varid{e})}}
{\ensuremath{\langle\Varid{is},\Varid{c}\mathbin{:}\Varid{cs},\mu,\rho\rangle\stepT{\mathbf{fetch}}{\epsilon}\langle\Varid{is}+{\mkern-9mu+}\ [\mskip1.5mu \Varid{i}\mskip1.5mu],\Varid{cs},\mu,\rho\rangle}}
\and
\inferrule[Exec-Protect$_1$]
{\ensuremath{\Varid{i}}\ \ensuremath{\mathrel{=}}\ \ensuremath{\Varid{x}\mathbin{:=}\mathbf{protect}(\Varid{e})} \\
 \ensuremath{\Varid{v}\mathrel{=}\llbracket \Varid{e}\rrbracket^{\rho}} \\
 \ensuremath{\Varid{i'}}\ \ensuremath{\mathrel{=}}\ \ensuremath{\Varid{x}\mathbin{:=}\mathbf{protect}(\Varid{v})} }
{\ensuremath{\langle\Varid{is}_{1},\Varid{i},\Varid{is}_{2},\Varid{cs}\rangle\xrsquigarrow{(\mu,\rho,\epsilon)}\langle\Varid{is}_{1}+{\mkern-9mu+}\ [\mskip1.5mu \Varid{i'}\mskip1.5mu]+{\mkern-9mu+}\ \Varid{is}_{2},\Varid{cs}\rangle}}
\and
\inferrule[Exec-Protect$_2$]
{\ensuremath{\Varid{i}}\ \ensuremath{\mathrel{=}}\ \ensuremath{\Varid{x}\mathbin{:=}\mathbf{protect}(\Varid{v})} \\
 \ensuremath{\mathbf{guard}(\anonymous ,\anonymous ,\anonymous )\;\not\in\;\Varid{is}_{1}} \\
 \ensuremath{\Varid{i'}}\ \ensuremath{\mathrel{=}}\ \ensuremath{\Varid{x}\mathbin{:=}\Varid{v}}}
{\ensuremath{\langle\Varid{is}_{1},\Varid{i},\Varid{is}_{2},\Varid{cs}\rangle\xrsquigarrow{(\mu,\rho,\epsilon)}\langle\Varid{is}_{1}+{\mkern-9mu+}\ [\mskip1.5mu \Varid{i'}\mskip1.5mu]+{\mkern-9mu+}\ \Varid{is}_{2},\Varid{cs}\rangle}}
\end{mathpar}
\captionsetup{format=caption-with-line}
\caption{Semantics of hardware-based \ensuremath{\mathbf{protect}}. }
\label{fig:full-semantics-protect}
\end{subfigure}%
\\
\begin{subfigure}{.9\textwidth}
\vspace{\baselineskip}
\begin{mathpar}
\inferrule[Fetch-Protect-SLH]
{\ensuremath{\Varid{c}}\ \ensuremath{\mathrel{=}}\ \ensuremath{\Varid{x}\mathbin{:=}\mathbf{protect}(\Varid{a}{[}\Varid{e}{]})} \\
 \ensuremath{\Varid{e}_{1}}\ \ensuremath{\mathrel{=}}\ \ensuremath{\Varid{e}_{1}<\mathit{length}(\Varid{a})} \\
 \ensuremath{\Varid{e}_{2}}\ \ensuremath{\mathrel{=}}\ \ensuremath{\Varid{base}(\Varid{a})\mathbin{+}\Varid{e}} \\
 \ensuremath{\ell\mathrel{=}\Varid{label}(\Varid{a})} \\
 \ensuremath{\Varid{c}_{1}}\ \ensuremath{\mathrel{=}}\ \ensuremath{\Varid{r}\mathbin{:=}\Varid{e}_{1}} \\
 \ensuremath{\Varid{c}_{2}}\ \ensuremath{\mathrel{=}}\ \ensuremath{\Varid{r}\mathbin{:=}\Varid{r}}\ \ensuremath{\mathbin{?}}\ \ensuremath{\overline{\mathrm{1}}}\ \ensuremath{\mathbin{:}}\ \ensuremath{\overline{\mathrm{0}}} \\
 \ensuremath{\Varid{c}_{3}}\ \ensuremath{\mathrel{=}}\ \ensuremath{\Varid{x}\mathbin{:=}\mathbin{*}_{\ell}\;(\Varid{e}_{2}\;\otimes\;\Varid{r})} \\
 \ensuremath{\Varid{c'}}\ \ensuremath{\mathrel{=}}\ \ensuremath{\Varid{c}_{1};}\ \ensuremath{\mathbf{if}}\ \ensuremath{\Varid{r}}\ \ensuremath{\mathbf{then}}\ \ensuremath{\Varid{c}_{2};}\ \ensuremath{\Varid{c}_{3}}\ \ensuremath{\mathbf{else}}\ \ensuremath{\mathbf{fail}}}
{\ensuremath{\langle\Varid{is},\Varid{c}\mathbin{:}\Varid{cs},\mu,\rho\rangle\stepT{\mathbf{fetch}}{\epsilon}\langle\Varid{is},\Varid{c'}\mathbin{:}\Varid{cs},\mu,\rho\rangle}}
\end{mathpar}
\captionsetup{format=caption-with-line}
\caption{Semantics of SLH-based \ensuremath{\mathbf{protect}}.}
\end{subfigure}

\begin{subfigure}{.9\textwidth}
\vspace{\baselineskip}
\begin{mathpar}
\inferrule[Done]
{}
{\ensuremath{\langle[\mskip1.5mu \mskip1.5mu],[\mskip1.5mu \mskip1.5mu],\mu,\rho\rangle\Downarrow_{\epsilon}^{[\mskip1.5mu \mskip1.5mu]}\langle[\mskip1.5mu \mskip1.5mu],[\mskip1.5mu \mskip1.5mu],\mu,\rho\rangle}}
\and
\inferrule[Step]
{\ensuremath{\langle\Varid{is},\Varid{cs},\mu,\rho\rangle\stepT{\Varid{d}}{\Varid{o}}\langle\Varid{is'},\Varid{cs'},\mu',\rho'\rangle} \\
 \ensuremath{\langle\Varid{is'},\Varid{cs'},\mu',\rho'\rangle\Downarrow_{\Conid{O}}^{\Conid{D}}\langle\Varid{is''},\Varid{cs''},\mu'',\rho''\rangle}}
{\ensuremath{\langle\Varid{is},\Varid{cs},\mu,\rho\rangle\Downarrow_{(\Varid{o}\;{{\mkern-2mu\cdot\mkern-2mu}}\;\Conid{O})}^{\Varid{d}\mathbin{:}\Conid{D}}\langle\Varid{is''},\Varid{cs''},\mu'',\rho''\rangle}}
\end{mathpar}
\captionsetup{format=caption-with-line}
\caption{Speculative big-step semantics.}
\label{fig:speculative-big-step}
\end{subfigure}

\caption{Full semantics (continued).}
\label{fig:full-calculus}
\end{figure}
\clearpage
\subsection{Full Type System}
\label{app:type-inference}
\begin{figure}[p]
\centering
\begin{subfigure}{\textwidth}
\begin{align*}
\text{Transient-flow Lattice :} & \quad \ensuremath{\mathscr{T}\mathrel{=}\{\mskip1.5mu \ensuremath{\Concrete},\ensuremath{\Transient}\mskip1.5mu\}}  \quad \text{ where } \quad \ensuremath{\ensuremath{\Concrete}\;\flows\;\tau}, \ensuremath{\ensuremath{\Transient}\;\flows\;\ensuremath{\Transient}}, \ensuremath{\ensuremath{\Transient}\;\not\flows\;\ensuremath{\Concrete}} \\
\text{Transient-flow types: } & \quad \ensuremath{\tau\;\in\;\mathscr{T}} \\
\text{Typing Context: } & \quad \ensuremath{\Gamma\;\in\;\Conid{Var}\rightharpoonup\mathscr{T}}
\end{align*}
\vspace{-\baselineskip}
\captionsetup{format=caption-with-line}
\caption{Transient-flow lattice and syntax.}
\end{subfigure}

\begin{subfigure}[b]{\textwidth}
    \rulesize
\begin{mathpar}
\inferrule[Value]
{}
{\ensuremath{\Gamma\vdash\Varid{v}\mathbin{:}\tau} \ \shadedCons{\cnsSep \emptyset}}
\and
\inferrule[Var]
{\ensuremath{\Gamma(\Varid{x})\mathrel{=}\tau}}
{\ensuremath{\Gamma\vdash\Varid{x}\mathbin{:}\tau} \; \shadedCons{\cnsSep \ x \flows \alpha_x} }
\and
\inferrule[Proj]
{\ensuremath{\Varid{f}\;\in\;\{\mskip1.5mu \mathit{length}(\cdot ),\Varid{base}(\cdot )\mskip1.5mu\}} \\
 \ensuremath{\Gamma\vdash\Varid{e}\mathbin{:}\tau_{1}} \; \shadedCons{\cnsSep k} \\
 \ensuremath{\tau_{1}\;\flows\;\tau} }
{\ensuremath{\Gamma\vdash\Varid{f}(\Varid{e})\mathbin{:}\tau} \; \shadedCons{\cnsSep k \cup (e \flows \ensuremath{\Varid{f}(\Varid{e})})}}
\and
\inferrule[Bop]
{ \ensuremath{\oplus\;\in\;\{\mskip1.5mu \mathbin{+},\leq \mskip1.5mu\}} \\
  \ensuremath{\Varid{i}\;\in\;\{\mskip1.5mu \mathrm{1},\mathrm{2}\mskip1.5mu\}} \\ \ensuremath{\Gamma\vdash\Varid{e}_{\Varid{i}}\mathbin{:}\tau_{\Varid{i}}\;\shadedCons{\cnsSep\;\Varid{k}_{\Varid{i}}}} \\ \ensuremath{\tau_{\Varid{i}}\;\flows\;\tau}}
{\ensuremath{\Gamma\vdash\Varid{e}_{1}\;\oplus\;\Varid{e}_{2}\mathbin{:}\tau}  \; \shadedCons{\cnsSep k_1 \cup k_2
    \cup (e_1 \flows e_1 \oplus e_2 )
     \cup (e_2 \flows e_1 \oplus e_2 )}
}
\and
\inferrule[Select]
{\ensuremath{\Varid{i}\;\in\;\{\mskip1.5mu \mathrm{1},\mathrm{2},\mathrm{3}\mskip1.5mu\}} \\ \ensuremath{\Gamma\vdash\Varid{e}_{\Varid{i}}\mathbin{:}\tau_{\Varid{i}}\;\shadedCons{\cnsSep\;\Varid{k}_{\Varid{i}}}} \\ \ensuremath{\tau_{\Varid{i}}\;\flows\;\tau}}
{\ensuremath{\Gamma\vdash\Varid{e}_{1}\mathbin{?}\Varid{e}_{2}\mathbin{:}\Varid{e}_{3}\mathbin{:}\tau\;\shadedCons{\cnsSep\;\Varid{k}_{1}\; \cup \;\Varid{k}_{2}\; \cup \;\Varid{k}_{3}\; \cup \;\Varid{e}_{\Varid{i}}\;\flows\;(\Varid{e}_{1}\mathbin{?}\Varid{e}_{2}\mathbin{:}\Varid{e}_{3})}}}
\and
\inferrule[Ptr-Read]
{\ensuremath{\Gamma\vdash\Varid{e}\mathbin{:}\ensuremath{\Concrete}} \; \shadedCons{\cnsSep k}}
{\ensuremath{\Gamma\vdash\mathbin{*}_{\ell}\;\Varid{e}\mathbin{:}\ensuremath{\Transient}} \; \shadedCons{\cnsSep k \cup (e \flows \Concrete) \cup (\Transient \flows e)}}
\and
\inferrule[Array-Read]
{\ensuremath{\Gamma\vdash\Varid{e}\mathbin{:}\ensuremath{\Concrete}} \; \shadedCons{\cnsSep k}}
{\ensuremath{\Gamma\vdash\Varid{a}{[}\Varid{e}{]}\mathbin{:}\ensuremath{\Transient}} \; \shadedCons{\cnsSep k \
 \cup ( e \flows \Concrete ) \cup  (\Transient \flows a[e])}
}
\end{mathpar}
\captionsetup{format=caption-with-line}
\caption{Typing Rules for Expressions and Arrays.}
\end{subfigure}%

\end{figure}

\begin{figure}[t]
\ContinuedFloat
\begin{subfigure}[b]{\textwidth}
\rulesize
\begin{mathpar}
\inferrule[Skip]
{}
{\ensuremath{\Gamma\vdash\mathbf{skip}} \  \shadedCons{\cnsSep \emptyset}}
\and
\inferrule[Fail]
{}
{\ensuremath{\Gamma\vdash\mathbf{fail}} \  \shadedCons{\cnsSep \emptyset}}
\and
\inferrule[Seq]
{\ensuremath{\Gamma\vdash\Varid{c}_{1}} \ \shadedCons{\cnsSep k_1} \\
 \ensuremath{\Gamma\vdash\Varid{c}_{2}} \ \shadedCons{\cnsSep k_2} }
{\ensuremath{\Gamma\vdash\Varid{c}_{1};\Varid{c}_{2}} \ \shadedCons{\cnsSep k_1 \cup k_2 }}
\and
\inferrule[Asgn]
{\ensuremath{\Gamma\vdash\Varid{r}\mathbin{:}\tau} \ \shadedCons{\cnsSep k} \\
 \ensuremath{\tau\;\flows\;\Gamma(\Varid{x})}}
{\ensuremath{\Gamma,\protectedSet\vdash\Varid{x}\mathbin{:=}\Varid{r}} \ \shadedCons{\cnsSep k \cup (r \flows x)}}
\and
\inferrule[Ptr-Write]
{\ensuremath{\Gamma\vdash\Varid{e}_{1}\mathbin{:}\ensuremath{\Concrete}} \ \shadedCons{\cnsSep k_1} \\
 \ensuremath{\Gamma\vdash\Varid{e}_{2}\mathbin{:}\tau} \ \shadedCons{\cnsSep k_2}  }
{\ensuremath{\Gamma,\protectedSet\vdash\ast\Varid{e}_{1}\mathbin{:=}\Varid{e}_{2}} \ \shadedCons{\cnsSep k_1 \cup k_2 \cup (e_1 \flows \Concrete)}}
\and
\inferrule[Ptr-Write-Spectre-1.1]
{\ensuremath{\Gamma\vdash\Varid{e}_{1}\mathbin{:}\ensuremath{\Concrete}} \ \shadedCons{\cnsSep k_1} \\
 \ensuremath{\Gamma\vdash\Varid{e}_{2}\mathbin{:}\ensuremath{\Concrete}} \ \shadedCons{\cnsSep k_2}  }
{\ensuremath{\Gamma,\protectedSet\vdash\ast\Varid{e}_{1}\mathbin{:=}\Varid{e}_{2}} \ \shadedCons{\cnsSep k_1 \cup k_2 \cup (e_1 \flows \Concrete) \cup (e_2 \flows \Concrete)}}
\and
\inferrule[Array-Write]
{\ensuremath{\Gamma\vdash\Varid{e}_{1}\mathbin{:}\ensuremath{\Concrete}} \ \shadedCons{\cnsSep k_1} \\
 \ensuremath{\Gamma\vdash\Varid{e}_{2}\mathbin{:}\tau} \ \shadedCons{\cnsSep k_2} }
{\ensuremath{\Gamma,\protectedSet\vdash\Varid{a}{[}\Varid{e}_{1}{]}\mathbin{:=}\Varid{e}_{2}} \ \shadedCons{\cnsSep k_1 \cup k_2 \cup (e_1 \flows \Concrete) }}
\and
\inferrule[Array-Write-Spectre-1.1]
{\ensuremath{\Gamma\vdash\Varid{e}_{1}\mathbin{:}\ensuremath{\Concrete}} \ \shadedCons{\cnsSep k_1} \\
 \ensuremath{\Gamma\vdash\Varid{e}_{2}\mathbin{:}\ensuremath{\Concrete}} \ \shadedCons{\cnsSep k_2} }
{\ensuremath{\Gamma,\protectedSet\vdash\Varid{a}{[}\Varid{e}_{1}{]}\mathbin{:=}\Varid{e}_{2}} \ \shadedCons{\cnsSep k_1 \cup k_2 \cup (e_1 \flows \Concrete) \cup (e_2 \flows \Concrete)}}
\and
\inferrule[Protect]
{\ensuremath{\Gamma\vdash\Varid{r}\mathbin{:}\tau} \ \shadedCons{\cnsSep k}}
{\ensuremath{\Gamma,\protectedSet\vdash\Varid{x}\mathbin{:=}\mathbf{protect}(\Varid{r})} \ \shadedCons{\cnsSep k}}
\and
\inferrule[Asgn-Prot]
{\ensuremath{\Gamma\vdash\Varid{r}\mathbin{:}\tau} \ \shadedCons{\cnsSep k}\\
\ensuremath{\Varid{x}\;\in\;\protectedSet} }
{\ensuremath{\Gamma,\protectedSet\vdash\Varid{x}\mathbin{:=}\Varid{r}} \ \shadedCons{\cnsSep k \cup (r \flows x)}}
\and
\inferrule[If-Then-Else]
{\ensuremath{\Gamma\vdash\Varid{e}\mathbin{:}\ensuremath{\Concrete}} \; \shadedCons{\cnsSep k} \\ \ensuremath{\Gamma,\protectedSet\vdash\Varid{c}_{\mathrm{1}}} \; \shadedCons{\cnsSep k_1} \\ \ensuremath{\Gamma,\protectedSet\vdash\Varid{c}_{\mathrm{2}}} \; \shadedCons{\cnsSep k_2}}
{\ensuremath{\Gamma,\protectedSet\vdash\mathbf{if}\;\Varid{e}\;\mathbf{then}\;\Varid{c}_{1}\;\mathbf{else}\;\Varid{c}_{2}} \; \shadedCons{\cnsSep k \cup
  k_1 \cup k_2 \cup (e \flows \Concrete)}}
\and
\inferrule[While]
{\ensuremath{\Gamma\vdash\Varid{e}\mathbin{:}\ensuremath{\Concrete}} \; \shadedCons{\cnsSep k_1} \\
 \ensuremath{\Gamma,\protectedSet\vdash\Varid{c}} \; \shadedCons{\cnsSep k_2}}
{\ensuremath{\Gamma,\protectedSet\vdash\mathbf{while}\;\Varid{e}\;\mathbf{do}\;\Varid{c}} \; \shadedCons{\cnsSep k_1 \cup k_2 \cup (e \flows \Concrete)}}
\end{mathpar}
\captionsetup{format=caption-with-line}
\caption{Typing Rules for Commands.}
\end{subfigure}
\begin{subfigure}{\textwidth}
\centering
\begin{align*}
\text{Atoms }&  \quad a\ \ensuremath{\Coloneqq} \  \alpha_x \mid\ r \\
\text{Constraints }& \quad k\ \ensuremath{\Coloneqq} \ a \flows \Concrete\ \mid\ \Transient \flows a\ \mid\ a \flows a \mid\ k \cup k\ \mid\ \emptyset \\
\text{Solutions }& \quad \sigma\ \in\  \textsc{Atoms}\ \uplus\ \ensuremath{\mathscr{T}}\ \mapsto\ \ensuremath{\mathscr{T}} \quad \text{where} \quad \ensuremath{\sigma(\ensuremath{\Concrete})\mathrel{=}\ensuremath{\Concrete}}, \ensuremath{\sigma(\ensuremath{\Transient})\mathrel{=}\ensuremath{\Transient}}
\end{align*}
\rulesize
  \begin{mathpar}
  \inferrule[Sol-Transient]
  {\Transient \flows \sigma(a_2)}
  {\sigma \vdash \Transient \flows a_2}
  \and
  \inferrule[Sol-Stable]
  {\sigma(a_1) \flows \Concrete}
  {\sigma \vdash a_1 \flows \Concrete}
  \and
  \inferrule[Sol-Flow]
  {\sigma(a_1) \flows \sigma(a_2)}
  {\sigma \vdash a_1 \flows a_2}
  \and
  \inferrule[Sol-Set]
  {\sigma \vdash c_1 \dots \sigma \vdash c_n}
  {\sigma \vdash \theSet{c_1, \dots,  c_n}}

\end{mathpar}
\captionsetup{format=caption-with-line}
\caption{Type Constraints and Satisfiability.}
\end{subfigure}%

\caption{Transient-flow type system and \colorbox{constColor}{type
    constraints generation}, and constraints satisfiability.\label{fig:app-transient-ts}}
\end{figure}

\begin{lemma}[Transient-Flow Type Preservation]
If \ensuremath{\Gamma\vdash\Conid{C}} and \ensuremath{\Conid{C}\;\stepT{\Varid{d}}{\Conid{O}}\;\Conid{C}'}, then \ensuremath{\Gamma\vdash\Conid{C}'}.
\label{lemma:tr-preservation}
\end{lemma}
\begin{proof}
  By case analysis on the reduction step and typing judgment.
\end{proof}

\begin{figure}[!th]
\centering

\begin{subfigure}{\textwidth}
\begin{mathpar}
\inferrule[Nop]
{}
{\ensuremath{\Gamma\vdash\mathbf{nop}}}
\and
\inferrule[Fail]
{}
{\ensuremath{\Gamma\vdash\mathbf{fail}(\Varid{p})}}
\and
\inferrule[Asgn]
{\ensuremath{\Gamma\vdash\Varid{e}\mathbin{:}\tau} \\ \ensuremath{\tau\;\flows\;\Gamma(\Varid{x})} }
{\ensuremath{\Gamma\vdash\Varid{x}\mathbin{:=}\Varid{e}}}
\and
\inferrule[Protect]
{\ensuremath{\Gamma\vdash\Varid{e}\mathbin{:}\tau}}
{\ensuremath{\Gamma\vdash\Varid{x}\mathbin{:=}\mathbf{protect}(\Varid{e})}}
\and
\inferrule[Load]
{\ensuremath{\Gamma\vdash\Varid{e}\mathbin{:}\ensuremath{\Concrete}}}
{\ensuremath{\Gamma\vdash\Varid{x}\mathbin{:=}\mathbf{load}_{\ell}(\Varid{e})}}
\and
\inferrule[Store]
{\ensuremath{\Gamma\vdash\Varid{e}_{1}\mathbin{:}\ensuremath{\Concrete}} \\
 \ensuremath{\Gamma\vdash\Varid{e}_{2}\mathbin{:}\tau}}
{\ensuremath{\Gamma\vdash\mathbf{store}_{\ell}(\Varid{e}_{1},\Varid{e}_{2})}}
\and
\inferrule[Store-Spectre-1.1]
{\ensuremath{\Gamma\vdash\Varid{e}_{1}\mathbin{:}\ensuremath{\Concrete}} \\
 \ensuremath{\Gamma\vdash\Varid{e}_{2}\mathbin{:}\ensuremath{\Concrete}}}
{\ensuremath{\Gamma\vdash\mathbf{store}_{\ell}(\Varid{e}_{1},\Varid{e}_{2})}}
\and
\inferrule[Guard]
{\ensuremath{\Gamma\vdash\Varid{e}\mathbin{:}\ensuremath{\Concrete}} \\
 \ensuremath{\Gamma\vdash\Varid{cs}}}
{\ensuremath{\Gamma\vdash\mathbf{guard}(\Varid{e}^{\Varid{b}},\Varid{cs},\Varid{p})}}
\end{mathpar}
\captionsetup{format=caption-with-line}
\caption{Instructions \ensuremath{\Gamma\vdash\Varid{i}}. }
\label{fig:full-ts-instr}
\end{subfigure}%

\begin{subfigure}{\textwidth}
\begin{mathpar}
\inferrule[Cmd-Stack-Empty]
{}
{\ensuremath{\Gamma\vdash[\mskip1.5mu \mskip1.5mu]}}
\and
\inferrule[Cmd-Stack-Cons]
{\ensuremath{\Gamma\vdash\Varid{c}} \\ \ensuremath{\Gamma\vdash\Varid{cs}}}
{\ensuremath{\Gamma\vdash\Varid{c}\mathbin{:}\Varid{cs}}}
\and
\inferrule[RB-Empty]
{}
{\ensuremath{\Gamma\vdash[\mskip1.5mu \mskip1.5mu]}}
\and
\inferrule[RB-Cons]
{\ensuremath{\Gamma\vdash\Varid{i}} \\ \ensuremath{\Gamma\vdash\Varid{is}}}
{\ensuremath{\Gamma\vdash\Varid{i}\mathbin{:}\Varid{is}}}
\end{mathpar}
\captionsetup{format=caption-with-line}
\caption{Command stack \ensuremath{\Gamma\vdash\Varid{cs}} and reorder buffer \ensuremath{\Gamma\vdash\Varid{is}} .}
\end{subfigure}%

\begin{subfigure}{\textwidth}
\begin{mathpar}
\inferrule[Conf]
{\ensuremath{\Gamma\vdash\Varid{is}} \\ \ensuremath{\Gamma\vdash\Varid{cs}}}
{\ensuremath{\Gamma\vdash\langle\Varid{is},\Varid{cs},\mu,\rho\rangle}}
\end{mathpar}
\captionsetup{format=caption-with-line}
\caption{Configurations \ensuremath{\Gamma\vdash\Conid{C}}.}
\end{subfigure}

\caption{Typing rules for the speculative processor.}
\label{fig:full-transient-flow-ts}
\end{figure}

\clearpage
\newpage
\subsection{Constant-Time Type System}
\begin{figure}[t]
\centering
\begin{subfigure}{.9\textwidth}
\begin{align*}
\text{Security Lattice :} & \quad \ensuremath{\mathscr{L}\mathrel{=}\{\mskip1.5mu \Blue{\Conid{L}},\Red{\Conid{H}}\mskip1.5mu\}}  \quad \text{ where } \quad \ensuremath{\Blue{\Conid{L}}\;\flows\;\ell}, \ensuremath{\Red{\Conid{H}}\;\flows\;\Red{\Conid{H}}}, \ensuremath{\Red{\Conid{H}}\;\not\flows\;\Blue{\Conid{L}}} \\
\text{Security Labels: } & \quad \ensuremath{\ell\;\in\;\mathscr{L}} \\
\text{Typing Context: } & \quad \ensuremath{\Gamma\;\in\;\Conid{Var}\rightharpoonup\mathscr{L}}
\end{align*}
\vspace{-\baselineskip}
\captionsetup{format=caption-with-line}
\caption{Two-point security lattice and syntax.}
\label{fig:2-lattice}
\end{subfigure}%

\begin{subfigure}{.9\textwidth}
\rulesize
\begin{mathpar}
\inferrule[Var]
{}
{\ensuremath{\Gamma\;\vdash_{\textsc{ct}}\;\Varid{x}\mathbin{:}\Gamma(\Varid{x})}}
\and
\inferrule[Array]
{\ensuremath{\Varid{a}\mathrel{=}\{\mskip1.5mu \anonymous ,\anonymous ,\ell\mskip1.5mu\}} \\ \ensuremath{\ell\mathrel{=}\Gamma(\Varid{a})}}
{\ensuremath{\Gamma\;\vdash_{\textsc{ct}}\;\Varid{a}\mathbin{:}\ell}}
\and
\inferrule[Val]
{\ensuremath{\Varid{v}\neq\Varid{a}} }
{\ensuremath{\Gamma\;\vdash_{\textsc{ct}}\;\Varid{v}\mathbin{:}\ell}}
\and
\inferrule[Proj]
{\ensuremath{\Varid{f}\;\in\;\{\mskip1.5mu \mathit{length}(\cdot ),\Varid{base}(\cdot )\mskip1.5mu\}} \\ \ensuremath{\Gamma\;\vdash_{\textsc{ct}}\;\Varid{e}\mathbin{:}\ell}}
{\ensuremath{\Gamma\;\vdash_{\textsc{ct}}\;\Varid{f}(\Varid{e})\mathbin{:}\Blue{\Conid{L}}}}
\and
\inferrule[Bop]
{\ensuremath{\Gamma\;\vdash_{\textsc{ct}}\;\Varid{e}_{1}\mathbin{:}\ell} \\
 \ensuremath{\Gamma\;\vdash_{\textsc{ct}}\;\Varid{e}_{2}\mathbin{:}\ell} }
{\ensuremath{\Gamma\;\vdash_{\textsc{ct}}\;(\Varid{e}_{1}\;\oplus\;\Varid{e}_{2})\mathbin{:}\ell}}
\and
\inferrule[Select]
{\ensuremath{\Gamma\;\vdash_{\textsc{ct}}\;\Varid{e}_{1}\mathbin{:}\Blue{\Conid{L}}} \\
 \ensuremath{\Gamma\;\vdash_{\textsc{ct}}\;\Varid{e}_{2}\mathbin{:}\ell} \\
 \ensuremath{\Gamma\;\vdash_{\textsc{ct}}\;\Varid{e}_{3}\mathbin{:}\ell}}
{\ensuremath{\Gamma\;\vdash_{\textsc{ct}}\;(\Varid{e}_{1}\mathbin{?}\Varid{e}_{2}\mathbin{:}\Varid{e}_{3})\mathbin{:}\ell}}
\and
\inferrule[Array-Read]
{\ensuremath{\Gamma\;\vdash_{\textsc{ct}}\;\Varid{e}_{1}\mathbin{:}\ell} \\
 \ensuremath{\Gamma\;\vdash_{\textsc{ct}}\;\Varid{e}_{2}\mathbin{:}\Blue{\Conid{L}}}}
{\ensuremath{\Gamma\;\vdash_{\textsc{ct}}\;\Varid{e}_{1}{[}\Varid{e}_{2}{]}\mathbin{:}\ell}}
\and
\inferrule[Ptr-Read]
{\ensuremath{\Gamma\;\vdash_{\textsc{ct}}\;\Varid{e}\mathbin{:}\Blue{\Conid{L}}}}
{\ensuremath{\Gamma\;\vdash_{\textsc{ct}}\;(\mathbin{*}_{\ell}\;\Varid{e})\mathbin{:}\ell}}
\and
\inferrule[Sub]
{\ensuremath{\Gamma\;\vdash_{\textsc{ct}}\;\Varid{e}\mathbin{:}\ell_{1}} \\ \ensuremath{\ell_{1}\;\flows\;\ell_{2}}}
{\ensuremath{\Gamma\;\vdash_{\textsc{ct}}\;\Varid{e}\mathbin{:}\ell_{2}}}
\end{mathpar}
\captionsetup{format=caption-with-line}
\caption{Expressions: \ensuremath{\Gamma\;\vdash_{\textsc{ct}}\;\Varid{e}\mathbin{:}\ell}.}
\label{fig:ct-ts-expr}
\end{subfigure}%

\begin{subfigure}{.9\textwidth}
\begin{mathpar}
\inferrule[Skip]
{}
{\ensuremath{\Gamma\;\vdash_{\textsc{ct}}\;\mathbf{skip}}}
\and
\inferrule[Fail]
{}
{\ensuremath{\Gamma\;\vdash_{\textsc{ct}}\;\mathbf{fail}}}
\and
\inferrule[Asgn]
{\ensuremath{\Gamma\;\vdash_{\textsc{ct}}\;\Varid{r}\mathbin{:}\ell} \\ \ensuremath{\ell\;\flows\;\Gamma(\Varid{x})}}
{\ensuremath{\Gamma\;\vdash_{\textsc{ct}}\;\Varid{x}\mathbin{:=}\Varid{r}}}
\and
\inferrule[Protect]
{\ensuremath{\Gamma\;\vdash_{\textsc{ct}}\;\Varid{r}\mathbin{:}\ell} \\ \ensuremath{\ell\;\flows\;\Gamma(\Varid{x})}}
{\ensuremath{\Gamma\;\vdash_{\textsc{ct}}\;\Varid{x}\mathbin{:=}\mathbf{protect}(\Varid{r})}}
\and
\inferrule[Array-Write]
{\ensuremath{\Gamma\;\vdash_{\textsc{ct}}\;\Varid{e}_{1}\mathbin{:}\ell_{1}} \\
 \ensuremath{\Gamma\;\vdash_{\textsc{ct}}\;\Varid{e}_{2}\mathbin{:}\Blue{\Conid{L}}} \\
 \ensuremath{\Gamma\;\vdash_{\textsc{ct}}\;\Varid{e}\mathbin{:}\ell} \\ \ensuremath{\ell\;\flows\;\ell_{1}}}
{\ensuremath{\Gamma\;\vdash_{\textsc{ct}}\;\Varid{e}_{1}{[}\Varid{e}_{2}{]}\mathbin{:=}\Varid{e}}}
\and
\inferrule[Ptr-Write]
{\ensuremath{\Gamma\;\vdash_{\textsc{ct}}\;\Varid{e}_{1}\mathbin{:}\Blue{\Conid{L}}} \\
 \ensuremath{\Gamma\;\vdash_{\textsc{ct}}\;\Varid{e}_{2}\mathbin{:}\ell_{2}} \\ \ensuremath{\ell_{2}\;\flows\;\ell}}
{\ensuremath{\Gamma\;\vdash_{\textsc{ct}}\;(\mathbin{*}_{\ell})\;\Varid{e}_{1}\mathbin{:=}\Varid{e}_{2}}}
\and
\inferrule[If]
{\ensuremath{\Gamma\;\vdash_{\textsc{ct}}\;\Varid{e}\mathbin{:}\Blue{\Conid{L}}} \\
 \ensuremath{\Gamma\;\vdash_{\textsc{ct}}\;\Varid{c}_{1}} \\
 \ensuremath{\Gamma\;\vdash_{\textsc{ct}}\;\Varid{c}_{2}} }
{\ensuremath{\Gamma\;\vdash_{\textsc{ct}}\;\mathbf{if}\;\Varid{e}\;\mathbf{then}\;\Varid{c}_{1}\;\mathbf{else}\;\Varid{c}_{2}}}
\and
\inferrule[While]
{\ensuremath{\Gamma\;\vdash_{\textsc{ct}}\;\Varid{e}\mathbin{:}\Blue{\Conid{L}}} \\
 \ensuremath{\Gamma\;\vdash_{\textsc{ct}}\;\Varid{c}} }
{\ensuremath{\Gamma\;\vdash_{\textsc{ct}}\;\mathbf{while}\;\Varid{e}\;\Varid{c}}}
\and
\inferrule[Seq]
{\ensuremath{\Gamma\;\vdash_{\textsc{ct}}\;\Varid{c}_{1}} \\
 \ensuremath{\Gamma\;\vdash_{\textsc{ct}}\;\Varid{c}_{1}} }
{\ensuremath{\Gamma\;\vdash_{\textsc{ct}}\;\Varid{c}_{1};\Varid{c}_{2}}}
\end{mathpar}
\captionsetup{format=caption-with-line}
\caption{Commands \ensuremath{\Gamma\;\vdash_{\textsc{ct}}\;\Varid{c}}.}
\label{fig:ct-ts-cmds}
\end{subfigure}%
\end{figure}

\begin{figure}[t]
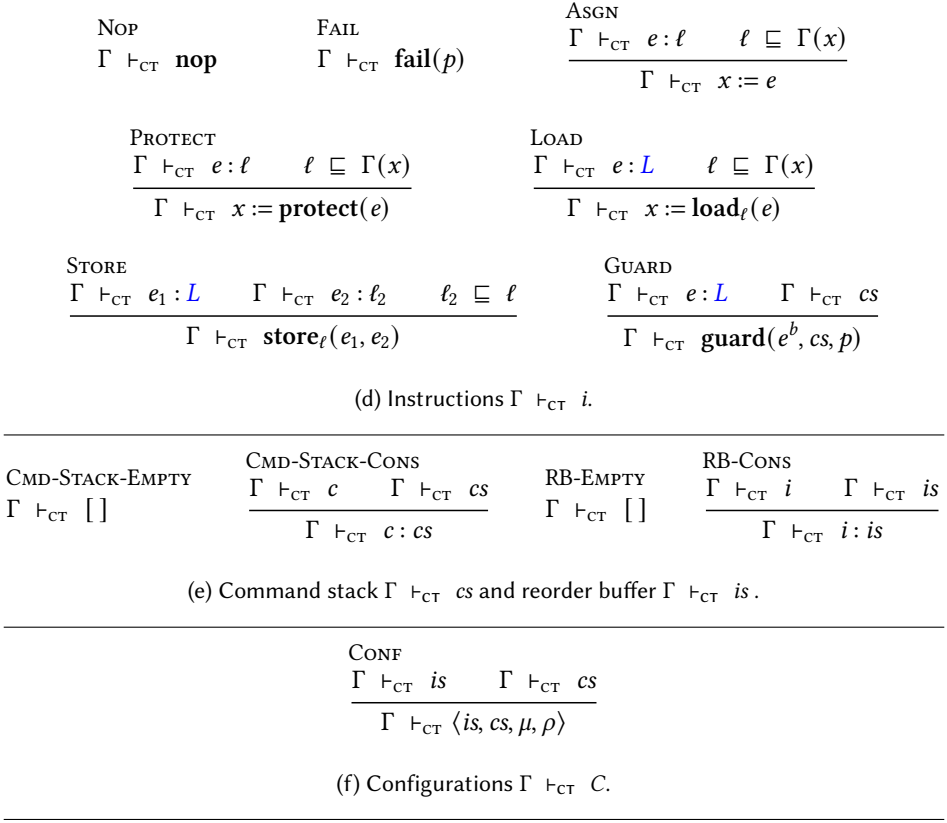

\ContinuedFloat
\begin{subfigure}{.9\textwidth}
\begin{mathpar}
\inferrule[Nop]
{}
{\ensuremath{\Gamma\;\vdash_{\textsc{ct}}\;\mathbf{nop}}}
\and
\inferrule[Fail]
{}
{\ensuremath{\Gamma\;\vdash_{\textsc{ct}}\;\mathbf{fail}(\Varid{p})}}
\and
\inferrule[Asgn]
{\ensuremath{\Gamma\;\vdash_{\textsc{ct}}\;\Varid{e}\mathbin{:}\ell} \\ \ensuremath{\ell\;\flows\;\Gamma(\Varid{x})} }
{\ensuremath{\Gamma\;\vdash_{\textsc{ct}}\;\Varid{x}\mathbin{:=}\Varid{e}}}
\and
\inferrule[Protect]
{\ensuremath{\Gamma\;\vdash_{\textsc{ct}}\;\Varid{e}\mathbin{:}\ell} \\ \ensuremath{\ell\;\flows\;\Gamma(\Varid{x})}}
{\ensuremath{\Gamma\;\vdash_{\textsc{ct}}\;\Varid{x}\mathbin{:=}\mathbf{protect}(\Varid{e})}}
\and
\inferrule[Load]
{\ensuremath{\Gamma\;\vdash_{\textsc{ct}}\;\Varid{e}\mathbin{:}\Blue{\Conid{L}}} \\ \ensuremath{\ell\;\flows\;\Gamma(\Varid{x})}}
{\ensuremath{\Gamma\;\vdash_{\textsc{ct}}\;\Varid{x}\mathbin{:=}\mathbf{load}_{\ell}(\Varid{e})}}
\and
\inferrule[Store]
{\ensuremath{\Gamma\;\vdash_{\textsc{ct}}\;\Varid{e}_{1}\mathbin{:}\Blue{\Conid{L}}} \\
 \ensuremath{\Gamma\;\vdash_{\textsc{ct}}\;\Varid{e}_{2}\mathbin{:}\ell_{2}} \\ \ensuremath{\ell_{2}\;\flows\;\ell}}
{\ensuremath{\Gamma\;\vdash_{\textsc{ct}}\;\mathbf{store}_{\ell}(\Varid{e}_{1},\Varid{e}_{2})}}
\and
\inferrule[Guard]
{\ensuremath{\Gamma\;\vdash_{\textsc{ct}}\;\Varid{e}\mathbin{:}\Blue{\Conid{L}}} \\
 \ensuremath{\Gamma\;\vdash_{\textsc{ct}}\;\Varid{cs}}}
{\ensuremath{\Gamma\;\vdash_{\textsc{ct}}\;\mathbf{guard}(\Varid{e}^{\Varid{b}},\Varid{cs},\Varid{p})}}
\end{mathpar}
\captionsetup{format=caption-with-line}
\caption{Instructions \ensuremath{\Gamma\;\vdash_{\textsc{ct}}\;\Varid{i}}.}
\label{fig:ct-ts-instr}
\end{subfigure}%

\begin{subfigure}{.9\textwidth}
\begin{mathpar}
\inferrule[Cmd-Stack-Empty]
{}
{\ensuremath{\Gamma\;\vdash_{\textsc{ct}}\;[\mskip1.5mu \mskip1.5mu]}}
\and
\inferrule[Cmd-Stack-Cons]
{\ensuremath{\Gamma\;\vdash_{\textsc{ct}}\;\Varid{c}} \\ \ensuremath{\Gamma\;\vdash_{\textsc{ct}}\;\Varid{cs}}}
{\ensuremath{\Gamma\;\vdash_{\textsc{ct}}\;\Varid{c}\mathbin{:}\Varid{cs}}}
\and
\inferrule[RB-Empty]
{}
{\ensuremath{\Gamma\;\vdash_{\textsc{ct}}\;[\mskip1.5mu \mskip1.5mu]}}
\and
\inferrule[RB-Cons]
{\ensuremath{\Gamma\;\vdash_{\textsc{ct}}\;\Varid{i}} \\ \ensuremath{\Gamma\;\vdash_{\textsc{ct}}\;\Varid{is}}}
{\ensuremath{\Gamma\;\vdash_{\textsc{ct}}\;\Varid{i}\mathbin{:}\Varid{is}}}
\end{mathpar}
\captionsetup{format=caption-with-line}
\caption{Command stack \ensuremath{\Gamma\;\vdash_{\textsc{ct}}\;\Varid{cs}} and reorder buffer \ensuremath{\Gamma\;\vdash_{\textsc{ct}}\;\Varid{is}} .}
\end{subfigure}%

\begin{subfigure}{.9\textwidth}
\begin{mathpar}
\inferrule[Conf]
{\ensuremath{\Gamma\;\vdash_{\textsc{ct}}\;\Varid{is}} \\ \ensuremath{\Gamma\;\vdash_{\textsc{ct}}\;\Varid{cs}}}
{\ensuremath{\Gamma\;\vdash_{\textsc{ct}}\langle\Varid{is},\Varid{cs},\mu,\rho\rangle}}
\end{mathpar}
\captionsetup{format=caption-with-line}
\caption{Configurations \ensuremath{\Gamma\;\vdash_{\textsc{ct}}\;\Conid{C}}.}
\end{subfigure}
\caption{Constant-time Type System.}
\label{fig:ct-ts}
\end{figure}

The patches computed by \tool enforce speculative constant time
(\Cref{def:STC}) only for programs that are already \emph{sequential}
constant-time.
Enforcing sequential constant time is not a goal of \tool because this
problem has already been addressed in previous work
\cite{Watt:2019,libsignalSP}.
Therefore, the soundness guarantees of \tool (\Cref{thm:ts-sound})
rely on a separate, but standard, type-system to enforce sequential
constant time.
To this end, we simply adopt the constant-time programming discipline
from~\cite{Watt:2019,libsignalSP} by disallowing secret-dependent
branches and memory accesses.
Given a security policy \ensuremath{\Blue{\Conid{L}}} that specifies the set of public
variables and arrays (memory addresses containing public data), a
program satisfies sequential constant time if it is well-typed
according to the type-system from Figure~\ref{fig:ct-ts}.
Figure \ref{fig:2-lattice} defines the classic two-point lattice
consisting of public (\ensuremath{\Blue{\Conid{L}}}) and secret (\ensuremath{\Red{\Conid{H}}}) security levels, which
disallows secret-to-public flows of information (\ensuremath{\Red{\Conid{H}}\;\not\flows\;\Blue{\Conid{L}}}).
The type-system relies on a typing environment \ensuremath{\Gamma} to map
variables and arrays to their security level, which we derive from the
given security policy as follows:
\begin{align*}
\ensuremath{CT_{\mkern-1mu\Blue{\Conid{L}}}(\Varid{c})} \ \triangleq \ \ensuremath{\Gamma\;\vdash_{\textsc{ct}}\;\Varid{c}} \quad \text{where} \quad \ensuremath{\Gamma(\Varid{x})} =
\begin{cases}
\ensuremath{\Blue{\Conid{L}}}, & \text{if } \ensuremath{\Varid{x}\;\in\;\Blue{\Conid{L}}}  \\
\ensuremath{\Red{\Conid{H}}}, & \text{otherwise }  \\
\end{cases}
\end{align*}

%
The typing rules from Figure \ref{fig:ct-ts-expr} and
\ref{fig:ct-ts-cmds} are fairly standard.
Figure \ref{fig:ct-ts-expr} defines the typing judgment for expressions,
i.e., \ensuremath{\Gamma\;\vdash_{\textsc{ct}}\;\Varid{e}\mathbin{:}\ell}, which indicates that expression \ensuremath{\Varid{e}}
has sensitivity at most \ensuremath{\ell} under typing context \ensuremath{\Gamma}.
Rule \srule{Var} types variables
according to the typing context
\ensuremath{\Gamma} and ground values can assume any label in rule \srule{Value}; rule \srule{Array} projects the label \ensuremath{\ell} contained in the array \ensuremath{\Varid{a}}, which must coincide with the label specified by the typing context, i.e., \ensuremath{\ell\mathrel{=}\Gamma(\Varid{a})}.
These labels can be upgraded via rule \srule{Sub} and are otherwise
propagated in rules \srule{Fun,Bop,Select,Array-Read,Ptr-Read}.\footnote{
  Technically, our type-system does not prevent an attacker from
  forging a public pointer to secret memory via rules \srule{Val} and
  \srule{Ptr-Read}.
  In the following, we assume that pointers are not directly casted
  from integers, but they only appear when evaluating array accesses
  and thus have the same security level of the array.}
Additionally, rule \srule{Array-Read} disallows memory reads that may
depend on secret data by typing the index \ensuremath{\Varid{e}_{2}} as public (\ensuremath{\Blue{\Conid{L}}}).

Figure \ref{fig:ct-ts-cmds} defines the typing judgment for commands,
i.e., \ensuremath{\Gamma\;\vdash_{\textsc{ct}}\;\Varid{c}}, which indicates that program \ensuremath{\Varid{c}} is
constant time under sequential execution.
Rules \srule {Asgn,Protect} disallow (protected) secret-to-public
assignments (\ensuremath{\Red{\Conid{H}}\;\not\flows\;\Blue{\Conid{L}}}).
Rule \srule{Array-Write} is similar but additionally disallows memory-writes that depend on secret data by typing the index \ensuremath{\Varid{e}_{2}} as public (\ensuremath{\Blue{\Conid{L}}}), like rule \srule{Array-Read}.

%
The combination of rules \srule{Array-Read} and \srule{Array-Write}
makes programs that exhibit secret-dependent memory accesses
ill-typed.
In a similar way, rules \srule{If} and \srule{While} forbid branches
that may depend on secret data by typing the conditional \ensuremath{\Varid{e}} as public
(\ensuremath{\Blue{\Conid{L}}}).

\begin{lemma}[Constant-Time Type Preservation]
If \ensuremath{\Gamma\;\vdash_{\textsc{ct}}\;\Conid{C}} and \ensuremath{\Conid{C}\;\stepT{\Varid{d}}{\Conid{O}}\;\Conid{C}'}, then \ensuremath{\Gamma\;\vdash_{\textsc{ct}}\;\Conid{C}'}.
\label{lemma:ct-preservation}
\end{lemma}
\begin{proof}
  By case analysis on the reduction step and typing judgment.
\end{proof}

\clearpage
\newpage
\section{Proofs}
\label{app:proofs}

\subsection{Consistency}

\begin{figure}[t]
\begin{mathpar}
\inferrule[Skip]
{}
{\ensuremath{\langle\mu,\rho\rangle\Downarrow_{\epsilon}^{\mathbf{skip}}\langle\mu,\rho\rangle}}
\and
\inferrule[Fail]
{}
{\ensuremath{\langle\mu,\rho\rangle\Downarrow_{\mathbf{fail}}^{\mathbf{fail}}\langle\mu,\rho\rangle}}
\and
\inferrule[Asgn]
{\ensuremath{\Varid{v}\mathrel{=}\llbracket \Varid{e}\rrbracket^{\rho}}}
{\ensuremath{\langle\mu,\rho\rangle\Downarrow_{\epsilon}^{\Varid{x}\mathbin{:=}\Varid{e}}\langle\mu,\update{\rho}{\Varid{x}}{\Varid{v}}\rangle}}
\and
\inferrule[Protect]
{\ensuremath{\langle\mu,\rho\rangle\Downarrow_{\Conid{O}}^{\Varid{x}\mathbin{:=}\Varid{r}}\langle\mu,\rho'\rangle}}
{\ensuremath{\langle\mu,\rho\rangle\Downarrow_{\Conid{O}}^{\Varid{x}\mathbin{:=}\mathbf{protect}(\Varid{r})}\langle\mu,\rho'\rangle}}
\and
\inferrule[Ptr-Read]
{\ensuremath{\Varid{n}\mathrel{=}\llbracket \Varid{e}\rrbracket^{\rho}} \\
 \ensuremath{\Varid{v}\mathrel{=}\mu(\Varid{n})} }
{\ensuremath{\langle\mu,\rho\rangle\Downarrow_{\mathbf{read}(\Varid{n})}^{\Varid{x}\mathbin{:=}\mathbin{*}\Varid{e}}\langle\mu,\update{\rho}{\Varid{x}}{\Varid{v}}\rangle}}
\and
\inferrule[Array-Read]
{\ensuremath{\Varid{n}\mathrel{=}\llbracket \Varid{e}\rrbracket^{\rho}} \\
 \ensuremath{\Varid{n}<\mathit{length}(\Varid{a})}  \\\\ \ensuremath{\Varid{n'}\mathrel{=}\Varid{base}(\Varid{a})\mathbin{+}\Varid{n}} \\
 \ensuremath{\Varid{v}\mathrel{=}\mu(\Varid{n'})}}
{\ensuremath{\langle\mu,\rho\rangle\Downarrow_{\mathbf{read}(\Varid{n'})}^{\Varid{x}\mathbin{:=}\Varid{a}{[}\Varid{e}{]}}\langle\mu,\update{\rho}{\Varid{x}}{\Varid{v}}\rangle}}
\and
\inferrule[Array-Read-Fail]
{\ensuremath{\Varid{n}\mathrel{=}\llbracket \Varid{e}\rrbracket^{\rho}} \\
 \ensuremath{\Varid{n}\geq\mathit{length}(\Varid{a})}}
{\ensuremath{\langle\mu,\rho\rangle\Downarrow_{\mathbf{fail}}^{\Varid{x}\mathbin{:=}\Varid{a}{[}\Varid{e}{]}}\langle\mu,\rho\rangle}}
\and
\inferrule[Ptr-Write]
{\ensuremath{\Varid{n}\mathrel{=}\llbracket \Varid{e}_{1}\rrbracket^{\rho}} \\ \ensuremath{\Varid{v}\mathrel{=}\llbracket \Varid{e}_{2}\rrbracket^{\rho}}}
{\ensuremath{\langle\mu,\rho\rangle\Downarrow_{\mathbf{write}(\Varid{n})}^{\mathbin{*}\Varid{e}_{1}\mathbin{:=}\Varid{e}_{2}}\langle\update{\mu}{\Varid{n}}{\Varid{v}},\rho\rangle}}
\and
\inferrule[Array-Write]
{\ensuremath{\Varid{n}\mathrel{=}\llbracket \Varid{e}_{1}\rrbracket^{\rho}} \\
 \ensuremath{\Varid{v}\mathrel{=}\llbracket \Varid{e}_{2}\rrbracket^{\rho}} \\\\ \ensuremath{\Varid{n}<\mathit{length}(\Varid{a})} \\
 \ensuremath{\Varid{n'}\mathrel{=}\Varid{base}(\Varid{a})\mathbin{+}\Varid{n}}}
{\ensuremath{\langle\mu,\rho\rangle\Downarrow_{\mathbf{write}(\Varid{n'})}^{\Varid{a}{[}\Varid{e}_{1}{]}\mathbin{:=}\Varid{e}_{2}}\langle\update{\mu}{\Varid{n'}}{\Varid{v}},\rho\rangle}}
\and
\inferrule[Array-Write-Fail]
{\ensuremath{\Varid{n}\mathrel{=}\llbracket \Varid{e}_{1}\rrbracket^{\rho}} \\
 \ensuremath{\Varid{n}\geq\mathit{length}(\Varid{a})}}
{\ensuremath{\langle\mu,\rho\rangle\Downarrow_{\mathbf{fail}}^{\Varid{a}{[}\Varid{e}_{1}{]}\mathbin{:=}\Varid{e}_{2}}\langle\mu,\rho\rangle}}
\and
\inferrule[If-Then-Else]
{\ensuremath{\Varid{c}\mathrel{=}\mathbf{if}\;\Varid{e}\;\mathbf{then}\;\Varid{c}_{\mathbf{true}}\;\mathbf{else}\;\Varid{c}_{\mathbf{false}}} \\ \ensuremath{\Varid{b}\mathrel{=}\llbracket \Varid{e}\rrbracket^{\rho}} \\ \ensuremath{\langle\mu,\rho\rangle\Downarrow_{\Conid{O}}^{\Varid{c}_{\Varid{b}}}\langle\mu',\rho'\rangle}}
{\ensuremath{\langle\mu,\rho\rangle\Downarrow_{\Conid{O}}^{\Varid{c}}\langle\mu',\rho'\rangle}}
\and
\inferrule[While-True]
{\ensuremath{\llbracket \Varid{e}\rrbracket^{\rho}\mathrel{=}\mathbf{true}} \\ \ensuremath{\Varid{c'}\mathrel{=}\Varid{c};\mathbf{while}\;\Varid{e}\;\mathbf{do}\;\Varid{c}} \\ \ensuremath{\langle\mu,\rho\rangle\Downarrow_{\Conid{O}}^{\Varid{c'}}\langle\mu',\rho'\rangle}}
{\ensuremath{\langle\mu,\rho\rangle\Downarrow_{\Conid{O}}^{\mathbf{while}\;\Varid{e}\;\mathbf{do}\;\Varid{c}}\langle\mu',\rho'\rangle}}
\and
\inferrule[While-False]
{\ensuremath{\llbracket \Varid{e}\rrbracket^{\rho}\mathrel{=}\mathbf{false}}}
{\ensuremath{\langle\mu,\rho\rangle\Downarrow_{\Conid{O}}^{\mathbf{while}\;\Varid{e}\;\mathbf{do}\;\Varid{c}}\langle\mu,\rho\rangle}}
\and
\inferrule[Seq]
{\ensuremath{\langle\mu,\rho\rangle\Downarrow_{\Conid{O}_{1}}^{\Varid{c}_{1}}\langle\mu',\rho'\rangle} \\ \ensuremath{\mathbf{fail}\;\not\in\;\Conid{O}_{1}} \\
 \ensuremath{\langle\mu',\rho'\rangle\Downarrow_{\Conid{O}_{2}}^{\Varid{c}_{2}}\langle\mu'',\rho''\rangle} }
{\ensuremath{\langle\mu,\rho\rangle\Downarrow_{(\Conid{O}_{1}\;{{\mkern-2mu\cdot\mkern-2mu}}\;\Conid{O}_{2})}^{\Varid{c}_{1};\Varid{c}_{2}}\langle\mu'',\rho''\rangle}}
\and
\inferrule[Seq-Fail]
{\ensuremath{\langle\mu,\rho\rangle\Downarrow_{\Conid{O}_{1}}^{\Varid{c}_{1}}\langle\mu',\rho'\rangle} \\ \ensuremath{\mathbf{fail}\;\in\;\Conid{O}_{1}}}
{\ensuremath{\langle\mu,\rho\rangle\Downarrow_{\Conid{O}_{1}}^{\Varid{c}_{1};\Varid{c}_{2}}\langle\mu',\rho'\rangle}}
\end{mathpar}
\caption{Sequential big-step semantics with observations.} \label{fig:sequential-semantics}
\end{figure}

%
%

\mypara{Notation}
In the following, we write \ensuremath{\mathbf{rollback}\;\not\in\;\Conid{O}} to denote that no rollback observation occurs in the observation trace \ensuremath{\Conid{O}}, i.e., \ensuremath{\forall\;\Varid{p}.\mathbf{rollback}(\Varid{p})\;\not\in\;\Conid{O}} and similarly \ensuremath{\mathbf{fail}\;\not\in\;\Conid{O}} for \ensuremath{\forall\;\Varid{p}.\mathbf{fail}(\Varid{p})\;\not\in\;\Conid{O}}.
The notation \ensuremath{\Conid{C}\;\stepT{\Conid{D}}{\Conid{O}}\;\Conid{C}'} denotes a multi-step speculative reduction, i.e., a sequence of zero or more small-step reductions from configuration \ensuremath{\Conid{C}} to \ensuremath{\Conid{C}'}
that follow the directives in schedule \ensuremath{\Conid{D}} and generate observation trace \ensuremath{\Conid{O}}.
We consider equivalence of observation traces \emph{up to silent observations} \ensuremath{\epsilon}, i.e., \ensuremath{\forall\;\Conid{O}.\epsilon\;{{\mkern-2mu\cdot\mkern-2mu}}\;\Conid{O}\mathrel{=}\Conid{O}\;{{\mkern-2mu\cdot\mkern-2mu}}\;\epsilon\mathrel{=}\Conid{O}}, and ignoring fail identifiers, i.e., \ensuremath{\mathbf{fail}(\Varid{p})\mathrel{=}\mathbf{fail}}.

We begin by defining the notion of \emph{valid directive}, i.e., a directive that can be followed by a given configuration without getting stuck.

\begin{definition}[Valid Directive]
A directive \ensuremath{\Varid{d}} is valid for a configuration \ensuremath{\Conid{C}} iff there exists a configuration \ensuremath{\Conid{C}'} and an observation \ensuremath{\Varid{o}} such that \ensuremath{\Conid{C}\;\stepT{\Varid{d}}{\Varid{o}}\;\Conid{C}'}.
\label{def:valid-directive}
\end{definition}

Similarly, we define the notion of \emph{valid schedule}, i.e., a list of valid directives that completely evaluates a given program without getting the processor stuck and re-executing any instruction.

\begin{definition}[Valid Schedule]
A schedule \ensuremath{\Conid{D}} is \emph{valid} for a configuration \ensuremath{\Conid{C}} iff there exists a final configuration \ensuremath{\Conid{C}'} and a sequence of observations \ensuremath{\Conid{O}} such that \ensuremath{\Conid{C}\;\Downarrow_{\Conid{O}}^{\Conid{D}}\;\Conid{C}'} and every fetched instruction is executed at most once.
\label{def:valid}
\end{definition}

\begin{lemma}[Consistency of Valid Schedules] For all
    schedules \ensuremath{\Conid{D}_{1}} and \ensuremath{\Conid{D}_{2}} valid for configuration \ensuremath{\Conid{C}}, if \ensuremath{\Conid{C}\;\Downarrow_{\Conid{O}_{1}}^{\Conid{D}_{1}}\;\Conid{C}_{1}}, \ensuremath{\Conid{C}\;\Downarrow_{\Conid{O}_{2}}^{\Conid{D}_{2}}\;\Conid{C}_{2}}, \ensuremath{\mathbf{rollback}\;\not\in\;\Conid{O}_{1}} and
    \ensuremath{\mathbf{rollback}\;\not\in\;\Conid{O}_{2}}, then \ensuremath{\Conid{C}_{1}\mathrel{=}\Conid{C}_{2}} and \ensuremath{\Conid{O}_{1}\cong\Conid{O}_{2}}.
\label{lemma:valid-consistent}
\end{lemma}
\begin{proofsketch}
By design, our processor fetch and retire instructions \emph{in-order} and
  can only execute them \emph{out-of-order}.
In this case, we know additionally that the schedules \ensuremath{\Conid{D}_{1}} and \ensuremath{\Conid{D}_{2}} are valid
(Def.~\ref{def:valid}) and no rollbacks occur during the executions, therefore
we deduce that (1) each fetched instruction is executed \emph{exactly once}
before it is eventually retired, and (2) the \ensuremath{\mathbf{fetch}} and \ensuremath{\mathbf{retire}} directives in
the schedules fetch and retire \emph{corresponding instructions} in the
\emph{same order}.
Therefore, to prove that the executions reach the same final configurations (\ensuremath{\Conid{C}_{1}\mathrel{=}\Conid{C}_{2}}), it suffices to show that the executions evaluate the operands of
corresponding instructions to the same value.
Intuitively, the particular order with which instructions are executed does not
affect the value of their operands, which is determined only by the data they
depend on.
In particular, if the schedule does not respect these data dependencies the
processor simply gets \emph{stuck}.\footnote{As explained in
  \Cref{sec:semantics}, the transient variable map and the semantics rules of
  the processor inhibit (1) instructions whose dependencies have not yet been
  resolved, which manifests as the value of an operand \ensuremath{\Varid{e}} being undefined,
  i.e., \ensuremath{\llbracket \Varid{e}\rrbracket^{\rho}\mathrel{=}\bot}, and (2) instructions that may access stale data,
  i.e., rule \srule{Exec-Load} requires that no store is pending in the
  buffer.}
However the two executions cannot get stuck because schedules \ensuremath{\Conid{D}_{1}} and \ensuremath{\Conid{D}_{2}} are
valid by assumption, and thus each \ensuremath{\mathbf{exec}} directive in the schedules is also
\emph{valid} (Def.~\ref{def:valid-directive}).
Therefore, in the executions \ensuremath{\Conid{C}\;\Downarrow_{\Conid{O}_{1}}^{\Conid{D}_{1}}\;\Conid{C}_{1}} and \ensuremath{\Conid{C}\;\Downarrow_{\Conid{O}_{2}}^{\Conid{D}_{2}}\;\Conid{C}_{2}}
above, the dependencies of each instruction executed must have already been resolved
(out-of-order) or committed (in-order) by previous \ensuremath{\mathbf{exec}} and \ensuremath{\mathbf{retire}}
directives in schedules \ensuremath{\Conid{D}_{1}} and \ensuremath{\Conid{D}_{2}}.
Thus, the operands of each instruction must evaluate to the same value and we
conclude that \ensuremath{\Conid{C}_{1}\mathrel{=}\Conid{C}_{2}}.
%
%
%


To prove that the observation traces are equal up to permutation, i.e., \ensuremath{\Conid{O}_{1}\cong\Conid{O}_{2}}, we \emph{match} individual observations generated by corresponding directives and
instructions.
%
The observations generated by \ensuremath{\mathbf{fetch}} and \ensuremath{\mathbf{retire}} directives are easier to
match because these directives proceed \emph{in-order}.
%
%
To relate their observations, we compute the pairs of matching \ensuremath{\mathbf{fetch}} and
\ensuremath{\mathbf{retire}} directives from the schedules and use the indexes of individual
directives to relate the corresponding observations.
Formally, let \ensuremath{\Conid{D}_{1}\mathrel{=}[\mskip1.5mu (\Varid{d}_{1})_{\mathrm{1}},\mathbin{...},(\Varid{d}_{1})_{\Varid{n}}\mskip1.5mu]} and \ensuremath{\Conid{D}_{2}\mathrel{=}[\mskip1.5mu (\Varid{d}_{2})_{\mathrm{1}},\mathbin{...},(\Varid{d}_{2})_{\Varid{n}}\mskip1.5mu]} be the list of directives of the first and second
schedule, respectively.
Let \ensuremath{((\Varid{d}_{1})_{\Varid{i}},(\Varid{d}_{2})_{\Varid{j}})_{\Varid{k}}} be the pair matching
the \ensuremath{\Varid{k}}-th \ensuremath{\mathbf{fetch}} directives, i.e., where directives \ensuremath{(\Varid{d}_{1})_{\Varid{i}}\mathrel{=}\mathbf{fetch}\mathrel{=}(\Varid{d}_{2})_{\Varid{j}}} (or \ensuremath{(\Varid{d}_{1})_{\Varid{i}}\mathrel{=}\mathbf{fetch}\;\Varid{b}\mathrel{=}(\Varid{d}_{2})_{\Varid{j}}}) are the \ensuremath{\Varid{k}}-th \ensuremath{\mathbf{fetch}} directives in schedules \ensuremath{\Conid{D}_{1}} and
\ensuremath{\Conid{D}_{2}}.
Then, the indexes \ensuremath{\Varid{i}\leftrightarrow\Varid{j}} in each pair indicate that the \ensuremath{\Varid{i}}-th and \ensuremath{\Varid{j}}-th observations in the observation traces \ensuremath{\Conid{O}_{1}} and \ensuremath{\Conid{O}_{2}} are related.
In particular, for each \ensuremath{\mathbf{fetch}} matching pair \ensuremath{((\Varid{d}_{1})_{\Varid{i}},(\Varid{d}_{2})_{\Varid{j}})_{\Varid{k}}} in \ensuremath{\Conid{D}_{1}} and \ensuremath{\Conid{D}_{2}}, we relate the individual observations \ensuremath{(\Varid{o}_{1})_{\Varid{i}}\leftrightarrow(\Varid{o}_{2})_{\Varid{j}}} (where \ensuremath{\Varid{o}_{1}\mathrel{=}\epsilon\mathrel{=}\Varid{o}_{2}}) of the observation traces \ensuremath{\Conid{O}_{1}\mathrel{=}(\Varid{o}_{1})_{\mathrm{1}}\mathbin{...}(\Varid{o}_{1})_{\Varid{n}}} and \ensuremath{\Conid{O}_{2}\mathrel{=}(\Varid{o}_{2})_{\mathrm{1}}\mathbin{...}(\Varid{o}_{2})_{\Varid{n}}}.
%
The procedure for matching observations generated by \ensuremath{\mathbf{retire}}
directives is similar.
Since the schedules can execute instruction \emph{out-of-order}, we cannot apply exactly the same technique to \ensuremath{\mathbf{exec}} directives.
Instead, we annotate the instructions in the reorder buffer with a unique index
when they are fetched and relate the observations generated by the instructions
with the same index.
%
%
Intuitively, these indexes identify uniquely matching observations because each
fetched instruction is executed \emph{exactly once} (as explained above).
Formally, for each instruction \ensuremath{(\Varid{i}_{1})_{\Varid{j}}} whose execution generates the
observation \ensuremath{(\Varid{o}_{1})_{\Varid{j'}}} in the first execution, we have a corresponding
instruction \ensuremath{(\Varid{i}_{2})_{\Varid{j}}} with observation \ensuremath{(\Varid{o}_{2})_{\Varid{j''}}} in the second
execution, where in general \ensuremath{\Varid{j}\neq\Varid{j'}\neq\Varid{j''}} because instructions can be executed out-of-order.
Since these instructions are annotated with the same index \ensuremath{\Varid{j}}, they are
identical, i.e., \ensuremath{\Varid{i}_{1}\mathrel{=}\Varid{i}_{2}}, their operands evaluate to the same value (as
explained above), and thus generate the same observations \ensuremath{\Varid{o}_{1}\mathrel{=}\Varid{o}_{2}} and hence
\ensuremath{(\Varid{o}_{1})_{\Varid{j'}}\leftrightarrow(\Varid{o}_{2})_{\Varid{j''}}}.

\end{proofsketch}





\begin{definition}[Sequential Schedule]
A schedule \ensuremath{\Conid{D}} is \emph{sequential} for a  configuration \ensuremath{\Conid{C}} iff there exists a final configuration \ensuremath{\Conid{C}'}
and observation trace \ensuremath{\Conid{O}} such that \ensuremath{\Conid{C}\;\Downarrow_{\Conid{O}}^{\Conid{D}}\;\Conid{C}'}
 and:
\begin{enumerate}
\item all instructions are executed in-order, as soon as they are fetched;
\item all instructions are retired as soon as they are executed;
\item no misprediction occurs, i.e., \ensuremath{\mathbf{rollback}\;\not\in\;\Conid{O}}.
\end{enumerate}
\end{definition}

The following lemma shows that we can simulate a sequential execution
on our processor through a sequential schedule.
%
%

\begin{lemma}[Sequential Consistency]
\label{lemma1}
Given a sequential execution \ensuremath{\langle\mu,\rho\rangle\Downarrow_{\Conid{O}}^{\Varid{c}}\langle\mu',\rho'\rangle},
there exists a valid, sequential schedule \ensuremath{\Conid{D}}, such that \ensuremath{\langle[\mskip1.5mu \mskip1.5mu],[\mskip1.5mu \Varid{c}\mskip1.5mu],\mu,\rho\rangle\Downarrow_{\Conid{O}}^{\Conid{D}}\langle[\mskip1.5mu \mskip1.5mu],[\mskip1.5mu \mskip1.5mu],\mu',\rho'\rangle}.
\end{lemma}
\begin{proof}
The proof is constructive.
We perform induction on the sequential execution judgment \ensuremath{\langle\mu,\rho\rangle\Downarrow_{\Conid{O}}^{\Varid{c}}\langle\mu',\rho'\rangle} defined in Figure~\ref{fig:sequential-semantics} and
we construct the corresponding sequential schedule for it.
Then, we translate the individual sequential reductions in the derivation tree into corresponding speculative reductions that follow the schedule.
In the following, we ignore the security labels that annotate commands.

\vspace{\baselineskip}
\begin{description}[font=\normalfont,wide, labelwidth=!, labelindent=0pt]
\setlength{\parskip}{\baselineskip}
\item[Case \srule{Skip}.] We define the sequential schedule \ensuremath{\Conid{D}\mathrel{=}[\mskip1.5mu \mathbf{fetch},\mathbf{retire}\mskip1.5mu]} and construct the corresponding speculative big-step reductions consisting of \srule{Step} applied to \srule{Fetch-Skip}, \srule{Step} applied to \srule{Retire-Nop} and \srule{Done}.
These reductions form the big-step \ensuremath{\langle[\mskip1.5mu \mskip1.5mu],[\mskip1.5mu \mathbf{skip}\mskip1.5mu],\mu,\rho\rangle\Downarrow_{\epsilon}^{\Conid{D}}\langle[\mskip1.5mu \mskip1.5mu],[\mskip1.5mu \mskip1.5mu],\mu,\rho\rangle} that corresponds to the sequential reduction \ensuremath{\langle\mu,\rho\rangle\Downarrow_{\epsilon}^{\mathbf{skip}}\langle\mu,\rho\rangle}.
\item[Case \srule{Fail}.] Analogous to the previous case, but using rules \srule{Fetch-Fail} and \srule{Retire-Fail}. Notice that in this case the  observation trace generated in the sequential execution contains only observation \ensuremath{\mathbf{fail}(\Varid{p})} for some fresh identifier \ensuremath{\Varid{p}}, which corresponds to the trace generated by the sequential schedule up to silent events and fail identifiers, i.e., \ensuremath{\epsilon\;{{\mkern-2mu\cdot\mkern-2mu}}\;\mathbf{fail}\mathrel{=}\mathbf{fail}(\Varid{p})}.
\item[Case \srule{Asgn}.] We define the sequential schedule \ensuremath{\Conid{D}\mathrel{=}[\mskip1.5mu \mathbf{fetch},\mathbf{exec}\;\mathrm{1},\mathbf{retire}\mskip1.5mu]} and construct the corresponding speculative big-step reduction using rules \srule{Fetch-Asgn}, \srule{Execute} applied to \srule{Exec-Asgn}, and \srule{Retire-Asgn}.
These reductions form the big-step \ensuremath{\langle[\mskip1.5mu \mskip1.5mu],\Varid{x}\mathbin{:=}\Varid{e},\mu,\rho\rangle\Downarrow_{\epsilon}^{\Conid{D}}\langle[\mskip1.5mu \mskip1.5mu],[\mskip1.5mu \mskip1.5mu],\mu,\update{\rho}{\Varid{x}}{\llbracket \Varid{e}\rrbracket^{\rho'}}\rangle}, where \ensuremath{\rho'} is the transient variable map computed in rule \srule{Execute}.
Since the reorder buffer is initially empty, the transient map is identical to the initial variable map, i.e., \ensuremath{\rho'\mathrel{=}\phi(\rho,[\mskip1.5mu \mskip1.5mu])\mathrel{=}\rho}, and thus the speculative reduction corresponds to the sequential reduction \ensuremath{\langle\mu,\rho\rangle\Downarrow_{\epsilon}^{\Varid{x}\mathbin{:=}\Varid{e}}\langle\mu,\update{\rho}{\Varid{x}}{\llbracket \Varid{e}\rrbracket^{\rho}}\rangle}.
%
%

\item[Case \srule{Ptr-Read}.]
The translation is similar to the previous case and  follows the same sequential schedule \ensuremath{\Conid{D}\mathrel{=}[\mskip1.5mu \mathbf{fetch},\mathbf{exec}\;\mathrm{1},\mathbf{retire}\mskip1.5mu]}.
%
First we apply rule \srule{Fetch-Ptr-Read}, which fetches command \ensuremath{\Varid{x}\mathbin{:=}\ast\Varid{e}} from the command stack and inserts the corresponding instruction \ensuremath{\Varid{x}\mathbin{:=}\mathbf{load}(\Varid{e})} in the empty reorder buffer.
Then, we apply rule \srule{Exec-Load}, which evaluates \ensuremath{\Varid{e}} to the same address \ensuremath{\Varid{n}} obtained in \srule{Ptr-Read} (as explained above).
Notice that the premise \ensuremath{\mathbf{store}(\anonymous ,\anonymous )\;\not\in\;\Varid{is}_{1}} of rule
\srule{Exec-Load} holds trivially, because the reorder buffer is initially empty and thus \ensuremath{\Varid{is}_{1}\mathrel{=}[\mskip1.5mu \mskip1.5mu]}.
The memory store in the sequential and speculative reductions are equal by assumption, therefore the same value \ensuremath{\Varid{v}\mathrel{=}\mu(\Varid{n})} is loaded from memory and assigned to variable \ensuremath{\Varid{x}} and the final variable maps remain equal after applying rule \srule{Retire-Asgn} like in the previous case.
As a result, we obtain the speculative reduction \ensuremath{\langle[\mskip1.5mu \mskip1.5mu],[\mskip1.5mu \Varid{x}\mathbin{:=}\ast\Varid{e}\mskip1.5mu],\mu,\rho\rangle\Downarrow_{\mathbf{read}(\Varid{n},[\mskip1.5mu \mskip1.5mu])}^{\Conid{D}}\langle[\mskip1.5mu \mskip1.5mu],[\mskip1.5mu \mskip1.5mu],\mu,\update{\rho}{\Varid{x}}{\mu(\Varid{n})}\rangle} corresponding to the sequential reduction \ensuremath{\langle\mu,\rho\rangle\Downarrow_{\mathbf{read}(\Varid{n})}^{\Varid{x}\mathbin{:=}\ast\Varid{e}}\langle\mu,\update{\rho}{\Varid{x}}{\mu(\Varid{n})}\rangle}.
Notice that the observations \ensuremath{\mathbf{read}(\Varid{n},[\mskip1.5mu \mskip1.5mu])} and \ensuremath{\mathbf{read}(\Varid{n})} are considered equivalent because the list of guard identifiers is \emph{empty}, i.e., \ensuremath{\Moon{\Varid{is}_{1}}\mathrel{=}[\mskip1.5mu \mskip1.5mu]} from \ensuremath{\Varid{is}_{1}\mathrel{=}[\mskip1.5mu \mskip1.5mu]} in rule \srule{Exec-Load}.






\item[Case \srule{Ptr-Write}.] Analogous to case \srule{Ptr-Read}.
\item[Case \srule{Array-Read}.]
The sequential reduction reveals that the array is read \emph{in-bounds}, therefore we ensure that the speculative execution follows the correct branch after the bounds check by supplying prediction \ensuremath{\mathbf{true}} for the bounds-check condition.
Formally, we define  the sequential schedule \ensuremath{\Conid{D}\mathrel{=}[\mskip1.5mu \mathbf{fetch},\mathbf{fetch}\;\mathbf{true},\mathbf{exec}\;\mathrm{1},\mathbf{retire},\mathbf{fetch},\mathbf{exec}\;\mathrm{1},\mathbf{retire}\mskip1.5mu]}  for rule \srule{Array-Read}.
The first \ensuremath{\mathbf{fetch}} directive is processed by rule \srule{Fetch-Array-Load}, which converts command \ensuremath{\Varid{x}\mathbin{:=}\Varid{a}{[}\Varid{e}_{1}{]}} into the corresponding bounds-checked pointer read, i.e., command \ensuremath{\mathbf{if}}\ \ensuremath{\Varid{e}}\ \ensuremath{\mathbf{then}}\ \ensuremath{\Varid{x}\mathbin{:=}\ast\Varid{e'}}\ \ensuremath{\mathbf{else}}\ \ensuremath{\mathbf{fail}}, where expression \ensuremath{\Varid{e}\mathrel{=}\Varid{e}_{1}<\mathit{length}(\Varid{a})} is the bounds-check condition and expression \ensuremath{\Varid{e'}\mathrel{=}\Varid{base}(\Varid{a})\mathbin{+}\Varid{e}_{1}} computes the memory address for the pointer.
Then, directive \ensuremath{\mathbf{fetch}\;\mathbf{true}} is processed by rule
\srule{Fetch-If-True}, which speculatively follows the \ensuremath{\mathbf{then}}
branch and produces the guard \ensuremath{\mathbf{guard}(\Varid{e}^{\mathbf{true}},[\mskip1.5mu \mathbf{fail}\mskip1.5mu],\Varid{p})} for some
fresh guard identifier \ensuremath{\Varid{p}}, and pushes command \ensuremath{\Varid{x}\mathbin{:=}\ast\Varid{e'}} on the
empty command stack.
The guard instruction is then immediately resolved by rule \srule{Execute} applied to \srule{Exec-Branch-Ok} which consumes the first directive \ensuremath{\mathbf{exec}\;\mathrm{1}}.
In rule \srule{Array-Read}, the premises \ensuremath{\Varid{n}\mathrel{=}\llbracket \Varid{e}\rrbracket^{\rho}} and \ensuremath{\Varid{n}<\mathit{length}(\Varid{a})} ensure that the bounds-check condition succeeds as predicted in rule
\srule{Exec-Branch-Ok}, i.e., \ensuremath{\llbracket \Varid{e}\rrbracket^{\rho}\mathrel{=}\mathbf{true}}, which then rewrites the guard instruction to \ensuremath{\mathbf{nop}}.
Then, \ensuremath{\mathbf{nop}} is retired by the next \ensuremath{\mathbf{retire}} directive and the remaining directives \ensuremath{[\mskip1.5mu \mathbf{fetch},\mathbf{exec}\;\mathrm{1},\mathbf{retire}\mskip1.5mu]} process the pointer read as in case \srule{Ptr-Read}.

\item[Case \srule{Array-Write}.] Analogous to case \srule{Array-Read}.

\item[Case \srule{Array-Read-Fail}.]
The sequential reduction reveals that the array is read  \emph{out-of-bounds}, therefore we supply prediction \ensuremath{\mathbf{false}} to the bounds-check guard in the speculative reduction.
Formally, we define  the sequential schedule \ensuremath{\Conid{D}\mathrel{=}[\mskip1.5mu \mathbf{fetch},\mathbf{fetch}\;\mathbf{false},\mathbf{exec}\;\mathrm{1},\mathbf{retire},\mathbf{fetch},\mathbf{retire}\mskip1.5mu]} for rule \srule{Array-Read-Fail}.
The first \ensuremath{\mathbf{fetch}} directive is processed by rule
\srule{Fetch-Array-Load} as in case \srule{Array-Read}.
Then, directive \ensuremath{\mathbf{fetch}\;\mathbf{false}} is processed by rule \srule{Fetch-If-False}, which generates guard \ensuremath{\mathbf{guard}(\Varid{e}^{\mathbf{false}},[\mskip1.5mu \Varid{x}\mathbin{:=}\ast\Varid{e'}\mskip1.5mu],\Varid{p})} and command stack \ensuremath{[\mskip1.5mu \mathbf{fail}\mskip1.5mu]}, instead.
Similarly to the previous case, the \ensuremath{\mathbf{exec}\;\mathrm{1}} directive resolves the guard correctly
via rule \srule{Exec-Branch-Ok}, rewriting it to \ensuremath{\mathbf{nop}}, which is then retired by the next \ensuremath{\mathbf{retire}} directive.
Finally, directive \ensuremath{\mathbf{fetch}} pops \ensuremath{\mathbf{fail}} from the stack, generates a fresh identifier \ensuremath{\Varid{p}} and inserts instruction \ensuremath{\mathbf{fail}(\Varid{p})} in the reorder buffer, which is then retired by directive \ensuremath{\mathbf{retire}}, halting the processor with observation \ensuremath{\mathbf{fail}(\Varid{p})}.

\item[Case \srule{Array-Write-Fail}.] Analogous to case \srule{Array-Read-Fail}.

\item[Case \srule{Protect}.] Analogous to case \srule{Asgn}, \srule{Array-Read}, and \srule{Array-Read-Fail}.

\item[Case \srule{If-Then-Else}.]
Let \ensuremath{\Varid{c}\mathrel{=}\mathbf{if}\;\Varid{e}\;\mathbf{then}\;\Varid{c}_{\mathbf{true}}\;\mathbf{else}\;\Varid{c}_{\mathbf{false}}}
be the command executed in rule \srule{If-Then-Else} and let \ensuremath{\Varid{b}\mathrel{=}\llbracket \Varid{e}\rrbracket^{\rho}} be the value of the conditional
under initial variable map \ensuremath{\rho}.
Then, we define \ensuremath{\Conid{D}_{1}\mathrel{=}[\mskip1.5mu \mathbf{fetch}\;\Varid{b},\mathbf{exec}\;\mathrm{1},\mathbf{retire}\mskip1.5mu]} as the first
part of the corresponding sequential schedule.
First, we consume directive \ensuremath{\mathbf{fetch}\;\Varid{b}} via rule \srule{Fetch-If-\ensuremath{\Varid{b}}}, which follows the correct prediction and inserts instruction \ensuremath{\mathbf{guard}(\Varid{e}^{\Varid{b}},[\mskip1.5mu \Varid{c}_{(\neg\Varid{b})}\mskip1.5mu],\Varid{p})} in the empty buffer and pushes command \ensuremath{\Varid{c}_{\Varid{b}}} on the empty command stack.
Then, we process directive \ensuremath{\mathbf{exec}\;\mathrm{1}}  through rule \srule{Execute} applied to \srule{Exec-Branch-Ok}, which rewrites the guard to instruction \ensuremath{\mathbf{nop}}, which is lastly retired by rule \srule{Retire-Nop}.
By composing these small-steps, we obtain the multi-step reduction \ensuremath{\langle[\mskip1.5mu \mskip1.5mu],[\mskip1.5mu \Varid{c}\mskip1.5mu],\rho,\mu\rangle\stepT{\Conid{D}_{1}}{\epsilon}\langle[\mskip1.5mu \mskip1.5mu],[\mskip1.5mu \Varid{c}_{\Varid{b}}\mskip1.5mu],\rho,\mu\rangle}.
Next, we apply our induction hypothesis to the reduction \ensuremath{\langle\rho,\mu\rangle\Downarrow_{\Conid{O}}^{\Varid{c}_{\Varid{b}}}\langle\rho',\mu'\rangle},
giving us the second part of the sequential schedule, i.e., \ensuremath{\Conid{D}_{2}}, and speculative big-step \ensuremath{\langle[\mskip1.5mu \mskip1.5mu],[\mskip1.5mu \Varid{c}_{\Varid{b}}\mskip1.5mu],\rho,\mu\rangle\Downarrow_{\Conid{O}}^{\Conid{D}_{2}}\langle\rho',\mu'\rangle}.
We then conclude the proof by defining the complete sequential schedule \ensuremath{\Conid{D}\mathrel{=}\Conid{D}_{1}+{\mkern-9mu+}\ \Conid{D}_{2}} and composing the multi-step reduction \ensuremath{\langle[\mskip1.5mu \mskip1.5mu],[\mskip1.5mu \Varid{c}\mskip1.5mu],\rho,\mu\rangle\stepT{\Conid{D}_{1}}{\epsilon}\langle[\mskip1.5mu \mskip1.5mu],[\mskip1.5mu \Varid{c}_{\Varid{b}}\mskip1.5mu],\rho,\mu\rangle} with the big-step reduction \ensuremath{\langle[\mskip1.5mu \mskip1.5mu],[\mskip1.5mu \Varid{c}_{\Varid{b}}\mskip1.5mu],\rho,\mu\rangle\Downarrow_{\Conid{O}}^{\Conid{D}_{2}}\langle\rho',\mu'\rangle}, thus obtaining  the big-step \ensuremath{\langle[\mskip1.5mu \mskip1.5mu],[\mskip1.5mu \Varid{c}\mskip1.5mu],\mu,\rho\rangle\Downarrow_{\Conid{O}}^{\Conid{D}}} \ensuremath{\langle[\mskip1.5mu \mskip1.5mu],[\mskip1.5mu \mskip1.5mu],\mu',\rho'\rangle}.

\item[Case \srule{While-\ensuremath{\Varid{b}}}.]
Both the sequential and the speculative semantics  progressively unroll \ensuremath{\mathbf{while}} loops into a sequence of conditionals, therefore this case follows similarly to case \srule{If-Then-Else}, i.e., using the conditional value \ensuremath{\Varid{b}} computed in the sequential reduction as prediction in the \ensuremath{\mathbf{fetch}\;\Varid{b}} directive of the sequential schedule.

\item[Case \srule{Seq}.]
From the sequential reduction \ensuremath{\langle\mu,\rho\rangle\Downarrow_{(\Conid{O}_{1}\;{{\mkern-2mu\cdot\mkern-2mu}}\;\Conid{O}_{2})}^{\Varid{c}_{1};\Varid{c}_{2}}\langle\mu'',\rho''\rangle}, we have two sub-reductions \ensuremath{\langle\mu,\rho\rangle\Downarrow_{\Conid{O}_{1}}^{\Varid{c}_{1}}\langle\mu',\rho'\rangle} and \ensuremath{\langle\mu',\rho'\rangle\Downarrow_{\Conid{O}_{2}}^{\Varid{c}_{2}}\langle\mu'',\rho''\rangle}, where \ensuremath{\mathbf{fail}\;\not\in\;\Conid{O}_{1}}.
First, we apply our induction hypothesis to the first reduction
\ensuremath{\langle\mu,\rho\rangle\Downarrow_{\Conid{O}_{1}}^{\Varid{c}_{1}}\langle\mu',\rho'\rangle}, obtaining the first part of the sequential schedule, i.e., \ensuremath{\Conid{D}_{1}}, and a big-step reduction
\ensuremath{\langle[\mskip1.5mu \mskip1.5mu],[\mskip1.5mu \Varid{c}_{1}\mskip1.5mu],\mu,\rho\rangle\Downarrow_{\Conid{O}_{1}}^{\Conid{D}_{1}}\langle[\mskip1.5mu \mskip1.5mu],[\mskip1.5mu \mskip1.5mu],\mu',\rho'\rangle}.
Intuitively, we can \emph{lift} this reduction to use initial stack \ensuremath{[\mskip1.5mu \Varid{c}_{1},\Varid{c}_{2}\mskip1.5mu]} and obtain the multi-step reduction \ensuremath{\langle[\mskip1.5mu \mskip1.5mu],[\mskip1.5mu \Varid{c}_{1},\Varid{c}_{2}\mskip1.5mu],\mu,\rho\rangle\stepT{\Conid{D}_{1}}{\Conid{O}_{1}}\langle[\mskip1.5mu \mskip1.5mu],[\mskip1.5mu \Varid{c}_{2}\mskip1.5mu],\mu',\rho'\rangle}, where the second command \ensuremath{\Varid{c}_{2}} remains unchanged on the resulting stack because no failures or rollbacks occur during the execution, i.e., \ensuremath{\mathbf{fail}\;\not\in\;\Conid{O}_{1}} and \ensuremath{\mathbf{rollback}\;\not\in\;\Conid{O}_{1}}
Then, we apply our induction hypothesis to the second reduction \ensuremath{\langle\mu',\rho'\rangle\Downarrow_{\Conid{O}_{2}}^{\Varid{c}_{2}}\langle\mu'',\rho''\rangle} and obtain the second part of the sequential schedule, i.e., \ensuremath{\Conid{D}_{2}}, and a big-step reduction
\ensuremath{\langle[\mskip1.5mu \mskip1.5mu],[\mskip1.5mu \Varid{c}_{2}\mskip1.5mu],\mu',\rho\rangle\Downarrow_{\Conid{O}_{2}}^{\Conid{D}_{2}}\langle[\mskip1.5mu \mskip1.5mu],[\mskip1.5mu \mskip1.5mu],\mu'',\rho''\rangle}.
By composing these reductions, we obtain the big-step reduction \ensuremath{\langle[\mskip1.5mu \mskip1.5mu],[\mskip1.5mu \Varid{c}_{1},\Varid{c}_{2}\mskip1.5mu],\mu,\rho\rangle\Downarrow_{(\Conid{O}_{1}\;{{\mkern-2mu\cdot\mkern-2mu}}\;\Conid{O}_{2})}^{\Conid{D}_{1}+{\mkern-9mu+}\ \Conid{D}_{2}}\langle[\mskip1.5mu \mskip1.5mu],[\mskip1.5mu \mskip1.5mu],\mu'',\rho''\rangle}.
Finally, we define the complete sequential schedule \ensuremath{\Conid{D}\mathrel{=}[\mskip1.5mu \mathbf{fetch}\mskip1.5mu]+{\mkern-9mu+}\ \Conid{D}_{1}+{\mkern-9mu+}\ \Conid{D}_{2}} and compose the small-step
\ensuremath{\langle[\mskip1.5mu \mskip1.5mu],[\mskip1.5mu \Varid{c}_{1};\Varid{c}_{2}\mskip1.5mu],\mu,\rho\rangle\stepT{\mathbf{fetch}}{\epsilon}\langle[\mskip1.5mu \mskip1.5mu],[\mskip1.5mu \Varid{c}_{1},\Varid{c}_{2}\mskip1.5mu],\mu,\rho\rangle} obtained via rule \srule{Fetch-Seq} with the big-step  reduction \ensuremath{\langle[\mskip1.5mu \mskip1.5mu],[\mskip1.5mu \Varid{c}_{1},\Varid{c}_{2}\mskip1.5mu],\mu,\rho\rangle\Downarrow_{(\Conid{O}_{1}\;{{\mkern-2mu\cdot\mkern-2mu}}\;\Conid{O}_{2})}^{\Conid{D}_{1}+{\mkern-9mu+}\ \Conid{D}_{2}}\langle[\mskip1.5mu \mskip1.5mu],[\mskip1.5mu \mskip1.5mu],\mu'',\rho''\rangle}, thus obtaining the corresponding big-step reduction \ensuremath{\langle[\mskip1.5mu \mskip1.5mu],[\mskip1.5mu \Varid{c}_{1};\Varid{c}_{2}\mskip1.5mu],\mu,\rho\rangle\Downarrow_{(\Conid{O}_{1}\;{{\mkern-2mu\cdot\mkern-2mu}}\;\Conid{O}_{2})}^{\Conid{D}}\langle[\mskip1.5mu \mskip1.5mu],[\mskip1.5mu \mskip1.5mu],\mu'',\rho''\rangle}.


\item[Case \srule{Seq-Fail}.]
Analogous to case \srule{Seq}.
First, we apply our induction hypothesis to the reduction
\ensuremath{\langle\mu,\rho\rangle\Downarrow_{\Conid{O}_{1}}^{\Varid{c}_{1}}\langle\mu',\rho'\rangle} where \ensuremath{\mathbf{fail}\;\in\;\Conid{O}_{1}}, and obtain part of the sequential schedule, i.e., \ensuremath{\Conid{D}_{1}}, and a big-step reduction
\ensuremath{\langle[\mskip1.5mu \mskip1.5mu],[\mskip1.5mu \Varid{c}_{1}\mskip1.5mu],\mu,\rho\rangle\Downarrow_{\Conid{O}_{1}}^{\Conid{D}_{1}}\langle[\mskip1.5mu \mskip1.5mu],[\mskip1.5mu \mskip1.5mu],\mu',\rho'\rangle}.
In contrast to the previous case, when we lift this reduction
to use initial stack  \ensuremath{[\mskip1.5mu \Varid{c}_{1};\Varid{c}_{2}\mskip1.5mu]}, we obtain the big-step reduction \ensuremath{\langle[\mskip1.5mu \mskip1.5mu],[\mskip1.5mu \Varid{c}_{1},\Varid{c}_{2}\mskip1.5mu],\mu,\rho\rangle\Downarrow_{\Conid{O}_{1}}^{\Conid{D}_{1}}\langle[\mskip1.5mu \mskip1.5mu],[\mskip1.5mu \mskip1.5mu],\mu',\rho'\rangle},
because a failure occurs during the execution of command \ensuremath{\Varid{c}_{1}}, i.e., \ensuremath{\mathbf{rollback}\;\in\;\Conid{O}_{1}}, which aborts the execution and empties the command stack.
Finally, we define the complete sequential schedule \ensuremath{\Conid{D}\mathrel{=}[\mskip1.5mu \mathbf{fetch}\mskip1.5mu]+{\mkern-9mu+}\ \Conid{D}_{1}}  and compose the small-step
\ensuremath{\langle[\mskip1.5mu \mskip1.5mu],[\mskip1.5mu \Varid{c}_{1};\Varid{c}_{2}\mskip1.5mu],\mu,\rho\rangle\stepT{\mathbf{fetch}}{\epsilon}\langle[\mskip1.5mu \mskip1.5mu],[\mskip1.5mu \Varid{c}_{1},\Varid{c}_{2}\mskip1.5mu],\mu,\rho\rangle} obtained via rule \srule{Fetch-Seq} with the big-step reduction \ensuremath{\langle[\mskip1.5mu \mskip1.5mu],[\mskip1.5mu \Varid{c}_{1};\Varid{c}_{2}\mskip1.5mu],\mu,\rho\rangle\Downarrow_{\Conid{O}_{1}}^{\Conid{D}_{1}}\langle[\mskip1.5mu \mskip1.5mu],[\mskip1.5mu \mskip1.5mu],\mu',\rho'\rangle}, thus obtaining the corresponding big-step reduction \ensuremath{\langle[\mskip1.5mu \mskip1.5mu],[\mskip1.5mu \Varid{c}_{1};\Varid{c}_{2}\mskip1.5mu],\mu,\rho\rangle\Downarrow_{\Conid{O}_{1}}^{\Conid{D}}\langle[\mskip1.5mu \mskip1.5mu],[\mskip1.5mu \mskip1.5mu],\mu',\rho'\rangle}.

\end{description}

\end{proof}

\mypara{$J$-Equivalence}
In order to prove consistency between the speculative and sequential
semantics, we make use of an auxiliary relation called
$J$-equivalence, defined in Figure~\ref{fig:j-equiv}.
This relation captures program configurations of two executions that \emph{may} have diverged due to a
single misprediction (without loss of generality we assume that the second
configuration in the relation is at fault).
If \ensuremath{J\mathrel{=}\bot}, then rule \srule{Synch} ensures that the configurations are identical, i.e., the two executions have not diverged.
If \ensuremath{J\;\in\;\mathbb{N}}, the two executions have diverged on the $J+1$-th
(guard) instruction in the reorder buffer, but agree on the first $J$
instructions and have identical architectural state (memory store and
variable map).
Formally, rule \srule{Suc} ensures that the first $J$ instructions in
the buffers are equal i.e., \ensuremath{\Varid{i}_{1}\mathrel{=}\Varid{i}_{2}}, while the command stacks are
unrelated because the executions have followed different paths after they
have diverged, i.e., in general \ensuremath{\Varid{cs}_{1}\neq\Varid{cs}_{2}}.
The $J+1$ instruction in the buffer of the first configuration
(representing the correct execution) is the pending guard instruction,
which is mispredicted in the second configuration.
This instruction represents a \emph{synchronization point} for the two executions, therefore rule \srule{Zero} requires the guard instructions in the
configurations to be unresolved, i.e., \ensuremath{\Varid{i}_{1}\mathrel{=}\mathbf{guard}(\Varid{e}^{\Varid{b}_{1}},\Varid{cs}_{1},\Varid{p})} and \ensuremath{\Varid{i}_{2}\mathrel{=}\mathbf{guard}(\Varid{e}^{\Varid{b}_{2}},\Varid{cs}_{1},\Varid{p})}, with different predictions, i.e., \ensuremath{\Varid{b}_{1}\neq\Varid{b}_{2}}, and that the rollback command stack in \ensuremath{\Varid{i}_{2}} is identical to the stack of the first configuration, i.e., \ensuremath{\Varid{cs}_{1}}.
%
Notice that the first buffer contains only the guard instruction buffer, which simulates the first execution waiting for the second execution to catch up.
 This ensures that the two executions will fully synchronize again via rule
 \srule{Synch} after the pending guard is executed.
%


\begin{figure}[t]
\begin{mathpar}
\inferrule[Zero]
{\ensuremath{\Varid{i}_{1}\mathrel{=}\mathbf{guard}(\Varid{e}^{\Varid{b}_{1}},\Varid{cs}_{2}',\Varid{p})} \\ \ensuremath{\Varid{i}_{2}\mathrel{=}\mathbf{guard}(\Varid{e}^{\Varid{b}_{2}},\Varid{cs}_{1},\Varid{p})} \\ \ensuremath{\Varid{b}_{1}\neq\Varid{b}_{2}} }
{\ensuremath{\langle\Varid{i}_{1}\mathbin{:}[\mskip1.5mu \mskip1.5mu],\Varid{cs}_{1},\mu,\rho\rangle\mathrel{=}_{\mathrm{0}}\langle\Varid{i}_{2}\mathbin{:}\Varid{is}_{2},\Varid{cs}_{2},\mu,\rho\rangle}}
\and
\inferrule[Suc]
{\ensuremath{\Varid{i}_{1}\mathrel{=}\Varid{i}_{2}} \\ \ensuremath{\langle\Varid{is}_{1},\Varid{cs}_{1},\mu,\rho\rangle\mathrel{=}_{J}\langle\Varid{is}_{2},\Varid{cs}_{2},\mu,\rho\rangle} }
{\ensuremath{\langle\Varid{i}_{1}\mathbin{:}\Varid{is}_{1},\Varid{cs}_{1},\mu,\rho\rangle\mathrel{=}_{\mathrm{1}\mathbin{+}J}\;(\Varid{i}_{2}\mathbin{:}\Varid{is}_{2},\Varid{cs}_{2},\mu,\rho\rangle}}
\and
\inferrule[Synch]
{\ensuremath{\Conid{C}_{1}\mathrel{=}\Conid{C}_{2}}}
{\ensuremath{\Conid{C}_{1}\;\mathrel{=}_{\bot}\;\Conid{C}_{2}}}
\end{mathpar}
\caption{$J$-Equivalence \ensuremath{\Conid{C}_{1}\;\mathrel{=}_{J}\;\Conid{C}_{2}}, where \ensuremath{J\;\in\;\mathbb{N}\;\uplus\;\{\mskip1.5mu \bot\mskip1.5mu\}}.}
\label{fig:j-equiv}
\end{figure}

\begin{lemma}[Determinism]
\begin{enumerate}
\item If \ensuremath{\Conid{C}\;\stepT{\Varid{d}}{\Varid{o}_{1}}\;\Conid{C}_{1}} and \ensuremath{\Conid{C}\;\stepT{\Varid{d}}{\Varid{o}_{2}}\;\Conid{C}_{2}}, then \ensuremath{\Conid{C}_{1}\mathrel{=}\Conid{C}_{2}} and \ensuremath{\Varid{o}_{1}\mathrel{=}\Varid{o}_{2}}.
\item If \ensuremath{\Conid{C}\;\Downarrow_{\Conid{O}_{1}}^{\Conid{D}}\;\Conid{C}_{1}} and \ensuremath{\Conid{C}\;\Downarrow_{\Conid{O}_{2}}^{\Conid{D}}\;\Conid{C}_{2}}, then \ensuremath{\Conid{C}_{1}\mathrel{=}\Conid{C}_{2}} and \ensuremath{\Conid{O}_{1}\mathrel{=}\Conid{O}_{2}}.
\end{enumerate}
\label{lemma:determinism}
\end{lemma}
\begin{proof}
  For Lemma~\ref{lemma:determinism}.1 (\emph{single-step determinism}),
  we observe that the directive \ensuremath{\Varid{d}} and the configuration \ensuremath{\Conid{C}}
  \emph{uniquely} and \emph{fully} determine which rule of the
  semantics can be executed and the observation generated. For
  Lemma~\ref{lemma:determinism}.2 (\emph{big-step determinism}), we
  apply single-step determinism (Lemma~\ref{lemma:determinism}.1) to
  each small-step reduction and induction.
\end{proof}

\begin{lemma}[Single-Step Consistency]
If \ensuremath{\Conid{C}_{1}\;=_{J}\;\Conid{C}_{2}}, \ensuremath{\Conid{C}_{1}\;\stepT{\Varid{d}}{\Varid{o}_{1}}\;\Conid{C}_{1}'}, \ensuremath{\Conid{C}_{2}\;\stepT{\Varid{d}}{\Varid{o}_{2}}\;\Conid{C}_{2}'} and either:
\begin{enumerate}
\item \ensuremath{J\mathrel{=}\bot};
\item \ensuremath{J\;\in\;\mathbb{N}} and \ensuremath{\Varid{d}\mathrel{=}\mathbf{retire}}
\item \ensuremath{J\;\in\;\mathbb{N}}, \ensuremath{\Varid{d}\mathrel{=}\mathbf{exec}\;\Varid{n}}, and \ensuremath{\Varid{n}<J\mathbin{+}\mathrm{1}};
\end{enumerate}
Then \ensuremath{\Varid{o}_{1}\mathrel{=}\Varid{o}_{2}} and there exists \ensuremath{\Conid{J'}} such that \ensuremath{\Conid{C}_{1}'\;\mathrel{=}_{\Conid{J'}}\;\Conid{C}_{2}'}.
\label{lemma:single-step-consistency}
\end{lemma}
\begin{proof}
\begin{enumerate}[wide, labelwidth=!, labelindent=0pt]
\item Trivial. We derive \ensuremath{\Conid{C}_{1}\mathrel{=}\Conid{C}_{2}} from \ensuremath{\Conid{C}_{1}\;\mathrel{=}_{\bot}\;\Conid{C}_{2}} (rule \srule{Synch}), then we apply \emph{single-step determinism} (Lemma~\ref{lemma:determinism}.1) and obtain \ensuremath{\Varid{o}_{1}\mathrel{=}\Varid{o}_{2}} and \ensuremath{\Conid{C}_{1}'\mathrel{=}\Conid{C}_{2}'}, which implies \ensuremath{\Conid{C}_{1}'\;\mathrel{=}_{\bot}\;\Conid{C}_{2}'} by rule \srule{Synch}.
\item Let \ensuremath{\Varid{is}_{1}} and \ensuremath{\Varid{is}_{2}} be the reorder buffers of \ensuremath{\Conid{C}_{1}\;=_{J}\;\Conid{C}_{2}}, respectively.
If \ensuremath{\Varid{d}\mathrel{=}\mathbf{retire}}, then the index \ensuremath{J} must be non-zero, i.e., \ensuremath{J>\mathrm{0}},\footnote{If \ensuremath{J\mathrel{=}\mathrm{0}}, then \ensuremath{\Conid{C}_{1}\;\mathrel{=}_{\mathrm{0}}\;\Conid{C}_{2}} implies that \ensuremath{\Conid{C}_{2}} retires a \emph{pending} guard instruction (rule \srule{Zero}), which
contradicts the single-step \ensuremath{\Conid{C}_{2}\;\stepT{\mathbf{retire}}{\Varid{o}_{2}}\;\Conid{C}_{2}'}.} therefore there exists \ensuremath{\Conid{J'}} such that \ensuremath{\Conid{C}_{1}\;\mathrel{=}_{\mathrm{1}\mathbin{+}\Conid{J'}}\;\Conid{C}_{2}}.
Then, rule \srule{Suc} implies that \ensuremath{\Varid{is}_{1}\mathrel{=}\Varid{i}_{1}\mathbin{:}\Varid{is}_{1}'} and \ensuremath{\Varid{is}_{2}\mathrel{=}\Varid{i}_{2}\mathbin{:}\Varid{is}_{2}'} where \ensuremath{\Varid{i}_{1}\mathrel{=}\Varid{i}_{2}}, and thus \ensuremath{\Conid{C}_{1}} and \ensuremath{\Conid{C}_{2}} retire \emph{identical} instructions.
Therefore, the reductions preserve equality of the variable map and memory store and generate the identical observations, i.e., \ensuremath{\Varid{o}_{1}\mathrel{=}\Varid{o}_{2}}.
Since instructions \ensuremath{\Varid{i}_{1}\mathrel{=}\Varid{i}_{2}} are removed from the respective reorder buffer, the resulting configurations are \ensuremath{\Conid{J'}}-equivalence, i.e., \ensuremath{\Conid{C}_{1}'\;\mathrel{=}_{\Conid{J'}}\;\Conid{C}_{2}'} for rules \srule{Retire-Nop}, \srule{Retire-Asgn}, and \srule{Retire-Store}, and \ensuremath{\Conid{C}_{1}'\;\mathrel{=}_{\bot}\;\Conid{C}_{2}'} for rule \srule{Retire-Fail}. \hfill
\item Let \ensuremath{\Varid{is}_{1}\mathrel{=}\Varid{is}_{1}'+{\mkern-9mu+}\ [\mskip1.5mu \Varid{i}_{1}\mskip1.5mu]+{\mkern-9mu+}\ \Varid{is}_{1}''} such that \ensuremath{{\vert}\Varid{is}_{1}'{\vert}\mathrel{=}\Varid{n}\mathbin{-}\mathrm{1}}
and thus \ensuremath{\Varid{i}_{1}} is the instruction executed by \ensuremath{\Conid{C}_{1}} and let \ensuremath{\Varid{is}_{2}} in \ensuremath{\Conid{C}_{2}} be similarly decomposed according to rule \srule{Execute} (Fig.~\ref{fig:full-execute-stage}).
Since \ensuremath{\Conid{C}_{1}\;=_{J}\;\Conid{C}_{2}}, then the configurations have equal memory stores, i.e., \ensuremath{\mu_{1}\mathrel{=}\mu_{2}}, variable maps, i.e., \ensuremath{\rho_{1}\mathrel{=}\rho_{2}}, executed instructions \ensuremath{\Varid{i}_{1}\mathrel{=}\Varid{i}_{2}}, and instruction prefixes \ensuremath{\Varid{is}_{1}'\mathrel{=}\Varid{is}_{2}'} since \ensuremath{\Varid{n}<J\mathbin{+}\mathrm{1}}.
Therefore, the \emph{transient variable maps} computed in rule
\srule{Execute} are the same, i.e., \ensuremath{\rho_{1}'\mathrel{=}\phi(\Varid{is}_{1},\rho_{1})\mathrel{=}\phi(\Varid{is}_{2},\rho_{2})\mathrel{=}\rho_{2}'}.
Thus, the configurations \ensuremath{\Conid{C}_{1}\;=_{J}\;\Conid{C}_{2}} execute the same instruction \ensuremath{\Varid{i}_{1}\mathrel{=}\Varid{i}_{2}}, under the same transient variable maps \ensuremath{\rho_{1}'\mathrel{=}\rho_{2}'} and memories \ensuremath{\mu_{1}\mathrel{=}\mu_{2}}, i.e., \ensuremath{\langle\Varid{is}_{1}',\Varid{i}_{1},\Varid{is}_{1}'',\Varid{cs}_{1}\rangle\xrsquigarrow{(\mu_{1},\rho_{1}',\Varid{o}_{1})}\langle\Varid{is}_{1}''',\Varid{cs}_{1}'\rangle} and
\ensuremath{\langle\Varid{is}_{2}',\Varid{i}_{2},\Varid{is}_{2}'',\Varid{cs}_{2}\rangle\xrsquigarrow{(\mu_{2},\rho_{2}',\Varid{o}_{2})}\langle\Varid{is}_{2}''',\Varid{cs}_{2}'\rangle}.
Since rule \srule{Execute} does not commit changes to the variable map and memory store, it suffices to show that \ensuremath{\Varid{o}_{1}\mathrel{=}\Varid{o}_{2}} and that there exists a \ensuremath{\Conid{J'}} such that the first \ensuremath{\Conid{J'}} instructions in \ensuremath{\Varid{is}_{1}'''} and \ensuremath{\Varid{is}_{2}'''} are identical.
Since \ensuremath{\Varid{is}_{1}'\mathrel{=}\Varid{is}_{2}'}, \ensuremath{\Varid{i}_{1}\mathrel{=}\Varid{i}_{2}}, and \ensuremath{\rho_{1}'\mathrel{=}\rho_{2}'}, the configurations step according to the same rule,
evaluate equal expressions with equal transient maps and  produce equal resolved instructions, i.e., \ensuremath{\Varid{i}_{1}'\mathrel{=}\Varid{i}_{2}'} in rules \srule{Exec-Asgn}, \srule{Exec-Load}, \srule{Exec-Store} and \ensuremath{\mathbf{nop}} in rule \srule{Exec-Branch-Ok} and \srule{Exec-Branch-Mispredict}, and thus \ensuremath{\Varid{o}_{1}\mathrel{=}\Varid{o}_{2}}.
Furthermore, for rule \srule{Exec-Branch-Mispredict} the subfixes \ensuremath{\Varid{is}_{1}''} and \ensuremath{\Varid{is}_{2}''} containing the mispredicted guard are flushed from \ensuremath{\Varid{is}_{1}'''} and \ensuremath{\Varid{is}_{2}'''} (since \ensuremath{\Varid{n}<J\mathbin{+}\mathrm{1}}), and thus the configurations become fully synchronized, i.e., \ensuremath{\Conid{C}_{1}'\;\mathrel{=}_{\bot}\;\Conid{C}_{2}'}.
For all the other rules, the resulting buffers \ensuremath{\Varid{is}_{1}'''} and \ensuremath{\Varid{is}_{2}'''}
are identical to \ensuremath{\Varid{is}_{1}} and \ensuremath{\Varid{is}_{2}}, respectively, except for the
resolved instructions \ensuremath{\Varid{i}_{1}'\mathrel{=}\Varid{i}_{2}'}, therefore the configurations remain
\ensuremath{J}-equivalent, i.e., \ensuremath{\Conid{C}_{1}'\;=_{J}\;\Conid{C}_{2}'}.

%

\end{enumerate}
\end{proof}

Before proving consistency between the sequential and speculative semantics, we define the \emph{filtering function} \ensuremath{\Conid{O}{\downarrow}}, which rewrites rollback and mispredicted observations to the \emph{silent} observation \ensuremath{\epsilon}.

\begin{definition}[Filtering Function]
Given an observation trace \ensuremath{\Conid{O}}, let \ensuremath{\Conid{P}} be the set of  identifiers of mispredicted guards and fail instructions in \ensuremath{\Conid{O}}, i.e., \ensuremath{\Conid{P}\mathrel{=}\{\mskip1.5mu \Varid{p}\;\;|\;\;\mathbf{rollback}(\Varid{p})\;\in\;\Conid{O}\mskip1.5mu\}\; \cup \;\{\mskip1.5mu \Varid{p}\;\;|\;\;\mathbf{fail}(\Varid{p})\;\in\;\Conid{O}\mskip1.5mu\}}. Then, we define the filtering function \ensuremath{\Conid{O}{\downarrow}} by case analysis on \ensuremath{\Conid{O}}:
\begin{align*}
&\ensuremath{\epsilon{\downarrow}\mathrel{=}\epsilon} \\
&\ensuremath{(\Varid{o}\;{{\mkern-2mu\cdot\mkern-2mu}}\;\Conid{O}){\downarrow}\mathrel{=}}
\begin{cases}
 \ensuremath{\epsilon\;{{\mkern-2mu\cdot\mkern-2mu}}\;\Conid{O}{\downarrow}} & \text{ if } \ensuremath{\Varid{o}\mathrel{=}\mathbf{rollback}(\Varid{p})} \\
 \ensuremath{\mathbf{fail}\;{{\mkern-2mu\cdot\mkern-2mu}}\;\Conid{O}{\downarrow}} & \text{ if } \ensuremath{\Varid{o}\mathrel{=}\mathbf{fail}\;\Varid{p}} \\
 \ensuremath{\epsilon\;{{\mkern-2mu\cdot\mkern-2mu}}\;\Conid{O}{\downarrow}} & \text{ if } \ensuremath{\Varid{o}\mathrel{=}\mathbf{load}(\Varid{n},\Varid{ps})}\ \ensuremath{\land}\ \ensuremath{\Varid{ps}\; \cap \;\Conid{P}\neq\varnothing} \\
 \ensuremath{\epsilon\;{{\mkern-2mu\cdot\mkern-2mu}}\;\Conid{O}{\downarrow}} & \text{ if } \ensuremath{\Varid{o}\mathrel{=}\mathbf{store}(\Varid{n},\Varid{ps})}\ \ensuremath{\land}\ \ensuremath{\Varid{ps}\; \cap \;\Conid{P}\neq\varnothing} \\
\ensuremath{\Varid{o}\;{{\mkern-2mu\cdot\mkern-2mu}}\;\Conid{O}{\downarrow}} & \text{ otherwise } \\
\end{cases}
\end{align*}
\label{def:filter-function}
\end{definition}

%

The proof of consistency relies on the following more general lemma.

\begin{lemma}[General Consistency]
For all configurations \ensuremath{\Conid{C}_{1}} and \ensuremath{\Conid{C}_{2}} such that \ensuremath{\Conid{C}_{1}\;\mathrel{=}_{J}\;\Conid{C}_{2}}, \emph{valid} schedules \ensuremath{\Conid{D}_{1}} and \ensuremath{\Conid{D}_{2}}, observations \ensuremath{\Conid{O}_{1}} and \ensuremath{\Conid{O}_{2}} such that \ensuremath{\mathbf{rollback}\;\not\in\;\Conid{O}_{1}}, if \ensuremath{\Conid{C}_{1}\;\Downarrow_{\Conid{O}_{1}}^{\Conid{D}_{1}}\;\Conid{C}_{1}'} and \ensuremath{\Conid{C}_{2}\;\Downarrow_{\Conid{O}_{2}}^{\Conid{D}_{2}}\;\Conid{C}_{2}'}, then \ensuremath{\Conid{C}_{1}'\mathrel{=}\Conid{C}_{2}'} and \ensuremath{\Conid{O}_{1}\cong\Conid{O}_{2}{\downarrow}}.
\label{lemma:general-consistency}
\end{lemma}

\mypara{Proof Outline}
The proof consists in \emph{synchronizing} the two
executions by carefully reordering the directives of one execution after the other.
The fact that the second schedule may interleave directives
concerning both correct \emph{and} misspeculated paths complicates the proof.
Therefore, we rely on \ensuremath{J}-equivalence
to deal with configurations that are only \emph{partially} synchronized.
In particular, \ensuremath{J}-equivalence denotes whether two
configurations are \emph{fully} synchronized (\ensuremath{J\mathrel{=}\bot}) or
synchronized \emph{up to \ensuremath{J} instructions} in the reorder buffers (\ensuremath{J\;\in\;\mathbb{N}}).
Using \ensuremath{J}-equivalence, we can distinguish directives that \emph{must}
be synchronized (e.g., those that retire or execute an instruction at
index \ensuremath{\Varid{n}<J\mathbin{+}\mathrm{1}}), from those that belong to a misspeculated path
(e.g., executing a mispredicted instruction at index \ensuremath{\Varid{n}>J\mathbin{+}\mathrm{1}}), as
well as synchronization points (i.e., resolving the pending
mispredicted guard at index \ensuremath{\Varid{n}\mathrel{=}J\mathbin{+}\mathrm{1}}).

\begin{proof}
By induction on the big-step reductions and the $J$-equivalence relation.
If the first reduction is \srule{Done}, then \ensuremath{\Conid{C}_{1}\mathrel{=}\langle[\mskip1.5mu \mskip1.5mu],[\mskip1.5mu \mskip1.5mu],\mu,\rho\rangle\mathrel{=}\Conid{C}_{1}'} and \ensuremath{\Conid{O}_{1}\mathrel{=}\epsilon}.
Since the reorder buffer in \ensuremath{\Conid{C}_{1}} is \emph{empty} and
\ensuremath{\Conid{C}_{1}\;\mathrel{=}_{J}\;\Conid{C}_{2}}, we deduce that \ensuremath{J\mathrel{=}\bot} (only rule \srule{Synch} applies to empty buffers) and thus the buffer of \ensuremath{\Conid{C}_{2}} is also empty.
Therefore, the second reduction is also \srule{Done}, thus \ensuremath{\Conid{C}_{2}\mathrel{=}\Conid{C}_{2}'} and \ensuremath{\Conid{O}_{2}\mathrel{=}\epsilon}, and hence \ensuremath{\Conid{C}_{1}'\mathrel{=}\Conid{C}_{2}'} and \ensuremath{\Conid{O}_{1}\mathrel{=}\Conid{O}_{2}{\downarrow}}.
In the inductive case both big-step reductions follow rule \srule{Step},
therefore for \ensuremath{\Varid{i}\;\in\;\{\mskip1.5mu \mathrm{1},\mathrm{2}\mskip1.5mu\}}, we have schedules \ensuremath{\Conid{D}_{\Varid{i}}\mathrel{=}\Varid{d}_{\Varid{i}}\mathbin{:}\Conid{D}'_{\Varid{i}}}, observation traces \ensuremath{\Conid{O}_{\Varid{i}}\mathrel{=}\Varid{o}_{\Varid{i}}\mathbin{:}\Conid{O}'_{\Varid{i}}}, and two pairs of small- and big-steps, i.e.,
\ensuremath{\Conid{C}_{\Varid{i}}\;\stepT{\Varid{d}_{\Varid{i}}}{\Varid{o}_{\Varid{i}}}\;\Conid{C}'_{\Varid{i}}} and \ensuremath{\Conid{C}'_{\Varid{i}}\;\Downarrow_{\Conid{O}'_{\Varid{i}}}^{\Conid{D}'_{\Varid{i}}}\;\Conid{C}''_{\Varid{i}}}.
We will refer to these names in the rest of the proof.

If \ensuremath{\Varid{d}_{1}\mathrel{=}\Varid{d}_{2}}, then we apply Lemma~\ref{lemma:single-step-consistency} (if the index \ensuremath{J} and the directives \ensuremath{\Varid{d}_{1}\mathrel{=}\Varid{d}_{2}} fulfill either of the conditions of the lemma)\footnote{ Otherwise, if \ensuremath{\Varid{d}_{1}\mathrel{=}\Varid{d}_{2}} but \ensuremath{J\neq\bot} and \ensuremath{\Varid{d}_{2}\mathrel{=}\mathbf{fetch}} (\ensuremath{\Varid{d}_{2}\mathrel{=}\mathbf{fetch}\;\Varid{b}} ), see case~\ref{case:d2-fetch}. If  \ensuremath{J\;\in\;\mathbb{N}}, \ensuremath{\Varid{d}_{1}\mathrel{=}\Varid{d}_{2}}, \ensuremath{\Varid{d}_{1}\mathrel{=}\mathbf{exec}\;\Varid{n}}, but \ensuremath{\Varid{n}\mathrel{=}J\mathbin{+}\mathrm{1}}, see case~\ref{case:exec-n-J}. In this case, notice that \ensuremath{\Varid{n}} cannot be greater than \ensuremath{J\mathbin{+}\mathrm{1}}, i.e., $n \ngtr J + 1$, because
the buffer of \ensuremath{\Conid{C}_{1}} has only \ensuremath{J\mathbin{+}\mathrm{1}} instructions, which contradicts the step \ensuremath{\Conid{C}_{1}\;\stepT{\mathbf{exec}\;\Varid{n}}{\Varid{o}_{1}}\;\Conid{C}_{1}'}.} and induction.
In the following, we focus on the general case where \ensuremath{\Varid{d}_{1}\neq\Varid{d}_{2}} and rely on $J$-equivalence to determine whether
 the configurations are fully synchronized (\ensuremath{J\mathrel{=}\bot}) or
a mispredicted guard is pending in the second configuration (\ensuremath{J\;\in\;\mathbb{N}}) at index \ensuremath{J\mathbin{+}\mathrm{1}}.
%
%
%
\vspace{\baselineskip}
\begin{enumerate}[wide, labelwidth=!, labelindent=0pt]
\setlength{\parskip}{\baselineskip}
\item \label{case:split} \ensuremath{J\mathrel{=}\bot} and \ensuremath{\Varid{d}_{1}\mathrel{=}\mathbf{fetch}\;\Varid{b}_{1}}, \ensuremath{\Varid{d}_{2}\mathrel{=}\mathbf{fetch}\;\Varid{b}_{2}}, and \ensuremath{\Varid{b}_{1}\neq\Varid{b}_{2}}. Since the first
  execution does not contain mispredictions (\ensuremath{\mathbf{rollback}\;\not\in\;\Conid{O}_{1}}), we know that the second configuration enters a mispredicted
  branch and diverges \emph{temporarily} from the first.
The initial configurations \ensuremath{\Conid{C}_{1}} and \ensuremath{\Conid{C}_{2}} are synchronized (\ensuremath{J\mathrel{=}\bot}) and step via rule \srule{Fetch-If-$b_1$} and \srule{Fetch-If-$b_2$}, respectively, i.e., \ensuremath{\Conid{C}_{1}\;\stepT{\mathbf{fetch}\;\Varid{b}_{1}}{\epsilon}\;\Conid{C}_{1}'} and
\ensuremath{\Conid{C}_{1}\;\stepT{\mathbf{fetch}\;\Varid{b}_{2}}{\epsilon}\;\Conid{C}_{1}'}, both generating a \emph{silent} observation \ensuremath{\epsilon}.
Since \ensuremath{J\mathrel{=}\bot} the reorder buffers of \ensuremath{\Conid{C}_{1}} and \ensuremath{\Conid{C}_{2}} are identical, i.e., \ensuremath{\Varid{is}_{1}\mathrel{=}\Varid{is}_{2}} from rule \srule{Synch}, and remain related after the small-step.
In particular, let \ensuremath{\Conid{J'}\mathrel{=}{\vert}\Varid{is}_{1}{\vert}\mathrel{=}{\vert}\Varid{is}_{2}{\vert}}, then the intermediate configurations \ensuremath{\Conid{C}_{1}'} and \ensuremath{\Conid{C}_{2}'} are $(1+J')$-equivalent, i.e., \ensuremath{\Conid{C}_{1}'\;\mathrel{=}_{\mathrm{1}\mathbin{+}\Conid{J'}}\;\Conid{C}_{2}'} by $J'$ applications of rule \srule{Suc} and then rule \srule{Zero}.
Then, we apply our induction hypothesis to
\ensuremath{\Conid{C}_{1}'\;\Downarrow_{\Conid{D}_{1}'}^{\Conid{O}_{1}'}\;\Conid{C}_{1}''} and \ensuremath{\Conid{C}_{2}'\;\Downarrow_{\Conid{D}_{2}'}^{\Conid{O}_{2}'}\;\Conid{C}_{2}''} and conclude that the final configurations are identical, i.e., \ensuremath{\Conid{C}_{1}''\mathrel{=}\Conid{C}_{2}''}, and the observation traces are equal up to permutation and filtering, i.e., \ensuremath{\Conid{O}_{1}'\cong\Conid{O}_{2}'{\downarrow}} thus \ensuremath{\epsilon\;{{\mkern-2mu\cdot\mkern-2mu}}\;\Conid{O}_{1}'\cong\epsilon\;{{\mkern-2mu\cdot\mkern-2mu}}\;\Conid{O}_{2}'{\downarrow}}, i.e., \ensuremath{\Conid{O}_{1}\mathrel{=}\Conid{O}_{2}{\downarrow}}.

\item \ensuremath{J\mathrel{=}\bot} and \ensuremath{\Varid{d}_{2}\mathrel{=}\mathbf{fetch}} (\ensuremath{\Varid{d}_{2}\mathrel{=}\mathbf{fetch}\;\Varid{b}_{2}}) and \ensuremath{\Varid{d}_{1}\neq\Varid{d}_{2}}.
Let \ensuremath{\Varid{c}_{2}} be the command fetched, i.e., the command on top of the stack \ensuremath{\Varid{cs}_{2}} of configuration \ensuremath{\Conid{C}_{2}}.
Similarly, let \ensuremath{\Varid{cs}_{1}} be the stack of the first configuration \ensuremath{\Conid{C}_{1}}.
Since the configurations are $J$-equivalent, i.e., \ensuremath{\Conid{C}_{1}\;=_{J}\;\Conid{C}_{2}} and synchronized (\ensuremath{J\mathrel{=}\bot}), then \ensuremath{\Varid{cs}_{1}\mathrel{=}\Varid{cs}_{2}} from rule \srule{Synch}, therefore there exists a command \ensuremath{\Varid{c}_{1}\mathrel{=}\Varid{c}_{2}} on top of \ensuremath{\Varid{cs}_{1}}.
Then, we anticipate the next \ensuremath{\mathbf{fetch}} directive in the first schedule
(one must exist because schedule \ensuremath{\Conid{D}_{1}} is \emph{valid})
in order to fetch \ensuremath{\Varid{c}_{1}} and preserve \ensuremath{J}-equivalence.
Formally, let \ensuremath{\Conid{D}_{1}\mathrel{=}\Conid{D}_{1}'+{\mkern-9mu+}\ [\mskip1.5mu \mathbf{fetch}\mskip1.5mu]+{\mkern-9mu+}\ \Conid{D}_{1}''}, such that sub-schedule \ensuremath{\Conid{D}_{1}'}
contains no \ensuremath{\mathbf{fetch}} directives and let \ensuremath{\Conid{D}_{1}'''\mathrel{=}[\mskip1.5mu \mathbf{fetch}\mskip1.5mu]+{\mkern-9mu+}\ \Conid{D}_{1}'+{\mkern-9mu+}\ \Conid{D}_{1}''} be the
alternative schedule where directive \ensuremath{\mathbf{fetch}} is anticipated.
Intuitively, schedule \ensuremath{\Conid{D}_{1}'''} remains valid (Def.~\ref{def:valid}) for \ensuremath{\Conid{C}_{1}} because the
new instruction fetched end up at the end of the reorder buffer \ensuremath{\Varid{is}_{1}}
and thus do not interfere with the \ensuremath{\mathbf{retire}} directive in \ensuremath{\Conid{D}_{1}'}, which
operate on \ensuremath{\Varid{is}_{1}}, and with the \ensuremath{\mathbf{exec}\;\Varid{n}}
instructions in \ensuremath{\Conid{D}_{1}'}, which remain valid at their original index \ensuremath{\mathrm{1}\leq \Varid{n}\leq {\vert}\Varid{is}_{1}{\vert}}.
Since \ensuremath{\Conid{D}_{1}} and \ensuremath{\Conid{D}_{1}'''} are valid, we apply
Lemma~\ref{lemma:valid-consistent} and deduce that there exists an intermediate
configuration \ensuremath{\Conid{C}_{1}'''} such that \ensuremath{\Conid{C}_{1}\;\stepT{\mathbf{fetch}}{\epsilon}\;\Conid{C}_{1}'''}, which then reduces
to the \emph{same} final configuration \ensuremath{\Conid{C}_{1}''}, i.e., \ensuremath{\Conid{C}_{1}'''\;\Downarrow_{\Conid{O}_{1}''}^{\Conid{D}_{1}'+{\mkern-9mu+}\ \Conid{D}_{1}''}\;\Conid{C}_{1}''} for some observation trace \ensuremath{\Conid{O}_{1}''}.
Furthermore, Lemma~\ref{lemma:valid-consistent} ensures that the observation
trace \ensuremath{\Conid{O}_{1}'''\mathrel{=}\epsilon\;{{\mkern-2mu\cdot\mkern-2mu}}\;\Conid{O}_{1}''} generated by the alternative schedule \ensuremath{\Conid{D}_{1}'''} is equivalent
up to permutation to the observation trace \ensuremath{\Conid{O}_{1}\mathrel{=}\Varid{o}_{1}\;{{\mkern-2mu\cdot\mkern-2mu}}\;\Conid{O}_{1}'} generated by the original schedule \ensuremath{\Conid{D}_{1}}, i.e., we have that \ensuremath{\Conid{O}_{1}\cong\Conid{O}_{1}'''}.
%
Since \ensuremath{\Conid{C}_{1}\;=_{J}\;\Conid{C}_{2}}, \ensuremath{J\mathrel{=}\bot}, \ensuremath{\Conid{C}_{1}\;\stepT{\mathbf{fetch}}{\epsilon}\;\Conid{C}_{1}'''}, and \ensuremath{\Conid{C}_{2}\;\stepT{\mathbf{fetch}}{\epsilon}\;\Conid{C}_{2}'},
we apply  Lemma~\ref{lemma:single-step-consistency} and obtain that there exists a \ensuremath{\Conid{J'}} such that the intermediate configurations are \ensuremath{\Conid{J'}}-equivalent, i.e., \ensuremath{\Conid{C}_{1}'''\;\mathrel{=}_{\Conid{J'}}\;\Conid{C}_{2}'}.
Then, we apply our induction hypothesis to reductions \ensuremath{\Conid{C}_{1}'''\;\Downarrow_{\Conid{O}_{1}''}^{\Conid{D}_{1}'+{\mkern-9mu+}\ \Conid{D}_{1}''}\;\Conid{C}_{1}''} and \ensuremath{\Conid{C}_{2}'\;\Downarrow_{\Conid{O}_{2}'}^{\Conid{D}_{2}'}\;\Conid{C}_{2}''}, deduce that the final configurations are identical, i.e., \ensuremath{\Conid{C}_{1}''\mathrel{=}\Conid{C}_{2}''}, and the observation traces are equal up to permutation and filtering \ensuremath{\Conid{O}_{1}''\cong\Conid{O}_{2}'{\downarrow}} and thus \ensuremath{\epsilon\;{{\mkern-2mu\cdot\mkern-2mu}}\;\Conid{O}_{1}''\cong\epsilon\;{{\mkern-2mu\cdot\mkern-2mu}}\;\Conid{O}_{2}'{\downarrow}}, i.e., \ensuremath{\Conid{O}_{1}'''\cong\Conid{O}_{2}{\downarrow}}.
Finally, from \ensuremath{\Conid{O}_{1}\cong\Conid{O}_{1}'''} and \ensuremath{\Conid{O}_{1}'''\cong\Conid{O}_{2}{\downarrow}} we derive \ensuremath{\Conid{O}_{1}\cong\Conid{O}_{2}{\downarrow}} by transitivity.
Notice that if \ensuremath{\Varid{d}_{2}\mathrel{=}\mathbf{fetch}\;\Varid{b}_{2}}, then the next fetch directive in \ensuremath{\Conid{D}_{1}} is of the same type, i.e., \ensuremath{\mathbf{fetch}\;\Varid{b}_{1}},
and the proof follows as above if \ensuremath{\Varid{b}_{1}\mathrel{=}\Varid{b}_{2}} and as in case \ref{case:split}) if \ensuremath{\Varid{b}_{1}\neq\Varid{b}_{2}}.

\item \ensuremath{J\mathrel{=}\bot} and \ensuremath{\Varid{d}_{1}\mathrel{=}\mathbf{fetch}} (\ensuremath{\Varid{d}_{1}\mathrel{=}\mathbf{fetch}\;\Varid{b}_{1}}) and \ensuremath{\Varid{d}_{1}\neq\Varid{d}_{2}}.
Symmetric to the previous case.

\item \label{case:d2-fetch} \ensuremath{J\neq\bot} and \ensuremath{\Varid{d}_{2}\mathrel{=}\mathbf{fetch}} (\ensuremath{\Varid{d}_{2}\mathrel{=}\mathbf{fetch}\;\Varid{b}}).
Since \ensuremath{J\neq\bot}, a mispredicted
  guard is pending in the buffer of \ensuremath{\Conid{C}_{2}} at
  index \ensuremath{J\mathbin{+}\mathrm{1}} and thus the instruction fetched
in the step \ensuremath{\Conid{C}_{2}\;\stepT{\mathbf{fetch}}{\epsilon}\;\Conid{C}_{2}'} is going to be flushed
  eventually.
  We observe that this instruction is appended \emph{at the end} of the buffer, thus at
  an index greater than \ensuremath{J\mathbin{+}\mathrm{1}}, and therefore \ensuremath{\Conid{C}_{1}} remains
  \ensuremath{J}-equivalent with \ensuremath{\Conid{C}_{2}'}, i.e., \ensuremath{\Conid{C}_{1}\;=_{J}\;\Conid{C}_{2}'}.
We apply our induction hypothesis to \ensuremath{\Conid{C}_{1}\;\Downarrow_{\Conid{O}_{1}}^{\Conid{D}_{1}}\;\Conid{C}_{1}''} and \ensuremath{\Conid{C}_{2}'\;\Downarrow_{\Conid{O}_{2}'}^{\Conid{D}_{2}'}\;\Conid{C}_{2}''} and deduce that the final configurations are identical, i.e., \ensuremath{\Conid{C}_{1}''\mathrel{=}\Conid{C}_{2}''}, and the observation traces are equivalent up to permutation and filtering, i.e., \ensuremath{\Conid{O}_{1}\cong\Conid{O}_{2}'{\downarrow}}.
Notice that the step from \ensuremath{\Conid{C}_{2}} to \ensuremath{\Conid{C}_{2}'} generates a  \emph{silent} observation \ensuremath{\epsilon}, i.e., \ensuremath{\Conid{O}_{2}\mathrel{=}\epsilon\;{{\mkern-2mu\cdot\mkern-2mu}}\;\Conid{O}_{2}'\mathrel{=}\Conid{O}_{2}'}, and therefore \ensuremath{\Conid{O}_{1}\cong\Conid{O}_{2}'{\downarrow}} implies
 \ensuremath{\Conid{O}_{1}\cong\Conid{O}_{2}{\downarrow}}.

\item \ensuremath{\Varid{d}_{2}\mathrel{=}\mathbf{exec}\;\Varid{n}} and \ensuremath{\Varid{n}>J\mathbin{+}\mathrm{1}}.
  The instruction executed in the step \ensuremath{\Conid{C}_{2}\;\stepT{\mathbf{exec}\;\Varid{n}}{\Varid{o}_{2}}\;\Conid{C}_{2}'}
  is at index \ensuremath{\Varid{n}>J\mathbin{+}\mathrm{1}} and thus it belongs to a mispredicted path and will be \emph{squashed}.
In particular, the second buffer contains a pending, mispredicted guard instruction at
  index \ensuremath{J\mathbin{+}\mathrm{1}} with prediction identifier \ensuremath{\Varid{p}} such that \ensuremath{\mathbf{rollback}(\Varid{p})\;\in\;\Conid{O}_{2}'}.
Therefore the observation \ensuremath{\Varid{o}_{2}} generated in the step is
  tainted with the prediction identifier \ensuremath{\Varid{p}} of the guard and
  thus rewritten to \emph{silent} event \ensuremath{\epsilon} (Def.~\ref{def:filter-function}), i.e., \ensuremath{(\Varid{o}_{2}\;{{\mkern-2mu\cdot\mkern-2mu}}\;\Conid{O}_{2}'){\downarrow}\mathrel{=}\epsilon\;{{\mkern-2mu\cdot\mkern-2mu}}\;\Conid{O}_{2}'{\downarrow}\mathrel{=}\Conid{O}_{2}'{\downarrow}}.
  The resolved instruction is reinserted in the buffer at its original
  index \ensuremath{\Varid{n}>J\mathbin{+}\mathrm{1}} and thus \ensuremath{\Conid{C}_{1}} remains \ensuremath{J}-equivalent to \ensuremath{\Conid{C}_{2}'}, i.e.,
  \ensuremath{\Conid{C}_{1}\;=_{J}\;\Conid{C}_{2}'}, and the lemma follows by induction on \ensuremath{\Conid{C}_{1}\;\Downarrow_{\Conid{O}_{1}}^{\Conid{D}_{1}}\;\Conid{C}_{1}''} and \ensuremath{\Conid{C}_{2}'\;\Downarrow_{\Conid{O}_{2}'}^{\Conid{D}_{2}'}\;\Conid{C}_{2}''} as in the previous case.
%
%
%
%
\item \ensuremath{\Varid{d}_{2}\mathrel{=}\mathbf{exec}\;\Varid{n}} and \ensuremath{\Varid{n}<J\mathbin{+}\mathrm{1}} or \ensuremath{J\mathrel{=}\bot}.
Since \ensuremath{\Varid{n}<J\mathbin{+}\mathrm{1}} or \ensuremath{J\mathrel{=}\bot}, the instruction
executed in step \ensuremath{\Conid{C}_{2}\;\stepT{\mathbf{exec}\;\Varid{n}}{\Varid{o}_{2}}\;\Conid{C}_{2}'}
 will \emph{not} be squashed and thus we synchronize
the first schedule to execute it as well.
Let \ensuremath{\Varid{is}_{2}\mathrel{=}\Varid{is}_{2}'+{\mkern-9mu+}\ [\mskip1.5mu \Varid{i}_{2}\mskip1.5mu]+{\mkern-9mu+}\ \Varid{is}_{2}''} be the reorder buffer of \ensuremath{\Conid{C}_{2}} such that \ensuremath{{\vert}\Varid{is}_{2}'{\vert}\mathrel{=}\Varid{n}\mathbin{-}\mathrm{1}}
and thus \ensuremath{\Varid{i}_{2}} is the executed instruction at index \ensuremath{\Varid{n}}.
From \ensuremath{J}-equivalence, i.e., \ensuremath{\Conid{C}_{1}\;=_{J}\;\Conid{C}_{2}}, we deduce that
the buffer \ensuremath{\Varid{is}_{1}} of \ensuremath{\Conid{C}_{1}} can be decomposed
like \ensuremath{\Varid{is}_{2}}, i.e., \ensuremath{\Varid{is}_{1}\mathrel{=}\Varid{is}_{1}'+{\mkern-9mu+}\ [\mskip1.5mu \Varid{i}_{1}\mskip1.5mu]+{\mkern-9mu+}\ \Varid{is}_{1}''} where \ensuremath{{\vert}\Varid{is}_{1}'{\vert}\mathrel{=}\Varid{n}\mathbin{-}\mathrm{1}<J} and thus \ensuremath{\Varid{is}_{1}'\mathrel{=}\Varid{is}_{2}'} and \ensuremath{\Varid{i}_{1}\mathrel{=}\Varid{i}_{2}}.
Then, we observe that directive \ensuremath{\Varid{d}_{2}\mathrel{=}\mathbf{exec}\;\Varid{n}} is \emph{valid} also for \ensuremath{\Conid{C}_{1}} (Def.~\ref{def:valid-directive}) and we adjust the schedule \ensuremath{\Conid{D}_{1}} to anticipate the directive that executes that instruction.\footnote{In particular, the fact that \ensuremath{\Varid{d}_{2}\mathrel{=}\mathbf{exec}\;\Varid{n}} is valid for \ensuremath{\Conid{C}_{2}} implies that the data-dependencies of \ensuremath{\Varid{i}_{1}\mathrel{=}\Varid{i}_{2}} are also resolved in the prefix \ensuremath{\Varid{is}_{1}'\mathrel{=}\Varid{is}_{2}'} and thus \ensuremath{\Varid{d}_{2}} is valid for \ensuremath{\Conid{C}_{1}}. }
%
%
%
Formally, let \ensuremath{\Conid{D}_{1}\mathrel{=}\Conid{D}_{1}'+{\mkern-9mu+}\ [\mskip1.5mu \Varid{d}_{1}'\mskip1.5mu]+{\mkern-9mu+}\ \Conid{D}_{1}''} be the original schedule, where \ensuremath{\Varid{d}_{1}'} is the \ensuremath{\mathbf{exec}\;\Varid{n'}}
directive that evaluates \ensuremath{\Varid{i}_{1}} and let \ensuremath{\Conid{D}_{1}'''\mathrel{=}[\mskip1.5mu \mathbf{exec}\;\Varid{n}\mskip1.5mu]+{\mkern-9mu+}\ \Conid{D}_{1}'+{\mkern-9mu+}\ \Conid{D}_{1}''} be
the alternative schedule where the execution of \ensuremath{\Varid{i}_{1}} is anticipated.\footnote{
  The position \ensuremath{\Varid{n'}} where instruction \ensuremath{\Varid{i}_{1}} lies in the buffer when it gets
  executed following \ensuremath{\Conid{D}_{1}} can be different from \ensuremath{\Varid{n}} because other instructions
  may be retired in \ensuremath{\Conid{D}_{1}'}, i.e., in general \ensuremath{\Varid{n}\neq\Varid{n'}}.  }
Intuitively, schedule \ensuremath{\Conid{D}_{1}'''} is \emph{valid} because
directive \ensuremath{\mathbf{exec}\;\Varid{n}} does not cause any rollback (\ensuremath{\mathbf{rollback}\;\not\in\;\Conid{O}_{1}}) and thus it cannot interfere with
any \ensuremath{\mathbf{fetch}}, \ensuremath{\mathbf{exec}}, or \ensuremath{\mathbf{retire}} directive in \ensuremath{\Conid{D}_{1}'},
which were already valid in the original schedule.
%
%
Since \ensuremath{\Conid{D}_{1}} and \ensuremath{\Conid{D}_{1}'''} are valid, we apply
Lemma~\ref{lemma:valid-consistent} and deduce that there exists an intermediate
configuration \ensuremath{\Conid{C}_{1}'''} such that \ensuremath{\Conid{C}_{1}\;\stepT{\mathbf{exec}\;\Varid{n}}{\Varid{o}_{1}'}\;\Conid{C}_{1}'''}, which then reduces
to the \emph{same} final configuration \ensuremath{\Conid{C}_{1}''}, i.e., \ensuremath{\Conid{C}_{1}'''\;\Downarrow_{\Conid{O}_{1}''}^{\Conid{D}_{1}'+{\mkern-9mu+}\ \Conid{D}_{1}''}\;\Conid{C}_{1}''} for some observation trace \ensuremath{\Conid{O}_{1}''}.
Furthermore, Lemma~\ref{lemma:valid-consistent} ensures that the observation
trace \ensuremath{\Conid{O}_{1}'''\mathrel{=}\Varid{o}_{1}'\mathbin{:}\Conid{O}_{1}''} generated by the alternative schedule \ensuremath{\Conid{D}_{1}'''} is equivalent
up to permutation to the observation trace \ensuremath{\Conid{O}_{1}\mathrel{=}\Varid{o}_{1}\mathbin{:}\Conid{O}_{1}'} generated by the original schedule \ensuremath{\Conid{D}_{1}}, i.e., \ensuremath{\Conid{O}_{1}\cong\Conid{O}_{1}'''}.
%
Since \ensuremath{\Conid{C}_{1}\;=_{J}\;\Conid{C}_{2}}, \ensuremath{\Conid{C}_{1}\;\stepT{\mathbf{exec}\;\Varid{n}}{\Varid{o}_{1}'}\;\Conid{C}_{1}'''}, and \ensuremath{\Conid{C}_{2}\;\stepT{\mathbf{exec}\;\Varid{n}}{\Varid{o}_{2}}\;\Conid{C}_{2}'}, we apply  Lemma~\ref{lemma:single-step-consistency} and obtain that the observations generated are equivalent, i.e.,  \ensuremath{\Varid{o}_{1}'\mathrel{=}\Varid{o}_{2}}, and there exists \ensuremath{\Conid{J'}} such that the intermediate configurations are \ensuremath{\Conid{J'}}-equivalent, i.e.,
\ensuremath{\Conid{C}_{1}'''\;\mathrel{=}_{\Conid{J'}}\;\Conid{C}_{2}'}.
Then, we apply our induction hypothesis to reductions \ensuremath{\Conid{C}_{1}'''\;\Downarrow_{\Conid{O}_{1}''}^{\Conid{D}_{1}'+{\mkern-9mu+}\ \Conid{D}_{1}''}\;\Conid{C}_{1}''} and \ensuremath{\Conid{C}_{2}'\;\Downarrow_{\Conid{O}_{2}'}^{\Conid{D}_{2}'}\;\Conid{C}_{2}''}
and  deduce that the final configurations are identical, i.e., \ensuremath{\Conid{C}_{1}''\mathrel{=}\Conid{C}_{2}''}, and the observation traces are  equivalent up to permutation and filtering, i.e., \ensuremath{\Conid{O}_{1}''\cong\Conid{O}_{2}'{\downarrow}}.
From \ensuremath{\Conid{O}_{1}''\cong\Conid{O}_{2}'{\downarrow}} and \ensuremath{\Varid{o}_{1}'\mathrel{=}\Varid{o}_{2}}, we have \ensuremath{\Varid{o}_{1}'\mathbin{:}\Conid{O}_{1}''\cong\Varid{o}_{2}\mathbin{:}\Conid{O}_{2}'{\downarrow}}, i.e., \ensuremath{\Conid{O}_{1}'''\cong\Conid{O}_{2}{\downarrow}}.
Finally, from \ensuremath{\Conid{O}_{1}\cong\Conid{O}_{1}'''} and \ensuremath{\Conid{O}_{1}'''\cong\Conid{O}_{2}{\downarrow}} we derive \ensuremath{\Conid{O}_{1}\cong\Conid{O}_{2}{\downarrow}} by transitivity.

\item \label{case:exec-n-J} \ensuremath{\Varid{d}_{2}\mathrel{=}\mathbf{exec}\;\Varid{n}} and \ensuremath{\Varid{n}\mathrel{=}J\mathbin{+}\mathrm{1}}.
  Since \ensuremath{\Varid{n}\mathrel{=}J\mathbin{+}\mathrm{1}}, the second configuration \ensuremath{\Conid{C}_{2}} resolves the
  mispredicted guard and performs a rollback, i.e., the second
  small-step is \ensuremath{\Conid{C}_{2}\;\stepT{\mathbf{exec}\;\Varid{n}}{\mathbf{rollback}(\Varid{p})}\;\Conid{C}_{2}'}, where
  \ensuremath{\Varid{p}} is the identifier of the guard.
Then, we anticipate the execution of the corresponding guard instruction in the first configuration \ensuremath{\Conid{C}_{1}}, so that the two executions synchronize again.
Let \ensuremath{\Varid{is}_{2}\mathrel{=}\Varid{is}_{2}'+{\mkern-9mu+}\ [\mskip1.5mu \Varid{i}_{2}\mskip1.5mu]+{\mkern-9mu+}\ \Varid{is}_{2}''} be the buffer of \ensuremath{\Conid{C}_{2}}, where \ensuremath{{\vert}\Varid{is}_{2}'{\vert}\mathrel{=}\Varid{n}\mathbin{-}\mathrm{1}} and thus \ensuremath{\Varid{i}_{2}} is the mispredicted guard instruction.
From \ensuremath{J}-equivalence, i.e., \ensuremath{\Conid{C}_{1}\;=_{J}\;\Conid{C}_{2}}, we deduce that
the buffer \ensuremath{\Varid{is}_{1}} of \ensuremath{\Conid{C}_{1}} can be decomposed similarly, i.e., \ensuremath{\Varid{is}_{1}\mathrel{=}\Varid{is}_{1}'+{\mkern-9mu+}\ [\mskip1.5mu \Varid{i}_{1}\mskip1.5mu]},  where \ensuremath{{\vert}\Varid{is}_{1}'{\vert}\mathrel{=}\Varid{n}\mathbin{-}\mathrm{1}} and thus \ensuremath{\Varid{is}_{1}'\mathrel{=}\Varid{is}_{2}'} and \ensuremath{\Varid{i}_{1}} is the corresponding  guard instruction, but correctly predicted.
In particular, from rule \srule{Zero}, we learn that \ensuremath{\Varid{i}_{1}\mathrel{=}\mathbf{guard}(\Varid{e}^{\Varid{b}_{1}},\Varid{cs}_{2}',\Varid{p})} and \ensuremath{\Varid{i}_{2}\mathrel{=}\mathbf{guard}(\Varid{e}^{\Varid{b}_{2}},\Varid{cs}_{1},\Varid{p})}, where  \ensuremath{\Varid{b}_{1}\neq\Varid{b}_{2}}
and \ensuremath{\Varid{cs}_{1}} is the command stack of \ensuremath{\Conid{C}_{1}}.
Let \ensuremath{\Conid{D}_{1}\mathrel{=}\Conid{D}_{1}'+{\mkern-9mu+}\ [\mskip1.5mu \Varid{d}_{1}'\mskip1.5mu]+{\mkern-9mu+}\ \Conid{D}_{1}''} be the original schedule where \ensuremath{\Varid{d}_{1}'}
is the \ensuremath{\mathbf{exec}\;\Varid{n'}} directive that evaluates the guard instruction
\ensuremath{\Varid{i}_{1}} and let \ensuremath{\Conid{D}_{1}'''\mathrel{=}[\mskip1.5mu \mathbf{exec}\;\Varid{n}\mskip1.5mu]+{\mkern-9mu+}\ \Conid{D}_{1}'+{\mkern-9mu+}\ \Conid{D}_{1}''} be the valid,
alternative schedule that anticipates it.
%
%
We apply Lemma~\ref{lemma:valid-consistent} to valid schedules \ensuremath{\Conid{D}_{1}} and \ensuremath{\Conid{D}_{1}'''}
and deduce that there exists an intermediate configuration \ensuremath{\Conid{C}_{1}'''} such that \ensuremath{\Conid{C}_{1}\;\stepT{\mathbf{exec}\;\Varid{n}}{\epsilon}\;\Conid{C}_{1}'''} via rule \srule{Exec-Branch-Ok}, which then reduces to the \emph{same} final
configuration \ensuremath{\Conid{C}_{1}''}, i.e., \ensuremath{\Conid{C}_{1}'''\;\Downarrow_{\Conid{O}_{1}''}^{\Conid{D}_{1}'+{\mkern-9mu+}\ \Conid{D}_{1}''}\;\Conid{C}_{1}''} for some
observation trace \ensuremath{\Conid{O}_{1}''}.
Furthermore, the lemma ensures that the observation trace
\ensuremath{\Conid{O}_{1}'''\mathrel{=}\epsilon\;{{\mkern-2mu\cdot\mkern-2mu}}\;\Conid{O}_{1}''} generated by the alternative schedule \ensuremath{\Conid{D}_{1}'''} is equivalent up to permutation to the  observation trace \ensuremath{\Conid{O}_{1}\mathrel{=}\Varid{o}_{1}\;{{\mkern-2mu\cdot\mkern-2mu}}\;\Conid{O}_{1}'} generated by the original schedule \ensuremath{\Conid{D}_{1}}, i.e., \ensuremath{\Conid{O}_{1}\cong\Conid{O}_{1}'''}.
Notice that rule \srule{Exec-Branch-Ok} rewrites the buffer \ensuremath{\Varid{is}_{1}} in \ensuremath{\Conid{C}_{1}} to \ensuremath{\Varid{is}_{1}'\mathrel{=}\Varid{is}_{1}+{\mkern-9mu+}\ [\mskip1.5mu \mathbf{nop}\mskip1.5mu]+{\mkern-9mu+}\ [\mskip1.5mu \mskip1.5mu]} in \ensuremath{\Conid{C}_{1}'}, which is identical to the buffer \ensuremath{\Varid{is}_{2}'\mathrel{=}\Varid{is}_{1}+{\mkern-9mu+}\ [\mskip1.5mu \mathbf{nop}\mskip1.5mu]} obtained from the step \ensuremath{\Conid{C}_{2}\;\stepT{\mathbf{exec}\;\Varid{n}}{\mathbf{rollback}(\Varid{p})}\;\Conid{C}_{2}'} via rule \srule{Exec-Branch-Mispredict}, which additionally replaces the command stack of \ensuremath{\Conid{C}_{2}} with \ensuremath{\Varid{cs}_{1}} in \ensuremath{\Conid{C}_{2}'}.
As a result, the intermediate configurations \ensuremath{\Conid{C}_{1}'''} and \ensuremath{\Conid{C}_{2}'} have identical reorder buffers (\ensuremath{\Varid{is}_{1}'\mathrel{=}\Varid{is}_{2}'}) and command stacks (\ensuremath{\Varid{cs}_{1}\mathrel{=}\Varid{cs}_{1}}), memory stores and variable maps, which are equal in \ensuremath{\Conid{C}_{1}\;=_{J}\;\Conid{C}_{2}} and
left unchanged in \ensuremath{\Conid{C}_{1}'''} and \ensuremath{\Conid{C}_{2}'} by the small steps.
Therefore, the configurations \ensuremath{\Conid{C}_{1}'''} and \ensuremath{\Conid{C}_{2}'} are fully synchronized
again, i.e., \ensuremath{\Conid{C}_{1}'''\;\mathrel{=}_{\bot}\;\Conid{C}_{2}''} via rule \srule{Synch} and we can apply our induction hypothesis to the reductions \ensuremath{\Conid{C}_{1}'''\;\Downarrow_{\Conid{O}_{1}''}^{\Conid{D}_{1}'+{\mkern-9mu+}\ \Conid{D}_{1}''}\;\Conid{C}_{1}''} and \ensuremath{\Conid{C}_{2}'\;\Downarrow_{\Conid{O}_{2}'}^{\Conid{D}_{2}'}\;\Conid{C}_{2}''}.
Then, we deduce that the final configurations are identical, i.e., \ensuremath{\Conid{C}_{1}''\mathrel{=}\Conid{C}_{2}''}, and that the observation traces are equal up to permutation and filtering, i.e., \ensuremath{\Conid{O}_{1}''\cong\Conid{O}_{2}'{\downarrow}}.
Notice that the filter function rewrites the rollback observation
\ensuremath{\mathbf{rollback}(\Varid{p})} to the \emph{silent} observation \ensuremath{\epsilon}
(Def~\ref{def:filter-function}) in the second observation trace, i.e.,
\ensuremath{\Conid{O}_{2}{\downarrow}\mathrel{=}(\mathbf{rollback}(\Varid{p})\;{{\mkern-2mu\cdot\mkern-2mu}}\;\Conid{O}_{2}'){\downarrow}\mathrel{=}\epsilon\;{{\mkern-2mu\cdot\mkern-2mu}}\;\Conid{O}_{2}'{\downarrow}\mathrel{=}\Conid{O}_{2}'{\downarrow}}, and that the alternative observation trace starts with a silent
observation as well, i.e., \ensuremath{\Conid{O}_{1}'''\mathrel{=}\epsilon\;{{\mkern-2mu\cdot\mkern-2mu}}\;\Conid{O}_{1}''}, therefore from
\ensuremath{\Conid{O}_{1}''\cong\Conid{O}_{2}'{\downarrow}} above we obtain \ensuremath{\Conid{O}_{1}'''\cong\Conid{O}_{2}{\downarrow}}.
%
%
%
Finally, from \ensuremath{\Conid{O}_{1}\cong\Conid{O}_{1}'''} and \ensuremath{\Conid{O}_{1}'''\cong\Conid{O}_{2}{\downarrow}} we derive \ensuremath{\Conid{O}_{1}\cong\Conid{O}_{2}{\downarrow}} by transitivity.

\item \ensuremath{\Varid{d}_{2}\mathrel{=}\mathbf{retire}}.
Let \ensuremath{\Varid{i}_{2}} be the instruction that gets retired, i.e., the
\emph{first} instruction in the reorder buffer of \ensuremath{\Conid{C}_{2}}, and let \ensuremath{\Varid{i}_{1}} be the first instruction in the buffer of \ensuremath{\Conid{C}_{1}}.
From $J$-equivalence, i.e., \ensuremath{\Conid{C}_{1}\;=_{J}\;\Conid{C}_{2}}, we know that either the two
configurations are fully synchronized, i.e., \ensuremath{J\mathrel{=}\bot}, or \ensuremath{\Varid{i}_{2}} \emph{precedes} the
mispredicted guard at index \ensuremath{J\mathbin{+}\mathrm{1}}, i.e., \ensuremath{J>\mathrm{0}}.\footnote{ In particular, if \ensuremath{J\mathrel{=}\mathrm{0}}, then
  rule \srule{Zero} implies that instruction \ensuremath{\Varid{i}_{2}} is the mispredicted
  \emph{guard}, which cannot be retired, contradicting the step \ensuremath{\Conid{C}_{2}\;\stepT{\mathbf{retire}}{\Varid{o}_{2}}\;\Conid{C}_{2}'}. }
Therefore, from \ensuremath{\Conid{C}_{1}\;=_{J}\;\Conid{C}_{2}} it follows that \ensuremath{\Varid{i}_{1}\mathrel{=}\Varid{i}_{2}} and thus also instruction \ensuremath{\Varid{i}_{1}} is resolved and can be retired, i.e., directive \ensuremath{\mathbf{retire}} is valid for \ensuremath{\Conid{C}_{1}}.
Then, we anticipate the next \ensuremath{\mathbf{retire}} directive in the first schedule
(one must exist because \ensuremath{\Conid{D}_{1}} is \emph{valid}), in order to retire \ensuremath{\Varid{i}_{2}} and preserve \ensuremath{J}-equivalence.
Formally, let \ensuremath{\Conid{D}_{1}\mathrel{=}\Conid{D}_{1}'+{\mkern-9mu+}\ [\mskip1.5mu \mathbf{retire}\mskip1.5mu]+{\mkern-9mu+}\ \Conid{D}_{1}''} be the original
schedule, where \ensuremath{\Conid{D}_{1}'} does not contain any \ensuremath{\mathbf{retire}} directive.
Then, we define the alternative, \emph{valid} schedule \ensuremath{\Conid{D}_{1}'''\mathrel{=}[\mskip1.5mu \mathbf{retire}\mskip1.5mu]+{\mkern-9mu+}\ \Conid{D}_{1}'\downarrow^1+{\mkern-9mu+}\ \Conid{D}_{1}''}, where \ensuremath{\Conid{D}_{1}'\downarrow^1} decreases by \ensuremath{\mathrm{1}} the
indexes of all \ensuremath{\mathbf{exec}\;\Varid{n}} directives in \ensuremath{\Conid{D}_{1}'}  to account for the early retirement of \ensuremath{\Varid{i}_{1}}.
%
%
We apply Lemma~\ref{lemma:valid-consistent} to valid schedules \ensuremath{\Conid{D}_{1}} and \ensuremath{\Conid{D}_{1}'''}
and deduce that there exists an intermediate configuration \ensuremath{\Conid{C}_{1}'''} such that \ensuremath{\Conid{C}_{1}\;\stepT{\mathbf{retire}}{\Varid{o}_{1}'}\;\Conid{C}_{1}'''}
which then reduces
to the \emph{same} final configuration \ensuremath{\Conid{C}_{1}''}, i.e., \ensuremath{\Conid{C}_{1}'''\;\Downarrow_{\Conid{O}_{1}''}^{\Conid{D}_{1}'\downarrow^1+{\mkern-9mu+}\ \Conid{D}_{1}''}\;\Conid{C}_{1}''}, for some observation trace \ensuremath{\Conid{O}_{1}''}.
Furthermore, Lemma~\ref{lemma:valid-consistent} ensures that the observation
trace \ensuremath{\Conid{O}_{1}'''\mathrel{=}\Varid{o}_{1}'\mathbin{:}\Conid{O}_{1}''} generated by the alternative schedule \ensuremath{\Conid{D}_{1}'''} is equivalent
up to permutation to the observation trace \ensuremath{\Conid{O}_{1}\mathrel{=}\Varid{o}_{1}\mathbin{:}\Conid{O}_{1}'} generated by the original schedule \ensuremath{\Conid{D}_{1}}, i.e., we have \ensuremath{\Conid{O}_{1}\cong\Conid{O}_{1}'''}.
%
Since \ensuremath{\Conid{C}_{1}\;=_{J}\;\Conid{C}_{2}}, \ensuremath{\Conid{C}_{1}\;\stepT{\mathbf{retire}}{\Varid{o}_{1}'}\;\Conid{C}_{1}'''}, and \ensuremath{\Conid{C}_{2}\;\stepT{\mathbf{retire}}{\Varid{o}_{2}}\;\Conid{C}_{2}'}, we apply  Lemma~\ref{lemma:single-step-consistency} and obtain that the observations generated are equivalent, i.e.,  \ensuremath{\Varid{o}_{1}'\mathrel{=}\Varid{o}_{2}}, and there exists \ensuremath{\Conid{J'}} such that the intermediate configurations remain \ensuremath{\Conid{J'}}-equivalent, i.e., \ensuremath{\Conid{C}_{1}'''\;\mathrel{=}_{\Conid{J'}}\;\Conid{C}_{2}'}.
%
Then, we apply our induction hypothesis to reductions \ensuremath{\Conid{C}_{1}'''\;\Downarrow_{\Conid{O}_{1}''}^{\Conid{D}_{1}'\downarrow^1+{\mkern-9mu+}\ \Conid{D}_{1}''}\;\Conid{C}_{1}''} and \ensuremath{\Conid{C}_{2}'\;\Downarrow_{\Conid{O}_{2}'}^{\Conid{D}_{2}'}\;\Conid{C}_{2}''} and
deduce that the final configurations are identical, i.e., \ensuremath{\Conid{C}_{1}''\mathrel{=}\Conid{C}_{2}''}, and the observation traces are equivalent up to
permutation and filtering, i.e., \ensuremath{\Conid{O}_{1}''\cong\Conid{O}_{2}'{\downarrow}}.
From \ensuremath{\Conid{O}_{1}''\cong\Conid{O}_{2}'{\downarrow}} and \ensuremath{\Varid{o}_{1}'\mathrel{=}\Varid{o}_{2}}, we have \ensuremath{\Varid{o}_{1}'\mathbin{:}\Conid{O}_{1}''\cong\Varid{o}_{2}\mathbin{:}\Conid{O}_{2}'{\downarrow}}, i.e., \ensuremath{\Conid{O}_{1}'''\cong\Conid{O}_{2}{\downarrow}}.
Finally, from \ensuremath{\Conid{O}_{1}\cong\Conid{O}_{1}'''} and \ensuremath{\Conid{O}_{1}'''\cong\Conid{O}_{2}{\downarrow}} we derive \ensuremath{\Conid{O}_{1}\cong\Conid{O}_{2}{\downarrow}} by transitivity.

\end{enumerate}

\end{proof}

\begin{thm}[Speculative Consistency]
For all programs \ensuremath{\Varid{c}}, initial memory stores \ensuremath{\mu}, variable maps \ensuremath{\rho},
and valid directives \ensuremath{\Conid{D}}, such that \ensuremath{\langle\mu,\rho\rangle\Downarrow_{\Conid{O}}^{\Varid{c}}\langle\mu',\rho'\rangle}
and \ensuremath{\langle[\mskip1.5mu \mskip1.5mu],[\mskip1.5mu \Varid{c}\mskip1.5mu],\mu,\rho\rangle\Downarrow_{\Conid{O}'}^{\Conid{D}}\langle[\mskip1.5mu \mskip1.5mu],[\mskip1.5mu \mskip1.5mu],\mu'',\rho''\rangle},
then \ensuremath{\mu'\mathrel{=}\mu''}, \ensuremath{\rho'\mathrel{=}\rho''}, and \ensuremath{\Conid{O}\cong\Conid{O}'{\downarrow}}.
\end{thm}
\begin{proof}
First, we apply Lemma~\ref{lemma1} (\emph{sequential consistency}) to the sequential reduction \ensuremath{\langle\mu,\rho\rangle\Downarrow_{\Conid{O}}^{\Varid{c}}\langle\mu',\rho'\rangle} and obtain a valid, sequential schedule \ensuremath{\Conid{D}'} such that \ensuremath{\langle[\mskip1.5mu \mskip1.5mu],[\mskip1.5mu \Varid{c}\mskip1.5mu],\mu,\rho\rangle\Downarrow_{\Conid{O}}^{\Conid{D}'}\langle[\mskip1.5mu \mskip1.5mu],[\mskip1.5mu \mskip1.5mu]\;\mu',\rho'\rangle} where \ensuremath{\mathbf{rollback}\;\not\in\;\Conid{O}}.
Let \ensuremath{\Conid{C}\mathrel{=}\langle[\mskip1.5mu \mskip1.5mu],[\mskip1.5mu \Varid{c}\mskip1.5mu],\mu,\rho\rangle}, \ensuremath{\Conid{C}'\mathrel{=}\langle[\mskip1.5mu \mskip1.5mu],[\mskip1.5mu \mskip1.5mu]\;\mu',\rho'\rangle}, and \ensuremath{\Conid{C}''\mathrel{=}\langle[\mskip1.5mu \mskip1.5mu],[\mskip1.5mu \mskip1.5mu],\mu'',\rho''\rangle}.
We now have two speculative reductions \ensuremath{\Conid{C}\;\Downarrow_{\Conid{O}'}^{\Conid{D}}\;\Conid{C}'} and \ensuremath{\Conid{C}\;\Downarrow_{\Conid{O}}^{\Conid{D}'}\;\Conid{C}''}, where \ensuremath{\Conid{C}\;\mathrel{=}_{J}\;\Conid{C}} from rule \srule{Synch}.
Then, we apply Lemma~\ref{lemma:general-consistency} (\emph{general consistency}) to these reductions and obtain
 \ensuremath{\Conid{O}\cong\Conid{O}'{\downarrow}} and \ensuremath{\Conid{C}'\mathrel{=}\Conid{C}''}, which implies \ensuremath{\mu'\mathrel{=}\mu''} and \ensuremath{\rho'\mathrel{=}\rho''}.
\end{proof}

\subsection{Security}


\mypara{\ensuremath{\Blue{\Conid{L}}}-Equivalence}
The soundness proof of our type system relies on a relation called
\ensuremath{\Blue{\Conid{L}}}-equivalence, which intuitively relates configurations that are
\emph{indistinguishable} to an attacker that can observe only public data, which we identify with security label \ensuremath{\Conid{L}}.
Figure \ref{fig:low-equivalence} formally defines \ensuremath{\Blue{\Conid{L}}}-equivalence for the categories of our calculus.
Two configurations \ensuremath{\Conid{C}_{1}\mathrel{=}\langle\Varid{is}_{1},\Varid{cs}_{1},\mu_{1},\rho_{1}\rangle} and \ensuremath{\Conid{C}_{2}\mathrel{=}\langle\Varid{is}_{2},\Varid{cs}_{2},\mu_{2},\rho_{2}\rangle} are \ensuremath{\Blue{\Conid{L}}}-equivalent under typing environment \ensuremath{\Gamma},
written \ensuremath{\Gamma\vdash\Conid{C}_{1}\approx_{\Blue{\Conid{L}}}\Conid{C}_{2}}, if their reorder buffers are
\ensuremath{\Blue{\Conid{L}}}-equivalent, i.e., \ensuremath{\Gamma\vdash\Varid{is}_{1}\approx_{\Blue{\Conid{L}}}\Varid{is}_{2}}, they have the same
command stacks, i.e., \ensuremath{\Varid{cs}_{1}\mathrel{=}\Varid{cs}_{2}}, and their architectural state,
i.e., variable maps and memory stores, are related, i.e., \ensuremath{\rho_{1}\approx_{\Blue{\Conid{L}}}\rho_{2}} and \ensuremath{\mu_{1}\approx_{\Blue{\Conid{L}}}\mu_{2}}, respectively.
Syntactic equivalence for the command stacks ensures that programs cannot leak secret data through the instruction cache, e.g., by branching on secret data.
For the architectural state \ensuremath{\Blue{\Conid{L}}}-equivalence is standard; values of
public variables and stored at public memory addresses must be equal
(rules \srule{VarMap} and \srule{Memory} in
Fig.~\ref{fig:equiv-varmaps}).
Figure \ref{fig:equiv-stacks} defines \emph{pointwise
}\ensuremath{\Blue{\Conid{L}}}-equivalence for reorder buffers, which ensures that related
buffers have the same length.
Instructions are related (\ensuremath{\Gamma\vdash\Varid{i}_{1}\approx_{\Blue{\Conid{L}}}\Varid{i}_{2}} in
Fig.~\ref{fig:equiv-instr}) only if they are of the same kind, e.g.,
both instructions are guards, loads, assignments etc.
For most instructions, \ensuremath{\Blue{\Conid{L}}}-equivalence is just syntactic equivalence, e.g., for rules \srule{Nop} and \srule{Fail}.
In particular, to avoid leaking through the instruction cache, we
demand syntactic equivalence for guard instructions (rule
\srule{Guard}), where all components (condition \ensuremath{\Varid{e}}, predicted value
\ensuremath{\Varid{b}}, rollback stack \ensuremath{\Varid{cs}}, and prediction identifier \ensuremath{\Varid{p}}) must be
identical.
Similarly, to avoid leaks through the data cache, rule \srule{Load} relates load
instructions as long as they read the same address and update the
same variable.
We impose a similar restriction for store instructions (\srule{Store$_{\ensuremath{\Blue{\Conid{L}}}}$, Store$_{\ensuremath{\Red{\Conid{H}}}}$}) and additionally require equal expressions for stores that update public addresses (\srule{Store$_{\ensuremath{\Blue{\Conid{L}}}}$}), which are identified by the label annotation \ensuremath{\Blue{\Conid{L}}} that decorates the instruction itself.
Protect instructions allow their argument to be different because they
allow their arguments are allowed to be evaluated transiently
(\srule{Protect}).
Assignments are related as long as they update the same variable.
If the variable is public and typed stable in the typing environment,
the values assigned must be identical
(\srule{Asgn$_{\ensuremath{\Blue{\Conid{L}}\land\ensuremath{\Concrete}}}$}), but can be different for secret or
transient variables (\srule{Asgn$_{\ensuremath{\Red{\Conid{H}}\land\ensuremath{\Transient}}}$}).
This relaxation permits public, but \emph{transient}, variables
to temporarily assume different, secret values resulting from off-bounds array reads.
Later, we prove that when these assignments are retired, these values are necessarily public, and therefore equal.
%

%
Since \emph{transient} variable maps are computed from the pending
assignments in the reorder buffer
(Fig.~\ref{fig:full-transient-var-map}), we relax \ensuremath{\Blue{\Conid{L}}}-equivalence
in rule \srule{Transient-VarMap} similarly to rule
\srule{Asgn$_{\ensuremath{\Blue{\Conid{L}}\land\ensuremath{\Concrete}}}$}.


%

\begin{figure}[p]
\centering
\begin{subfigure}{.85\textwidth}
\begin{mathpar}
\inferrule[Nop]
{}
{\ensuremath{\Gamma\vdash\mathbf{nop}\approx_{\Blue{\Conid{L}}}\mathbf{nop}}}
\and
\inferrule[Fail]
{}
{\ensuremath{\Gamma\vdash\mathbf{fail}(\Varid{p})\approx_{\Blue{\Conid{L}}}\mathbf{fail}(\Varid{p})}}
\and
\inferrule[Guard]
{}
{\ensuremath{\Gamma\vdash\mathbf{guard}(\Varid{e}^{\Varid{b}},\Varid{cs},\Varid{p})\approx_{\Blue{\Conid{L}}}\mathbf{guard}(\Varid{e}^{\Varid{b}},\Varid{cs},\Varid{p})}}
\and
\inferrule[Asgn$_{\ensuremath{\Blue{\Conid{L}}\land\ensuremath{\Concrete}}}$]
{\ensuremath{\Varid{x}\;\in\;\Blue{\Conid{L}}} \\ \ensuremath{\Gamma(\Varid{x})\mathrel{=}\ensuremath{\Concrete}}}
{\ensuremath{\Gamma\vdash\Varid{x}\mathbin{:=}\Varid{e}\approx_{\Blue{\Conid{L}}}\Varid{x}\mathbin{:=}\Varid{e}}}
\and
\inferrule[Asgn$_{\ensuremath{\Red{\Conid{H}}\lor\ensuremath{\Transient}}}$]
{\ensuremath{\Varid{x}\;\not\in\;\Blue{\Conid{L}}}\ \ensuremath{\lor}\ \ensuremath{\Gamma(\Varid{x})\mathrel{=}\ensuremath{\Transient}}}
{\ensuremath{\Gamma\vdash\Varid{x}\mathbin{:=}\Varid{e}_{1}\approx_{\Blue{\Conid{L}}}\Varid{x}\mathbin{:=}\Varid{e}_{2}}}
\and
\inferrule[Protect]
{}
{\ensuremath{\Gamma\vdash\Varid{x}\mathbin{:=}\mathbf{protect}(\Varid{e}_{1})\approx_{\Blue{\Conid{L}}}\Varid{x}\mathbin{:=}\mathbf{protect}(\Varid{e}_{2})}}
\and
\inferrule[Load]
{}
{\ensuremath{\Gamma\vdash\Varid{x}\mathbin{:=}\mathbf{load}(\Varid{e})\approx_{\Blue{\Conid{L}}}\Varid{x}\mathbin{:=}\mathbf{load}(\Varid{e})}}
\and
\inferrule[Store$_{\ensuremath{\Red{\Conid{H}}}}$]
{}
{\ensuremath{\Gamma\vdash\mathbf{store}_{\Red{\Conid{H}}}(\Varid{e},\Varid{e}_{1})\approx_{\Blue{\Conid{L}}}\mathbf{store}_{\Red{\Conid{H}}}(\Varid{e},\Varid{e}_{2})}}
\and
\inferrule[Store$_{\ensuremath{\Blue{\Conid{L}}}}$]
{}
{\ensuremath{\Gamma\vdash\mathbf{store}_{\Blue{\Conid{L}}}(\Varid{e}_{1},\Varid{e}_{2})\approx_{\Blue{\Conid{L}}}\mathbf{store}_{\Blue{\Conid{L}}}(\Varid{e}_{1},\Varid{e}_{2})}}
\and
\end{mathpar}
\captionsetup{format=caption-with-line}
\caption{Instructions \ensuremath{\Gamma\vdash\Varid{i}_{1}\approx_{\Blue{\Conid{L}}}\Varid{i}_{2}}.}
\label{fig:equiv-instr}
\end{subfigure}%


\begin{subfigure}{.85\textwidth}
\begin{mathpar}
\inferrule[RB-Empty]
{}
{\ensuremath{\Gamma\vdash[\mskip1.5mu \mskip1.5mu]\approx_{\Blue{\Conid{L}}}[\mskip1.5mu \mskip1.5mu]}}
\and
\inferrule[RB-Cons]
{\ensuremath{\Gamma\vdash\Varid{i}_{1}\approx_{\Blue{\Conid{L}}}\Varid{i}_{2}} \\ \ensuremath{\Gamma\vdash\Varid{is}_{1}\approx_{\Blue{\Conid{L}}}\Varid{is}_{2}}}
{\ensuremath{\Gamma\vdash(\Varid{i}_{1}\mathbin{:}\Varid{is}_{1})\approx_{\Blue{\Conid{L}}}(\Varid{i}_{2}\mathbin{:}\Varid{is}_{2})}}
\end{mathpar}
\captionsetup{format=caption-with-line}
\caption{Reorder buffers \ensuremath{\Gamma\vdash\Varid{is}_{1}\approx_{\Blue{\Conid{L}}}\Varid{is}_{2}}.}
\label{fig:equiv-stacks}
\end{subfigure}

\begin{subfigure}{.85\textwidth}
\begin{mathpar}
\inferrule[VarMap]
{
 \ensuremath{\forall\;\Varid{x}\;\in\;\Blue{\Conid{L}}}\ .\ \ensuremath{\rho_{1}(\Varid{x})\mathrel{=}\rho_{2}(\Varid{x})}}
{\ensuremath{\rho_{1}\approx_{\Blue{\Conid{L}}}\rho_{2}}}
\and
\inferrule[Memory]
{
 \ensuremath{\forall\;\Varid{a}\;\in\;\Blue{\Conid{L}}}\ .\ \ensuremath{\forall\;\Varid{n}\;\in\;\{\mskip1.5mu \Varid{base}(\Varid{a}),\mathbin{...},\Varid{base}(\Varid{a})\mathbin{+}\mathit{length}(\Varid{a})\mathbin{-}\mathrm{1}\mskip1.5mu\}}\ \ensuremath{.}\ \ensuremath{\mu_{1}(\Varid{n})\mathrel{=}\mu_{2}(\Varid{n})}}
{\ensuremath{\mu_{1}\approx_{\Blue{\Conid{L}}}\mu_{2}}}
\end{mathpar}
\captionsetup{format=caption-with-line}
\caption{Variable maps (\ensuremath{\rho_{1}\approx_{\Blue{\Conid{L}}}\rho_{2}}) and memories (\ensuremath{\mu_{1}\approx_{\Blue{\Conid{L}}}\mu_{2}}).}
\label{fig:equiv-varmaps}
\end{subfigure}

\begin{subfigure}{.85\textwidth}
\begin{mathpar}
\inferrule[Conf]
{\ensuremath{\Gamma\vdash\Varid{is}_{1}\approx_{\Blue{\Conid{L}}}\Varid{is}_{2}} \\ \ensuremath{\Varid{cs}_{1}\mathrel{=}\Varid{cs}_{2}} \\ \ensuremath{\rho_{1}\approx_{\Blue{\Conid{L}}}\rho_{2}} \\ \ensuremath{\mu_{1}\approx_{\Blue{\Conid{L}}}\mu_{2}}}
{\ensuremath{\Gamma\vdash\langle\Varid{is}_{1},\Varid{cs}_{1},\mu_{1},\rho_{1}\rangle\approx_{\Blue{\Conid{L}}}\langle\Varid{is}_{2},\Varid{cs}_{2},\mu_{2},\rho_{2}\rangle}}
\end{mathpar}
\captionsetup{format=caption-with-line}
\caption{Configurations \ensuremath{\Gamma\vdash\Varid{c}_{1}\approx_{\Blue{\Conid{L}}}\Varid{c}_{2}}.}
\end{subfigure}

\begin{subfigure}{.85\textwidth}
\begin{mathpar}
\inferrule[Transient-VarMap]
{\ensuremath{\forall\;\Varid{x}\;\in\;\Blue{\Conid{L}}}\ \ensuremath{\land}\ \ensuremath{\Gamma(\Varid{x})\mathrel{=}\ensuremath{\Concrete}}\ .\ \ensuremath{\rho_{1}(\Varid{x})\mathrel{=}\rho_{2}(\Varid{x})}}
{\ensuremath{\Gamma\vdash\rho_{1}\approx_{\Blue{\Conid{L}}}\rho_{2}}}
\end{mathpar}
\captionsetup{format=caption-with-line}
\caption{Transient variable map \ensuremath{\Gamma\vdash\rho_{1}\approx_{\Blue{\Conid{L}}}\rho_{2}}.}
\label{fig:equiv-transient-varmap}
\end{subfigure}

\caption{\ensuremath{\Blue{\Conid{L}}}-equivalence.}
\label{fig:low-equivalence}
\end{figure}

\begin{lemma}[\ensuremath{\Blue{\Conid{L}}}-Equivalence for Transient Variable Maps]
If \ensuremath{\Gamma\vdash\Varid{is}_{1}\approx_{\Blue{\Conid{L}}}\Varid{is}_{2}} and \ensuremath{\Gamma\vdash\rho_{1}\approx_{\Blue{\Conid{L}}}\rho_{2}}, then \ensuremath{\Gamma\vdash\phi(\Varid{is}_{1},\rho_{1})\approx_{\Blue{\Conid{L}}}\phi(\Varid{is}_{2},\rho_{2})}.
\label{lemma:trans-var-maps}
\end{lemma}
\begin{proof}
  By induction on the \ensuremath{\Blue{\Conid{L}}}-equivalence judgment for reorder buffers
  (Fig.~\ref{fig:equiv-stacks}). The base case \srule{RB-Empty}
  follows by hypothesis (\ensuremath{\Gamma\vdash\rho_{1}\approx_{\Blue{\Conid{L}}}\rho_{2}}).
  In the inductive case \srule{RB-Cons}, we know that the next
  instructions in the buffers are equivalent, i.e., \ensuremath{\Gamma\vdash\Varid{i}_{1}\approx_{\Blue{\Conid{L}}}\Varid{i}_{2}} and the rest of the buffers are equivalent, i.e., \ensuremath{\Gamma\vdash\Varid{is}_{1}\approx_{\Blue{\Conid{L}}}\Varid{is}_{2}}.
  In order to apply the induction hypothesis, we need to show first
  that updating equivalent maps with equivalent instructions gives
  equivalent variable maps. We proceed by case analysis on the
  judgment for equivalent instruction (\ensuremath{\Gamma\vdash\Varid{i}_{1}\approx_{\Blue{\Conid{L}}}\Varid{i}_{2}} in
  Fig.~\ref{fig:equiv-instr}).
  By definition of the judgment, \ensuremath{\Varid{i}_{1}} and \ensuremath{\Varid{i}_{2}} are the same kind of
  instruction and cases \srule{Nop,Fail,Guard,Store} do not update the
  transient variable maps (they are handled by the last case of the
  function \ensuremath{\phi} in Fig.~\ref{fig:full-transient-var-map}), which
  remain equivalent (\ensuremath{\Gamma\vdash\rho_{1}\approx_{\Blue{\Conid{L}}}\rho_{2}}).
  In case \srule{Load}, both instructions are identical unresolved
  loads (\ensuremath{\Varid{i}_{1}\mathrel{=}\Varid{i}_{2}\mathrel{=}\Varid{x}\mathbin{:=}\mathbf{load}(\Varid{e})}), therefore variable \ensuremath{\Varid{x}} becomes
  undefined in both variable maps, which remain equivalent regardless
  of the sensitivity and type of \ensuremath{\Varid{x}}, i.e., \ensuremath{\Gamma\vdash\update{\rho_{1}}{\Varid{x}}{\bot}\approx_{\Blue{\Conid{L}}}\update{\rho_{2}}{\Varid{x}}{\bot}} from rule \srule{Transient-VarMap}.
  Similarly, in case \srule{Asgn$_{\ensuremath{\Blue{\Conid{L}}\land\ensuremath{\Concrete}}}$} the two instructions
  are identical assignments (\ensuremath{\Varid{i}_{1}\mathrel{=}\Varid{i}_{2}\mathrel{=}\Varid{x}\mathbin{:=}\Varid{e}}) and variable \ensuremath{\Varid{x}} is public (\ensuremath{\Varid{x}\;\in\;\Blue{\Conid{L}}})
  and stable (\ensuremath{\Gamma(\Varid{x})\mathrel{=}\ensuremath{\Concrete}}).
Since the assignments are identical, variable \ensuremath{\Varid{x}} gets updated in the same way in the
  respective variable maps, which remain equivalent whether the
  assignments are resolved (\ensuremath{\Varid{e}\mathrel{=}\Varid{v}} and \ensuremath{\Gamma\vdash\update{\rho_{1}}{\Varid{x}}{\Varid{v}}\approx_{\Blue{\Conid{L}}}\update{\rho_{2}}{\Varid{x}}{\Varid{v}}}) or not (\ensuremath{\Varid{e}\neq\Varid{v}} and \ensuremath{\Gamma\vdash\update{\rho_{1}}{\Varid{x}}{\bot}\approx_{\Blue{\Conid{L}}}\update{\rho_{2}}{\Varid{x}}{\bot}}). 
  In case \srule{Asgn$_{\ensuremath{\Red{\Conid{H}}\lor\ensuremath{\Transient}}}$}, the two instructions update
  variable \ensuremath{\Varid{x}} with possibly different expressions (\ensuremath{\Varid{i}_{1}\mathrel{=}\Varid{x}\mathbin{:=}\Varid{e}_{1}\neq\Varid{x}\mathbin{:=}\Varid{e}_{2}\mathrel{=}\Varid{i}_{2}}). This includes the cases where one instruction is
  resolved (e.g., \ensuremath{\Varid{e}_{1}\mathrel{=}\Varid{v}_{1}}) and the other is not (\ensuremath{\Varid{e}_{2}\neq\Varid{v}_{2}}), or
  both are resolved, but to different values (\ensuremath{\Varid{e}_{1}\mathrel{=}\Varid{v}_{1}\neq\Varid{v}_{2}\mathrel{=}\Varid{e}_{2}}).
  In these cases, variable \ensuremath{\Varid{x}} assumes different values in the
  resulting transient variable maps, but the maps remain nevertheless
  equivalent according to rule \srule{Transient-VarMap}, e.g., \ensuremath{\Gamma\vdash\update{\rho_{1}}{\Varid{x}}{\Varid{v}_{1}}\approx_{\Blue{\Conid{L}}}\update{\rho_{2}}{\Varid{x}}{\bot}}, because \ensuremath{\Varid{x}} is secret (\ensuremath{\Varid{x}\;\not\in\;\Blue{\Conid{L}}}) or transient (\ensuremath{\Gamma(\Varid{x})\mathrel{=}\ensuremath{\Transient}}).
  In case \srule{Protect}, both variables are protected, i.e., \ensuremath{\Gamma\vdash\Varid{x}\mathbin{:=}\mathbf{protect}(\Varid{e}_{1})\approx_{\Blue{\Conid{L}}}\Varid{x}\mathbin{:=}\mathbf{protect}(\Varid{e}_{2})}, and therefore, they are
  undefined in the transient variable maps (regardless of whether the
  expressions are resolved or not) and thus the maps remain related,
  i.e., \ensuremath{\Gamma\vdash\update{\rho_{1}}{\Varid{x}}{\bot}\approx_{\Blue{\Conid{L}}}\update{\rho_{2}}{\Varid{x}}{\bot}}.
Now that we have established that after processing the next
instructions in the buffers the resulting transient variable maps are
\ensuremath{\Blue{\Conid{L}}}-equivalent, we conclude the proof of the lemma by applying the
induction hypothesis.

\end{proof}

\begin{lemma}[Evaluations of Public Stable Expressions]
If \ensuremath{\Gamma\vdash\Varid{e}\mathbin{:}\ensuremath{\Concrete}}, \ensuremath{\vdash_{\textsc{ct}}\;\Varid{e}\mathbin{:}\Blue{\Conid{L}}}, and \ensuremath{\Gamma\vdash\rho_{1}\approx_{\Blue{\Conid{L}}}\rho_{2}}, then \ensuremath{\llbracket \Varid{e}\rrbracket^{\rho_{1}}\mathrel{=}\llbracket \Varid{e}\rrbracket^{\rho_{2}}}.
\label{lemma:evaluation-public-stable}
\end{lemma}
\begin{proof}
  By induction on the typing judgments, where the transient-flow and
  constant-time type systems ensure that the sub-expressions of \ensuremath{\Varid{e}}
  are also typed as \ensuremath{\Blue{\Conid{L}}} and \ensuremath{\ensuremath{\Concrete}}.
  (The length and the base of secret arrays are assumed to be public
  information like in~\cite{libsignalSP}).
  All the inductive cases follow directly by induction hypothesis. The
  base case for values is trivial and we derive equality in the base
  case for variables from \ensuremath{\Blue{\Conid{L}}}-equivalence of the transient variable
  maps (rule \srule{Transient-VarMap} in
  Fig.~\ref{fig:equiv-transient-varmap}).

\end{proof}

\begin{lemma}[Equal Guard Identifiers]
If \ensuremath{\Gamma\vdash\Varid{is}_{1}\approx_{\Blue{\Conid{L}}}\Varid{is}_{2}}, then \ensuremath{\Moon{\Varid{is}_{1}}\mathrel{=}\Moon{\Varid{is}_{2}}}.
\label{lemma:equal-guard-identifiers}
\end{lemma}
\begin{proof}
  By induction on the \ensuremath{\Blue{\Conid{L}}}-equivalence judgment for reorder buffers
  (Fig.~\ref{fig:equiv-stacks}).
  The only interesting cases are the inductive cases, when the
  \ensuremath{\Blue{\Conid{L}}}-instructions at the beginning of the reorder buffer are guard or fail instructions.
  By inspection of the definitions from Figure~\ref{fig:equiv-instr},
  we see that when this occurs, both instructions are guards or fail instructions and share
  the same prediction identifier (i.e., rule \srule{Guard} and \srule{Fail}).
\end{proof}

In the following we write \ensuremath{\Gamma\;\vdash_{\textsc{sct}}\;\Conid{C}} to indicate that
configuration \ensuremath{\Conid{C}} is well-typed with respect to \emph{both} the
constant-time type system (Fig.~\ref{fig:ct-ts}), i.e., \ensuremath{\Gamma\;\vdash_{\textsc{ct}}\;\Conid{C}}, and the
transient-flow type system (Fig.~\ref{fig:full-transient-flow-ts}), i.e., \ensuremath{\Gamma\vdash\Conid{C}}, under the respective
typing contexts.

\begin{lemma}[\ensuremath{\Blue{\Conid{L}}}-equivalence Preservation]
If \ensuremath{\Gamma\;\vdash_{\textsc{sct}}\;\Conid{C}_{1}}, \ensuremath{\Gamma\;\vdash_{\textsc{sct}}\;\Conid{C}_{2}}, \ensuremath{\Gamma\vdash\Conid{C}_{1}\approx_{\Blue{\Conid{L}}}\Conid{C}_{2}}, \ensuremath{\Conid{C}_{1}\;\stepT{\Varid{d}}{\Varid{o}_{1}}\;\Conid{C}_{1}'}, and \ensuremath{\Conid{C}_{2}\;\stepT{\Varid{d}}{\Varid{o}_{2}}\;\Conid{C}_{2}'}, then \ensuremath{\Conid{C}_{2}\approx_{\Blue{\Conid{L}}}\Conid{C}_{2}'} and \ensuremath{\Varid{o}_{1}\mathrel{=}\Varid{o}_{2}}.
\label{lemma:preservation-small-step}
\end{lemma}
\begin{proof}
By case analysis on the reduction steps.
From \ensuremath{\Gamma\vdash\Conid{C}_{1}\approx_{\Blue{\Conid{L}}}\Conid{C}_{2}}, we know that the respective components of the configurations are also \ensuremath{\Blue{\Conid{L}}}-equivalent.
In particular, \ensuremath{\Gamma\vdash\Varid{is}_{1}\approx_{\Blue{\Conid{L}}}\Varid{is}_{2}} and \ensuremath{\Varid{cs}_{1}\mathrel{=}\Varid{cs}_{2}}, 
thus the reorder buffers and the command stacks have the same
structure, length, and are point-wise related.
Then, since the directive \ensuremath{\Varid{d}} is the same in both reduction,
equivalent commands (\ensuremath{\Varid{c}_{1}\mathrel{=}\Varid{c}_{2}}) 
are fetched from the command stacks and \ensuremath{\Blue{\Conid{L}}}-equivalent instructions
(\ensuremath{\Gamma\vdash\Varid{i}_{1}\approx_{\Blue{\Conid{L}}}\Varid{i}_{2}}) are executed and retired.
\begin{CompactItemize}
\item Fetch Stage (\ensuremath{\Varid{d}\mathrel{=}\mathbf{fetch}} or \ensuremath{\Varid{d}\mathrel{=}\mathbf{fetch}\;\Varid{b}}). The rules from this
  stage (Fig.~\ref{fig:full-fetch-stage}) always generate the empty
  event, i.e., \ensuremath{\Varid{o}_{1}\mathrel{=}\epsilon\mathrel{=}\Varid{o}_{2}}, and affect only the reorder buffers
  and the command stacks.
Thus, to prove \ensuremath{\Gamma\vdash\Conid{C}_{1}'\approx_{\Blue{\Conid{L}}}\Conid{C}_{2}'}, we only need to show that the
  resulting buffers and stacks are related.
First, we observe that the rules only pop and push commands from the
  commands stack and append instructions at the end of the reorder
  buffer.
  Since the definition of \ensuremath{\Blue{\Conid{L}}}-equivalence for buffers is
  \emph{structural} (Fig.~\ref{fig:equiv-stacks}) and for stacks is
  syntactic equivalence, these operations preserve (\ensuremath{\Blue{\Conid{L}}}-)equivalence
  as long as they are applied to (\ensuremath{\Blue{\Conid{L}}}-)equivalent arguments.
  This is exactly the case, for control-flow commands, which are
  popped, flattened into simpler, equal commands, and pushed back on
  the stack, e.g., in rules \srule{Fetch-Seq} and \srule{Fetch-While}.
%
%
Similarly, the other commands are translated
directly into related instructions (Fig.~\ref{fig:equiv-instr}),
e.g., by rule \srule{Nop} for \srule{Fetch-Skip}, \srule{Fail} for \srule{Fetch-Fail}, and \srule{Guard} for
\srule{Fetch-If-True} and \srule{Fetch-If-False}, where the prediction
\ensuremath{\Varid{b}\mathrel{=}\mathbf{true}} or \ensuremath{\Varid{b}\mathrel{=}\mathbf{false}} is determined from the same
attacker-provided directive \ensuremath{\mathbf{fetch}\;\Varid{b}}.\footnote{We assume that the
  generation of the fresh prediction identifier (\ensuremath{\fresh{\Varid{p}}}) is
  deterministic and secret independent. For example, the configuration
  could contain a counter \ensuremath{\Varid{p}} containing the next fresh identifier,
  which gets incremented each time a conditional is fetched.}
When fetching assignments, i.e., rule \srule{Fetch-Asgn}), the
corresponding assignments instructions appended at the end of the
reorder buffers are related, either by rule
\srule{Asgn$_{\ensuremath{\Blue{\Conid{L}}\land\ensuremath{\Concrete}}}$} or \srule{Asng$_{\ensuremath{\Red{\Conid{H}}\lor\ensuremath{\Transient}}}$} depending
on the sensitivity and type of the variable.
For rule \srule{Fetch-Ptr-Store}, we rely on the label annotation that
decorates the command (i.e., label \ensuremath{\ell} for \ensuremath{\mathbin{*}_{\ell}\;\Varid{e}_{1}\mathrel{=}\Varid{e}_{2}}) to relate the corresponding instructions via rule \srule{Store$_{\ensuremath{\Blue{\Conid{L}}}}$} if \ensuremath{\ell\mathrel{=}\Blue{\Conid{L}}}, or \srule{Store$_{\ensuremath{\Red{\Conid{H}}}}$}), otherwise.
For array reads (\srule{Fetch-Array-Load}) and writes (\srule{Fetch-Array-Store}), we observe that the \emph{same} bounds-checking code is generated and pushed on the commands stack.
Rules
\srule{Fetch-Protect-Ptr,Fetch-Protect-Array,Fetch-Protect-Expr,Fetch-Protect-SLH}
fetch \emph{identical} \ensuremath{\mathbf{protect}} commands and decompose them into
equal commands (\srule{Fetch-Protect-Ptr,Fetch-Protect-Array,
  Fetch-Protect-SLH}) or equivalent \ensuremath{\mathbf{protect}} instructions (\srule{Fetch-Protect-Expr}).\footnote{ We prove security for both the
    hardware- and SLH-based implementation of \ensuremath{\mathbf{protect}}, but we
    assume that the same implementation is used in both executions.}


\item Execute Stage (\ensuremath{\Varid{d}\mathrel{=}\mathbf{exec}\;\Varid{n}}).  Instruction \srule{Execute}
  from Figure~\ref{fig:full-execute-stage}, selects the \ensuremath{\Varid{n}}-th
  instruction from the reorder buffer, computes the transient variable
  map, and relies on a separate reduction relation to generate an
  observation and compute the resulting reorder buffer and command
  stack.
Since the initial buffers are \emph{structurally} related, i.e., \ensuremath{\Gamma\vdash\Varid{is}_{1}\approx_{\Blue{\Conid{L}}}\Varid{is}_{2}} and the directive \ensuremath{\mathbf{exec}\;\Varid{n}} is the same in both reductions, then the buffers are split in related components, i.e., prefixes \ensuremath{\Gamma\vdash\Varid{is}_{1}'\approx_{\Blue{\Conid{L}}}\Varid{is}_{2}'} such that \ensuremath{{\vert}\Varid{is}_{1}'{\vert}\mathrel{=}{\vert}\Varid{is}_{2}'{\vert}\mathrel{=}\Varid{n}\mathbin{-}\mathrm{1}}, suffixes \ensuremath{\Gamma\vdash\Varid{is}_{1}''\approx_{\Blue{\Conid{L}}}\Varid{is}_{2}''}, and $n$-th instructions \ensuremath{\Gamma\vdash\Varid{i}_{1}\approx_{\Blue{\Conid{L}}}\Varid{i}_{2}}.
Therefore, their transient variable maps are related, i.e., \ensuremath{\Gamma\vdash\phi(\Varid{is}_{1}',\rho_{1})\approx_{\Blue{\Conid{L}}}\phi(\Varid{is}_{2}',\rho_{2})} by
Lemma~\ref{lemma:trans-var-maps} (notice that \ensuremath{\rho_{1}\approx_{\Blue{\Conid{L}}}\rho_{2}} implies
\ensuremath{\Gamma\vdash\rho_{1}\approx_{\Blue{\Conid{L}}}\rho_{2}}) and we conclude by applying
Lemma~\ref{lemma:preservation-exec}.

\item Retire Stage (\ensuremath{\Varid{d}\mathrel{=}\mathbf{retire}}). The rules from this stage
  (Fig.~\ref{fig:full-retire-stage}) remove the resolved instruction
  at the beginning of the reorder buffer, and update the variable map
  and the memory store accordingly.
  Therefore, to prove \ensuremath{\Gamma\vdash\Conid{C}_{1}'\approx_{\Blue{\Conid{L}}}\Conid{C}_{2}'}, we only need to show that
  the resulting reorder buffers, variable maps, and memory stores are
  \ensuremath{\Blue{\Conid{L}}}-equivalent.
  Since the reorder buffers are initially related, we have \ensuremath{\Gamma\vdash(\Varid{i}_{1}\mathbin{:}\Varid{is}_{1})\approx_{\Blue{\Conid{L}}}(\Varid{i}_{2}\mathbin{:}\Varid{is}_{2})}, and from rule \srule{RB-Cons} we learn
  that the instructions retired are related, i.e., \ensuremath{\Gamma\vdash\Varid{i}_{1}\approx_{\Blue{\Conid{L}}}\Varid{i}_{2}}, and the tails of the buffers remain related in the resulting
  configuration, i.e., \ensuremath{\Gamma\vdash\Varid{is}_{1}\approx_{\Blue{\Conid{L}}}\Varid{is}_{2}}.
  Then, we prove that the architectural state (variable maps and
  memory stores) remain related when updated by the related
  instructions.
Since the instructions are related, the two configurations perform the same reduction step.
Cases \srule{Retire-Nop} and \srule{Retire-Fail} are trivial.
We only point out that in either case the two reductions produce the
same event (i.e., \ensuremath{\Varid{o}_{1}\mathrel{=}\epsilon\mathrel{=}\Varid{o}_{2}} and \ensuremath{\Varid{o}_{1}\mathrel{=}\mathbf{fail}(\Varid{p})\mathrel{=}\Varid{o}_{2}}) and update (or
empty) the reorder buffers and the command stack in the same way.
For \srule{Retire-Asgn}, if the assignments involve a \emph{public}
variable (e.g., \ensuremath{\Varid{x}\;\in\;\Blue{\Conid{L}}}) and the variable is stable (i.e., \ensuremath{\Gamma(\Varid{x})\mathrel{=}\ensuremath{\Concrete}}), then the resolved values are the same, i.e., \ensuremath{\Varid{i}_{1}\mathrel{=}(\Varid{x}\mathbin{:=}\Varid{v})\mathrel{=}\Varid{i}_{2}} from rule \srule{Asgn$_{\ensuremath{\Blue{\Conid{L}}\land\ensuremath{\Concrete}}}$}, and the resulting
variable maps are \ensuremath{\Blue{\Conid{L}}}-equivalent, i.e., \ensuremath{\update{\rho_{1}}{\Varid{x}}{\Varid{v}}\approx_{\Blue{\Conid{L}}}\update{\rho_{2}}{\Varid{x}}{\Varid{v}}} from rule \srule{VarMap} in Fig.~\ref{fig:equiv-varmaps}.
If the variable is \emph{secret} (\ensuremath{\Varid{x}\;\not\in\;\Blue{\Conid{L}}}), then the assignments
may have been resolved to different values, i.e., \ensuremath{\Varid{i}_{1}\mathrel{=}(\Varid{x}\mathbin{:=}\Varid{v}_{1})\neq(\Varid{x}\mathbin{:=}\Varid{v}_{2})\mathrel{=}\Varid{i}_{2}} from rule \srule{Asgn$_{\ensuremath{\Red{\Conid{H}}\lor\ensuremath{\Transient}}}$}) and the
variable maps remain \ensuremath{\Blue{\Conid{L}}}-equivalent, i.e., \ensuremath{\update{\rho_{1}}{\Varid{x}}{\Varid{v}_{1}}\approx_{\Blue{\Conid{L}}}\update{\rho_{2}}{\Varid{x}}{\Varid{v}_{2}}}.
%
If the variable is public (\ensuremath{\Varid{x}\;\in\;\Blue{\Conid{L}}}), but typed \emph{transient}
according to our transient-flow type system (i.e., \ensuremath{\Gamma(\Varid{x})\mathrel{=}\ensuremath{\Transient}}), we show that the assignments necessarily must have been resolved
to the same value.
Intuitively, public variables may assume different, secret values
(e.g., \ensuremath{\Varid{v}_{1}\neq\Varid{v}_{2}}) only \emph{transiently} via a data-flow
dependency to a \emph{previous} command that reads a public array
\emph{off-bounds}.
Therefore, the array read command must have been fetched first via rule \srule{Fetch-Array-Load} and then secret data must have been transiently loaded via rule \srule{Exec-Load}.
Since rule \srule{Fetch-Array-Load} automatically prepends bounds-checking instructions to the load, its guard instruction   must necessarily have been fetched, executed, and \emph{retired} before the current assignment can reach the beginning of the reorder buffer and be retired.
From this, we deduce that the bounds-check guard must have succeeded,
the subsequent load instruction has (correctly) read \emph{public}
data from memory, and therefore (\srule{Memory} in
Fig.~\ref{fig:equiv-varmaps}) the current values are equal, i.e., \ensuremath{\Varid{v}_{1}\mathrel{=}\Varid{v}_{2}} and the variable maps remain \ensuremath{\Blue{\Conid{L}}}-equivalent, i.e., \ensuremath{\update{\rho_{1}}{\Varid{x}}{\Varid{v}_{1}}\approx_{\Blue{\Conid{L}}}\update{\rho_{2}}{\Varid{x}}{\Varid{v}_{1}}}.

In case \srule{Retire-Store}, two related store instructions
are retired, i.e., \ensuremath{\Gamma\vdash\mathbf{store}_{\ell}(\Varid{n},\Varid{v})\approx_{\Blue{\Conid{L}}}} \ensuremath{\mathbf{store}_{\ell}(\Varid{n'},\Varid{v'})}.
If the label that decorate the instructions is public, (\ensuremath{\ell\mathrel{=}\Blue{\Conid{L}}}), then the same public memory address is updated with the same values,
i.e., \ensuremath{\Varid{n}\mathrel{=}\Varid{n'}} and \ensuremath{\Varid{v}\mathrel{=}\Varid{v'}} from rule \srule{Store$_{\ensuremath{\Blue{\Conid{L}}}}$}, and the memory stores remain related, i.e., \ensuremath{\update{\mu_{1}}{\Varid{n}}{\Varid{v}}\approx_{\Blue{\Conid{L}}}\update{\mu_{2}}{\Varid{n}}{\Varid{v}}} from \srule{Memory}.
Otherwise (\ensuremath{\ell\mathrel{=}\Red{\Conid{H}}}), the instruction update the same \emph{secret}
location with possibly different values, i.e., \ensuremath{\Gamma\vdash\mathbf{store}_{\Red{\Conid{H}}}(\Varid{n},\Varid{v}_{1})\approx_{\Blue{\Conid{L}}}\mathbf{store}_{\Red{\Conid{H}}}(\Varid{n},\Varid{v}_{2})} from rule \srule{Store$_{\ensuremath{\Red{\Conid{H}}}}$}, and the stores
remain related, \ensuremath{\update{\mu_{1}}{\Varid{n}}{\Varid{v}_{1}}\approx_{\Blue{\Conid{L}}}\update{\mu_{2}}{\Varid{n}}{\Varid{v}_{2}}}.
\end{CompactItemize}

\end{proof}

\begin{lemma}[\ensuremath{\Blue{\Conid{L}}}-equivalence Preservation (Execute Stage)]
  Let \ensuremath{\Conid{C}_{1}\mathrel{=}\langle\Varid{is}_{1},\Varid{i}_{1},\Varid{is}_{1}',\Varid{cs}_{1}\rangle} and \ensuremath{\Conid{C}_{2}\mathrel{=}\langle\Varid{is}_{2},\Varid{i}_{2},\Varid{is}_{2}',\Varid{cs}_{2}\rangle},
  such that \ensuremath{\Gamma\;\vdash_{\textsc{sct}}\;\Conid{C}_{1}} and \ensuremath{\Gamma\;\vdash_{\textsc{sct}}\;\Conid{C}_{2}}. If \ensuremath{\Gamma\vdash\Conid{C}_{1}\approx_{\Blue{\Conid{L}}}\Conid{C}_{2}},
  \ensuremath{\Conid{C}_{1}\;\xrsquigarrow{(\mu_{1},\rho_{1},\Varid{o}_{1})}\;\Conid{C}_{1}'}, \ensuremath{\Conid{C}_{2}\;\xrsquigarrow{(\mu_{2},\rho_{2},\Varid{o}_{2})}\;\Conid{C}_{2}'},
  \ensuremath{\mu_{1}\approx_{\Blue{\Conid{L}}}\mu_{2}}, \ensuremath{\Gamma\vdash\rho_{1}\approx_{\Blue{\Conid{L}}}\rho_{2}}, then \ensuremath{\Varid{o}_{1}\mathrel{=}\Varid{o}_{2}} and \ensuremath{\Gamma\vdash\Conid{C}_{1}'\approx_{\Blue{\Conid{L}}}\Conid{C}_{2}'}.\footnote{Also for these configurations,
    \ensuremath{\Blue{\Conid{L}}}-equivalence and typing is defined component-wise.
  }
\label{lemma:preservation-exec}
\end{lemma}
\begin{proof}
  By case analysis on the \ensuremath{\Blue{\Conid{L}}}-equivalence judgment \ensuremath{\Gamma\vdash\Varid{i}_{1}\approx_{\Blue{\Conid{L}}}\Varid{i}_{2}} and the reduction steps.
  Since the judgment relates only instructions of the same kind, the
  two configurations perform the same reduction step.
\begin{CompactItemize}
\item \srule{Asgn$_{\ensuremath{\Blue{\Conid{L}}\land\ensuremath{\Concrete}}}$} and \srule{Exec-Asgn}.  Let the
  \ensuremath{\Blue{\Conid{L}}}-equivalent assignment instructions executed be \ensuremath{\Gamma\vdash\Varid{x}\mathbin{:=}\Varid{e}_{1}\approx_{\Blue{\Conid{L}}}\Varid{x}\mathbin{:=}\Varid{e}_{2}}, where \ensuremath{\Varid{x}} is public (\ensuremath{\Varid{x}\;\in\;\Blue{\Conid{L}}}) \emph{and} stable
  (\ensuremath{\Gamma(\Varid{x})\mathrel{=}\ensuremath{\Concrete}}).
  From \ensuremath{\Blue{\Conid{L}}}-equivalence, we know that the expressions are identical
  \ensuremath{\Varid{e}_{1}\mathrel{=}\Varid{e}_{2}} and public and stable from typing, therefore \ensuremath{\llbracket \Varid{e}_{1}\rrbracket^{\rho_{1}}\mathrel{=}\llbracket \Varid{e}_{2}\rrbracket^{\rho_{2}}}, by Lemma~\ref{lemma:evaluation-public-stable}
  and the resolved instructions reinserted in the reorder buffers are
  related, i.e., \ensuremath{\Gamma\vdash\Varid{x}\mathbin{:=}\llbracket \Varid{e}\rrbracket^{\rho_{1}}\approx_{\Blue{\Conid{L}}}\Varid{x}\mathbin{:=}\llbracket \Varid{e}\rrbracket^{\rho_{2}}} by
  rule \srule{Asgn$_{\ensuremath{\Blue{\Conid{L}}\land\ensuremath{\Concrete}}}$}.
  %
%
\item \srule{Asgn$_{\ensuremath{\Red{\Conid{H}}\lor\ensuremath{\Transient}}}$} and \srule{Exec-Asgn}. Similar to
  the previous case, but variable \ensuremath{\Varid{x}} is secret (\ensuremath{\Varid{x}\;\not\in\;\Blue{\Conid{L}}})
  \emph{or} transient (\ensuremath{\Gamma(\Varid{x})\mathrel{=}\ensuremath{\Transient}}).
  In this case, the expressions are different and may evaluate to
  different values, but the resolved instructions remain related,
  i.e., \ensuremath{\Gamma\vdash\Varid{x}\mathbin{:=}\llbracket \Varid{e}_{1}\rrbracket^{\rho_{1}}\approx_{\Blue{\Conid{L}}}\Varid{x}\mathbin{:=}\llbracket \Varid{e}_{2}\rrbracket^{\rho_{2}}} by rule
  \srule{Asgn$_{\ensuremath{\Red{\Conid{H}}\lor\ensuremath{\Transient}}}$}.
\item \srule{Guard} and \srule{Exec-Branch-Ok}.  Let the
  \ensuremath{\Blue{\Conid{L}}}-equivalent guard instructions be \ensuremath{\Gamma\vdash\mathbf{guard}(\Varid{e}_{1}^{\Varid{b}_{1}},\Varid{cs}_{1},\Varid{p}_{1})} \ensuremath{\approx_{\Blue{\Conid{L}}}\mathbf{guard}(\Varid{e}_{2}^{\Varid{b}_{2}},\Varid{cs}_{2},\Varid{p}_{2})}.
  From \ensuremath{\Blue{\Conid{L}}}-equivalence, we know that the guard instructions are
  identical (i.e., same condition expression \ensuremath{\Varid{e}_{1}\mathrel{=}\Varid{e}_{2}}, predicted
  outcome \ensuremath{\Varid{b}_{1}\mathrel{=}\Varid{b}_{2}}, roll-back stack \ensuremath{\Varid{cs}_{1}\mathrel{=}\Varid{cs}_{2}}, and prediction
  identifier \ensuremath{\Varid{p}_{1}\mathrel{=}\Varid{p}_{2}}) and from typing we know that the guard
  expression is public, i.e., \ensuremath{\Gamma\;\vdash_{\textsc{ct}}\;\Varid{e}_{\Varid{i}}\mathbin{:}\Blue{\Conid{L}}} and
  stable \ensuremath{\Gamma\vdash\Varid{e}_{\Varid{i}}\mathbin{:}\ensuremath{\Concrete}} from rule \srule{Guard} in
  Figure~\ref{fig:ct-ts-instr} and Fig.~\ref{fig:full-ts-instr},
  respectively.
  Then, the prediction is correct in both reductions, i.e., \ensuremath{\llbracket \Varid{e}_{1}\rrbracket^{\rho_{1}}\mathrel{=}\Varid{b}_{1}\mathrel{=}\Varid{b}_{2}\mathrel{=}\llbracket \Varid{e}_{2}\rrbracket^{\rho_{2}}} by
  Lemma~\ref{lemma:evaluation-public-stable}, and both guard
  instructions are resolved to \ensuremath{\Gamma\vdash\mathbf{nop}\approx_{\Blue{\Conid{L}}}\mathbf{nop}} and the resulting
  buffers remain related.
\item \srule{Guard} and \srule{Exec-Branch-Mispredict}. Similar to the
  previous case. For the same reasons, both predictions are wrong,
  i.e., \ensuremath{\llbracket \Varid{e}_{1}\rrbracket^{\rho_{1}}\mathrel{=}\Varid{b'}\mathrel{=}\llbracket \Varid{e}_{2}\rrbracket^{\rho_{2}}} and \ensuremath{\Varid{b'}\neq\Varid{b}}, and the
  rules generate the same observations \ensuremath{\mathbf{rollback}(\Varid{p})}, flush the rest of
  the reorder buffer (the guard instruction and the suffix) and
  restore the same rollback command stack.

\item \srule{Load} and \srule{Exec-Load}.  Let the \ensuremath{\Blue{\Conid{L}}}-equivalent
  prefixes of the reorder buffers be \ensuremath{\Gamma\vdash\Varid{is}_{1}\approx_{\Blue{\Conid{L}}}\Varid{is}_{2}}.
  From \ensuremath{\Blue{\Conid{L}}}-equivalence, no store is pending in either buffers (one
  configuration steps if and only if the other steps too), and the
  same guards are pending in each buffer, i.e., \ensuremath{\Moon{\Varid{is}_{1}}\mathrel{=}\Varid{ps}\mathrel{=}\Moon{\Varid{is}_{2}}} by Lemma~\ref{lemma:equal-guard-identifiers}.
  Let the \ensuremath{\Blue{\Conid{L}}}-equivalent load instructions executed in the steps be
  \ensuremath{\Gamma\vdash\Varid{x}\mathbin{:=}\mathbf{load}(\Varid{e}_{1})\approx_{\Blue{\Conid{L}}}\Varid{x}\mathbin{:=}\mathbf{load}(\Varid{e}_{2})}.
  From \ensuremath{\Blue{\Conid{L}}}-equivalence, both instructions the address expressions are
  identical, i.e., \ensuremath{\Varid{e}_{1}\mathrel{=}\Varid{e}_{2}}, and public and stable (from typing), and
  therefore evaluate to the same address \ensuremath{\llbracket \Varid{e}_{1}\rrbracket^{\rho_{1}}\mathrel{=}\Varid{n}\mathrel{=}\llbracket \Varid{e}_{2}\rrbracket^{\rho_{2}}} by Lemma~\ref{lemma:evaluation-public-stable}.
  As a result, both reductions generate the same observation \ensuremath{\Varid{o}_{1}\mathrel{=}\mathbf{read}(\Varid{n},\Varid{ps})\mathrel{=}\Varid{o}_{2}}.
  Lastly, the resolved instructions are \ensuremath{\Blue{\Conid{L}}}-equivalent, i.e., \ensuremath{\Gamma\vdash\Varid{x}\mathbin{:=}\mu_{1}(\Varid{n})\approx_{\Blue{\Conid{L}}}\Varid{x}\mathbin{:=}\mu_{2}(\Varid{n})} by rule
  \srule{Asgn$_{\ensuremath{\Red{\Conid{H}}\lor\ensuremath{\Transient}}}$}, because the variable \ensuremath{\Varid{x}} is
  \emph{transient} (\ensuremath{\Gamma(\Varid{x})\mathrel{=}\ensuremath{\Transient}}) by typing
  (Fig.~\ref{fig:full-ts-instr}).

\item \srule{Store$_{\ensuremath{\ell}}$} and \srule{Exec-Store}.  Let the two
  \ensuremath{\Blue{\Conid{L}}}-equivalent store instructions be \ensuremath{\mathbf{store}_{\ell}(\Varid{e}_{1},\Varid{e}_{1}')\approx_{\Blue{\Conid{L}}}\mathbf{store}_{\ell}(\Varid{e}_{2},\Varid{e}_{2}')}.
  First, notice that the expressions that compute the address of the
  stores are identical in both instructions from rules
  \srule{Store$_{\ensuremath{\ell}}$}.
  Furthermore, we know that these expressions are public and stable
  from typing, i.e., \ensuremath{\Gamma\;\vdash_{\textsc{ct}}\;\Varid{e}_{\Varid{i}}\mathbin{:}\Blue{\Conid{L}}} and \ensuremath{\Gamma\vdash\Varid{e}_{\Varid{i}}\mathbin{:}\ensuremath{\Concrete}} from rule \srule{Store} in
  Fig.~\ref{fig:full-ts-instr} and \ref{fig:ct-ts-instr}, respectively, and
  therefore the expression evaluates to the same address by
  Lemma~\ref{lemma:evaluation-public-stable}, i.e., \ensuremath{\llbracket \Varid{e}_{1}\rrbracket^{\rho_{1}}\mathrel{=}\Varid{n}\mathrel{=}\llbracket \Varid{e}_{2}\rrbracket^{\rho_{2}}}.
  We then proceed by further case distinction on the security label
  \ensuremath{\ell} that decorates the store instructions.
  If \ensuremath{\ell\mathrel{=}\Blue{\Conid{L}}}, then \ensuremath{\Varid{e}_{1}'\mathrel{=}\Varid{e}_{2}'} from \srule{Store$_{\ensuremath{\Blue{\Conid{L}}}}$}, and thus
  evaluate to the same value, i.e., \ensuremath{\llbracket \Varid{e}_{1}'\rrbracket^{\rho_{1}}\mathrel{=}\Varid{v}\mathrel{=}\llbracket \Varid{e}_{2}'\rrbracket^{\rho_{2}}} by Lemma~\ref{lemma:evaluation-public-stable}, and the
  resolved instructions in the resulting buffers remain related, i.e.,
  \ensuremath{\Gamma\vdash\mathbf{store}_{\Blue{\Conid{L}}}(\Varid{n},\Varid{v})\approx_{\Blue{\Conid{L}}}\mathbf{store}_{\Blue{\Conid{L}}}(\Varid{n},\Varid{v})} by rule
  \srule{Store$_{\ensuremath{\Blue{\Conid{L}}}}$}.
  Otherwise \ensuremath{\ell\mathrel{=}\Red{\Conid{H}}}, \ensuremath{\Varid{e}_{1}'\neq\Varid{e}_{2}'} and the resolved instructions
  remain related because the store writes a secret address, i.e.,
  \ensuremath{\Gamma\vdash\mathbf{store}_{\Red{\Conid{H}}}(\Varid{n},\llbracket \Varid{e}_{1}'\rrbracket^{\rho_{1}})\approx_{\Blue{\Conid{L}}}\mathbf{store}_{\Red{\Conid{H}}}(\Varid{n},\llbracket \Varid{e}_{2}'\rrbracket^{\rho_{2}})} by
  rule \srule{Store$_{\ensuremath{\Red{\Conid{H}}}}$}.

\item \srule{Protect} and \srule{Exec-Protect$_1$}.
Let the two
  \ensuremath{\Blue{\Conid{L}}}-equivalent \ensuremath{\mathbf{protect}} instructions be \ensuremath{\Gamma\vdash\Varid{x}\mathbin{:=}\mathbf{protect}(\Varid{e}_{1})\approx_{\Blue{\Conid{L}}}\mathbf{protect}(\Varid{e}_{2})}.
By rule \srule{Exec-Protect$_1$}, we know that \ensuremath{\Varid{v}_{1}\mathrel{=}\llbracket \Varid{e}_{1}\rrbracket^{\rho_{1}}} and \ensuremath{\Varid{v}_{2}\mathrel{=}\llbracket \Varid{e}_{2}\rrbracket^{\rho_{2}}} and thus the resulting instructions \ensuremath{\Gamma\vdash\Varid{x}\mathbin{:=}\mathbf{protect}(\Varid{v}_{1})\approx_{\Blue{\Conid{L}}}\Varid{x}\mathbin{:=}\mathbf{protect}(\Varid{v}_{2})} are \ensuremath{\Blue{\Conid{L}}}-equivalent by rule \srule{Protect}.

\item \srule{Protect} and \srule{Exec-Protect$_2$}.
Let the two
  \ensuremath{\Blue{\Conid{L}}}-equivalent \ensuremath{\mathbf{protect}} instructions be \ensuremath{\Gamma\vdash\Varid{x}\mathbin{:=}\mathbf{protect}(\Varid{v}_{1})\approx_{\Blue{\Conid{L}}}\Varid{x}\mathbin{:=}\mathbf{protect}(\Varid{v}_{2})} and we must show that \ensuremath{\Gamma\vdash\Varid{x}\mathbin{:=}\Varid{v}_{1}\approx_{\Blue{\Conid{L}}}\Varid{x}\mathbin{:=}\Varid{v}_{2}}.
  If variable \ensuremath{\Varid{x}} is secret, i.e., \ensuremath{\Varid{x}\;\not\in\;\Blue{\Conid{L}}}, or transient, i.e., \ensuremath{\Gamma(\Varid{x})\mathrel{=}\ensuremath{\Transient}}, then the two resolved instructions are related by
  rule \srule{Asgn$_{\ensuremath{\Red{\Conid{H}}\lor\ensuremath{\Transient}}}$}.
  If variable \ensuremath{\Varid{x}} is public, i.e., \ensuremath{\Varid{x}\;\in\;\Blue{\Conid{L}}}, and stable, i.e.,
  \ensuremath{\Gamma(\Varid{x})\mathrel{=}\ensuremath{\Concrete}}, then we must show that \ensuremath{\Varid{v}_{1}\mathrel{=}\Varid{v}_{2}} to apply rule \srule{Asgn$_{\ensuremath{\Blue{\Conid{L}}\land\ensuremath{\Concrete}}}$}.
  Intuitively, public variables may assume different, secret values
  (e.g., \ensuremath{\Varid{v}_{1}\neq\Varid{v}_{2}}), only \emph{transiently}, due to a data-flow
  dependency to a previous command that reads a public array
  \emph{off-bounds}.
Array reads are automatically \emph{guarded} by bounds check instructions when fetched (\srule{Fetch-Array-Load}) and by rule \srule{Exec-Protect$_2$}, we know that no
guards are pending in either prefix of the reorder buffer, i.e.,  \ensuremath{\mathbf{guard}(\anonymous ,\anonymous ,\anonymous )\;\not\in\;\Varid{is}_{1}} and \ensuremath{\mathbf{guard}(\anonymous ,\anonymous ,\anonymous )\;\not\in\;\Varid{is}_{2}}.
From this, we conclude that all previous bounds-check have been
successfully resolved and retired, and therefore the values \ensuremath{\Varid{v}_{1}} and
\ensuremath{\Varid{v}_{2}} are truly public and stable, and thus equal, i.e., \ensuremath{\Varid{v}_{1}\mathrel{=}\Varid{v}_{2}}, and the corresponding, resolved assignments are likewise related,
i.e., \ensuremath{\Gamma\vdash\Varid{x}\mathbin{:=}\Varid{v}_{1}\approx_{\Blue{\Conid{L}}}\Varid{x}\mathbin{:=}\Varid{v}_{1}}.
\end{CompactItemize}

\end{proof}

\begin{lemma}[Big-step \ensuremath{\Blue{\Conid{L}}}-equivalence Preservation]
If \ensuremath{\Gamma\;\vdash_{\textsc{sct}}\;\Conid{C}_{1}}, \ensuremath{\Gamma\;\vdash_{\textsc{sct}}\;\Conid{C}_{2}}, \ensuremath{\Gamma\vdash\Conid{C}_{1}\approx_{\Blue{\Conid{L}}}\Conid{C}_{2}}, \ensuremath{\Conid{C}_{1}\;\Downarrow_{\Conid{O}_{1}}^{\Conid{D}}\;\Conid{C}_{1}'}, and \ensuremath{\Conid{C}_{2}\;\Downarrow_{\Conid{O}_{2}}^{\Conid{D}}\;\Conid{C}_{2}'}, then \ensuremath{\Gamma\vdash\Conid{C}_{2}\approx_{\Blue{\Conid{L}}}\Conid{C}_{2}'} and \ensuremath{\Conid{O}_{1}\mathrel{=}\Conid{O}_{2}}.
\label{lemma:preservation-big-step}
\end{lemma}
\begin{proof}
By induction on the speculative big-step reductions (Fig.~\ref{fig:speculative-big-step}).
Since the reductions follow the same list of directives \ensuremath{\Conid{D}}, either
both configurations \ensuremath{\Conid{C}_{1}} and \ensuremath{\Conid{C}_{2}} reduce, i.e., rule \srule{Step}, or
have terminated, i.e., rule \srule{Done}.
The base case \srule{Done} is trivial; the lists of observation are empty, i.e., \ensuremath{\Conid{O}_{1}\mathrel{=}\epsilon\mathrel{=}\Conid{O}_{2}} and \ensuremath{\Conid{C}_{1}\mathrel{=}\Conid{C}_{1}'}, \ensuremath{\Conid{C}_{2}\mathrel{=}\Conid{C}_{2}'}, and thus \ensuremath{\Gamma\vdash\Conid{C}_{1}\approx_{\Blue{\Conid{L}}}\Conid{C}_{2}} implies \ensuremath{\Gamma\vdash\Conid{C}_{1}'\approx_{\Blue{\Conid{L}}}\Conid{C}_{2}'}.
In the inductive case \srule{Step}, both configurations perform a
small step, i.e., \ensuremath{\Conid{C}_{1}\;\stepT{\Varid{d}}{\Varid{o}_{1}}\;\Conid{C}_{1}''} and \ensuremath{\Conid{C}_{2}\;\stepT{\Varid{d}}{\Varid{o}_{2}}\;\Conid{C}_{2}''},
and a big-step, i.e., \ensuremath{\Conid{C}_{1}''\;\Downarrow_{\Conid{O}_{1}}^{\Conid{D}}\;\Conid{C}_{1}'}, and \ensuremath{\Conid{C}_{2}''\;\Downarrow_{\Conid{O}_{2}}^{\Conid{D}}\;\Conid{C}_{2}'}.
To prove that the observations \ensuremath{\Varid{o}_{1}\mathbin{:}\Conid{O}_{1}} and \ensuremath{\Varid{o}_{2}\mathbin{:}\Conid{O}_{2}} are identical
and the final configurations are \ensuremath{\Blue{\Conid{L}}}-equivalent, we first apply
\ensuremath{\Blue{\Conid{L}}}-equivalence preservation for small-steps, i.e.,
Lemma~\ref{lemma:preservation-small-step}, and deduce that \ensuremath{\Varid{o}_{1}\mathrel{=}\Varid{o}_{2}}
and that the intermediate configurations are \ensuremath{\Blue{\Conid{L}}}-equivalent, i.e.,
\ensuremath{\Gamma\vdash\Conid{C}_{1}''\approx_{\Blue{\Conid{L}}}\Conid{C}_{2}''}.
Furthermore, these configurations remain well-typed, i.e., \ensuremath{\Gamma\;\vdash_{\textsc{sct}}\;\Conid{C}_{1}''} and \ensuremath{\Gamma\;\vdash_{\textsc{sct}}\;\Conid{C}_{2}''}, by typing preservation, i.e., Lemma~\ref{lemma:tr-preservation} and Lemma~\ref{lemma:ct-preservation}.
At this point, we can apply the induction hypothesis to the big-steps \ensuremath{\Conid{C}_{1}''\;\Downarrow_{\Conid{O}_{1}}^{\Conid{D}}\;\Conid{C}_{1}'} and \ensuremath{\Conid{C}_{2}''\;\Downarrow_{\Conid{O}_{2}}^{\Conid{D}}\;\Conid{C}_{2}'}, and conclude that \ensuremath{\Conid{O}_{1}\mathrel{=}\Conid{O}_{2}}, and thus \ensuremath{\Varid{o}_{1}\mathbin{:}\Conid{O}_{1}\mathrel{=}\Varid{o}_{2}\mathbin{:}\Conid{O}_{2}}, and \ensuremath{\Gamma\vdash\Conid{C}_{1}'\approx_{\Blue{\Conid{L}}}\Conid{C}_{2}'}.

\end{proof}

\begin{thm}[Speculative Constant Time]
 For all programs \ensuremath{\Varid{c}} and security policies \ensuremath{\Blue{\Conid{L}}}, if \ensuremath{CT_{\mkern-1mu\Blue{\Conid{L}}}(\Varid{c})} and \ensuremath{\Gamma\vdash\Varid{c}},
  then \ensuremath{SCT_{\mkern-1mu\Blue{\Conid{L}}}(\Varid{c})}.
\end{thm}
\begin{proof}
  First, we expand the conclusion \ensuremath{SCT_{\mkern-1mu\Blue{\Conid{L}}}(\Varid{c})} (Definition~\ref{def:STC})
  and let \ensuremath{\Conid{C}_{\Varid{i}}\mathrel{=}\langle[\mskip1.5mu \mskip1.5mu],[\mskip1.5mu \Varid{c}\mskip1.5mu],\mu_{\Varid{i}},\rho_{\Varid{i}}\rangle} for \ensuremath{\Varid{i}\;\in\;\{\mskip1.5mu \mathrm{1},\mathrm{2}\mskip1.5mu\}}. Then, we assume \ensuremath{\Gamma\vdash\Conid{C}_{1}\approx_{\Blue{\Conid{L}}}\Conid{C}_{2}}, \ensuremath{\Conid{C}_{1}\;\Downarrow_{\Conid{O}_{1}}^{\Conid{D}}\;\Conid{C}_{1}'},
  and \ensuremath{\Conid{C}_{2}\;\Downarrow_{\Conid{O}_{2}}^{\Conid{D}}\;\Conid{C}_{2}'} and derive \ensuremath{\Conid{O}_{1}\mathrel{=}\Conid{O}_{2}} and \ensuremath{\Conid{C}_{1}'\approx_{\Blue{\Conid{L}}}\Conid{C}_{2}'} from
  Lemma~\ref{lemma:preservation-big-step}.

\end{proof}
\clearpage
\newpage
\section{Additional analysis of performace overheads}
\label{app:perf-overheads}

\begin{figure}
  \begin{subfigure}{0.49\textwidth}
  \includegraphics[width=0.99\textwidth]{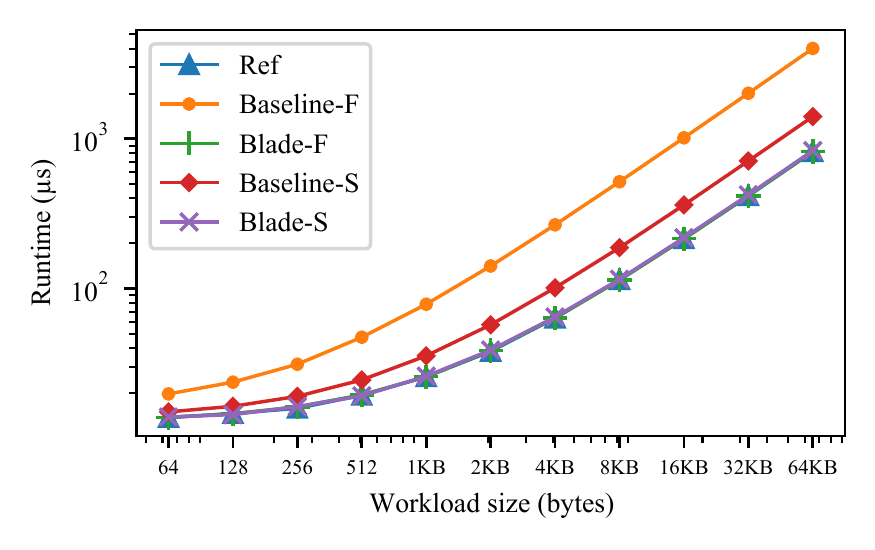}
  \caption{Without v1.1 protections}
  \label{subfig:runtime_no_v1_1}
  \end{subfigure}
  \begin{subfigure}{0.49\textwidth}
  \includegraphics[width=0.99\textwidth]{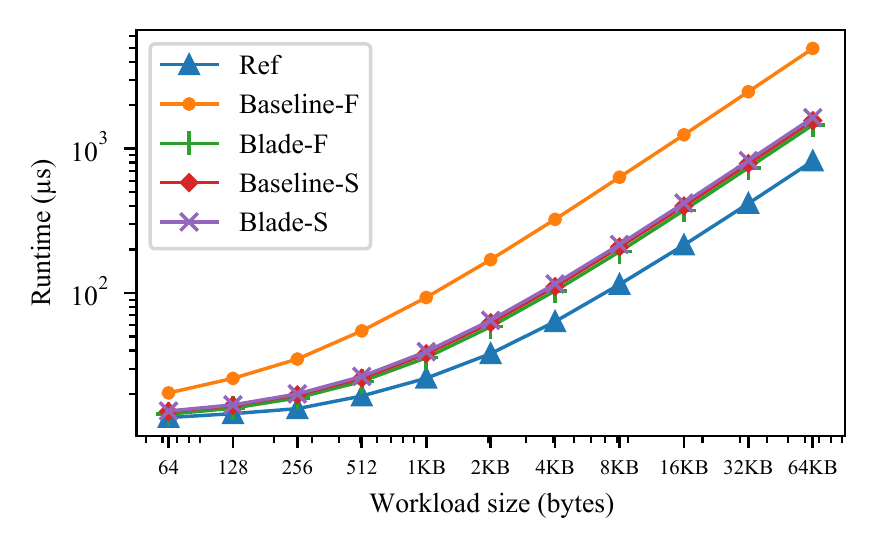}
  \caption{With v1.1 protections}
  \label{subfig:runtime_with_v1_1}
  \end{subfigure}
  \caption{
    Runtime of SHA256 (CT-Wasm) as the workload size varies.
  }
  \label{fig:sha_scaling_runtime}
\end{figure}

\begin{figure}
  \begin{subfigure}{0.49\textwidth}
  \includegraphics[width=0.99\textwidth]{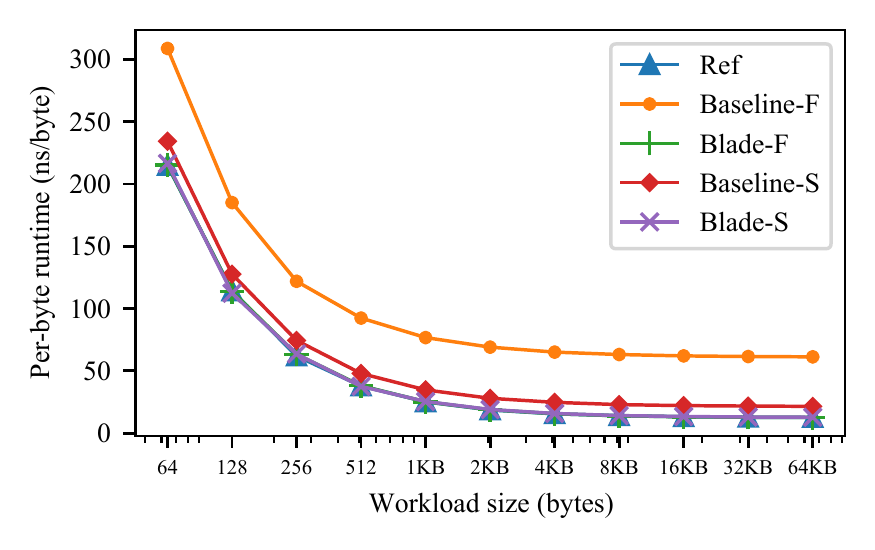}
  \caption{Without v1.1 protections}
  \label{subfig:runtime_per_byte_no_v1_1}
  \end{subfigure}
  \begin{subfigure}{0.49\textwidth}
  \includegraphics[width=0.99\textwidth]{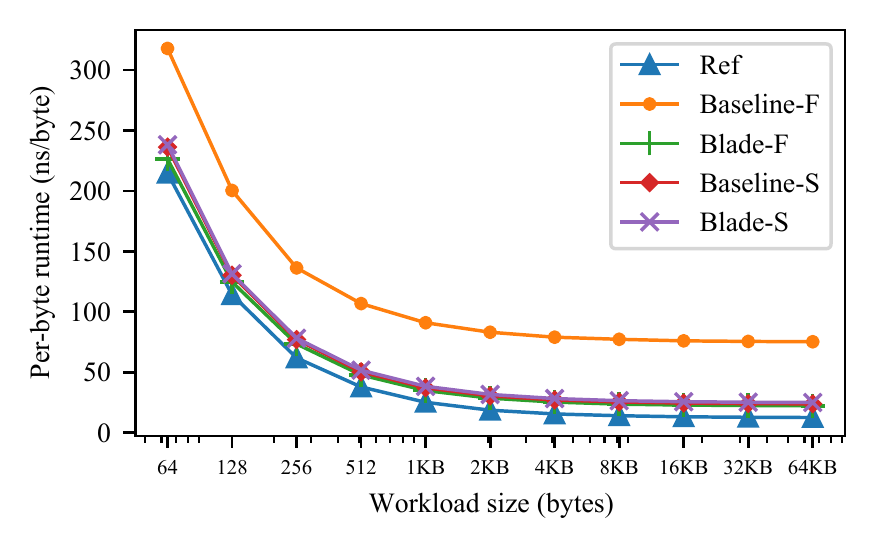}
  \caption{With v1.1 protections}
  \label{subfig:runtime_per_byte_with_v1_1}
  \end{subfigure}
  \caption{
    Runtime of SHA256 (CT-Wasm) as the workload size varies, presented on a per-byte basis.
  }
  \label{fig:sha_scaling_runtime_per_byte}
\end{figure}

\New{
\Cref{tab:evaluation} showed that the performance overhead for some
benchmarks depends heavily on the workload size.
We explore this relationship in more detail in
\Cref{fig:sha_scaling_runtime}.
Specifically, we see that for low workload sizes, the runtime is dominated by
fixed costs for sandbox setup and teardown; the execution of the Wasm code
contributes little to the performance.
As the workload size increases, the overall performance overhead
asymptotically approaches the overhead of the Wasm execution itself.
This is shown even more clearly in \Cref{fig:sha_scaling_runtime_per_byte},
where we see that the asymptotic overhead for SHA-256 with v1.1 protections
is approximately 78\% and 99\% respectively for \tool-F and \tool-S, while
the asymptotic overheads without v1.1 protections are unsurprisingly
approximately zero for \tool-F and \tool-S, as they insert no defenses.
}
\fi

\end{document}